%% file: main.tex
\documentclass[acmsmall,screen,nonacm]{acmart}

\usepackage{tikz}
\usetikzlibrary {calc,positioning,arrows.meta}
\usepackage{quantikz}    
\usepackage{url}
\usepackage{hyperref}
\usepackage{longtable}
\usepackage{booktabs}   
\usepackage{float}      
\usepackage{makecell}   
\usepackage{enumitem}
\usepackage{wrapfig}
\usepackage{slashed}    
\usepackage{cancel}     
\usepackage{listings}   

\lstdefinestyle{granthi}{
  basicstyle=\ttfamily\footnotesize,
  keywordstyle=\bfseries,
  columns=fullflexible,
  keepspaces=true,
  showstringspaces=false,
  breaklines=true,
  morekeywords={datatype,let,in,case,pair,split,fun,lolli,self}
}
\usepackage{varwidth} 

\newsavebox{\qswitchcode}
\newsavebox{\zfivecodeA}
\newsavebox{\zfivecodeB}

\usepackage{mathpartir}
\usepackage{mathtools}  
\usepackage{stmaryrd}   
\ifdefined\usepackage{unicode-math}\fi
\theoremstyle{plain}
\newtheorem{theorem}{Theorem}[section]
\newtheorem{lemma}[theorem]{Lemma}
\newtheorem{corollary}[theorem]{Corollary}
\newtheorem{proposition}[theorem]{Proposition}

\theoremstyle{definition}
\newtheorem{definition}[theorem]{Definition}

\theoremstyle{remark}
\newtheorem{remark}[theorem]{Remark}

\newcommand{\sjudge}{\vdash}
\newcommand{\intjudge}{\vdash_{\mathsf{i}}}
\newcommand{\ijudge}{\vdash_{\mathsf{inv}}}

\newcommand{\tensor}{\otimes}
\newcommand{\plus}{\oplus}
\newcommand{\expi}[2]{\expinv{#1}{#2}}

\newcommand{\ElabTy}[1]{\langle\!\langle #1 \rangle\!\rangle}   
\newcommand{\LayoutTy}[1]{\lfloor #1 \rfloor}
\newcommand{\sem}[1]{\llbracket #1 \rrbracket}

\newcommand{\SEM}[1]{\llbracket #1 \rrbracket}  
\newcommand{\lay}{\mathcal{L}}                         

\newcommand{\nwires}{\mathsf{nwires}}

\renewcommand{\ket}[1]{\lvert #1 \rangle}

\newcommand{\QBool}{\mathsf{QBool}}

\newcommand{\Framed}[1]{\mathsf{Framed}(#1)}

\newcommand{\qSwitch}{\mathsf{QSwitch}}

\newcommand{\id}{\mathsf{id}}

\newcommand{\lam}[2]{\lambda #1.\,#2}

\newcommand{\bigplus}{\mathop{\bigoplus}}

\newcommand{\NAME}{\mathsf{D}}
\newcommand{\caseofn}{\mathsf{case}_n}
\newcommand{\Wit}{\mathit{Wit}}
\newcommand{\Aux}{\mathit{Aux}}
\newcommand{\Case}{\mathsf{Case}}

\newcommand{\oplusmap}[4]{[\,#1\cdot #2 \mid #3\cdot #4\,]}

\newcommand{\caseof}[5]{\mathbf{case}\; #1\;\mathbf{of}\; #2 \Rightarrow #3 \mid #4 \Rightarrow #5}

\newcommand{\unitjudge}{\vdash_{\mathsf{unit}}}

\DeclareRobustCommand{\base}{\slashed{b}}      
\DeclareRobustCommand{\Qn}[1]{\base^{\oplus #1}} 
\newcommand{\expinv}[2]{\mathsf{exp}(i#1\cdot #2)} 

\newcommand{\size}[1]{\mathsf{n}(#1)}

\newcommand{\letpair}[4]{%
  \mathbf{let}\,(#1 \tensor #2)=#3\,\mathbf{in}\,#4%
}

\newcommand{\distL}{\mathsf{dist}_L}

\newcommand{\distR}{\mathsf{dist}_R}
\newcommand{\distRi}{\mathsf{dist}_R^{-1}}

\newcommand{\lmark}{\multimap}

\newif\ifappendices
\appendicestrue          

\setcopyright{cc}
\setcctype{by}
\acmDOI{10.1145/3839496}
\acmYear{2026}
\acmJournal{PACMPL}
\acmVolume{10}
\acmNumber{OOPSLA2}
\acmArticle{364}
\acmMonth{10}
\acmSubmissionID{oopslab26main-p876-p}
\received{2026-03-17}
\received[accepted]{2026-08-06}

\begin{document}

\title{Granthi: Higher-Order Quantum Programming via Unitary Wiring}

\author{Samson Abramsky}
\orcid{0000-0003-3921-6637}
\affiliation{%
  \institution{University College London}
  \city{London}
  \country{United Kingdom}
}
\email{s.abramsky@ucl.ac.uk}

\author{Radha Jagadeesan}
\orcid{0000-0002-4525-1589}
\affiliation{%
  \institution{DePaul University}
  \city{Chicago}
  \country{USA}
}
\email{rjagadee@depaul.edu}

\begin{abstract}
Many mainstream quantum programming languages confine higher-order structure to a
classical host while restricting the quantum layer to first-order
operations on qubits.  This paper presents Granthi, a purely unitary
higher-order quantum programming language built on three design
commitments: quantum programs are first-class values that may be passed,
returned, and coherently composed; additive structure is tag-preserving
routing rather than observational branching, so control may remain in
superposition; and programmer-facing finite label types with staged
reversible-operation bindings provide domain-level control spaces without
exposing tag management.  These bindings are eliminated by elaboration
before Source typing.

Granthi deterministically normalizes each Source program to a canonical
wiring form.  Every well-typed Source program---including a term of
function type---has a unitary boundary interpretation.  Under backend
correctness~\textup{(BC)}, the reference compiler produces a unitary
circuit realizing that interpretation.

Granthi's currently supported executable fragment is implemented
end-to-end: an OCaml
DSL elaborates surface programs through a higher-order Core IR to
executable quantum circuits via pytket.  The language directly supports the pure-unitary quantum switch for
explicitly supplied operations; closed instances compile to static
circuits.  It also supports interference on control-flow history and
structured finite control, all within the purely unitary fragment.
\end{abstract}



\maketitle

\noindent\textit{Note.}
This is the full version, including the appendices, of the article
published in \emph{Proceedings of the ACM on Programming Languages}
10, OOPSLA2, Article 364 (October 2026),
\url{https://doi.org/10.1145/3839496}.  The appendices contain the
proofs and are the supplementary material of the published version.
\par\smallskip

\section{Introduction}
\input{intro}

\input{WHYGOI}

\input{design}

\input{intro_object_language}
\input{core-lang}

\input{metatheory}
\input{SEM-NEW}
\input{compiler_new}
\input{datatypes-new}
\input{toolchain}

\input{expressiveness}

\section{Conclusion}

Granthi shows that higher-order quantum abstraction need not remain
confined to a classical circuit-building layer: under a disciplined
unitary semantics, it can compile directly to static circuits while
preserving coherent control and interference.  This provides a
concrete foundation for embedding higher-order unitary kernels within
larger adaptive quantum languages, where measurement, mixed states,
and recursion remain the next challenge.

\section*{Data-Availability Statement}

A prototype implementation accompanies this paper.  It provides the
OCaml-embedded source language, the higher-order Core IR, and the
compiler to executable \texttt{pytket} circuits described in
\S\ref{sec:toolchain}.  The artifact includes regression and matrix
tests together with the runnable demonstrations listed in
Table~\ref{tab:validation}.  Its current backend scope and size
limits are documented under Limitations in \S\ref{subsec:impl-limits};
those bounds do not constrain the formal language or the reference
compiler.

The evaluated artifact is archived on
Zenodo~\citep{GranthiArtifact2026}.  The current implementation is
maintained under the MIT license at
\begin{center}
\url{https://github.com/radhajagadeesan/Granthi}
\end{center}
which includes enforcement of the Source typing restrictions and the
regression tests for the counterexample term of
Remark~\ref{rem:why-first-order-witness}.  Appendices~\ref{app:focused-rules}--\ref{app:elaboration-proofs},
containing the proofs, are included in this full version; they are the
supplementary material of the published article.

\bibliographystyle{ACM-Reference-Format}
\bibliography{main_published_preferred}

\ifappendices
\clearpage
\appendix
\input{app-focused-rules}

\input{appendix-summand-completeness}
\input{appendix-normal}
\input{appendix-determinacy}

\input{app-unitarityNEW}
\input{appendix-sum-encoding}
\input{appendix-compilation-soundness-planned}
\input{narymonoidal}
\input{elaboration-proofs-appendix}
\fi

\end{document}

%% file: intro.tex

Quantum algorithms were developed well before suitable hardware existed, from
early breakthroughs such as Shor's factoring algorithm to more recent
variational and sampling-based methods~\citep{Peruzzo2014VQE,Aaronson2011BosonSampling}.
Their study also raises a programming-language question independent
of any particular device technology: how should quantum programs be
structured?  As in classical
computing~\citep{HennessyPatterson2018}, abstraction, architecture,
and programming models are central to that question.
Many mainstream quantum programming languages adopt a split
architecture: higher-order structure resides entirely in a classical
host language, while the quantum layer exposes only first-order
operations on qubits and registers~\citep{Gay2006Survey}.
Classical languages provide functions, control abstractions, and program
generators used to construct and manipulate quantum circuits as
classical data.  Callable operations and circuit values may be passed,
returned, and composed by the host, but coherent quantum programs
themselves do not inhabit higher types of the quantum object language.
The resulting separation yields a coherent and widely adopted
programming model, but one in which abstraction over coherent quantum
programs themselves is limited by design.

This common confinement of higher-order structure to the classical host is not
mandated by quantum mechanics itself. Quantum theory admits coherent control
over program composition, as demonstrated by the quantum
switch, in which the order of composition of two operations is placed in
superposition. Such examples show that higher-order structure is
physically meaningful within the purely unitary model, and that
first-order treatments of the quantum layer reflect design choices
rather than fundamental limitations.

A second limitation concerns abstraction over data. In classical
programming, user-defined data abstractions let programmers work
with problem-domain objects rather than low-level
representations. By contrast, quantum programming languages largely
expose physical resources---qubits, registers, and low-level operations---as
their primary data model, with little support for domain-level organization
compatible with quantum coherence.

Many familiar classical elimination mechanisms permit destructive
use of values through inspection, case analysis, or iteration,
collapsing control flow and discarding unused alternatives.  Their
unrestricted transfer to a purely unitary quantum setting is
incompatible with coherence: quantum data cannot be freely
inspected, copied, or discarded without irreversibly altering
program behavior. Even in the absence of measurement, these operations
conflict with unitarity by introducing implicit erasure or duplication.
Familiar datatype mechanisms such as pattern matching or
equality testing do not carry over directly to unitary quantum programs.

\medskip\noindent
This paper presents Granthi\footnote{From the Sanskrit for `knot,' evoking the wiring at the core of the system.}, a purely unitary higher-order quantum programming language built on three commitments:
quantum programs are first-class values that may be passed, returned, and coherently composed;
additive structure is tag-preserving routing rather than observational branching, so control may remain in superposition;
and programmer-facing finite label types with staged reversible-operation bindings provide domain-level control spaces without exposing tag management: elaboration substitutes those bindings and maps the datatypes to canonical $n$-ary sums, with no recursive types.

Granthi deterministically normalizes Source programs to canonical wiring
forms, which the boundary semantics interprets.  Every well-typed Source
program---including a term of function type---acts unitarily at its
interface.  Under backend correctness~\textup{(BC)}, the reference
compiler realizes this interface action as a unitary circuit.
An OCaml DSL compiles the currently supported fragment to
executable quantum circuits via pytket.

\paragraph{Contributions.}
This paper makes four contributions:
\begin{itemize}[nosep,leftmargin=*]
  \item \textbf{Language design.} A purely unitary linear language in
    which open quantum components are first-class values and coherent sums
    provide tag-preserving control.
  \item \textbf{Canonicalization and semantics.} A deterministic
    normalization procedure and a boundary interpretation for open and
    higher-order terms.
  \item \textbf{Circuit realization.} A type-directed compiler for
    canonical wiring forms and physical layouts for tags and payloads,
    together with a realization theorem under the stated backend assumption.
  \item \textbf{Implementation and validation.} An OCaml-to-\texttt{pytket}
    implementation supporting controlled lifting, $\qSwitch$,
    phase-sensitive routing, and finite control datatypes, tested by
    regression checks and matrix comparisons.
\end{itemize}

%% file: WHYGOI.tex
\section{Higher-Order Reversibility and Quantum Wiring}
\label{sec:ho-wiring}

In the purely unitary fragment, first-order reversibility is well
understood: the interface can always be enlarged so that no information
is lost~\citep{Bennett1973,Toffoli1980}.
Granthi extends this idea to higher order by treating programs
themselves as open unitary components with explicit interfaces.
The central claim is that higher-order reversibility is not a new
kind of unitarity, but ordinary unitarity relocated from closed
data to exposed interfaces.

\paragraph{The higher-order challenge.}

A term $f:A\lmark B$ is not generally interpreted as an invertible map
from the data space of $A$ to that of $B$.  Granthi instead places the
reversible invariant on the complete exposed boundary: in
$\mathbf{FdHilb}$, its boundary denotation is a unitary between the
negative and positive sectors selected by its canonical normal
derivation.  Thus no information is globally erased~\citep{HEUNEN2015217}.
The main obstacle to reversible higher-order computation is not
higher-order functions themselves but the usual operational
interpretation of application.
In a standard operational account, evaluating $f\,a$ consumes both
the function and its argument to produce a result; the interfaces
through which $f$ and $a$ interact are left implicit.
From the perspective of reversible systems, this hides the information
flow that must be preserved.

\paragraph{Application as connection.}

In the canonical-form semantics, an application that survives
normalization is interpreted as \emph{connection}, not consumption.
A function is
not a black box that eats an argument; it is a component with ports.
Figure~\ref{fig:boundary-wiring} shows the proof net for
$\mathsf{eval}$ and makes the connection topology
explicit: the arcs identify $\alpha_2$ with $\alpha_1$ and $\beta_1$ with
$\beta_2$, wiring the argument's output port to the function's
input port and exposing the function's output as the overall
result.
An irreducible application is a connectivity structure;
$\beta$-redexes are discharged by the fixed normalizer before semantic
interpretation.  Currying is the corresponding rewiring in the other direction: abstracting
over an argument
re-exports its port on the output boundary, so that the
caller supplies it from outside rather than consuming it internally (a perspective with precedent in functional languages~\citep{Wadsworth1971,Turner1979}).

\begin{figure}[t]
\centering
\begin{tikzpicture}[scale=0.82, transform shape,
  wire/.style={line width=0.9pt},
  box/.style={draw, rounded corners=2pt, line width=0.9pt,
              minimum width=1.0cm, minimum height=1.0cm}]
  \begin{scope}[xshift=-4.4cm]
    \node[box] (t) at (0,0) {$t$};
    \draw[wire] (t.east) -- ++(0.4,0) node[right] {$B$};
    \draw[wire] (0.9,0.35) node[right] {$A$}
      -- (0.75,0.35)
      arc[start angle=0, end angle=180, x radius=0.85cm, y radius=0.45cm]
      -- ($(t.west)+(0,0.35)$);
    \node[font=\scriptsize, anchor=north] at (0,-0.85) {Currying};
  \end{scope}
  \begin{scope}
    \node[box] (f) at (0,0) {$f$};
    \draw[wire] (f.west) -- ++(-0.8,0) node[left] {$A$};
    \draw[wire] (f.east) -- ++(0.4,0) node[right] {$B$};
    \node[font=\scriptsize, anchor=north] at (0,-0.85) {Function interface};
  \end{scope}
  \begin{scope}[xshift=4.4cm]
    \node[box] (ev) at (0,0) {\small$\mathsf{eval}$};
    \draw[wire] (-0.9, 0.25) node[left] {$\alpha_2$}
      -- (0.9, 0.25) node[right] {$\alpha_1$};
    \draw[wire] (-0.9,-0.25) node[left] {$\beta_1$}
      -- (0.9,-0.25) node[right] {$\beta_2$};
    \node[font=\scriptsize, anchor=north] at (0,-0.85) {Eval};
  \end{scope}
\end{tikzpicture}
\caption{Boundary wiring for currying, a function interface, and evaluation.}
\label{fig:boundary-wiring}
\end{figure}
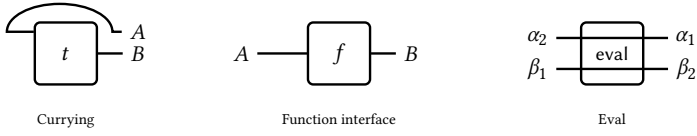

A term of type $A \lmark B$ is an open circuit fragment with a
structured interface: it consumes an $A$-bundle and produces a $B$-bundle.
Application connects interfaces rather than ``running'' a function.
Closed programs ($A \lmark A$) have matching interfaces; at first
order they execute as standalone input--output circuits.

\paragraph{Boundary unitarity.}

For closed first-order programs, the familiar correctness invariant is
ordinary unitarity: a closed $p : A \lmark A$ denotes a unitary on the
data interface of~$A$.  At higher types this is ill-posed: a term of type
$A \lmark B$ is an open component, and once functions appear at interfaces
we cannot ask for a unitary $A \to B$ ``on data''.

The solution is to interpret a component's exposed ports as
finite-dimensional Hilbert spaces and require the induced map to be unitary.
A canonical normal derivation selects the compatible input and output spaces
within that typed boundary.  For $\mathsf{eval}$, this map is simply the two
identity connections in Figure~\ref{fig:boundary-wiring}.  At
qubit-register endpoints, execution recovers ordinary circuit
unitarity;
\S\ref{sec:boundary-semantics} formalizes the higher-order boundary map, and
\S\ref{sec:compilation} proves that compilation realizes it on its typed
physical layout under backend correctness~\textup{(BC)}.

This viewpoint accommodates coherent control over composition
order, including the quantum switch~\citep{Chiribella2013Switch}
and interference on control flow; Granthi is designed
to make the wiring view explicit in syntax and typing.

%% file: design.tex
\section{Scope: A Canonical-Form Unitary Core}
\label{sec:design}

Granthi is a purely unitary core; measurement, mixed states,
classical feedback, and recursion are all absent.  A full quantum
programming language needs each of these, but each brings
equational laws under which unitary equality is no longer the
governing invariant.  Restricting to the unitary layer lets
higher-order abstraction be studied with a circuit-level unitary
interpretation.

The intended deployment is inside a larger adaptive quantum
language, one that prepares inputs, invokes a unitary kernel,
measures, branches on the outcome, prepares fresh state, and calls
the next kernel.  Measurement, classical control, and adaptive
scheduling belong to that surrounding language; what Granthi
governs is the kernel.  Inside a kernel, a programmer may inline,
abstract, apply, and normalize under Granthi's fixed discipline; the
metatheory proves that these transformations preserve the component's
boundary behavior.

The move is the standard PL response to effects, transposed to the
boundary between an adaptive host and a unitary core.  With
classical effects, unrestricted $\beta$-reduction becomes sound
only once an evaluation discipline, say call-by-value, fixes the
operational order; Granthi fixes the analogous discipline for
higher-order unitary kernels by normalizing the linear functional
structure first, holding quantum operations as opaque constants,
and exposing only afterwards the boundary on which those constants
must act unitarily.  On this reading, the source equations determine
the canonical wiring presented to the quantum operations, allowing a
kernel to sit cleanly under the measurement and classical control of
its host.

A function value in Granthi is thus a component with a typed
boundary interface, and the interface may carry phases, coherent
sums, coherent control, and unitary action across its ports.
Source evaluation rearranges the interface under the fixed
normalization discipline.  Normalization can therefore change the
higher-order presentation of a component while preserving its
boundary behavior, in the sense made precise by the boundary
semantics (\S\ref{sec:boundary-semantics}).  We expect the
interface-preserving discipline, more than any particular surface
syntax, to be the transferable design: higher-order functional
code as a source language for the unitary components of a larger
quantum system.

Concretely, conjugation and coherent reordering, including the quantum
switch, are higher-order Source combinators.  Controlled lifting and
finite-family dispatch over fixed closed operations are staging
constructions that produce sealed Source programs.  These
patterns recur throughout the algorithms literature---in phase
estimation and qubitization, in the basis changes and reflections of
amplitude amplification, and in SELECT/multiplexed oracles for LCU
and Hamiltonian simulation.  Both the Source combinators and the
staging constructions are reusable; the Source type system enforces
the interface discipline, and each certified exponential denotes its
prescribed first-order unitary.

%% file: intro_object_language.tex
\section{An Informal Tour of the Object Language}
\label{sec:lang-tour}

This section introduces the language by example, to make terms
readable \emph{as programs} before the formal development in
\S\ref{sec:core-language}.  The language treats higher-order programs as
first-class unitary components: types describe interfaces, and terms
describe reversible transformations between interfaces.

The source language is a linear $\lambda$-calculus with tensor ($\tensor$),
coherent sum ($\plus$), and linear function space ($\lmark$). The base
type $\base$ is a one-dimensional seed; finite-dimensional quantum
systems are built as coherent sums of $\base$ (e.g.\
$\QBool := \base \plus \base$). The Source language is linear: every
variable must be used exactly once, so copying, discarding, and
projection are unavailable.
We write $\Gamma\sjudge t:A$ when $t:A$ uses precisely the linear
variables in $\Gamma$; a type is first-order when it contains no
$\lmark$.
Programs are built from linear $\lambda$-terms, structural combinators,
and certified quantum exponentials, each interpreted by its prescribed
unitary.  We begin with programs that only rearrange data.

The boxed OCaml listings below are literal programs accepted by the
PPX (OCaml preprocessor-extension) Source frontend, with module
imports omitted.

\subsection{Basic Examples}

\paragraph{Pairing without projection.}
Given values of types $A$ and $B$, we can form a pair. The destructor
binds both components simultaneously:
$
  \mathsf{swap} : A \tensor B \lmark B \tensor A:= \lam{p}{\letpair{x}{y}{p}{y \tensor x}}
$.
There is no first or second projection: both components must be used
exactly once.

\paragraph{Coherent Choice}
\label{subsec:coherent-choice}

The simplest control space is the quantum boolean
$\QBool := \base \plus \base$, where $\base$ is the base type carrying no payload.  Although written as a sum, $\plus$ does
\emph{not} support classical, observational branching: a $\QBool$ may be in
superposition, so programs cannot test it and then discard the evidence of
which branch was taken.

Instead, the surface language permits (\S\ref{cohRouting}) a coherent, \emph{tag-preserving} case expression
\[
  \caseof{e}{x}{u}{y}{v}
  \;:\; (A \plus B) \tensor C,
\]
in which the two branches $u$ and $v$ share one linear context, with
$A$, $B$, and $C$ first-order.
This $\mathbf{case}$ construct
should be read as \emph{routing}, not inspection: it runs one of two
branches \emph{without measuring}, and it preserves the choice (result type records $A \plus B$) so the
computation remains reversible.  In the surface examples
below, branch labels such as \texttt{zero}/\texttt{one} for $\QBool$
are pattern binders carried through the case, not freely constructible
constants---the language has no injection constants for sum types.

This is close in spirit to Qunity's use of sum types for coherent
control~\citep{VoichickLiRandHicks2023}.  Granthi pushes the discipline
into the primitive sum interface itself: the orthogonality needed for
coherent branching is carried by the type, because branch alternatives
live in disjoint summands and tag preservation keeps the additive
structure visible on output.  Granthi's
tag preservation echoes the symmetric pattern clauses
of~\citet{SabryValironVizzotto2018}: reversibility comes from
respecting the additive structure on both sides.

\subsection{Higher-Order Quantum Control}
This section develops two forms of coherent control: composition order
for first-class programs, and interference on control-flow
history.

\paragraph{The Quantum Switch: coherent control over composition order.}
\label{subsec:quantum-switch}

The quantum switch is the flagship example of coherent control over
composition order.  Granthi expresses its pure-unitary form as a
higher-order term whose coherent branches share a linear function
context: each branch uses each supplied operation exactly once.
Let $f,g : A \lmark A$ be two programs, with $A$ first-order.  The
quantum switch
$\qSwitch : (A \lmark A) \tensor (A \lmark A) \tensor
(\QBool \tensor A) \lmark (\QBool \tensor A)$
applies them in a control-dependent order.
In the listing below, the host-language witness \texttt{(a : 'a P.t)} records that $A$ is
first-order.  Canonical uncurrying gives the displayed tensor type.
On the left branch, $\qSwitch$ applies $g$ then $f$; on the right,
$f$ then $g$.  The type
$(\QBool\tensor A)\lmark(\QBool\tensor A)$ returns both the control
qubit and the payload.  Externally preparing the control $b$ in
$\ket{+}=H\ket{0}$, with $H:\QBool\lmark\QBool$, yields coherent
control of the two composition orders~\citep{Chiribella2013Switch}.
$\qSwitch$ is the paper's running example.  In the mathematical
notation and the PPX Source notation used by the artifact:

\begin{lrbox}{\qswitchcode}%
\scriptsize
\begin{varwidth}{0.53\textwidth}
\begin{lstlisting}
let%source qswitch (a : 'a P.t)
    (f : ('a, 'a) lolli)
    (g : ('a, 'a) lolli)
    (p : (qbool, 'a) tensor) =
  let (b, x) = split p in
  case b
    ~zero:(f (g x))
    ~one_:(g (f x))
\end{lstlisting}
\end{varwidth}%
\end{lrbox}

\begin{center}
\small
\fbox{%
\begin{tabular}{@{\hspace{6pt}}m{0.31\textwidth}@{\hspace{6pt}}|@{\hspace{6pt}}m{0.52\textwidth}@{\hspace{6pt}}}
$\begin{array}{@{}l@{}}
\qSwitch(f, g, p) = \\
\quad \mathbf{let}\ (b, x) = p\ \mathbf{in} \\
\quad \mathbf{case}\ b\ \mathbf{of} \\
\quad\quad b_0 \Rightarrow f(g(x)) \\
\quad\quad b_1 \Rightarrow g(f(x))
\end{array}$
&
\usebox{\qswitchcode}
\end{tabular}}
\end{center}

\paragraph{Interference on Control Flow: a distinctly quantum effect.}
\label{subsec:interference-preview}

Beyond controlling the \emph{order} of operations, the language supports
\emph{interference on control history}---a phenomenon with no classical
analogue.  The pattern is to route a witness, phase one route, undo the
routing, and expose the relative phase by interference.  For example,
classical short-circuit evaluation on $b_1,b_2:\QBool$ erases whether
the second input was skipped.  A Source program instead records that
information by routing a witness $w:\Wit$, where $\Wit=\Aux\plus\QBool$
for an auxiliary first-order type~$\Aux$, according to $b_1$, while
threading $b_2$ through unchanged (linearity forbids discarding it).
A phase-marked variant applies $-\mathrm{id}_{\Wit}$ to $w$ in the
$b_1=0$ branch, marking which route fired without observation; an
unphased inverse route then restores the witness while leaving the
relative phase on $b_1$, and a Hadamard converts that phase into a
computational-basis distinction.  Control-flow history thus participates
in interference without collapsing the computation
(\S\ref{sec:quantum-core}).

\subsection{Finite Datatypes and Structured Control}
\label{subsec:finite-datatypes}

This subsection previews a second language-level construct:
programmer-facing finite control datatypes with named labels and
reversible operations, elaborated into the core's sum-and-wiring
machinery in \S\ref{sec:datatypes}.
We illustrate with the cyclic group $\mathbb Z_5$, represented
mathematically by a five-way first-order sum:
\[
  \mathbb Z_5 := \bigplus_{i=0}^{4}\base,
  \qquad
  \mathsf{neg}(a)=-a \pmod 5,
  \qquad
  \mathsf{add}(a,b)=(a,\,a+b \pmod 5).
\]
The PPX declaration
introduces its ordered labels, while the generated \texttt{permute}
and \texttt{select} combinators construct closed Source operations:

\begin{lrbox}{\zfivecodeA}%
\scriptsize
\begin{varwidth}{0.45\linewidth}
\begin{lstlisting}
type z5 = Z0 | Z1 | Z2 | Z3 | Z4
[@@source.datatype]

let shift1 =
  Z5.permute [ Z1; Z2; Z3; Z4; Z0 ]
let shift2 =
  Z5.permute [ Z2; Z3; Z4; Z0; Z1 ]
let shift3 =
  Z5.permute [ Z3; Z4; Z0; Z1; Z2 ]
let shift4 =
  Z5.permute [ Z4; Z0; Z1; Z2; Z3 ]
\end{lstlisting}
\end{varwidth}%
\end{lrbox}
\begin{lrbox}{\zfivecodeB}%
\scriptsize
\begin{varwidth}{0.5\linewidth}
\begin{lstlisting}
let neg_op =
  Z5.permute [ Z0; Z4; Z3; Z2; Z1 ]
let add_op =
  Z5.select ~target:Z5.p
    [ Op.id Z5.s; shift1; shift2;
      shift3; shift4 ]

let%source neg (x : Z5.t) =
  neg_op x
let%source add (p : (Z5.t, Z5.t) tensor) =
  add_op p
\end{lstlisting}
\end{varwidth}%
\end{lrbox}

\begin{center}
\small
\fbox{%
\begin{tabular}{@{\hspace{6pt}}m{0.36\textwidth}@{\hspace{6pt}}|@{\hspace{6pt}}m{0.46\textwidth}@{\hspace{6pt}}}
\usebox{\zfivecodeA}
&
\usebox{\zfivecodeB}
\end{tabular}}
\end{center}
Thus
$\mathsf{neg}:\mathbb Z_5\lmark\mathbb Z_5$ and
$\mathsf{add}:\mathbb Z_5\tensor\mathbb Z_5
 \lmark\mathbb Z_5\tensor\mathbb Z_5$.
The permutation lists are bijections.  The call to
\texttt{Z5.select} fixes five closed Source operations before
\texttt{add} is formed; it preserves the nominal tag and applies the
selected operation to the one shared payload.  More generally, closed
operations
$\emptyset\sjudge f_i:A\lmark A$ give the staged program
\[
 \mathsf{select}_{5,A}[f_0,\ldots,f_4]:
 \mathbb Z_5\tensor A\lmark\mathbb Z_5\tensor A ,
\]
for first-order $A$.  The brackets denote staging arguments, not
linear inputs of the resulting Source program.

\noindent

The controlled phase kick is an instance of \emph{group-coherent
pattern matching}:
\[
  \mathsf{kick}_5 : \mathbb{Z}_5 \tensor \QBool \lmark
                    \mathbb{Z}_5 \tensor \QBool .
\]
It applies the rotation $\mathrm{Rz}(2\pi k/5)=e^{-i\pi kZ/5}$ to the target qubit
under control of the group element~$k$ and returns $k$ unchanged:
\[
  |k\rangle \tensor |\psi\rangle \;\mapsto\;
  |k\rangle \tensor \mathrm{Rz}(2\pi k/5)\,|\psi\rangle.
\]
The phase depends on $k$, but $k$ itself may remain in superposition.
The programmer writes group labels and operations;
\S\ref{sec:datatypes} makes their elaboration precise.

%% file: core-lang.tex
\section{Core Language}
\label{sec:core-language}

This section develops the typed core of Granthi.  The core is a linear
$\lambda$-calculus that is already expressive for higher-order
resource-sensitive programming; the quantum extension enriches it
modularly with certified involutions and unitary primitives.
Conjugation by structural type isomorphisms moves unitaries along
type isomorphisms by ordinary composition, with no separate
primitive.
The type system supports both the boundary semantics
(\S\ref{sec:boundary-semantics}) and circuit compilation
(\S\ref{sec:compilation}).

\subsection{Linear Core}
\label{subsec:linear-core}

The linear $\lambda$-calculus forms the foundation of our language. The type
system enforces that every variable is used exactly once; the term language
provides higher-order functions, pairs, and coherent sums, using a disciplined form
of branching suited to reversible linear control.

\subsubsection{Types}
\label{subsubsec:types}

The grammar of Raw types is:
\[
T, U \;::=\; \base \;\mid\; T \tensor U \;\mid\; T \plus U \;\mid\; T \lmark U
\]
(base, tensor, sum, and linear implication respectively).

The language includes a primitive base type $\base$, which is the
seed object of the type theory carrying no payload.
No type has term-level value constructors such as
$\mathsf{true}$ or $\mathsf{false}$: this language constructs unitaries,
not data.  Accordingly, basis-state preparation is not a term former.

We use the notation
$\Qn{n} := \bigoplus_{i=0}^{n-1}\base$ for the fixed
left-associated $n$-fold sum; in particular
$\Qn{1} = \base$ and $\Qn{2} = \QBool$.
$\QBool := \base \plus \base$ is the qubit type and serves as the standard
base for the user-facing type theory.

Tensor~($\tensor$) reflects pairing without projection:
both components must be consumed. Sum~($\plus$) is a \emph{coherent sum}:
it carries no injections $A \to A \plus B$ and no projections, because
injections and projections are irreversible and incompatible with unitarity.
Instead, $\plus$ routes values coherently
through one of two branches while preserving, at the type level, which
branch was taken.
Linear implication~($\lmark$) is the type of functions that consume their
argument exactly once.

\subsubsection{Raw Language}
\label{subsubsec:raw-language}
\label{subsubsec:terms}

The Raw typing judgment $\Gamma\vdash_{\mathsf r}t:A$ classifies
linear terms.  A context is a finite multiset of distinctly named
typed variables; comma denotes disjoint multiset union, and exchange
is implicit.  Linearity is enforced by context splitting
($\Gamma=\Gamma_1\uplus\Gamma_2$) in multiplicative rules.
Table~\ref{tab:typing-rules-linear} gives the general
natural-deduction calculus.  Appendix~\ref{app:focused-rules}
defines the internal administrative judgment used by normalization
and the semantic proofs.  The Source language below selects the
programmer-facing fragment.

\begin{definition}[Raw formation values]
\label{def:raw-formation-values}
The Raw formation values are generated by
\[
\begin{split}
 W ::= {}&x \mid W\tensor W \mid \lam{x}{t}
 \mid [\,W\mid W\,]
 \mid s \mid x\,W \mid W\plus W .
\end{split}
\]
Here $s$ ranges over
\[
\alpha^\plus \mid (\alpha^\plus)^{-1}
\mid \sigma^\plus \mid (\sigma^\plus)^{-1}
\mid \mathsf{dist}_L \mid \mathsf{dist}_L^{-1}
\mid \mathsf{dist}_R \mid \mathsf{dist}_R^{-1},
\]
at their declared type indices.
Brackets form Raw coherent sums; infix $W\plus W$ is the Raw
branchwise map former.
\end{definition}

Type indices on structural atoms are implicit; $\eta$ is reserved
for $\eta$-expansion.  Formation values need not be irreducible: in
particular, a lambda body $t$ is arbitrary.  We reserve $V$ for the
normal-form value grammar of Definition~\ref{def:nf-grammar}.  The
value restriction keeps Raw sum formation stable under the
let-floating normalizer; the internal calculus of
Appendix~\ref{app:focused-rules} lifts it.

\begin{table*}[!ht]
\caption{Raw typing rules for the linear core.  All binary rules have an
implicit context split $\Gamma = \Gamma_1 \uplus \Gamma_2$.
\textsc{$\plus$-Map} shows the phases-suppressed ($\alpha{=}\beta{=}1$) case.}
\label{tab:typing-rules-linear}
\centering
\footnotesize
\resizebox{\linewidth}{!}{%
\begin{tabular}{@{}l@{\;}c@{\;}l@{\;}c@{}}
\toprule
\multicolumn{4}{@{}l}{\emph{Raw Linear Core} ($\Gamma\vdash_{\mathsf r}t:A$)} \\
\midrule

\textsc{Var}
& $\inferrule{ }{x:A \vdash_{\mathsf r} x : A}$
& &
\\[2.5ex]

\textsc{$\lmark$-I}
& $\inferrule{
    \Gamma, x:A \vdash_{\mathsf r} t : B
  }{
    \Gamma \vdash_{\mathsf r} \lam{x}{t} : A \lmark B
  }$
& \textsc{$\lmark$-E}
& $\inferrule{
    \Gamma_1 \vdash_{\mathsf r} f : A \lmark B \\
    \Gamma_2 \vdash_{\mathsf r} u : A
  }{
    \Gamma \vdash_{\mathsf r} f\,u : B
  }$
\\[3ex]

\textsc{$\tensor$-I}
& $\inferrule{
    \Gamma_1 \vdash_{\mathsf r} t : A \\
    \Gamma_2 \vdash_{\mathsf r} u : B
  }{
    \Gamma \vdash_{\mathsf r} t \tensor u : A \tensor B
  }$
& \textsc{$\tensor$-E}
& $\inferrule{
    \Gamma_1 \vdash_{\mathsf r} t : A \tensor B \\
    \Gamma_2, x:A, y:B \vdash_{\mathsf r} u : C
  }{
    \Gamma \vdash_{\mathsf r} \letpair{x}{y}{t}{u} : C
  }$
\\[3ex]

\textsc{$\plus$-I}
& $\inferrule{
    \Gamma_1 \vdash_{\mathsf r} W_1 : A \\
    \Gamma_2 \vdash_{\mathsf r} W_2 : B
  }{
    \Gamma \vdash_{\mathsf r} [\,W_1 \mid W_2\,] : A \plus B
  }$
& \textsc{$\plus$-Map}
& $\inferrule{
    \Gamma_1 \vdash_{\mathsf r} f : A \lmark C \\
    \Gamma_2 \vdash_{\mathsf r} g : B \lmark D
  }{
    \Gamma \vdash_{\mathsf r} f \plus g : A \plus B \lmark C \plus D
  }$
\\[2ex]

\bottomrule
\end{tabular}%
}

\vspace{-1em}
\end{table*}

\noindent

The multiplicative rules ($\lmark$, $\tensor$) are standard for a
linear $\lambda$-calculus.  Raw \textsc{$\plus$-I} forms a tagged
branch pair $[\,W_1\mid W_2\,]:A\plus B$ from two Raw formation
values; both branches are present in the syntax.  Raw
\textsc{$\plus$-Map} routes the two summands through their respective
maps and preserves the tag.

\paragraph{Structural type isomorphisms.}
\label{par:structural-isos}
We write $s : T \cong S$ for a \emph{primitive} structural type
isomorphism.  The primitives are the coherence maps of $\plus$ and
the distributivities:
\[
\begin{aligned}
  \alpha^\plus &: (A \plus B) \plus C \cong A \plus (B \plus C)
    &\qquad \sigma^\plus &: A \plus B \cong B \plus A, \\
  \mathsf{dist}_L &: A \tensor (B \plus C)
    \cong (A \tensor B) \plus (A \tensor C), \\
  \mathsf{dist}_R &: (A \plus B) \tensor C
    \cong (A \tensor C) \plus (B \tensor C).
\end{aligned}
\]
The displayed $\plus$-coherence maps and distributivities (and
their inverses) are primitive atomic constants at all of their
declared types.  Identity, inverse, composition, and $\tensor$-side
composites are built by application and the abbreviations below.
Branchwise composites are formed by Raw \textsc{$\plus$-Map}.
The $\tensor$-side coherence maps are
\emph{not} primitives: they are definable abbreviations,
\[
  \sigma^\tensor \;:=\; \lam{p}{\letpair{x}{y}{p}{y \tensor x}},
  \qquad
  \alpha^\tensor \;:=\; \lam{p}{\letpair{q}{z}{p}
    {\letpair{x}{y}{q}{x \tensor (y \tensor z)}}},
\]
(the $\mathsf{swap}$ of \S\ref{sec:lang-tour}); their applications
$\beta$-reduce, so they never block normalization.

\subsubsection{Source Language}
\label{subsubsec:source-language}

\paragraph{Source types and first-order interfaces.}
\label{par:first-order-types}
First-order data types and Source types are generated by
\[
  P ::= \base \mid P\tensor P \mid P\plus P,
  \qquad
  S ::= P \mid S\tensor S \mid S\lmark S .
\]
Thus every occurrence of $\plus$ in a Source type belongs to a
first-order data type.  First-order types are exactly those whose
boundary interfaces carry no reversed-polarity ports
(\S\ref{sec:boundary-semantics}).

The Source judgment $\Gamma\sjudge t:S$ uses Source types throughout.
It inherits the variable, implication, and tensor rules of
Table~\ref{tab:typing-rules-linear}, together with declared structural
atoms whose source and target are Source types.  Consequently, additive
coherence atoms and distributors may occur explicitly in Source only
at first-order instances.  Tensor coherence remains available at
arbitrary Source types through the defining linear terms of
\S\ref{par:structural-isos}.  The coherent branching construct is the
tag-preserving case rule below.  The quantum extension adds
\textsc{Exp}.

Throughout, a display $f:=u$ is only a metalanguage abbreviation.
Every occurrence of $f$ denotes a fresh capture-avoiding copy of $u$,
expanded before Source typing and before $(-)^\circ$ is formed.  Thus
$f$ is not a Source term constructor; the Source calculus has no
operation-name typing rule or closed-operation environment.

\paragraph{Tag-preserving case.}
\label{cohRouting}
(Appendix~\ref{app:case-derivation})
The Source case rule is
\[
\inferrule{
  \Gamma_1\sjudge e:A\plus B \\
  \Gamma\sjudge u:C \\
  \Gamma\sjudge v:C \\
  A,B,C\text{ first-order}
}{
  \Gamma_1,\Gamma\sjudge
  \caseof{e}{x}{u}{y}{v}:(A\plus B)\tensor C
}.
\]
The two branch premises contain the same complete typed context
$\Gamma$.  This same-context discipline also appears in coherent
branching calculi of \citet{BarssePechouxPerdrix2026} and
\citet{HirataTsukada2026LICS}; those languages do not use Granthi's
sum types.

The binders $x$ and $y$ record the routed summand and are paired with
the branch result on the output side.  They are not in scope as linear resources in the branch bodies $u$ and $v$.  Fix an ordering of nonempty
$\Gamma$, and write $G_\Gamma$ for the corresponding left-associated
tensor of its types and $\langle\Gamma\rangle$ for the matching tensor
of its variables.  If $z:G_\Gamma$, then $u[z/\Gamma]$ unpacks $z$ and
substitutes its components for the variables of $\Gamma$.

Write $t^\circ$ for the recursive Raw expansion of a Source term.  For
nonempty $\Gamma$, its case clause uses the closed composite
\[
\begin{aligned}
  G_\Gamma \tensor (A \plus B)
  &\xrightarrow{\;\distL\;}
  (G_\Gamma \tensor A) \plus (G_\Gamma \tensor B)\\
  &\xrightarrow{\;\hat f\,\plus\,\hat g\;}
  (A \tensor C) \plus (B \tensor C)\\
  &\xrightarrow{\;\distRi\;}
  (A \plus B) \tensor C,
\end{aligned}
\qquad\text{where}\quad
\begin{aligned}
  \hat f&=\lam{p}{\letpair{z}{x}{p}{x\tensor u^\circ[z/\Gamma]}},\\
  \hat g&=\lam{q}{\letpair{z}{y}{q}{y\tensor v^\circ[z/\Gamma]}}.
\end{aligned}
\]
The case occurrence expands to this composite applied to
$\langle\Gamma\rangle\tensor e^\circ$.  When $\Gamma=\varnothing$,
take $\hat f=\lam{x}{x\tensor u^\circ}$ and
$\hat g=\lam{y}{y\tensor v^\circ}$ and omit the first distributor.  When $G_\Gamma$ contains a function
type, the first distributor and its intermediate sum are Raw-only.
This higher-order distributor instance occurs solely inside the fixed
case expansion.

\begin{proposition}[Source expansion]
\label{prop:source-expansion}
If $\Gamma\sjudge t:A$, then
$\Gamma\vdash_{\mathsf r}t^\circ:A$.
\end{proposition}
Appendix~\ref{app:case-derivation} gives the typed derivation of the
case clause.

\begin{remark}[Why both Source restrictions are needed]
\label{rem:why-first-order}
The first-order condition excludes a branch-selected function that is
later applied to its own branch tag; the resulting Raw term can have
rank one.  Equality of the branch contexts excludes a different
failure: the Raw term
$\letpair{a}{b}{z}{[\,a\mid b\,]}$ separates the two components of
$z:A\tensor B$ between coherent alternatives.  Both terms belong to
the Raw calculus but not to the Source language.
Appendix~\ref{app:unitarityNEW},
Remark~\ref{rem:why-first-order-witness}, gives the first calculation.
\end{remark}

Higher-order computations can still be routed by applying them before
case formation.  Let $A$, $B$, and $P$ be first-order and $D$ a Source
type.  Given $\Delta\sjudge e:A\plus B$ and closed Source operations
$u,v:D\lmark P$, function $\eta$-expansion gives
\[
  \Delta\sjudge
  \lam{z}{\caseof{e}{x}{u\,z}{y}{v\,z}}
  :D\lmark((A\plus B)\tensor P)
\]
Iteration covers arrow chains: the tag follows every argument, while
the returned coherent payload remains first-order.

\subsubsection{Programming with Linear Control}
\label{subsubsec:linear-control}

The continuation and yield constructions below already exhibit
higher-order linear control using only the multiplicative rules.  The
final short-circuit routing example additionally combines
tag-preserving case with the Source \textsc{Exp} constructor.

\begin{table}[thb]
\caption{Continuation-based type constructions.}
\label{tab:continuation-types}
\centering
\small
\begin{tabular}{@{}l@{\;\;}l@{\;\;}l@{}}
\toprule
\textbf{Name} & \textbf{Type} & \textbf{Description} \\
\midrule
$\mathsf{Cont}_R(A)$ & $(A \lmark R) \lmark R$ & CPS computation \\[0.5ex]
$\mathsf{return}$ & $A \lmark \mathsf{Cont}_R(A)$ & wrap value \\[0.5ex]
$\mathsf{bind}$ & $\mathsf{Cont}_R(A) \tensor (A \lmark \mathsf{Cont}_R(B)) \lmark \mathsf{Cont}_R(B)$
  & sequence \\[0.8ex]
$\mathsf{Yield}(A,B)$ & $(A \lmark B) \tensor (B \lmark A)$ & coroutine channel \\[0.5ex]
$\mathsf{sendA}$ & $\mathsf{Yield}(A,B) \tensor A \lmark B \tensor (B \lmark A)$ & send from $A$, return peer \\[0.5ex]
$\mathsf{sendB}$ & $\mathsf{Yield}(A,B) \tensor B \lmark A \tensor (A \lmark B)$ & send from $B$, return peer \\[0.8ex]
\midrule
$\mathsf{and}_{\mathrm{sc}}$ & $(\QBool \tensor \QBool) \tensor \Wit \lmark (\QBool \tensor \QBool) \tensor \Wit$ & reversible short-circuit routing \\[0.5ex]
\bottomrule
\end{tabular}

\vspace{-1em}
\end{table}

Table~\ref{tab:continuation-types} summarizes key type constructions:
linear continuations and CPS ($\mathsf{Cont}_R$, $\mathsf{return}$,
$\mathsf{bind}$), and bidirectional control transfer via yield
channels ($\mathsf{Yield}$, $\mathsf{sendA}$, $\mathsf{sendB}$).
All entries are well typed in Source; the continuation and yield
entries use only the linear core.  Linearity enforces that each
continuation resource is used exactly once: it is either invoked or
returned explicitly, so control transfer is by function application,
not by special primitives.
Thus Granthi couples coherent quantum data to an already-expressive
higher-order control substrate.

With $\Wit:=\Aux\plus\QBool$ for an arbitrary auxiliary first-order
type $\Aux$, define the certified involution
\[
  J_{\Wit}
  :=\oplusmap{1}{\mathrm{id}_{\Aux}}
                   {1}{\sigma^\plus_{\base,\base}},
  \qquad
  \ijudge J_{\Wit}:\Wit\lmark\Wit .
\]
The brackets here denote the function-typed, branchwise $\plus$-map, not the sum-typed coherent-sum former; the result type distinguishes the two constructors in the typed abstract syntax.  This follows from \textsc{Inv-Id}, \textsc{Inv-$\sigma^\plus$}, and
\textsc{Inv-$\plus$}; $J_{\Wit}$ is metanotation for that certificate,
not a Source operation name (\S\ref{subsubsec:involutions}).  Using the Source \textsc{Exp} constructor
of \S\ref{subsubsec:exponentiation}, let
\[
  \mathsf{toggle}_{\Wit}
  :=\lam{w}{
    \expi{(-\pi/2)}{\mathrm{id}_{\Wit}}
      \bigl(\expi{\pi/2}{J_{\Wit}}\,w\bigr)}
  :\Wit\lmark\Wit .
\]
This is a Source abbreviation using only \textsc{Exp}, application,
and abstraction.  Since $J_{\Wit}^{2}=I$, its first-order action is
\[
 e^{-i\pi I/2}e^{i\pi J_{\Wit}/2}
 =(-iI)(iJ_{\Wit})
 =J_{\Wit},
\]
so it has the intended block action and is involutive.
Define the Source routing operation directly by case:
\[
\begin{aligned}
  \mathsf{route}_{\Wit}
  &:=\lam{p}{
      \letpair{b}{w}{p}{
        \caseof{b}{x}{\mathsf{toggle}_{\Wit}\,w}{y}{w}}} \\
  &:\QBool\tensor\Wit\lmark\QBool\tensor\Wit .
\end{aligned}
\]
Both branches have the exact context $w:\Wit$.

The short-circuit witness-routing kernel is
\[
  \mathsf{and}_{\mathrm{sc}}
  :=\lam{p}{
    \letpair{c}{w}{p}{
      \letpair{b_1}{b_2}{c}{
        \letpair{b_1'}{w'}{\mathsf{route}_{\Wit}(b_1\tensor w)}{
          (b_1'\tensor b_2)\tensor w'
        }
      }
    }
  }.
\]
The first boolean controls whether the $\QBool$ summand of the witness
is toggled; the second boolean and the witness are both returned.
Thus the output records the route without discarding a linear
resource.  The choice of $\Aux$ is immaterial; only its position as
the other summand of $\Wit$ matters.  The quantum extension below
adds a branch phase and then unroutes the witness.

\subsection{Quantum Extension}
\label{subsec:quantum-extension}
\label{sec:quantum-core}

We now extend the Raw language with one new term constructor
(\textsc{Exp}), a phase-enriched Raw \textsc{$\plus$-Map}, and two
new judgment forms (certified involutions and unitary primitives).
The Source language adds \textsc{Exp}; coherent phase control remains
expressed through case.  The sets
$\mathbb R_{\mathrm{static}}$ and $\mathbb C_{\mathrm{static}}$
contain scalars fixed at elaboration time.  The judgment
$\ijudge J:B\lmark B$, defined in
\S\ref{subsubsec:involutions}, certifies that $J$ is a closed
self-inverse map.
Accordingly, the Raw formation-value grammar of
Definition~\ref{def:raw-formation-values} gains the productions
\[
  W ::= \cdots \mid \expi{\theta}{J}
  \mid \oplusmap{\alpha}{W}{\beta}{W},
  \qquad
  \alpha,\beta\in\mathbb C_{\mathrm{static}},
  \quad |\alpha|=|\beta|=1.
\]
\paragraph{Type restriction.}
Throughout this section, $B$ (and subscripted variants) ranges over
the first-order types $P$ of \S\ref{par:first-order-types}.  We use
$B$ to emphasize that primitive quantum operations act on
first-order data; the surrounding Source language remains
higher-order.

\subsubsection{From Involutions to Unitaries}
\label{subsubsec:exponentiation}

Involutions---self-inverse operators satisfying $J^2 = \mathrm{id}$---serve
as generators for continuous families of unitaries via matrix
exponentiation. The \textsc{Exp} term rule constructs a unitary from a
certified involution:

\begin{center}
\textsc{Exp}\quad
$\inferrule{
    \ijudge J : B \lmark B \\
    \theta \in \mathbb{R}_{\mathrm{static}}
  }{
    \sjudge \expi{\theta}{J} : B \lmark B
  }$
\end{center}

The identical Raw term is typed by the same rule with
$\vdash_{\mathsf r}$ in the conclusion.

\noindent
The first-order action of $\expi{\theta}{J}$ is
\[
 e^{i\theta J}=\cos\theta\,\mathrm{id}+i\sin\theta\,J,
\]
using $J^2=\mathrm{id}$.  Its boundary denotation is derived from the
identity graph by applying this unitary to the output leg
(\S\ref{sec:boundary-semantics}).  The generators produced by the
certified-involution judgment are self-inverse and Hermitian by
induction on that judgment, so $\expi{\theta}{J}$ is unitary. The restriction to statically known
angles ensures that exponentiation does not introduce data-dependent control or
measurement into the term language.

The Raw quantum language also admits phased
\textsc{$\plus$-Map}:\phantomsection\label{subsubsec:quantum-plus-map}
\[
\inferrule{
  \Gamma_1\vdash_{\mathsf r}f:A_1\lmark B_3 \\
  \Gamma_2\vdash_{\mathsf r}g:A_2\lmark B_4 \\
  \Gamma=\Gamma_1\uplus\Gamma_2 \\
  \alpha,\beta\in\mathbb C_{\mathrm{static}},
  \quad |\alpha|=|\beta|=1
}{
  \Gamma\vdash_{\mathsf r}
  \oplusmap{\alpha}{f}{\beta}{g}
  :A_1\plus A_2\lmark B_3\plus B_4
}.
\]
The phases act on their respective branches.  This Raw rule permits
arbitrary summand types.  Source phase control is expressed through
case; its expansion produces closed branch maps with first-order
targets, as required by the internal judgment.

\paragraph{Example: phase-marked short-circuit routing.}
Let
\[
  \mathsf{phase}_{\Wit}
  :=\expi{\pi}{\mathrm{id}_{\Wit}}
  =-\mathrm{id}_{\Wit}.
\]
The phase-marked route is another Source case:
\[
\begin{aligned}
  \mathsf{route}_{\Wit}^{\mathrm q}
  &:=\lam{p}{
      \letpair{b}{w}{p}{
        \caseof{b}{x}{
          \mathsf{phase}_{\Wit}(\mathsf{toggle}_{\Wit}\,w)}
        {y}{w}}} \\
  &:\QBool\tensor\Wit\lmark\QBool\tensor\Wit .
\end{aligned}
\]
The left branch toggles the witness and contributes a $-1$ phase;
the right branch passes the witness unchanged.  Since
$\mathsf{toggle}_{\Wit}$ is involutive, follow this route by the
unphased $\mathsf{route}_{\Wit}$.  Their composite restores the
witness and retains the relative branch phase.  Substituting it for
$\mathsf{route}_{\Wit}$ in $\mathsf{and}_{\mathrm{sc}}$ gives the
quantum short-circuit operator $\mathsf{and}_{\mathrm{sc}}^{\mathrm q}$.
For every witness state $|w\rangle$, the composite sends
$|{+}\rangle\tensor|w\rangle$ to
$-|{-}\rangle\tensor|w\rangle$.  A subsequent Hadamard on the first
boolean therefore yields $|1\rangle$ up to global phase; the unphased
route--unroute yields $|0\rangle$.

\subsubsection{Certified Involutions}
\label{subsubsec:involutions}

The judgment $\ijudge J : B \lmark B$ certifies that $J$ is a closed
self-inverse map: $J \circ J = \mathrm{id}$. Restricting involutions to a
syntactic judgment allows precise control over which generators may be
exponentiated, keeping the bridge to continuous unitaries explicit.
Table~\ref{tab:typing-rules-involutions} presents the rules.

\begin{table}[!ht]
\caption{Typing rules for certified involutions.}
\label{tab:typing-rules-involutions}
\centering
\small
\begin{tabular}{@{}l@{}}
\toprule
\emph{Certified Involutions} ($\ijudge J : B \lmark B$) \\
\midrule

\textsc{Inv-Id}\;
$\inferrule{ }{\ijudge \mathrm{id}_B : B \lmark B}$
\qquad
\textsc{Inv-Scalar}\;
$\inferrule{
    \ijudge J : B \lmark B \\
    \alpha \in \{1, -1\}
  }{
    \ijudge {[\alpha \cdot J]} : B \lmark B
  }$
\\[2ex]

\textsc{Inv-$\sigma^\tensor$}\;
$\inferrule{ }{\ijudge \sigma^\tensor_{B,B} : B \tensor B \lmark B \tensor B}$
\qquad
\textsc{Inv-$\tensor$}\;
$\inferrule{
    \ijudge J : B_1 \lmark B_1 \quad
    \ijudge K : B_2 \lmark B_2
  }{
    \ijudge J \tensor K : B_1 \tensor B_2 \lmark B_1 \tensor B_2
  }$
\\[2ex]

\textsc{Inv-$\sigma^\plus$}\;
$\inferrule{ }{\ijudge \sigma^\plus_{B,B} : B \plus B \lmark B \plus B}$
\qquad
\textsc{Inv-$\plus$}\;
$\inferrule{
    \ijudge J : B_1 \lmark B_1 \quad
    \ijudge K : B_2 \lmark B_2 \\
    \alpha, \beta \in \{1, -1\}
  }{
    \ijudge \oplusmap{\alpha}{J}{\beta}{K} : B_1 \plus B_2 \lmark B_1 \plus B_2
  }$
\\[2ex]

\bottomrule
\end{tabular}

\end{table}

In \textsc{Inv-$\tensor$} and \textsc{Inv-$\plus$}, the symbols $\tensor$
and $\oplusmap{\alpha}{-}{\beta}{-}$ are map formers on involutions,
not type constructors: if $J : B_1 \lmark B_1$ and $K : B_2 \lmark B_2$, then
$J \tensor K : B_1 \tensor B_2 \lmark B_1 \tensor B_2$ and
$\oplusmap{\alpha}{J}{\beta}{K} : B_1 \plus B_2 \lmark B_1 \plus B_2$.

The general tensor symmetry
$\sigma^\tensor : B_1 \tensor B_2 \lmark B_2 \tensor B_1$
and the general sum symmetry
$\sigma^\plus : B_1 \plus B_2 \lmark B_2 \plus B_1$
are not endomorphisms when $B_1 \neq B_2$, so the certified
involution judgment does not apply. Their endomorphic specializations
$\sigma^\tensor_{B,B} : B \tensor B \lmark B \tensor B$ and
$\sigma^\plus_{B,B} : B \plus B \lmark B \plus B$ are involutive by
definition ($\sigma_{B,B} \circ \sigma_{B,B} = \mathrm{id}$), making
them admissible as generators for exponentiation. Swaps inside larger
types are obtained by the closure rules \textsc{Inv-$\tensor$} and
\textsc{Inv-$\plus$}, up to structural associativity and symmetry. The scalar restriction to $\{1, -1\}$ in
\textsc{Inv-Scalar} and \textsc{Inv-$\plus$} is necessary:
$(\alpha J)^2 = \alpha^2 J^2 = \alpha^2 \cdot \mathrm{id}$ equals
$\mathrm{id}$ only when $\alpha^2 = 1$.
Every certified involution is Hermitian: the symmetries
$\sigma^\tensor_{B,B}$ and $\sigma^\plus_{B,B}$ are real permutation
matrices (hence Hermitian), and the rules preserve Hermiticity since
$\{1,-1\}$ scalars are real and $J \tensor K$,
$\oplusmap{\alpha}{J}{\beta}{K}$ are Hermitian when $J$, $K$ are.
Combined with $J^2 = \mathrm{id}$, this ensures $\expi{\theta}{J}$ is
unitary.

\subsubsection{Unitary Primitives}
\label{subsubsec:unitary-primitives}

The auxiliary judgment
$\unitjudge U:B\lmark B$ records static certificates for unitary
operators at first-order type $B$; it does not extend the
source-term grammar.  Table~\ref{tab:unitary-generation} bridges
from the certified involutions of
Table~\ref{tab:typing-rules-involutions}, closes certificates under
composition and adjoint, and forms the one-parameter families
generated by involutions.  This presentation isolates the available
quantum actions from any particular backend gate set.

\begin{table*}[!ht]
\caption{Static unitary certificates at first-order type $B$.}
\label{tab:unitary-generation}
\centering
\small
\begin{tabular}{@{}l@{\;}c@{\quad}l@{\;}c@{}}
\toprule
\multicolumn{4}{@{}l}{\emph{Unitary Primitives} ($\unitjudge U : B \lmark B$)} \\
\midrule

\textsc{Unit-Inv}
& $\inferrule{
    \ijudge J : B \lmark B
  }{
    \unitjudge J : B \lmark B
  }$
& \textsc{Unit-Comp}
& $\inferrule{
    \unitjudge U : B \lmark B \\
    \unitjudge V : B \lmark B
  }{
    \unitjudge (V \circ U) : B \lmark B
  }$
\\[2ex]

\textsc{Unit-Dag}
& $\inferrule{
    \unitjudge U : B \lmark B
  }{
    \unitjudge U^\dagger : B \lmark B
  }$
& \textsc{Unit-Exp}
& $\inferrule{
    \ijudge J : B \lmark B \\
    \theta \in \mathbb{R}_{\mathrm{static}}
  }{
    \unitjudge \expi{\theta}{J} : B \lmark B
  }$
\\[2ex]

\bottomrule
\end{tabular}

\vspace{-1em}
\end{table*}
Interpreting the four rules in $\mathbf{FdHilb}$ as inclusion of a
certified involution, operator composition, Hilbert-space adjoint,
and matrix exponentiation, respectively, each rule preserves
unitarity.  Thus the denotations of the certificates form a subgroup
of $U(\sem B)$.  Here $(-)^\dagger$ belongs to the certificate
language; the later notation $(-)^*$ denotes type duality and is
unrelated.
On an $n$-fold coherent sum of $B$, exponentiated sign and swap
involutions, transported by structural conjugation, generate every
summand-index action $U\tensor I_{\sem B}$; for $B=\base$, this is
every unitary on $\Qn{n}$
(Appendix~\ref{app:summand-index-completeness}).

\paragraph{Backend correctness (BC)}
The semantic results do not assume a gate implementation.  Circuit
realization uses the external hypothesis \textup{(BC)}.  It requires
the declared artifact for every
$a=\expi{\theta}{J}:P\lmark P$ to realize
$V_a=e^{i\theta J}$ through its declared source and target
coordinates.  The declaration supplies computational-basis isometries onto the
valid source and target codewords and type-canonical logical-frame
basis bijections.  The primitive artifact acts on exactly the
canonical physical carrier of $P$ and contains no backend-owned
workspace or residual register.  Prescribed tag and zero-padding
coordinates remain part of this carrier; the source and target
placements may differ.  The target constructors used by the reference emitter
realize their stated operator operations: exact phases and controls,
lifting, retargeting, and sequential and parallel composition.  The
closed boundary graph and the applied action are
derived from the same first-order map.

%% file: metatheory.tex
\section{Metatheory}
\label{sec:metatheory}

Normalization is Granthi's canonicalization pass.  It converts the Raw
expansions of Source programs into the restricted intermediate form
consumed by both the boundary
semantics and the circuit emitter.  This section fixes that deterministic
pass and proves that any irreducible reduct is unique modulo alpha-conversion,
independent-let exchange, and equality of static phases.

\subsection{Normal Forms}
\label{subsubsec:normal-forms}

\begin{definition}[Atomic operations]
\label{def:atoms}
The atomic operations are generated by
\[
  \mathsf{Atom} ::= \expi{\theta}{J} \mid s,
\]
where $s$ ranges over the primitive structural isomorphisms of
\S\ref{par:structural-isos}, viewed as closed constants at their
canonical types.
\end{definition}

We call $s$-atoms \emph{structural} and $\expi{\theta}{J}$-atoms
\emph{quantum}.  Normalization treats both classes opaquely.

The Raw formation values $W$ of
Definition~\ref{def:raw-formation-values} are a typing-side
syntactic category and may contain reducible lambda bodies.  They are
distinct from the irreducible values $V$ classified below.

\begin{definition}[Values, neutrals, normal forms]
\label{def:nf-grammar}
The classes $V$, $E$, $R$, and $N$ are the least mutually inductive
classes generated by Table~\ref{tab:nf-grammar}.
\end{definition}

The production $x\,V$ belongs to both the value and neutral
grammars.  Variable-headed neutrals support the induction on open
terms.  Such a neutral cannot be the outer form of a closed normal
term, although it may occur under a binder.

\begin{table}[!ht]
\caption{Grammar of values ($V$), neutrals ($E$), results ($R$), and
normal forms ($N$).}
\label{tab:nf-grammar}
\centering
\small
\renewcommand{\arraystretch}{1.3}
\begin{tabular}{@{}r@{\;\;}c@{\;\;}l@{\qquad}l@{}}
\toprule
\textbf{Class} & & \textbf{Productions} & \textbf{Description} \\
\midrule
$V$ & $::=$
    & $x \mid V \tensor V \mid \lam{x}{N}$
    & variables, pairs, abstractions \\
    & $\mid$
    & ${[\,\alpha \cdot V_1 \mid \beta \cdot V_2\,]} \mid \mathsf{Atom} \mid x\,V$
    & phased sums (sum-typed), atoms, stuck applications \\
    & $\mid$
    & $\oplusmap{\alpha}{V_1}{\beta}{V_2}$
    & $\plus$-map values (function-typed) \\[1ex]
$E$ & $::=$
    & $x \mid \mathsf{Atom} \mid E\;R$
    & variables, atoms, applications \\
    & $\mid$
    & $\oplusmap{\alpha}{R_1}{\beta}{R_2}\;E$
    & $\plus$-map on neutral \\[1ex]
$R$ & $::=$
    & $V \mid E \mid R_1 \tensor R_2$
    & values, neutrals, pairs \\
    & $\mid$
    & ${[\,\alpha \cdot R_1 \mid \beta \cdot R_2\,]}
       \mid \oplusmap{\alpha}{R_1}{\beta}{R_2}$
    & blocked sums and maps \\[1ex]
$N$ & $::=$
    & $R \mid \letpair{x}{y}{R}{N}$
    & results, sequenced lets \\
\bottomrule
\end{tabular}

\end{table}

Sum formers carry static unit-modulus phases $\alpha, \beta$
recorded by normalization (Appendix~\ref{app:normalization});
$[\,V_1 \mid V_2\,]$ abbreviates $[\,1 \cdot V_1 \mid 1 \cdot V_2\,]$;
Raw $\plus$-introduction forms only unphased formation values.

A \emph{blocked
sum} $[\,\alpha \cdot R_1 \mid \beta \cdot R_2\,]$, in which at least
one branch is not a value, is a result but not a value.
Symmetrically, a \emph{blocked map}
$\oplusmap{\alpha}{R_1}{\beta}{R_2}$ has at least one non-value
function operand.  It has function type and is justified by the
$\plus$-Map target condition: its targets are first-order, while
its sources remain unrestricted.  The result productions deliberately overlap the value productions
in the all-value case; the word \emph{blocked} names the proper
subclass having at least one non-value operand.

The full floater family (Appendix~\ref{app:normalization})
extracts let-binders by the fixed left-to-right prefix-float
strategy, producing the selected normal form; alternative orders
differ only by exchange of independent prefixes.
Sum-former summand types are first-order.

The four conditions below state what the pass guarantees
(see Appendix~\ref{app:normalization}):
\begin{description}
\item[(NF1)] \textbf{Eliminate $\beta$-redexes.} No subterm matches the
  left-hand side of $(\beta_{\lmark})$ or $(\beta_{\tensor})$.
\item[(NF2)] \textbf{Fuse nested branch maps.} No subterm matches the
  left-hand side of $(\plus\text{-comp})$.
\item[(NF3)] \textbf{Push maps through explicit sums.} No subterm is a
  $\plus$-map applied to a (possibly phased) sum former
  $[\,\alpha' \cdot t_1 \mid \beta' \cdot t_2\,]$ with arbitrary
  branch terms; in particular no
  subterm matches $\oplusmap{\alpha}{f}{\beta}{g}$ $[V_1 \mid V_2]$.
\item[(NF4*)] \textbf{Hoist independent lets deterministically.} In every
  $\lambda z.L_1\cdots L_k[R]$, where each $L_i[-]$ is a
  tensor-let prefix, no prefix let $L_j$ has a scrutinee independent
  of $z$ while also being independent of every preceding $L_i$.
  The exact prefix and independence conditions are
  Definition~\ref{def:nf4star}.
\end{description}
Source case expands to the distributor--closed-map--inverse-distributor
composite of \S\ref{cohRouting}; hence no separate case constructor
appears in the normal-form grammar.

The computational equality used by normalization is generated by the
oriented $\beta$-rules and the administrative and branch-map
conversions of Appendix~\ref{app:normalization}.  Eta-expansion is an
admissible type-directed identity witness; it is not a step of the
normalizer $\mathsf{NF}_{\mathrm{LO}}$ of Theorem~\ref{thm:normalization}.  Appendix~\ref{app:unitarityNEW} proves
its identity action and the fresh-application compatibility required
for neutral function operands
(Lemmas~\ref{lem:semantic-eta-identity}
and~\ref{lem:classified-cut-eta}).

\subsection{Normalization and Determinacy}
\label{subsubsec:main-results}

\paragraph{The three judgments.}
The Source judgment is the programmer-facing calculus and enforces the
discipline of \S\ref{subsubsec:source-language}.  The recursive
translation $(-)^\circ$ expands Source into Raw syntax; Raw is the
syntax on which normalization is defined.  The proof-only internal
administrative judgment
$\Gamma\intjudge t:A$, defined in Appendix~\ref{app:focused-rules},
types the phase-bearing intermediates produced during normalization.
It admits phased coherent sums with arbitrary typed branches while
retaining the first-order summand restriction.  If
$\Gamma\sjudge t:A$, then its Raw expansion satisfies
$\Gamma\intjudge t^\circ:A$
(Lemma~\ref{lem:source-internal-inclusion}).  Thus a Source program
normalizes through internally typed intermediates to the normal form
consumed by the semantics and the compiler.

\begin{theorem}[Deterministic normalization]
\label{thm:normalization}
If $\Gamma\intjudge t:A$, then the fixed leftmost-outermost
normalizer $\mathsf{NF}_{\mathrm{LO}}(t)$ is defined, is internally
well typed and irreducible, and
$t\to^*\mathsf{NF}_{\mathrm{LO}}(t)$.  Hence every well-typed
Source term has a normal form after Raw expansion.
\end{theorem}
The appendix fixes left-before-right branch normalization,
left-to-right let-prefix extraction, and a deterministic
capture-avoiding naming convention.  Its strong-induction theorem
also bounds the $\plus$-map-constructor count $\Phi$
(Definition~\ref{def:phi}) and preserves the first-order side
conditions, unit-modulus phases, and disjoint branch contexts needed
by hereditary contraction
(Theorem~\ref{thm:lo-normalizer-total}).  The grammar theorem is
separate: irreducible terms are exactly the terms of
Table~\ref{tab:nf-grammar} satisfying (NF1)--(NF4*)
(Lemma~\ref{lem:nf-iff-irreducible}).

Write $\approx$ for the least congruence generated by
alpha-equivalence, exchange of adjacent independent tensor lets,
and extensional equality of static unit-modulus phase annotations.

\begin{theorem}[Determinacy]
\label{thm:determinacy}
If $\Gamma\intjudge t:A$, $t\to^*n$, and $t\to^*n'$, with $n$ and
$n'$ irreducible, then $n\approx n'$.
\end{theorem}
The proof is a staged induction on the normalization measure:
local compatibility with $\approx$, one-step peak joining, and
uniqueness at a rank use uniqueness only at strictly smaller ranks
(Appendix~\ref{app:determinacy-proofs},
Theorem~\ref{thm:ranked-determinacy}).  In particular, every
irreducible reduct is $\approx$-equal to
$\mathsf{NF}_{\mathrm{LO}}(t)$.

A conclusion merely up to the symmetric closure of reduction would
be immediate.  The compiler normalizes before emission.

\begin{corollary}[Source normalization]
\label{cor:source-normalization}
If $\Gamma\sjudge t:A$, then $\mathsf{NF}_{\mathrm{LO}}(t^\circ)$
is defined, internally well typed, irreducible, contains no Raw
coherent-sum former $[\,-\mid-\,]$, every $\plus$-map in it has
closed operands, and it contains no inverse distributor whose source
or target is non-first-order.
\end{corollary}
The proof is in Appendix~\ref{app:normalization}.

%% file: SEM-NEW.tex
\section{Boundary Semantics}
\label{sec:boundary-semantics}

Granthi assigns a boundary operator to each Source judgment
$\Gamma\sjudge t:A$.  Its Raw expansion $t^\circ$ is deterministically
normalized and canonically retyped, and the resulting internal normal
derivation is interpreted below.  We write the resulting operator as
$\SEM{\Gamma\sjudge t:A}$: the context is part of the denotation.
All Hilbert spaces and linear maps in this section live in
$\mathbf{FdHilb}$; thus $\tensor$ is its monoidal tensor and $\oplus$
its additive biproduct.

\paragraph{Signed boundaries and the flat envelope.}
A signed boundary formula is generated by
\[
  \varphi,\psi ::= \mathbf{1}\mid b^-\mid b^+
       \mid \varphi\tensor\psi\mid \varphi\oplus\psi .
\]
Here $b^\pm$ are signed base ports and $\mathbf{1}$ is the empty
tensor.  Raw types---and hence Source types and the Raw-only intermediate
types introduced by Source expansion---determine signed formulas by
polarity reversal:
\[
\begin{array}{r@{\;}c@{\;}l@{\qquad}r@{\;}c@{\;}l}
\mathsf{sgn}(\base)        &=& b^+, &
\mathsf{sgn}(A\tensor B)   &=& \mathsf{sgn}(A)\tensor\mathsf{sgn}(B),\\
\mathsf{sgn}(A\plus B)     &=& \mathsf{sgn}(A)\oplus\mathsf{sgn}(B),&
\mathsf{sgn}(A\lmark B)    &=& \mathsf{sgn}(A)^*\tensor\mathsf{sgn}(B).
\end{array}
\]
The involution $(-)^*$ flips signs and distributes over both
operations.  Put $\varphi$ in the fixed left-to-right rig normal
form
\[
  \varphi\cong\bigoplus_{i\in I_\varphi}M_i,
  \qquad
  M_i=b^{\epsilon_{i,1}}\tensor\cdots\tensor b^{\epsilon_{i,k_i}},
\]
and sort each monomial by sign:
\[
  M_i^-:=\bigotimes_{\epsilon_{i,j}=-}b,
  \qquad
  M_i^+:=\bigotimes_{\epsilon_{i,j}=+}b,
  \qquad
  \partial^\pm(\varphi):=\bigoplus_i M_i^\pm .
\]
Empty products are $\mathbf{1}$.  In the projected monomials $M_i^\pm$,
$b$ denotes the unsigned carrier underlying either signed port.  Hilbert
evaluation forgets polarity:
\[
  \sem{\mathbf{1}}=\sem b=\sem{b^-}=\sem{b^+}=\mathbb C,\qquad
  \sem{\varphi\tensor\psi}=\sem\varphi\otimes\sem\psi,\qquad
  \sem{\varphi\oplus\psi}=\sem\varphi\oplus\sem\psi .
\]
For any Raw type $T$ we use the shorthand
\[
  \sem{T}:=\sem{\mathsf{sgn}(T)}.
\]

For a judgment shape
$\mathfrak J=(x_1{:}A_1,\ldots,x_n{:}A_n\intjudge A)$, set
\[
  \varphi_{\mathfrak J}:=\mathsf{sgn}(A_1)^*\tensor\cdots\tensor
       \mathsf{sgn}(A_n)^*\tensor\mathsf{sgn}(A),
  \qquad
  \mathcal E_{\mathfrak J}^\pm:=\sem{\partial^\pm(\varphi_{\mathfrak J})} .
\]
The pair $(\mathcal E_{\mathfrak J}^-,\mathcal E_{\mathfrak J}^+)$ is the
\emph{ambient flat envelope} of the judgment shape.  It expands the
sum indices of $\varphi_{\mathfrak J}$ and supplies fixed coordinates
for the derivation-selected boundary spaces.

A canonical normal derivation $\mathcal N$ (the canonical retyping
derivation of Lemma~\ref{lem:canonical-normal-retyping}) selects its
boundary spaces
$B_{\mathcal N}^\pm\subseteq\mathcal E_{\mathfrak J}^\pm$.
Different derivations of the same judgment shape may select different
subspaces.  Such a subspace may have support on every ambient
coordinate while remaining strictly lower-dimensional than the
envelope.

\paragraph{Semantic clauses.}
Table~\ref{tab:sem-compositional} gives the interpretation of the
canonical internal normal judgments used by Source semantics.  Every
bracket displays the complete judgment, and all displayed context
unions are disjoint.  The equations are read through the fixed
canonical reassociations, symmetries, polarity repartitions, and
distributivities of the boundary spaces.  $R$ may include the blocked
maps of Table~\ref{tab:nf-grammar}.  The eliminator clauses take
canonical premise derivations
$\mathcal D_E,\mathcal D_R,\mathcal D_N,\mathcal D_F$ as arguments:
$\mathsf{AppCut}_A$ connects a function component to its argument,
$\mathsf{TenCut}_{A,B}$ connects a tensor producer to the body that
destructures it, and $\mathsf{MapApply}_{A\plus B}$ routes source
summands through their branch maps; Appendix~\ref{app:unitarityNEW}
defines these derivation-directed cuts and proves that they are well
founded and unitary.  In the applied-map row, $\mathcal F$ is the
canonical derivation of the displayed $\plus$-map former.

For a canonical normal derivation
$\mathcal N:(\Gamma\intjudge N:A)$ of shape $\mathfrak J$, the
construction simultaneously selects
$B_{\mathcal N}^{\pm}\subseteq\mathcal E_{\mathfrak J}^{\pm}$ and
assigns a map $U_{\mathcal N}:B_{\mathcal N}^{-}\to B_{\mathcal N}^{+}$;
write $\SEM{\Gamma\intjudge N:A}_{\mathrm{NF}}:=U_{\mathcal N}$.
Premise brackets denote previously assigned premise data, while
eliminators assign the conclusion spaces and map together.
The variable leaf uses the equal-address paired
port $\mathcal Y_A^\pm$ with
$\mathsf{yank}_A:\mathcal Y_A^-\to\mathcal Y_A^+$, the direct sum
over the rig normal form of the canonical tensor symmetries
(Definition~\ref{def:variable-graph-sector},
Appendix~\ref{app:unitarityNEW}).  Atomic leaves use their selected
graph boundaries.  Eliminators close the displayed whole port
while retaining all spectator factors.  Write $Y_\Gamma$ for the
typed identity through-map on an inactive context $\Gamma$
(\textup{(TY)} in Appendix~\ref{app:unitarityNEW}).

The boundary of a first-order unitary $V:\sem P\to\sem P$ is derived
from the identity graph: its source is $\mathcal Y_P^-$, its target is
$(I\tensor V)\mathcal Y_P^+$, and its map is
$(I\tensor V)\mathsf{yank}_P$.  For $V_{\theta,J}:=e^{i\theta J}$
from \S\ref{subsubsec:exponentiation}, \textsc{Exp} applies
$V_{\theta,J}$ to the positive leg of the identity graph.  In the
structural row below,
$\mathfrak S=(\cdot\intjudge s:T\lmark S)$ denotes the displayed
judgment shape.

\begin{table*}[t]
\caption{Derivation-indexed semantic clauses for the canonical
internal normal judgments.}
\label{tab:sem-compositional}
\centering
\scriptsize
\renewcommand{\arraystretch}{1.35}
\setlength{\tabcolsep}{3pt}
\begin{tabular}{@{}p{0.10\linewidth}p{0.42\linewidth}p{0.42\linewidth}@{}}
\toprule
\textbf{Rule} & \textbf{Canonical normal judgment} &
\textbf{Judgment denotation}\\
\midrule
\textsc{Var}
&
$x{:}A\intjudge x:A$
&
$\SEM{x{:}A\intjudge x:A}=\mathsf{yank}_A$
\\[1.2ex]

\textsc{$\lmark$-I}
&
$\displaystyle
 \frac{\Gamma,x{:}A\intjudge N:B}
      {\Gamma\intjudge\lambda x.N:A\lmark B}$
&
$\displaystyle
 \SEM{\Gamma\intjudge\lambda x.N:A\lmark B}
 =
 \SEM{\Gamma,x{:}A\intjudge N:B}$
\\[2.4ex]

\textsc{$\lmark$-E}
&
$\displaystyle
 \frac{\Gamma_1\intjudge E:A\lmark B\qquad
       \Gamma_2\intjudge R:A}
      {\Gamma_1,\Gamma_2\intjudge E\,R:B}$
&
$\displaystyle
 \SEM{\Gamma_1,\Gamma_2\intjudge E\,R:B}
 =
 \mathsf{AppCut}_{A}(\mathcal D_E,\mathcal D_R)$
\\[2.8ex]

\textsc{$\tensor$-I}
&
$\displaystyle
 \frac{\Gamma_1\intjudge R_1:A\qquad
       \Gamma_2\intjudge R_2:B}
      {\Gamma_1,\Gamma_2\intjudge R_1\tensor R_2:A\tensor B}$
&
$\displaystyle
 \SEM{\Gamma_1,\Gamma_2\intjudge R_1\tensor R_2:A\tensor B}
 =
 \SEM{\Gamma_1\intjudge R_1:A}\tensor
 \SEM{\Gamma_2\intjudge R_2:B}$
\\[2.8ex]

\textsc{$\tensor$-E}
&
$\displaystyle
 \frac{\Gamma_1\intjudge R:A\tensor B\qquad
       \Gamma_2,x{:}A,y{:}B\intjudge N:C}
      {\Gamma_1,\Gamma_2\intjudge\letpair{x}{y}{R}{N}:C}$
&
$\displaystyle
 \SEM{\Gamma_1,\Gamma_2\intjudge\letpair{x}{y}{R}{N}:C}
 =
 \mathsf{TenCut}_{A,B}(\mathcal D_N,\mathcal D_R)$
\\[3.0ex]

\textsc{$\plus$-Map}
&
$\displaystyle
 \frac{\Gamma_1\intjudge R_1:A\lmark C\qquad
       \Gamma_2\intjudge R_2:B\lmark D}
      {\Gamma_1,\Gamma_2\intjudge
       \oplusmap{\alpha}{R_1}{\beta}{R_2}:
       (A\plus B)\lmark(C\plus D)}$
\quad
$C,D$ first-order;\;
$\alpha,\beta\in\mathbb C_{\mathrm{static}}$;\;
$|\alpha|=|\beta|=1$
&
$\displaystyle
 \SEM{\Gamma_1,\Gamma_2\intjudge
 \oplusmap{\alpha}{R_1}{\beta}{R_2}:
 (A\plus B)\lmark(C\plus D)}
 =
 \begin{aligned}
 &\alpha\bigl(
   \SEM{\Gamma_1\intjudge R_1:A\lmark C}\tensor Y_{\Gamma_2}
  \bigr)\\
 &\quad\oplus\ \beta\bigl(
   Y_{\Gamma_1}\tensor
   \SEM{\Gamma_2\intjudge R_2:B\lmark D}
  \bigr)
 \end{aligned}$
\\[3.4ex]

\textsc{Applied $\plus$-Map}
&
$\displaystyle
 \frac{\Gamma_1\intjudge R_1:A\lmark C\;
       \Gamma_2\intjudge R_2:B\lmark D\;
       \Delta\intjudge E:A\plus B}
      {\Gamma_1,\Gamma_2,\Delta\intjudge
       \oplusmap{\alpha}{R_1}{\beta}{R_2}\,E:C\plus D}$
\quad
$C,D$ first-order;\;
$\alpha,\beta\in\mathbb C_{\mathrm{static}}$;\;
$|\alpha|=|\beta|=1$
&
$\displaystyle
 \begin{aligned}
 &\SEM{\begin{gathered}
   \Gamma_1,\Gamma_2,\Delta\intjudge
   \oplusmap{\alpha}{R_1}{\beta}{R_2}\,E\\[-0.3ex]
   {}:C\plus D
 \end{gathered}}
 \\[-0.2ex]
 &\quad={}\mathsf{MapApply}_{A\plus B}
   (\mathcal D_F,\mathcal D_E)
 \end{aligned}$
\\[1.2ex]

Structural
&
$\cdot\intjudge s:T\lmark S$
&
$\SEM{\cdot\intjudge s:T\lmark S}
  =\mathsf{Rig}_{\mathfrak S}(s)$
\\[1.2ex]

\textsc{Exp}
&
$\cdot\intjudge\expi{\theta}{J}:P\lmark P$
\quad
$P$ first-order;\; $\ijudge J:P\lmark P$;\; $\theta\in\mathbb R_{\mathrm{static}}$
&
$\SEM{\cdot\intjudge\expi{\theta}{J}:P\lmark P}
 =(I\tensor V_{\theta,J})\mathsf{yank}_P$
\\
\bottomrule
\end{tabular}

\end{table*}

The $Y_{\Gamma_i}$ factors in the \textsc{$\plus$-Map} row give the
derivation-indexed completion through the inactive branch context
(Appendix~\ref{app:unitarityNEW},
Definition~\ref{def:inactive-completion}).
A primitive structural atom $s:T\lmark S$ carries the
canonical rig-coherence isomorphism
$\mathsf{Rig}(s):\sem T\to\sem S$, a basis permutation.  Its judgment
denotation is the induced graph-sector map
$\mathsf{Rig}_{\mathfrak S}(s)$ of \textup{(SG)--(SG-map)} in
Appendix~\ref{app:unitarityNEW}; the compilation appendix uses
$\mathsf{Rig}(s)$ as its coordinate action.  Tensor-side coherence
maps are interpreted through their defining linear programs.

\paragraph{Why the cuts remain unitary.}
At a first-order port, the classified cut (in the sense of
Appendix~\ref{app:unitarityNEW}) contracts the producer and
consumer graph presentations over their common port coordinate
(componentwise for a completed family), yielding the graph
presentation of their composite coordinate unitary.  At a
higher-order port the cut follows the corresponding normal-form
clause, and orthogonal branch cases close blockwise.
Appendix~\ref{app:unitarityNEW} proves these claims
(Lemmas~\ref{lem:typed-graph-composition}
and~\ref{lem:whole-port-collapse},
Theorem~\ref{thm:nf-boundary-unitarity}).

\begin{definition}[Canonical-form semantics]
\label{def:canonical-form-semantics}
For $\Gamma\sjudge t:A$, let
$N_t=\mathsf{NF}_{\mathrm{LO}}(t^\circ)$ and set
\[
  \SEM{\Gamma\sjudge t:A}
  :=
  \SEM{\Gamma\intjudge N_t:A}_{\mathrm{NF}}.
\]
\end{definition}
Ordinary clauses recurse on proper normal sub\-derivations, and the
applied-map source-\allowbreak consumer and coherent-\allowbreak sharing clauses normalize
their formal branch applications at strictly smaller $\Phi$
(Appendix~\ref{app:unitarityNEW}).
Corollary~\ref{cor:source-normalization} and
Lemma~\ref{lem:canonical-normal-retyping} supply the displayed
canonical internal derivation.  Theorem~\ref{thm:determinacy} and
Corollary~\ref{cor:canonical-invariance} show that reduction of
$t^\circ$ and choice of reduction strategy preserve the denotation up
to canonical unitary boundary transport.

\begin{theorem}[Boundary unitarity]
\label{thm:boundary-unitarity}
If $\Gamma\sjudge t:A$, then $\SEM{\Gamma\sjudge t:A}$ is unitary
between the negative and positive boundary spaces selected by the
canonical derivation of $\mathsf{NF}_{\mathrm{LO}}(t^\circ)$.
\end{theorem}

This follows from Corollary~\ref{cor:source-normalization} and
Theorem~\ref{thm:nf-boundary-unitarity} in
Appendix~\ref{app:unitarityNEW}.

\begin{corollary}[Register unitarity]
\label{cor:register-unitarity}
For every $n\geq1$, a closed derivation
$\cdot\sjudge t:\QBool^{\tensor n}\lmark\QBool^{\tensor n}$
induces a unique unitary
$U_t:\sem{\QBool^{\tensor n}}\to\sem{\QBool^{\tensor n}}$
on the full register space.
If $x:\QBool^{\tensor n}$ and
$N=\mathsf{NF}_{\mathrm{LO}}\bigl((t\,x)^\circ\bigr)$, then
$\SEM{x:\QBool^{\tensor n}\intjudge N:\QBool^{\tensor n}}_{\mathrm{NF}}$
presents $U_t$.
\end{corollary}

The register readback is proved in
Appendix~\ref{app:unitarityNEW}.

%% file: compiler_new.tex
\section{Circuit Realization}
\label{sec:compilation}

Compilation has three stages: normalize and canonically retype the program;
assign typed physical layouts; then emit symbolic wiring and gates, adding
controls where required.  Given
$\Gamma\sjudge t:A$, the compiler computes
\[
  N_t=\mathsf{NF}_{\mathrm{LO}}(t^\circ),
  \qquad
  \mathcal N_t=\mathsf{CanDer}_{\Gamma,A}(N_t),
\]
where $\mathsf{CanDer}_{\Gamma,A}(N_t)$ is the canonical retyping
derivation of Lemma~\ref{lem:canonical-normal-retyping}, and emits a
circuit from $\mathcal N_t$; thus the semantics and the compiler use the
same canonical normal derivation.
At applications and tensor lets, the derivation-wide physical-address
assignment \textup{CarrierPlan} fixes the matched producer and consumer
incidences before recursive emission, so both sides are emitted directly
into the same physical presentation.  Emission recurses on canonical
normal derivations, whose premises are themselves normal, normalizing
and canonically retyping any auxiliary term generated by a clause;
Appendix~\ref{app:compilation-soundness} proves exhaustive coverage.

Physically, the selected semantic boundary $B_{\mathcal N}^{\pm}$ of a
derivation is carried, through the layout maps
$\lambda_{\mathcal N}^{\pm}$, by a code sector $C_{\mathcal N}^{\pm}$
inside a padded register $\mathcal H_{\mathcal N}$, on which the
emitted target program $G_{\mathcal N}$ acts by its gate operator
$\mathcal U(G_{\mathcal N})\in U(\mathcal H_{\mathcal N})$.
Under \textup{(BC)}, the correctness theorem proves that this gate
operator, read through the typed input and output frame coordinates,
maps $C_{\mathcal N}^{-}$ unitarily onto $C_{\mathcal N}^{+}$ and
realizes the boundary denotation
$B_{\mathcal N}^{-}\to B_{\mathcal N}^{+}$.  Two artifacts share the
three stages: the formal theorem verifies the recursive-binary
\emph{reference} emitter of Appendix~\ref{app:compilation-soundness},
whereas the executable prototype uses the flat normalized-frame
lowering of Remark~\ref{rem:ref-vs-exec} and is exercised by the
regression and matrix tests of \S\ref{sec:toolchain}.

\subsection{Wire Layouts}
\label{sec:impl-layout}

The compiler represents boundary spaces using a wire encoding.  Each
type has a canonical \emph{wire layout}---an ordered list of qubits
allocated by the compiler, summarized in
Table~\ref{tab:wire-layouts}.  Write
$\size{T}:=\dim\sem T$, and write $\nwires(T)$ for the number of
physical qubits in the logical interface encoding of~$T$.

\begin{table}[t]
\caption{Wire layouts for the binary core.  A sum has one root-tag
qubit and a payload wide enough for either recursively encoded
summand.}
\label{tab:wire-layouts}
\centering
\footnotesize
\renewcommand{\arraystretch}{1.3}
\begin{tabular}{@{}l@{\quad}l@{\quad}c@{}}
\toprule
\textbf{Type} & $\mathsf{nwires}(T)$ & \textbf{Layout} \\
\midrule
$A \tensor B$
& $\nwires(A) {+} \nwires(B)$
& \raisebox{-4pt}{%
  \begin{tikzpicture}[x=4mm, y=3mm, font=\footnotesize]
    \draw[fill=blue!15] (0,0) rectangle (2.5,1);
    \draw[fill=green!15] (2.5,0) rectangle (5,1);
    \node at (1.25,0.5) {$A$};
    \node at (3.75,0.5) {$B$};
  \end{tikzpicture}}
\\[1.5ex]
$A \plus B$
& $1 {+} \max\{\nwires(A),\nwires(B)\}$
& \raisebox{-4pt}{%
  \begin{tikzpicture}[x=4mm, y=3mm, font=\footnotesize]
    \draw[fill=orange!25] (0,0) rectangle (1.5,1);
    \draw[fill=gray!15] (1.5,0) rectangle (5,1);
    \node at (0.75,0.5) {tag};
    \node at (3.25,0.5) {payload};
  \end{tikzpicture}}
\\[1.5ex]
$\QBool$
& $1$
& \raisebox{-4pt}{%
  \begin{tikzpicture}[x=4mm, y=3mm, font=\footnotesize]
    \draw[fill=orange!25] (0,0) rectangle (1.5,1);
    \node at (0.75,0.5) {$b$};
  \end{tikzpicture}}
\\[1.5ex]
$\base$
& $0$
& (no wires: $\sem{\base} = \mathbb{C}$)
\\[1.5ex]
$A \lmark B$
& $\nwires(A) {+} \nwires(B)$
& \raisebox{-4pt}{%
  \begin{tikzpicture}[x=4mm, y=3mm, font=\footnotesize]
    \draw[fill=red!15] (0,0) rectangle (2.5,1);
    \draw[fill=blue!15] (2.5,0) rectangle (5,1);
    \node at (1.25,0.5) {arg};
    \node at (3.75,0.5) {res};
  \end{tikzpicture}}
\\
\bottomrule
\end{tabular}

\end{table}

For $\tensor$ and $\lmark$ the layout is literal juxtaposition of the
component bundles, so neither constructor compresses data.  Polarity
is not visible in this physical layout; it enters only in the
correctness theorem through the two layout maps
$\lay_{\mathfrak J}^-$ and $\lay_{\mathfrak J}^+$.

For sum types, however, the tag-plus-payload encoding introduces a
valid-subspace issue.  The semantic dimension is additive,
$\size{A \plus B} = \size{A} + \size{B} \leq 2^{\nwires(A \plus B)}$,
while the compiler uses one tag register and a shared payload bundle.
The payload wires are reused across summands: the same physical
payload region is interpreted as an $A$-payload or a $B$-payload
according to the tag.  This does not make the sum classical: valid
states may be coherent superpositions across tag values, such as
$\alpha\,\ket{0}\ket{a} + \beta\,\ket{1}\ket{b}$.

\begin{remark}[Reference vs.\ executable lowering]
\label{rem:ref-vs-exec}
The reference compiler and the correctness theorem operate on the
binary core type: a derived $n$-ary sum uses the fixed left-associated
binary expansion of Appendix~\ref{app:nary-plus}, and its register
layout is obtained recursively from the binary row of
Table~\ref{tab:wire-layouts}.  The executable prototype instead uses a
flat $\lceil\log_2 n\rceil$-bit tag layout, so
\[
 \nwires(A_1\plus\cdots\plus A_n)
 =\lceil\log_2 n\rceil+\max_i\nwires(A_i);
\]
either way, each logical type interface width is a function of the
type alone.  The reference register width is exactly
\[
 Q_{\mathcal N}=\max_s\overline w_s^{\mathcal N},
\]
where $s$ ranges over the derivation's retained contextual semantic
faces (Appendix~\ref{app:compilation-soundness}); it is determined by
the canonical derivation and can exceed the root-interface width when a
prescribed recursive-binary face is wider.
No block or splice adds emitter-owned workspace, a derivation tag, or a
carrier bank.  When a structural atom's recursive-binary source and
target codeword images differ by a genuine basis-word permutation,
planned endpoint transport lowers that permutation exactly; literal
wire-address permutations remain symbolic, so genuine word
re-encodings cost gates but no additional wires.  The prototype lowers
syntactically exposed nested branches directly to \texttt{NPlusMap},
its $n$-ary sum-map node, and may emit $V_a=e^{i\theta J}$ directly at
a first-order quantum endpoint; these optimizations lie outside the
recursive-binary reference theorem.
\end{remark}

The prototype lays out sums in an $\plus$-outermost normalized
logical frame, defined per sum node (tensor nodes are laid out
componentwise, by juxtaposition): a live tag precedes its shared
payload, tensor
components of the selected summand lie inside that payload, and
padding comes last.  In this implementation frame, the multiplicative structural
isomorphisms, additive associativity, distributivity, and their inverses
emit no gates.  Composition through unequal-width distributors is
supported exactly.  At an incident endpoint, a frame mismatch that is only a wire
permutation is absorbed symbolically; a genuine computational-basis word
permutation emits an exact basis-word adapter.  Additive symmetry changes tag values, so binary
$\sigma^\plus$ emits $X$ on the tag and a permutation of more summands
emits the corresponding tag-register permutation.
Appendix~\ref{app:sum-encoding} states the exact invariant.

The surplus states are the physical bit patterns not corresponding to
a well-formed type encoding: nonzero padding in a smaller branch
payload and invalid states inherited from branch layouts.  The type-valid subspace $V_T \subseteq
\mathbb{C}^{2^{\nwires(T)}}$ is the span of the well-formed
encodings, and the layout isomorphisms $\lay_T^-, \lay_T^+$
identify it with the two polarities of the flat type envelope.
A term's canonical normal derivation can select a codeword sector
that is a proper subspace of $V_T$; in particular, branch pairing
correlates indices that the type layout leaves independent.

\subsection{Compiler Correctness}
\label{sec:compiler-correctness}

Fix a Source judgment $\Gamma\sjudge t:A$ of shape
$\mathfrak J$.  For the canonical derivation $\mathcal N_t$ of
$N_t=\mathsf{NF}_{\mathrm{LO}}(t^\circ)$, let
$C_{\mathcal N_t}^\pm$ and $\lambda_{\mathcal N_t}^\pm$ be the
common-register code sectors and restricted layouts of
Definition~\ref{def:layout-iso}
(Appendix~\ref{app:compilation-soundness}); $\lambda_{\mathcal N_t}^\pm$
restricts the flat layout $\lay_{\mathfrak J}^\pm$, itself assembled
from the type layouts $\lay_T^\pm$, to the code sector.  The emitter
returns a target program
$G_{\mathcal N_t}\in\mathsf{TCirc}_{Q_{\mathcal N_t}}$ (the target IR
of Appendix~\ref{app:compilation-soundness})
and total unitary frame extensions $L_{\mathcal N_t}^\pm$ chosen by
the fixed finite-dimensional unitary-extension convention
\textup{(INV-4)} of Appendix~\ref{app:sum-encoding};
its framed operator is
$\Framed{t}:=\Framed{\mathcal N_t}
 =(L_{\mathcal N_t}^{+})^{\dagger}
  \mathcal U_{Q_{\mathcal N_t}}(G_{\mathcal N_t})
  L_{\mathcal N_t}^{-}$,
a unitary on the common padded register.

Here \textup{(BC)} is the external obligation of
\S\ref{subsubsec:unitary-primitives}: declared quantum artifacts and
the target constructors of the reference emitter realize their
stated operators.

\noindent
\begin{minipage}[t]{0.60\linewidth}
\begin{theorem}[Circuit Realization]
\label{thm:compilation-soundness-main}
Under assumption~\textup{(BC)} of
\S\ref{subsubsec:unitary-primitives},
$\mathcal U(G_{\mathcal N_t})$ maps
$L_{\mathcal N_t}^{-}(C_{\mathcal N_t}^{-})$
unitarily onto
$L_{\mathcal N_t}^{+}(C_{\mathcal N_t}^{+})$, and
\[
  \Framed{t}\!\upharpoonright_{C_{\mathcal N_t}^-}
  \;=\;
  (\lambda_{\mathcal N_t}^+)^{-1}
  \SEM{\Gamma\sjudge t:A}\lambda_{\mathcal N_t}^- .
\]
\end{theorem}
\end{minipage}\hfill
\begin{minipage}[t]{0.36\linewidth}
\vspace{0em}
\centering
\footnotesize
$\begin{array}{ccc}
B_{\mathcal N_t}^-
  & \xrightarrow{\;\SEM{\Gamma\sjudge t:A}\;}
  & B_{\mathcal N_t}^+ \\[3pt]
{\scriptstyle\lambda_{\mathcal N_t}^-}\uparrow
  & & \uparrow{\scriptstyle\lambda_{\mathcal N_t}^+} \\[3pt]
C_{\mathcal N_t}^-
  & \xrightarrow{\;\Framed{t}\;}
  & C_{\mathcal N_t}^+
\end{array}$
\end{minipage}

\medskip\noindent
Compiler correctness is therefore a change-of-coordinates statement
between the selected semantic boundary spaces and their physical code
sectors.  The proof is in Appendix~\ref{app:compilation-soundness}.

\paragraph{Qubit-register specialization.}
Write \(\mathsf{QReg}_n=\QBool^{\tensor n}\).  A closed register
endomorphism is compiled by applying it to a symbolic register input,
normalizing, and compiling the resulting endpoint derivation.

\begin{corollary}[Qubit-register execution]
\label{cor:first-order-readback-main}
Assume \textup{(BC)}.  For every closed derivation
\[
  \cdot\sjudge
  t:\mathsf{QReg}_n\lmark\mathsf{QReg}_n,
\]
let \(U_t\) be the unitary of
Corollary~\ref{cor:register-unitarity}.  The reference compiler emits
a target program whose induced action on the canonical data register
is \(U_t\).
\end{corollary}

The precise change of coordinates is given in
Appendix~\ref{app:qubit-execution}.

\subsection{Compilation Rules}
\label{sec:impl-terms}

After LO normalization and canonical retyping, the canonical normal
derivation \(\Gamma\intjudge N:A\) reached from the Source term emits
a physical gate list together with symbolic input and output wire
layouts, and \(\Framed{\mathcal N}\) interprets the gate list between
those layouts.  Layout-only structural maps update the metadata;
additive symmetry instead contributes its tag permutation to the gate
list.  Later gates are placed through the current layout.
Appendix~\ref{app:compilation-soundness} gives the formal artifact
definition.  Table~\ref{tab:compilation-rules} summarizes the clauses.

\begin{table}[t]
\caption{Compilation of the core constructs.  Colored boxes are the
wire bundles of types; left is the context (input wires) and right
the result (output wires).  The five multiplicative constructs are
pure wiring; the $\plus$-map guards $f$ under root tag~$0$
(anti-control) and $g$ under tag~$1$ (control).}
\label{tab:compilation-rules}
\centering
\small
\setlength{\tabcolsep}{6pt}
\begin{tabular}{@{}ccc@{}}
\toprule
\begin{tabular}{@{}c@{}}$x : A \vdash x : A$\\[3pt]
\scalebox{1.4}{\begin{tikzpicture}[x=5mm, y=3mm, font=\tiny]
    \draw[fill=blue!20] (0,0) rectangle (1.5,1);
    \node at (0.75,0.5) {$A$};
    \draw[thick] (1.5,0.5) -- (2.5,0.5);
    \draw[fill=blue!20] (2.5,0) rectangle (4,1);
    \node at (3.25,0.5) {$A$};
  \end{tikzpicture}}\end{tabular}
&
\begin{tabular}{@{}c@{}}$t \tensor u$\\[3pt]
\scalebox{1.4}{\begin{tikzpicture}[x=4mm, y=3mm, font=\tiny]
    \draw[fill=gray!20] (0,1.5) rectangle (1.5,2.5);
    \node at (0.75,2) {$\Gamma_1$};
    \draw[fill=gray!20] (0,0) rectangle (1.5,1);
    \node at (0.75,0.5) {$\Gamma_2$};
    \draw[fill=white] (2,1.5) rectangle (3.5,2.5);
    \node at (2.75,2) {$\mathcal{C}_t$};
    \draw[fill=white] (2,0) rectangle (3.5,1);
    \node at (2.75,0.5) {$\mathcal{C}_u$};
    \draw[thick] (1.5,2) -- (2,2);
    \draw[thick] (1.5,0.5) -- (2,0.5);
    \draw[fill=blue!20] (4,1.5) rectangle (5.2,2.5);
    \node at (4.6,2) {$A$};
    \draw[fill=green!20] (4,0) rectangle (5.2,1);
    \node at (4.6,0.5) {$B$};
    \draw[thick] (3.5,2) -- (4,2);
    \draw[thick] (3.5,0.5) -- (4,0.5);
  \end{tikzpicture}}\end{tabular}
&
\begin{tabular}{@{}c@{}}$\letpair{x}{y}{t}{u}$\\[3pt]
\scalebox{1.4}{\begin{tikzpicture}[x=4mm, y=3mm, font=\tiny]
    \draw[fill=gray!20] (0,1) rectangle (1.2,2);
    \node at (0.6,1.5) {$\Gamma_1$};
    \draw[fill=white] (1.5,1) rectangle (2.7,2);
    \node at (2.1,1.5) {$\mathcal{C}_t$};
    \draw[thick] (1.2,1.5) -- (1.5,1.5);
    \draw[fill=blue!20] (3,1.5) rectangle (3.8,2);
    \node at (3.4,1.75) {$A$};
    \draw[fill=green!20] (3,1) rectangle (3.8,1.5);
    \node at (3.4,1.25) {$B$};
    \draw[thick] (2.7,1.7) -- (3,1.75);
    \draw[thick] (2.7,1.3) -- (3,1.25);
    \draw[fill=gray!20] (3,0) rectangle (3.8,0.7);
    \node at (3.4,0.35) {$\Gamma_2$};
    \draw[fill=white] (4.3,0) rectangle (5.5,2);
    \node at (4.9,1) {$\mathcal{C}_u$};
    \draw[thick] (3.8,1.75) -- (4.3,1.7);
    \draw[thick] (3.8,1.25) -- (4.3,1.3);
    \draw[thick] (3.8,0.35) -- (4.3,0.35);
    \draw[fill=purple!20] (5.8,0.5) rectangle (6.6,1.5);
    \node at (6.2,1) {$C$};
    \draw[thick] (5.5,1) -- (5.8,1);
  \end{tikzpicture}}\end{tabular}
\\[6pt]
\begin{tabular}{@{}c@{}}$\lambda x.\, t : A \lmark B$\\[3pt]
\scalebox{1.4}{\begin{tikzpicture}[x=4mm, y=3mm, font=\tiny]
    \draw[fill=gray!20] (0,0) rectangle (1.2,1);
    \node at (0.6,0.5) {$\Gamma$};
    \node[left] at (0,1.9) {\tiny$(x)$};
    \draw[thick] (0,1.9) -- (1.8,1.2);
    \draw[fill=white] (1.8,0) rectangle (3.2,1.5);
    \node at (2.5,0.75) {$\mathcal{C}_t$};
    \draw[thick] (1.2,0.5) -- (1.8,0.5);
    \draw[fill=blue!20] (3.8,0) rectangle (4.8,1);
    \node at (4.3,0.5) {$B$};
    \draw[thick] (3.2,0.5) -- (3.8,0.5);
    \draw[fill=red!20] (3.8,1.4) rectangle (4.8,2.2);
    \node at (4.3,1.8) {$A$};
    \draw[thick, rounded corners=3pt] (0,1.9) -- (-0.3,1.9) -- (-0.3,2.8) -- (4.3,2.8) -- (4.3,2.2);
  \end{tikzpicture}}\end{tabular}
&
\begin{tabular}{@{}c@{}}$f\, u$\\[3pt]
\scalebox{1.4}{\begin{tikzpicture}[x=4mm, y=3mm, font=\tiny]
    \draw[fill=gray!20] (0,1.8) rectangle (1.2,2.6);
    \node at (0.6,2.2) {$\Gamma_1$};
    \draw[fill=gray!20] (0,0) rectangle (1.2,0.8);
    \node at (0.6,0.4) {$\Gamma_2$};
    \draw[fill=white] (1.6,1.8) rectangle (2.8,2.6);
    \node at (2.2,2.2) {$\mathcal{C}_f$};
    \draw[thick] (1.2,2.2) -- (1.6,2.2);
    \draw[fill=white] (1.6,0) rectangle (2.8,0.8);
    \node at (2.2,0.4) {$\mathcal{C}_u$};
    \draw[thick] (1.2,0.4) -- (1.6,0.4);
    \draw[fill=red!20] (3.2,2.3) rectangle (4,2.8);
    \node at (3.6,2.55) {$A$};
    \draw[fill=blue!20] (3.2,1.6) rectangle (4,2.1);
    \node at (3.6,1.85) {$B$};
    \draw[thick] (2.8,2.4) -- (3.2,2.55);
    \draw[thick] (2.8,2.0) -- (3.2,1.85);
    \draw[fill=red!20] (3.2,0.2) rectangle (4,0.7);
    \node at (3.6,0.45) {$A$};
    \draw[thick] (2.8,0.4) -- (3.2,0.45);
    \draw[thick, rounded corners=3pt] (4,0.45) -- (4.8,0.45) -- (4.8,2.55) -- (4,2.55);
    \draw[thick] (4,1.85) -- (5.4,1.85);
  \end{tikzpicture}}\end{tabular}
&
\begin{tabular}{@{}c@{}}$f \plus g : A \plus B \lmark C \plus D$\\[3pt]
\scalebox{1.4}{\begin{tikzpicture}[x=5mm, y=4mm, font=\tiny]
    \draw[fill=orange!25] (0,1.6) rectangle (0.8,2.3);
    \node at (0.4,1.95) {tag};
    \draw[fill=blue!12] (0,0) rectangle (0.8,1);
    \node at (0.4,0.5) {$A/B$};
    \draw[thick] (0.8,1.95) -- (6,1.95);
    \draw[thick] (0.8,0.5) -- (6,0.5);
    \draw[fill=white] (2,0.25) rectangle (3,0.75);
    \node at (2.5,0.5) {$\mathcal{C}_f$};
    \fill (2.5,1.95) circle (1pt);
    \draw[thick] (2.5,1.95) -- (2.5,0.75);
    \node[above] at (2.5,2.05) {\tiny$0$};
    \draw[fill=white] (4,0.25) rectangle (5,0.75);
    \node at (4.5,0.5) {$\mathcal{C}_g$};
    \fill (4.5,1.95) circle (1pt);
    \draw[thick] (4.5,1.95) -- (4.5,0.75);
    \node[above] at (4.5,2.05) {\tiny$1$};
    \draw[fill=orange!25] (6,1.6) rectangle (6.8,2.3);
    \node at (6.4,1.95) {tag};
    \draw[fill=blue!12] (6,0) rectangle (6.8,1);
    \node at (6.4,0.5) {$C/D$};
  \end{tikzpicture}}\end{tabular}
\\[2pt]
\bottomrule
\end{tabular}

\end{table}

\paragraph{Multiplicatives.}
Multiplicative clauses add neither gates nor work\-space: variables
re\-target their planned ranges; tensor introduction emits its premises on
disjoint planned ranges; abstraction re-exports the argument range on
the output boundary; and tensor elimination and application emit
producer and consumer premises into the common midpoint fixed by
\textup{CarrierPlan}.  Thus only the premise circuits contribute
gates, and multiplicative equations are circuit-invisible: the fixed
leftmost--outermost normalizer contracts $\beta$-redexes before
emission, so presentations with the same canonical normal form emit
the same topology.

\paragraph{Sum constructs as coherent control.}
\label{sec:impl-case}
Sum maps and Source cases compile branchwise programs into coherent
controlled gates.  A Source case expands to the
composite of a distributor, a $\plus$-map, and an inverse
distributor (\S\ref{cohRouting}); the $\plus$-map row of
Table~\ref{tab:compilation-rules} supplies its controlled branch
computation.  The root tag coherently controls the two branch
circuits: the left branch fires under tag $\ket{0}$, implemented by an
$X$-sandwich ($X$; controlled-$G_f$; $X$), and the right branch under
standard control, so a tag in superposition executes both branches
coherently.  The payload bundle is shared and tag-dependent: the
labels $A/B$ and $C/D$ in the diagram are two interpretations of one
physical bundle, not simultaneous disjoint wires, and unequal branch
widths use recorded padding within it
(Appendix~\ref{app:sum-encoding} states the root-tag and padding
invariant).  Gates and static phases act branch-locally, inside the
selected tag block, while context ports of the other branch bypass
them; the $X$-sandwich toggles the tag only transiently, so the whole
circuit is the identity on the tag register.

\paragraph{Coherent sharing.}
When both branches use a shared component, CarrierPlan first gives the
component and its two branch restrictions one parent address in the
common tagged payload.  Both branches are then emitted directly into
those inherited presentations, and the component is connected once
through the single planned midpoint at the whole sum port.
Appendix~\ref{app:compilation-soundness} gives the direct recursion and
its well-foundedness proof.

\paragraph{$n$-ary sums.}
For $n\geq2$ the reference emitter recursively emits the binary case
rule along the fixed left-associated elaboration, testing one root tag
at each binary node (for $n=1$ the derived case is direct application,
with neither tag nor control), whereas the executable prototype uses
one flat $\lceil\log_2 n\rceil$-bit equality control
(Remark~\ref{rem:ref-vs-exec}); Appendix~\ref{app:sum-encoding}
verifies the valid-subspace invariant for the reference construction
and for each structural combinator.

\subsection{Running Example: Compiling QSwitch}
\label{sec:impl-qswitch}

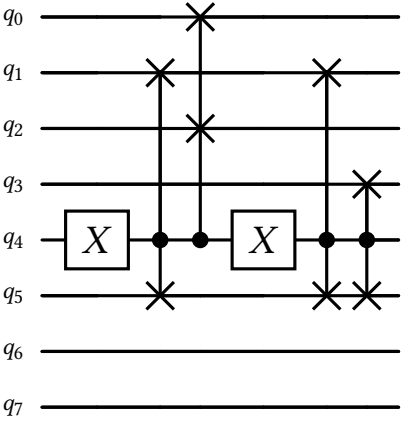
\begin{figure}[t]
\begin{minipage}[c]{0.40\columnwidth}
\centering
\resizebox{\linewidth}{!}{%
\begin{quantikz}[column sep=6pt, row sep={0.5cm,between origins}]
\lstick{\tiny$q_0$} & \qw      & \qw       & \swap{2}  & \qw      & \qw       & \qw       & \qw \\
\lstick{\tiny$q_1$} & \qw      & \swap{4}  & \qw       & \qw      & \swap{4}  & \qw       & \qw \\
\lstick{\tiny$q_2$} & \qw      & \qw       & \targX{}  & \qw      & \qw       & \qw       & \qw \\
\lstick{\tiny$q_3$} & \qw      & \qw       & \qw       & \qw      & \qw       & \swap{2}  & \qw \\
\lstick{\tiny$q_4$} & \gate{X} & \ctrl{-3} & \ctrl{-2} & \gate{X} & \ctrl{-3} & \ctrl{-1} & \qw \\
\lstick{\tiny$q_5$} & \qw      & \targX{}  & \qw       & \qw      & \targX{}  & \targX{}  & \qw \\
\lstick{\tiny$q_6$} & \qw      & \qw       & \qw       & \qw      & \qw       & \qw       & \qw \\
\lstick{\tiny$q_7$} & \qw      & \qw       & \qw       & \qw      & \qw       & \qw       & \qw
\end{quantikz}%
}
\end{minipage}\hfill
\begin{minipage}[c]{0.57\columnwidth}
\caption{Eight-wire abstract $\qSwitch$ artifact.  Immediately before
the coherent case, the recorded placement assigns $q_0,q_1$
(argument, result) to the boundary of $f:A\lmark A$, $q_2,q_3$
(argument, result) to the boundary of $g:A\lmark A$, $q_4$ to the Boolean tag $b$, and $q_5$ to the payload
$x$.  Wires $q_6,q_7$ are the result side of the surrounding
curried-function boundary; they are gate-idle because the final
structural movement is retained as the pending output permutation
$\pi_{\mathrm{out}}^{\mathrm{new}\to\mathrm{old}}=(6,7,0,1,2,3,4,5)$,
and they are not auxiliary wires.  No additional function-layout
wires are allocated; all eight wires belong to the semantic boundary.
The $X$-sandwich on $q_4$ selects the $b=0$ sector for the first two
Fredkin gates, which splice the payload through $g$ and then $f$; the
last two Fredkin gates select the opposite order.}
\label{fig:qswitch-emitted}
\end{minipage}
\end{figure}

With $A=\QBool$, the PPX-elaborated abstract quantum switch has type
$(A\lmark A)\lmark(A\lmark A)\lmark
 \bigl((\QBool\otimes A)\lmark(\QBool\otimes A)\bigr)$
and compiles to the eight-wire circuit of
Figure~\ref{fig:qswitch-emitted}, whose caption records the wire
assignment; the operations supplied for $f$ and $g$ occupy the
application boundaries routed by its Fredkin gates.

%% file: datatypes-new.tex
\section{Finite Control Datatypes}
\label{sec:datatypes}

Granthi's finite label types and staged reversible-operation bindings
elaborate to canonical $n$-ary sums in the coherent sum-and-wiring core
(\S\S\ref{sec:core-language}--\ref{sec:compilation}).
Their metatheoretic and semantic properties follow from the
elaboration theorems of Appendix~\ref{app:elaboration-proofs} under
the corresponding core hypotheses~\citep{pfpl}.

\paragraph{Surface syntax and typing.}
\label{subsec:surface-datatypes}
Assume a datatype environment $\mathcal D$ mapping each datatype name
to an ordered nonempty list
$(l_0,\ldots,l_{n-1})$ of pairwise-distinct labels; thus every
declaration below has $n\geq1$.  Resolved surface identifiers for
closed reversible terms are recorded only in an elaboration environment
$\Sigma$.  They are staging nodes, not Source terms: elaboration
replaces each occurrence directly by its assigned closed Source term.
\[
\begin{array}{l}
\mathbf{datatype}\ \NAME\ \mathbf{where} \\
\qquad \mathbf{labels}\ l_0 \mid \cdots \mid l_{n-1}.
\end{array}
\]
Labels are \emph{not constructors}: there are no injections
$l_i : \base \to \NAME$, since they are irreversible.
Values of~$\NAME$ arise only from inputs or from closed reversible
terms supplied at staging.  Surface identifiers are fully qualified
(equivalently, resolution is recorded in the typechecked surface
syntax).  The surface language is the Source language of
\S\ref{subsubsec:source-language}, extended with these staging nodes
and datatype cases.  Every staging node is eliminated before a Source
term is produced.  Derived forms expand before elaboration.

The complete surface typing rules and canonical elaboration appear in
Appendix~\ref{app:elaboration-proofs}.  In outline, a datatype with
$n$ labels elaborates to $\Qn n$; a resolved staging identifier $f$ is
replaced by the closed Source term $\Sigma(f)$; and an $n$-ary datatype
case elaborates to the fixed Source form $\mathsf{TagCase}_n$.  We
write $\ElabTy S$ and $\mathsf{Elab}_\Sigma(e)$ for the resulting
translations.

\paragraph{Metatheory.}
The meaning of a surface program $e$ is the canonical-form meaning of
the Raw expansion $(\mathsf{Elab}_\Sigma(e))^\circ$: elaborate to
Source, expand to Raw, and then apply $\mathsf{NF}_{\mathrm{LO}}$.
Under a well-formed environment $\Sigma$, elaboration is a typed
function
(Appendix~\ref{app:elaboration-proofs},
Theorem~\ref{thm:elab-sound-full}, Corollary~\ref{thm:elab-unique-full},
and Corollary~\ref{cor:surface-transfer}).
Normalization (Theorem~\ref{thm:normalization}), determinacy
(Theorem~\ref{thm:determinacy}), and boundary unitarity
(Theorem~\ref{thm:boundary-unitarity}) are therefore inherited by
well-typed surface programs; compilation soundness
(Theorem~\ref{thm:compilation-soundness-main}) is inherited under
backend correctness~\textup{(BC)}.

\paragraph{Staging perspective.}
\label{subsec:staging}
Staging is unrestricted compile-time program construction, admitted
exactly when its elaboration is
Source-well-typed~\citep{TahaSheard2000,DaviesPfenning2001}.

%% file: toolchain.tex
\section{Implementation and Validation}
\label{sec:toolchain}

Granthi's currently supported executable fragment is implemented
end-to-end.  The OCaml PPX elaborates Source notation into a sealed
(abstract-interface), intrinsically typed linear GADT, exports a higher-order JSON Core IR,
and invokes a Python backend that maintains symbolic layouts and emits
\texttt{pytket} circuits.  The GADT enforces linearity, context
splitting, and the first-order restrictions; the frontend exposes
Source rather than Raw sum formation.  Wire permutations remain
metadata, while genuine codeword mismatches emit exact basis-word
adapters as described in \S\ref{sec:compilation}.  The prototype
implements the flat lowering of Remark~\ref{rem:ref-vs-exec} rather
than the reference emitter verified in
Appendix~\ref{app:compilation-soundness}.
The prototype was developed with LLM coding assistance.

\paragraph{Validation.}
\label{subsec:impl-validation}
Unitary equalities use the artifact's numerical comparison of compiled
unitaries up to global phase; golden tests are byte-for-byte
regressions against committed outputs; and compile pins check widths,
operation counts, and boundary permutations without claiming semantic
equality (materialized mode emits pending permutations as swap gates;
meta-level references are hand-constructed circuit-level terms).
Operation counts are pre-decomposition \texttt{pytket}
operations and exclude pending \texttt{WirePerm} metadata; all listed
checks pass at v1.0.0 (1570 backend tests, the 33-row, 119-check
counterpart ledger, and 31 goldens).
Table~\ref{tab:validation} lists the evidence by example family.

\begin{table}[t]
\caption{Validation evidence in the artifact (v1.0.0).}
\label{tab:validation}
\footnotesize
\setlength{\tabcolsep}{4pt}
\renewcommand{\arraystretch}{1.1}
\begin{tabular}{@{}p{0.27\linewidth}p{0.46\linewidth}p{0.23\linewidth}@{}}
\toprule
\textbf{Example family} & \textbf{Validation evidence} &
\textbf{Representative result}\\
\midrule
\textbf{Quantum switch.}  Abstract
(Figure~\ref{fig:qswitch-emitted}); partial and closed
specializations; $H/S$, $X/Z$, $H/Y$, $H/H$, $R_z/R_z$; two- and
three-fold sequential compositions
& Unitary equality with hand-written meta-level and Raw references;
widths pinned at $8/6/4$ (abstract/partial/closed) in symbolic and
materialized modes; goldens
& Closed $H,S$: four-wire frame, two active wires, six operations
before permutation materialization\\
\addlinespace
\textbf{Phase-marked short-circuit routing}
(\S\ref{subsubsec:quantum-plus-map})
& Source/Raw unitary equality for $\mathsf{toggle}_{\Wit}$,
$\mathsf{and}_{\mathrm{sc}}$, and $\mathsf{phase}_{\Wit}$;
$\mathsf{toggle}_{\Wit}^2=\mathsf{and}_{\mathrm{sc}}^2
 =\mathsf{phase}_{\Wit}^2=I$; goldens
& $\mathsf{route}^{\mathrm q}_{\Wit}$ followed by the unphased
$\mathsf{route}_{\Wit}$, and $\mathsf{and}^{\mathrm q}_{\mathrm{sc}}$:
13 operations each\\
\addlinespace
\textbf{Finite control.}  $\mathsf{select}_n$, $2\le n\le5$;
$\mathbb Z_n$ shift, negation, and addition, $2\le n\le11$; phase
kick, $n\in\{2,4,5,8\}$
& Raw unitary equality for all selectors and shifts and for
representative negations and additions; larger additions
compile-pinned; $\mathsf{shift}^5=I$ ($\mathbb Z_5$),
$\mathsf{neg};\mathsf{shift};\mathsf{neg}=\mathsf{shift}^{-1}$
($\mathbb Z_8$), $\mathsf{kick};\mathsf{kick}^{-1}=I$
& $\mathbb Z_5$ selection: 25 operations; $\mathbb Z_8$ phased
dispatch: 32\\
\addlinespace
\textbf{Nested control and composition.}  $C^k(H)$, $1\le k\le3$;
$\mathsf{compose}_k$, $2\le k\le4$
& $C^k(H)$ against NumPy-built reference unitaries; Source
counterparts against those circuits; $\mathsf{compose}_k$ against
sequential Raw composition
& ---\\
\addlinespace
\textbf{Certified exponentials.}  $\expi{\theta}{\mathsf{twist}}$,
$\expi{\theta}{\mathsf{swap}_{ij}}$
& Source/Raw unitary equality;
$\expi{\pi/4}{J}^2=\expi{\pi/2}{J}$ for
$J\in\{\mathsf{twist},\mathsf{swap}_{12},\mathsf{swap}_{23}\}$;
$\expi{\theta}{\mathsf{swap}_{12}}$ and
$\expi{\theta}{\mathsf{swap}_{23}}$ verified non-commuting
& ---\\
\addlinespace
\textbf{Algorithm kernels}
& Deutsch--Jozsa (constant, balanced) and the two-qubit HSP core:
Source/Raw unitary equality; Simon, Bell, GHZ: goldens only
& ---\\
\addlinespace
\textbf{Fixed-order $n$-switch simulators}, $n=2,3$
& $n=2$: controlled-SWAP, one round, and the two-round simulator
Raw-equal; $n=3$: Raw demo equal to a direct semantic reference,
Source counterpart component-equal and compile-pinned
& $n=3$ simulator: 14 boundary wires; 201 operations (retained Raw
demo)\\
\addlinespace
\textbf{Source typing restrictions}
& Exact located diagnostics: higher-order case summands and results,
unequal branch contexts, higher-order distributors, dropped tensor
components, and related violations
& 23 PPX reject fixtures and 25 sealed-Source conformance reject
fixtures\\
\bottomrule
\end{tabular}

\end{table}

\paragraph{Limitations.}
\label{subsec:impl-limits}
The prototype prioritizes exact lowering over circuit optimization.
It applies no gate-level simplification, leaves ordinary peephole and
reordering passes to downstream tools, and does not yet cancel every
adjacent pair of basis-word adapters.  It automatically flattens
syntactically decomposable nested $\plus$-maps; other cases take the
completed-block (Appendix~\ref{app:compilation-soundness}) or generic
controlled path.
The remaining coverage restrictions concern dense synthesis.
Certified exponentials and the asymmetric block-diagonal fallback are
limited to a total width of at most three qubits, and tag
permutations that reach the dense fallback support at most eight
summands (three tag qubits), because \texttt{pytket} provides only
\texttt{Unitary1qBox}, \texttt{Unitary2qBox}, and
\texttt{Unitary3qBox}.  A closed block requiring dense synthesis
between noncoincident source and target frames is rejected before
emission.  Explicit permutation (including the $\mathbb Z_n$
operations of Table~\ref{tab:validation}), completed open-block, and
controlled-dispatch paths do not use that dense fallback.  These are
limitations of the executable prototype, not of the reference compiler
or its correctness theorem.

%% file: expressiveness.tex
\section{Related Work}
\label{sec:expressiveness}

\subsection{Classical Control with First-Order Quantum Data}
\label{subsec:rw-classical-control}

Many of the systems below share one architectural commitment: higher-order
structure remains classical while quantum data is first-order.

Quipper~\citep{Green2013Quipper} embeds circuit construction in
Haskell; Q\#~\citep{Svore2018QSharp} supports classical higher-order
programming but restricts qubits to first-order resources;
QML~\citep{AltenkirchGrattage2005} is first-order: its finite quantum
types include sums and products but no function types.
Proto-Quipper-M~\citep{RiosSelinger2017} has general linear function
types, while boxed circuits have tensor-generated simple $M$-type
interfaces.  The quantum
$\lambda$-calculi of Selinger and
Valiron~\citep{SelingerValiron2009} provide a
foundation for reasoning about quantum data in a higher-order setting.

More recent work enriches this model without changing its architecture.
Qunity~\citep{VoichickLiRandHicks2023} provides a unified syntax where
classical constructs have both quantum and classical effects.
Twist~\citep{YuanMcNallyCarbin2022Twist} contributes a type system for
purity and entanglement.
\citet{HeunenLemonnierMcNallyRice2026} introduce structured constructs
for generating unitaries from phase operations and pattern matching.
\citet{PaykinWinnick2025} develop $\lambda\mathrm{P}_c$, a typed
calculus for projective Cliffords with additive biproducts (direct
sums) and linear higher-order structure; while restricted to Clifford
operations, its
use of additive structure and higher-order types is closely related to
ours.

A second difference concerns where circuit construction happens.  In
the host-language tradition, higher-order classical code builds
first-order circuits as data: given concrete operations
$f,g:A\to A$, these languages can build and package the corresponding
switched circuit.  Granthi instead places higher-order structure in
the object language, so the switch itself is a source-level quantum
component before $f$ and $g$ are fixed, with a boundary denotation
independent of any circuit encoding.

Qunity is the closest comparison.  Its sum types denote direct
sums of Hilbert spaces, and its control constructs support coherent
tag-dependent routing: for eligible non-overlapping classical
patterns, Qunity's orthogonality judgment certifies orthogonal basis
states (closed patterns) or orthogonal subspaces (patterns with
variables).  Granthi instead exposes the branch decomposition as
$\plus$ itself (\S\ref{subsec:coherent-choice},
\S\ref{cohRouting}), so branch orthogonality is automatic, and places
the resulting routing inside a higher-order linear language.

\subsection{Higher-Order Quantum Control}
\label{subsec:rw-ho-quantum}

Programming-language accounts of higher-order quantum control remain
sparse, especially accounts combining syntax, typing, unitary semantics,
and circuit compilation.
Quantum supermaps and networks provide a general framework for
higher-order quantum
transformations~\citep{Chiribella2008Supermaps,Chiribella2009}.  The
quantum switch realizes coherent control over composition
order~\citep{Chiribella2013Switch}, and its $n$-ary form yields query
advantages for suitable tasks~\citep{AraujoCostaBrukner2014}.
\citet{DiazCaroMalherbe2022} develop $\lambda_{S_1}$, a typed
calculus for quantum control in the unitary sphere; closed abstractions
at the qubit endomorphism type are exactly isometries, while no
analogous characterization is stated for arbitrary higher-order
interfaces.

\citet{YuanVillanyiCarbin2024} show that direct lifts of classical
control-flow abstractions such as conditional jump are invalid in
general, due to disruptive entanglement between the program counter
and data.
Granthi avoids this by placing quantum structure \emph{at higher types}
rather than lifting classical control flow: the controlled-order
term of \S\ref{subsec:quantum-switch} arises without an operational
program counter; after its component operations are supplied, a
closed instance compiles to a static circuit.

\subsection{Categorical Semantics and Reversible Languages}
\label{subsec:rw-semantics}

Categorical quantum
mechanics~\citep{AbramskyCoecke2004,CoeckeKissinger2017,Selinger2007CPM}
(see, e.g., \citet{HeunenVicary2019} for a textbook reference)
provides the semantic vocabulary of dagger compact closed categories
and the CPM construction.  The $\Pi$ family of reversible languages~\citep{Carette2024Pi} provides a
complete equational theory for first-order unitary computation over
finite types; subsequent
work~\citep{CaretteHeunenKaarsgaardSabry2024Bake} extends $\Pi$ to
measurement and mixed states while remaining first-order.

Our companion categorical work develops a boundary-centric semantics
for higher-order quantum computation from Kelly--Laplaza
linkings~\citep{KellyLaplaza1980} and Abramsky's execution
account~\citep{AbramskyJagadeesan2026EU}.  It defines a categorical
model $\mathsf{QC}$ whose morphisms are essentially unitary,
realizing the coherent quantum switch and the unitary stages of
equal-ratio one-slot supermap dilations with explicit memory.
Granthi gives the complementary programming-language account: a
typed source calculus, deterministic normalization,
derivation-indexed Hilbert-space boundaries, concrete circuit
layouts, and compilation to executable circuits.

\subsection{Causality and First-Order Restrictions}

Granthi's first-order restrictions on Source sums and case
(\S\ref{subsubsec:source-language},
Remark~\ref{rem:why-first-order}) prevent a branch-selected function
payload from being applied to its own branch tag.  Kengo Hirata and
Takeshi Tsukada pointed out this failure mode to us (personal
communication), using the central counterexample of their LICS
paper~\citep{HirataTsukada2026LICS}.  Their calculus with
higher-order functions and quantum conditional branching enforces
causality through intuitionistic BV logic; Granthi uses the local
syntactic discipline above.

Concurrent work by Barsse, P\'echoux, and
Perdrix~\citep{BarssePechouxPerdrix2026}, posted to arXiv in July
2026, develops a linear higher-order language for coherent control
and indefinite causal order.  Granthi explores a related purely
unitary design point: their controlled outputs are qubit registers,
whereas Granthi's first-order sums provide finite-dimensional coherent
data and control spaces.  Their language additionally
supports measurement and arbitrary quantum channels, with synchronized
operational semantics and denotational semantics in
$\mathbf{Caus}[\mathbf{CPM}]$.  Granthi complements this semantic
breadth with normalization and an implemented compiler to static
unitary circuits.

%% file: app-focused-rules.tex
\section{Eta-long Focused Linear Core}
\label{app:focused-rules}

The Raw judgment $\Gamma\vdash_{\mathsf r}t:A$ is the
natural-deduction system of Table~\ref{tab:typing-rules-linear},
including its quantum extension in \S\ref{sec:quantum-core}.
Reduction is typed in a separate internal administrative judgment
$\Gamma \intjudge t:A$.  It has the same multiplicative rules and
atoms, together with the following phase-bearing additive rules:
\begingroup\small
\[
\inferrule{
  \Gamma \intjudge t:A
  \\
  \Delta \intjudge u:B
  \\
  A,B \text{ first-order}
  \\
  \alpha,\beta\in\mathbb{C}_{\mathrm{static}},
  \quad |\alpha|=|\beta|=1
}{
  \Gamma,\Delta \intjudge
  [\,\alpha\cdot t\mid\beta\cdot u\,]:A\plus B
}
\quad\textsc{$\plus$-I${}_{\mathsf{i}}$}
\]
\[
\inferrule{
  \Gamma \intjudge f:A\lmark C
  \\
  \Delta \intjudge g:B\lmark D
  \\
  C,D \text{ first-order}
  \\
  \alpha,\beta\in\mathbb{C}_{\mathrm{static}},
  \quad |\alpha|=|\beta|=1
}{
  \Gamma,\Delta \intjudge
  \oplusmap{\alpha}{f}{\beta}{g}:
  (A\plus B)\lmark(C\plus D)
}
\quad\textsc{$\plus$-Map${}_{\mathsf{i}}$}.
\]
\endgroup
Contexts displayed side by side are disjoint.  The first rule has
arbitrary term premises: the value restriction belongs only to Raw
$\plus$-introduction.  The map sources $A,B$ are
deliberately unrestricted; only its targets $C,D$ are first-order.
The unit-modulus phase premises are part of both internal rules.

\begin{lemma}[Source expansion is internal]
\label{lem:source-internal-inclusion}
If $\Gamma \sjudge t:A$, then $\Gamma \intjudge t^\circ:A$.
\end{lemma}

\begin{proof}
By induction on the Source derivation.  The multiplicative and atomic
cases use the corresponding internal rules.  In the case clause of
the Raw expansion, the two branch maps are closed, their source types
are unrestricted, and their targets $A\tensor C$ and $B\tensor C$ are
first-order.  Hence the distributor--map--inverse-distributor composite
of \S\ref{cohRouting} is internally typed.
\end{proof}

Table~\ref{tab:focused-rules} is an eta-long
\emph{witness calculus} inside the internal judgment.  It exposes
type-directed identities and the local function operand used by the
unitarity proof, while normalization is carried out in $\intjudge$.
Each connective has one right rule and one left rule; the primitive
axioms are the atomic axiom and the axiom at $\eta$-atomic sums.
A context-cut rule supports composition.  An $\eta$-atomic sum has a
non-first-order summand, as defined in the eta-expansion subsection.
Later arguments use only these explicit witnesses; focused
completeness is not an assumption of the development.

\begin{table*}[!ht]
\caption{Eta-long focused presentation of the linear core.  Contexts
in binary rules are linearly split.  The $\plus$-left rule is
monoidal: it routes the two summands to possibly different result
types and returns $C \plus D$.  The $\plus R$ rule is the internal
first-order refinement of Raw $\plus$-introduction and retains its
formation-value branches $W_1,W_2$.}
\label{tab:focused-rules}
\centering
\small
\renewcommand{\arraystretch}{1.45}
\begin{tabular}{@{}l@{\qquad}c@{}}
\toprule
\multicolumn{2}{@{}l}{\emph{Eta-long focused linear core}}\\
\midrule
\textsc{Ax$_{\base}$}
&
$\inferrule{ }{x:\base \intjudge x:\base}$
\\[2.5ex]
\textsc{Ax$_{\plus}$}
&
$\inferrule{A \plus B \text{ has a non-first-order summand}}
  {x:A \plus B \intjudge x:A \plus B}$
\\[2.5ex]
\textsc{Cut}
&
$\inferrule{
  \Gamma \intjudge u:A
  \\
  \Delta,x:A \intjudge t:B
}{
  \Gamma,\Delta \intjudge t[u/x]:B
}$
\\[3ex]
\textsc{$\lmark R$}
&
$\inferrule{
  \Gamma,x:A \intjudge t:B
}{
  \Gamma \intjudge \lam{x}{t}:A\lmark B
}$
\\[3ex]
\textsc{$\lmark L$}
&
$\inferrule{
  \Gamma \intjudge u:A
  \\
  \Delta,y:B \intjudge n:C
}{
  \Gamma,\Delta,f:A\lmark B
  \intjudge n[f\,u/y]:C
}$
\\[3ex]
\textsc{$\tensor R$}
&
$\inferrule{
  \Gamma \intjudge t:A
  \\
  \Delta \intjudge u:B
}{
  \Gamma,\Delta \intjudge t\tensor u:A\tensor B
}$
\\[3ex]
\textsc{$\tensor L$}
&
$\inferrule{
  \Gamma,x:A,y:B \intjudge n:C
}{
  \Gamma,p:A\tensor B
  \intjudge \letpair{x}{y}{p}{n}:C
}$
\\[3ex]
\textsc{$\plus R$}
&
$\inferrule{
  \Gamma \intjudge W_1:A
  \\
  \Delta \intjudge W_2:B
  \\
  A, B \text{ first-order}
}{
  \Gamma,\Delta \intjudge [\,W_1\mid W_2\,]:A\plus B
}$
\\[3ex]
\textsc{$\plus L_{\mathrm{mon}}$}
&
$\inferrule{
  \Gamma,x:A \intjudge n:C
  \\
  \Delta,y:B \intjudge m:D
  \\
  C, D \text{ first-order}
}{
  \Gamma,\Delta,s:A\plus B
  \intjudge
  \bigl(\oplusmap{1}{\lam{x}{n}}{1}{\lam{y}{m}}\bigr)\,s
  : C\plus D
}$
\\[2ex]
\bottomrule
\end{tabular}

\end{table*}

\begin{lemma}[Internal substitution]
\label{lem:internal-substitution}
Let the variable domains of $\Gamma$ and $\Delta$ be disjoint,
with $x$ fresh for $\Delta$.  If
$\Gamma,x:A \intjudge t:B$ and
$\Delta \intjudge u:A$, then
$\Gamma,\Delta \intjudge t[u/x]:B$.
\end{lemma}
\begin{proof}
By induction on the derivation of
$\Gamma,x:A \intjudge t:B$.  In the variable case, linearity
forces $t=x$, and the conclusion is the second premise.  In every
rule with split premises, the unique occurrence of $x$ lies in
exactly one premise; apply the induction hypothesis there and
rebuild the rule with all other premises unchanged.  This covers
application, tensor formation and elimination, abstraction, and
$\plus$-Map.  For
\textsc{$\plus$-I${}_{\mathsf{i}}$}, the arbitrary-term premises
are essential: after applying the induction hypothesis to the
unique branch containing $x$, the same internal sum rule rebuilds
the conclusion.  Substitution changes neither the branch types nor
the static phases, so the first-order and unit-modulus premises are
preserved.  The phased-map case is identical: its target types and
phases are unchanged, and its source types remain unrestricted.
Atomic operations have no free linear term variables.  Capture is
avoided by alpha-renaming binders before the induction step.
\end{proof}

\subsection*{Eta-expansion}

For every type $A$, define $\eta_A(x)$ by induction:
\[
\begin{array}{rcl}
\eta_{\base}(x) &=& x,\\[0.5ex]
\eta_{A\tensor B}(p) &=& \letpair{x}{y}{p}{\eta_A(x)\tensor\eta_B(y)},\\[0.5ex]
\eta_{A\plus B}(s) &=& \bigl(
  \oplusmap{1}{\lam{x}{\eta_A(x)}}{1}{\lam{y}{\eta_B(y)}}
\bigr)\,s
  \quad (A, B \text{ first-order}),\\[0.2ex]
\eta_{A\plus B}(s) &=& s
  \quad (\text{otherwise: such sums are $\eta$-atomic}),\\[0.5ex]
\eta_{A\lmark B}(f) &=& \lam{x}{\eta_B(f\,\eta_A(x))}.
\end{array}
\]
Each compound case is read through the focused left rule
introducing the principal connective.

\begin{lemma}[Eta identity admissibility]
\label{lem:eta-identity-admissible}
For every type $A$, the judgment $x:A \intjudge \eta_A(x):A$ is
derivable in the eta-long focused presentation.  Hence the
arbitrary identity rule $x:A \intjudge x:A$ is admissible in
eta-expanded form.
\end{lemma}

\begin{proof}
By induction on $A$.

\emph{Base.}  $\eta_{\base}(x) = x$ via $\mathrm{Ax}_{\base}$.

\emph{Tensor.}  Apply $\tensor R$ to the IH judgments
$x:A \intjudge \eta_A(x):A$ and $y:B \intjudge \eta_B(y):B$,
obtaining
$x:A, y:B \intjudge \eta_A(x) \tensor \eta_B(y) : A \tensor B$.
Then apply $\tensor L$ to derive
$p:A \tensor B \intjudge \letpair{x}{y}{p}{\eta_A(x) \tensor \eta_B(y)} : A \tensor B$.

\emph{Sum, first-order summands.}  The two premises of
$\plus L_{\mathrm{mon}}$ are exactly the IH judgments
$x:A \intjudge \eta_A(x):A$ and $y:B \intjudge \eta_B(y):B$; the rule's
side condition holds since $A, B$ are first-order.
The conclusion is
$s:A \plus B \intjudge \eta_{A \plus B}(s) : A \plus B$.

\emph{Sum, otherwise.}  When a summand is not first-order,
$\eta_{A \plus B}(s) = s$ and the judgment is an axiom leaf: such
sums are $\eta$-atomic.  Decompose-and-rebuild is unavailable at
these types --- the $\plus R$ side condition forbids re-forming the
sum --- and is not needed: assumptions of such types are consumed
whole by their eliminating $\plus$-map, whose sources are
unrestricted.

\emph{Implication.}  Derive
$f:A \lmark B,\, x:A \intjudge \eta_B(f\,\eta_A(x)) : B$ by
$\lmark L$ with premises $x:A \intjudge \eta_A(x):A$ (IH for $A$)
and $y:B \intjudge \eta_B(y):B$ (IH for $B$), substituting
$f\,\eta_A(x)$ for $y$.  Then apply $\lmark R$ to abstract $x$.
\end{proof}

\subsection*{Eta-exposed neutral function operands}
\label{subsec:focused-representatives}

The eta-long spine cap belongs to the witness calculus, not to the
raw internal normal-form grammar.  We record the corresponding local
syntactic rep\-re\-sen\-ta\-tive for neutral function-\allowbreak typed operands.
Eta-expansion supplies type-directed identity witnesses.  At Source
domains, its function-type form is the fresh-application
compatibility of
Lemma~\ref{lem:classified-cut-eta}.

For any internally typed \(r:T\), extend type-directed eta-expansion
from variables to terms by
\[
  \eta_T(r):=\eta_T(x)[r/x]
  \qquad(x\notin\mathrm{fv}(r)).
\]
This expression need not itself be normal: a tensor expansion can
expose a commuting conversion.  For an irreducible neutral
\(r:A\lmark B\), define its \emph{eta-exposed representative} by
\begin{equation}
  r^\eta
  :=\mathsf{NF}_{\mathrm{LO}}\bigl(\eta_{A\lmark B}(r)\bigr).
\label{eq:local-function-focus}
\end{equation}
Thus \(r^\eta\), never the raw eta-expansion, is used below.

\begin{proposition}[Local function-operand focusing]
\label{prop:focused-representative}
If \(\Gamma\intjudge r:A\lmark B\) and \(r\) is an irreducible
neutral, then \(r^\eta\) is defined, internally well typed, and
irreducible.  Its outer constructor is an abstraction.  If \(A\)
is a tensor, the bound argument is immediately destructured before
being passed to the neutral spine.  Replacing \(r\) by \(r^\eta\)
in an otherwise irreducible bare \(\plus\)-map former preserves
irreducibility.

For an occurrence in a Source-generated normal form or in a
lower-$\Phi$ branch generated during its interpretation, if \(A\) is
a Source type and \(B\) is first-order, the normalized fresh
applications of \(r\) and \(r^\eta\) agree up to the canonical
boundary transports of Lemma~\ref{lem:classified-cut-eta}.
\end{proposition}
\begin{proof}
Internal typing of \(\eta_{A\lmark B}(r)\) follows from
Lemma~\ref{lem:eta-identity-admissible} and
Lemma~\ref{lem:internal-substitution}.  Totality,
preservation, and irreducibility of its fixed normal form are
Theorem~\ref{thm:lo-normalizer-total}.  The outer implication eta-clause is a lambda.  Every eta-generated
tensor let exposed in its body has a scrutinee that either contains
the bound argument or uses a variable introduced by an earlier
eta-generated prefix.  Hence no such let satisfies \textup{(H$^*$)}
across the full prefix and the enclosing lambda; all remaining
rewrites stay inside its body.  The outer constructor therefore
remains an abstraction.  When $A$ is a tensor, its argument
elimination is the first prefix immediately inside that abstraction.
Normalization preserves typing and the displayed shape
properties.
Finally, a bare map former has no root rewrite.  Replacing one of
its operands by an irreducible abstraction therefore leaves the
former irreducible; only its value/result classification may
change.  The final semantic statement, under the stated occurrence,
Source-type, and first-order hypotheses, is
Lemma~\ref{lem:classified-cut-eta}.
\end{proof}

This replacement may change the normal-form classification from a
blocked map to a map value.

The construction is local: tensor eta-expansion in an arbitrary
normal context may expose a let-floater.  Raw terms continue to
normalize in $\intjudge$; the eta-exposed representative is the local
internal witness used for neutral function operands of blocked maps.

%% file: appendix-summand-completeness.tex
\section{Summand-Index Completeness}
\label{app:summand-index-completeness}

This appendix gives the construction summarized in
\S\ref{subsubsec:unitary-primitives}.

The corresponding source programs needed for summand-index
completeness use only \textsc{Exp}, definable structural
isomorphisms, and ordinary source composition.  Sign involutions
(\textsc{Inv-Scalar}, \textsc{Inv-Id}, \textsc{Inv-$\plus$}) give
diagonal phases, while swap involutions
(\textsc{Inv-$\sigma^\plus$}, \textsc{Inv-Id},
\textsc{Inv-$\plus$}) give the required plane rotations.
Non-adjacent summands and isomorphic first-order types are handled
by structural conjugation,
\[
  s^{-1}\circ u\circ s .
\]
This is an ordinary source program rather than a
$\unitjudge$-derivation, since $s$ may change the intermediate type.

Under the canonical identification
\[
  \sem{B^{\plus n}}\cong\mathbb C^n\tensor\sem B,
\]
let $X_{ij}$ exchange summands $i$ and $j$ and fix the others.  Then
$e^{i\theta X_{ij}}$ performs an $X$-rotation on the $(i,j)$ plane
and contributes the scalar $e^{i\theta}$ on the remaining summands;
the available diagonal phases cancel that scalar.  Consequently,
diagonal--$X$--diagonal Euler decomposition supplies every embedded
$U(2)$, and the standard two-level decomposition supplies
\[
  U\tensor I_{\sem B}
  \qquad (U\in U(n))
\]
on $B^{\plus n}$~\citep{NielsenChuang10}.  Each factor is an actual
source \textsc{Exp} term, transported and composed as above.  No
source-language dagger is required: adjoint exponential factors use
the angle $-\theta$.  For $B=\base$, this gives every unitary on $\Qn{n}$.

%% file: appendix-normal.tex
\section{Normalization}
\label{app:normalization}
\noindent

\subsection{Normalization as rewriting}

\paragraph{Rewrite relation.}
We orient the equations to define a context-closed rewrite relation
$t\to t'$, which permits every generating redex.  The deterministic
leftmost--outermost strategy used by the hereditary normalizer is
fixed below in Definition~\ref{def:lo-normalizer}.

\paragraph{Eta-expanded witnesses.}
The rewrite relation consists of $\beta$-contractions, commuting
conversions, and the monoidal $\plus$ equations.  Eta expansion is
handled by the witness construction rather than by a generating
reduction.  Eta-exposed function representatives are constructed and
normalized separately
in Proposition~\ref{prop:focused-representative}.

\paragraph{Hereditary sum normalization.}
The branch contraction for a phased $\plus$-map uses a syntactic
auxiliary normalizer
$\mathsf{SumNF}^{C,D}_{\alpha,\beta}(e_L,e_R)$, indexed by
first-order output types $C,D$ and unit scalars $\alpha,\beta$: the
normal form of the coherent sum whose branches are $e_L,e_R$ with
accumulated phases $\alpha,\beta$.  It is a purely syntactic
operation, defined simultaneously with the deterministic normalizer
in Definition~\ref{def:lo-normalizer}; its single generating call
site is the hereditary contraction \textup{(G)} below.
The operation $\mathsf{SumNF}$ has no independent semantic clause.
When it is invoked while normalizing a Source expansion, or a
lower-$\Phi$ branch generated during its interpretation, the resulting
canonical normal form is interpreted in
Appendix~\ref{app:unitarityNEW}.  Lemma~\ref{lem:administrative-naturality}
supplies its canonical boundary transports, and
Corollary~\ref{cor:canonical-invariance} specializes the result to
Source semantics.

\begin{definition}[{One-step rewrite $\to$}]
\label{def:rewrite}
Let $\to$ be the smallest relation closed under term contexts such
that the following instances rewrite left-to-right.  In addition to
the local rewrite rules below --- all rules except the hereditary
contraction, including the indexed commuting variants ---
normalization uses one hereditary branch contraction (G)
for branchwise sums.

\medskip
\noindent\textbf{$\beta$-reductions:}
\begin{description}
\item[\textbf{(A)}] \textbf{($\beta_{\lmark}$)} \quad
  $(\lam{x}{t})\,e \;\to\; t[e/x]$
\item[\textbf{(B)}] \textbf{($\beta_{\tensor}$)} \quad
  $\letpair{x}{y}{(e \tensor t)}{u} \;\to\; u[e/x, t/y]$
\end{description}

\medskip
\noindent\textbf{Commuting conversions for $\tensor$-let:}
For the indexed $\lambda$-conversion below, write
\[
  L_i[-] \;:=\; \letpair{x_i}{y_i}{e_i}{[-]},
  \qquad
  L_1\cdots L_k[t] \;:=\; L_1[L_2[\cdots L_k[t]\cdots]].
\]
Two prefix lets $L_i,L_j$ are \emph{independent} when
\[
  \{x_i,y_i\}\cap\mathrm{fv}(e_j)=\varnothing
  \quad\text{and}\quad
  \{x_j,y_j\}\cap\mathrm{fv}(e_i)=\varnothing .
\]
Bound variables are renamed apart before applying a rule.
\begin{description}
\item[\textbf{(C)}] \textbf{($c_{\tensor/\mathrm{app}}$)} \quad
  $(\letpair{x}{y}{e}{f})\,g \;\to\; \letpair{x}{y}{e}{(f\,g)}$
\item[\textbf{(C$'$)}] \textbf{($c_{\tensor/\mathrm{arg}}$)} \quad
  $f\,(\letpair{x}{y}{e}{t}) \;\to\; \letpair{x}{y}{e}{(f\,t)}$
\item[\textbf{(D)}] \textbf{($c_{\tensor/\tensor}$)} \quad
  $\letpair{p}{q}{(\letpair{x}{y}{e}{f})}{g} \;\to\;
   \letpair{x}{y}{e}{(\letpair{p}{q}{f}{g})}$
\item[\textbf{(C$_\tensor^L$)}] \textbf{($c_{\tensor/\tensor L}$)} \quad
  $(\letpair{x}{y}{e}{f}) \tensor g \;\to\;
   \letpair{x}{y}{e}{(f \tensor g)}$
\item[\textbf{(C$_\tensor^R$)}] \textbf{($c_{\tensor/\tensor R}$)} \quad
  $g \tensor (\letpair{x}{y}{e}{f}) \;\to\;
   \letpair{x}{y}{e}{(g \tensor f)}$
\item[\textbf{(H$^*$)}] \textbf{($c^*_{\tensor/\lambda}$, deep $\lambda$-float)}
  For $j\geq 1$,
  \[
    \lam{z}{L_1\cdots L_{j-1}
      [\letpair{x_j}{y_j}{e_j}{t}]}
    \;\to\;
    \letpair{x_j}{y_j}{e_j}
      {\lam{z}{L_1\cdots L_{j-1}[t]}},
  \]
  provided $z\notin\mathrm{fv}(e_j)$ and $L_j$ is independent of every
  crossed prefix let $L_i$ with $i<j$.  The case $j=1$ is the
  ordinary $c_{\tensor/\lambda}$ conversion.
\end{description}

\medskip
\noindent\textbf{Commuting conversions for sum formers:}
\begin{description}
\item[\textbf{(S$_L$)}] \textbf{($c_{\tensor/[\,]L}$)} \quad
  $[\,\alpha \cdot (\letpair{x}{y}{e}{t}) \mid \beta \cdot s\,]
   \;\to\;
   \letpair{x}{y}{e}{\,[\,\alpha \cdot t \mid \beta \cdot s\,]}$
\item[\textbf{(S$_R$)}] \textbf{($c_{\tensor/[\,]R}$)} \quad
  $[\,\alpha \cdot s \mid \beta \cdot (\letpair{x}{y}{e}{t})\,]
   \;\to\;
   \letpair{x}{y}{e}{\,[\,\alpha \cdot s \mid \beta \cdot t\,]}$
\end{description}
With (H$^*$) and (S$_{L/R}$) the let-floaters comprise the complete
family of multiplicative commuting conversions: the
tensor-elimination let commutes with every surrounding construct,
including $\lambda$ ($\lmark$-introduction) and the sum former.  In
particular, a $z$-independent let cannot be hidden behind an
independent prefix let merely by exchanging the prefix order.
The float-out clause of $\mathsf{SumNF}$ is exactly iterated
(S$_{L/R}$).

\medskip
\noindent\textbf{Monoidal $\plus$ equations:}
\begin{description}
\item[\textbf{(E)}] \textbf{($\plus$-comp)} \quad
  $\oplusmap{\alpha}{f}{\beta}{g}(\oplusmap{\alpha'}{h}{\beta'}{k}\;t)
   \;\to\;
   \oplusmap{\alpha\alpha'}{(f \circ h)}{\beta\beta'}{(g \circ k)}\;t$
\end{description}

\medskip
\noindent\textbf{Branchwise sum contraction (hereditary):}
\begin{description}
\item[\textbf{(G)}] \textbf{($\beta_{\plus}$, hereditary branch contraction)} \quad
  If $f : A \lmark C$ and $g : B \lmark D$, then
  \[
    \oplusmap{\alpha}{f}{\beta}{g}\,[\,\alpha' \cdot t_1 \mid \beta' \cdot t_2\,]
    \;\to\;
    \mathsf{SumNF}^{C,D}_{\alpha\alpha',\beta\beta'}(f\,t_1,\; g\,t_2),
  \]
  for any sum former in scrutinee position, with \emph{arbitrary}
  branch terms $t_1,t_2$ --- possibly phased, possibly blocked,
  possibly let-headed or otherwise reducible; the unphased value case
  $\oplusmap{\alpha}{f}{\beta}{g}\,[V_1 \mid V_2]
  \to \mathsf{SumNF}^{C,D}_{\alpha,\beta}(f\,V_1,\, g\,V_2)$
  is the instance $\alpha'=\beta'=1$.  (Since the redex is rooted at
  the application node, leftmost--outermost selects (G) there even
  when a branch contains a floatable let; this matches the
  $(G){\times}(S_{L/R})$ overlap analysis of
  Appendix~\ref{app:determinacy-proofs}.)  The type superscripts $C,D$ are typing bookkeeping only: the
  $\mathsf{SumNF}$ computation does not inspect them.  Rule~(G) is
  invoked only at internally typed instances;
  Theorem~\ref{thm:lo-normalizer-total} proves that its hereditary
  output is defined there.
\end{description}

\medskip
\noindent\textbf{Commuting conversions for $\plus$-map:}
\begin{description}
\item[\textbf{(F)}] \textbf{($c_{\tensor/\plus}$)} \quad
  $\oplusmap{\alpha}{f}{\beta}{g}(\letpair{x}{y}{e}{t})
   \;\to\;
   \letpair{x}{y}{e}{(\oplusmap{\alpha}{f}{\beta}{g}\;t)}$
\item[\textbf{(F$_L$)}] \textbf{($c_{\tensor/\plus L}$)} \quad
  $\oplusmap{\alpha}{(\letpair{x}{y}{e}{h})}{\beta}{g}
   \;\to\;
   \letpair{x}{y}{e}{(\oplusmap{\alpha}{h}{\beta}{g})}$
\item[\textbf{(F$_R$)}] \textbf{($c_{\tensor/\plus R}$)} \quad
  $\oplusmap{\alpha}{f}{\beta}{(\letpair{x}{y}{e}{h})}
   \;\to\;
   \letpair{x}{y}{e}{(\oplusmap{\alpha}{f}{\beta}{h})}$
\end{description}
Rule (F) is the map-former instance of (C$'$); the overlap is
deliberate and the two reducts coincide, so the
leftmost--outermost priority, which lists (C$'$) first, is
unaffected.
\end{definition}

\paragraph{Role.}
Definition~\ref{def:rewrite} fixes an oriented normalization relation by
choosing a direction for each generating equation.

\begin{definition}[Deterministic leftmost--outermost normalizer]
\label{def:lo-normalizer}
Fix an infinite ordered supply of variables.  Capture-avoiding
renaming always chooses the least name in this supply that is fresh
for the ambient term.  Order child positions syntactically from left
to right: function before argument, left tensor/sum/map branch before
right, let scrutinee before body, and map branches before the
scrutinee to which the map is applied.

A redex occurrence is selected \emph{leftmost--outermost}: a redex at
a node precedes every redex properly below that node, and otherwise
the preceding child order is used.  If several generating rules match
the selected node, choose the first applicable rule in the fixed order
\[
 (A),(B),(E),(G),(C),(C'),(D),(C_\tensor^L),(C_\tensor^R),
 (H^*),(S_L),(S_R),(F_L),(F_R),(F).
\]
Thus contractions precede commuting conversions, while ties between
commuting conversions at distinct child slots follow the child order
above.  An (H$^*$)-instance counts as a redex at its enclosing
$\lambda$-node; when several indices $j$ are eligible at that node,
choose the least one.  Write
$t\mathrel{\longrightarrow_{\mathrm{LO}}} t'$ for the resulting selected instance of a generating
rewrite.  Thus $\mathrel{\longrightarrow_{\mathrm{LO}}}\subseteq{\to}$ and $\mathrel{\longrightarrow_{\mathrm{LO}}}$ is a partial
function once the hereditary (G)-output is defined.

The partial functions $\mathsf{NF}_{\mathrm{LO}}(-)$ and $\mathsf{SumNF}$ are defined
simultaneously.  If $t$ has no selected redex, set $\mathsf{NF}_{\mathrm{LO}}(t)=t$.
If $t\mathrel{\longrightarrow_{\mathrm{LO}}} t'$, set $\mathsf{NF}_{\mathrm{LO}}(t)=\mathsf{NF}_{\mathrm{LO}}(t')$; in the (G)-case the
right-hand side $t'$ is first obtained from the indicated
$\mathsf{SumNF}$ call.  To compute
$\mathsf{SumNF}^{C,D}_{\alpha,\beta}(e_L,e_R)$, compute
$\mathsf{NF}_{\mathrm{LO}}(e_L)$ completely before $\mathsf{NF}_{\mathrm{LO}}(e_R)$, and write the results
uniquely as
\[
  \mathsf{NF}_{\mathrm{LO}}(e_L)=P_L[R_L],
  \qquad
  \mathsf{NF}_{\mathrm{LO}}(e_R)=P_R[R_R],
\]
where $P_L,P_R$ are their maximal outer tensor-let prefixes and
$R_L,R_R$ the non-let-headed remainders (a purely syntactic split;
their classification as grammar results is
Lemma~\ref{lem:nf-iff-irreducible}, used only later).  Rename the two prefixes apart by
the fixed convention and set
\[
  \mathsf{SumNF}^{C,D}_{\alpha,\beta}(e_L,e_R)
  :=
  P_L\!\left[P_R\!\left[
    [\,\alpha\cdot R_L\mid\beta\cdot R_R\,]
  \right]\right].
\]
This fixes left-before-right branch normalization and left-to-right
prefix extraction.  The definition is initially partial; its
totality, including the hereditary calls made while computing the
two branch normal forms, is Theorem~\ref{thm:lo-normalizer-total}.
\end{definition}

\subsection{Typing scope of normalization}
\label{subsec:structured-reduction}

All preservation and normalization results below are stated for the
internal judgment $\intjudge$.  Rule~\textup{(E)} may expose a
compound-typed bound variable in $f\circ h=\lambda z.f(hz)$ while
preserving internal typing.  Canonical normal derivations arising
from Source expansions, together with the lower-$\Phi$ branches
generated during their interpretation, supply the sectors used in
Appendix~\ref{app:unitarityNEW}.  Lemma~\ref{lem:administrative-naturality}
supplies their $\approx$-transports; fixed normalization, determinacy,
and Corollary~\ref{cor:canonical-invariance} give Source reduction
invariance.
The eta-long witness calculus supplies the type-directed identities
and the local neutral-function representative of
Proposition~\ref{prop:focused-representative}.  The computational
inductions use the internal typing and normalization lemmas; boundary
sectors are attached afterward to canonical normal derivations.

\subsection{Termination}

We prove totality of the fixed normalizer and termination of the
ambient rewrite relation separately.  The proof uses a
well-founded lexicographic measure with three components.

\begin{definition}[$\plus$-map constructor count $\Phi$]
\label{def:phi}
Let $\Phi(t)$ be the total number of occurrences in $t$ of the
$\plus$-map constructor $\oplusmap{\alpha}{f}{\beta}{g}$, whether or
not it is applied to an argument.  Application nodes, sum formers,
phases, $\lambda$-abstractions, and let-binders contribute nothing:
only the map constructor itself is counted.
\end{definition}

\begin{definition}[Binder count $\mathsf{Bind}$]
\label{def:bind-count}
Let $\mathsf{Bind}(t)$ be the total number of $\lambda$-abstractions
$(\lam{x}{-})$ and $\tensor$-constructors $(-\tensor-)$ occurring in $t$.
\end{definition}

\begin{definition}[Let-depth $\mathsf{LetDepth}$]
\label{def:letdepth}
For a let-binder $\ell = \letpair{x}{y}{-}{-}$ occurring in $t$, define
$\mathrm{depth}(\ell,t)$ to be the number of nodes of the following kinds
on the path from the root of $t$ to $\ell$:
\begin{itemize}
\item application nodes $(-)\,(-)$;
\item tensor-constructor nodes $(-)\tensor(-)$;
\item $\plus$-map application nodes $\oplusmap{-}{-}{-}{-}(-)$;
\item branch-function slots of $\plus$-map constructors
  $\oplusmap{-}{-}{-}{-}$;
\item scrutinee slots of let-binders, i.e.\ positions where the path
  enters the scrutinee $e$ of some $\letpair{p}{q}{e}{u}$;
\item $\lambda$-abstraction nodes $\lam{z}{-}$;
\item branch slots of sum formers
  $[\,\alpha \cdot {-} \mid \beta \cdot {-}\,]$.
\end{itemize}
Set $\mathsf{LetDepth}(t) := \sum_{\ell \in \mathrm{Lets}(t)}
\mathrm{depth}(\ell,t)$.
\end{definition}

\begin{definition}[Termination measure]
\label{def:measure}
Define
\[
  \mu(t) := \bigl(\Phi(t),\; \mathsf{Bind}(t),\; \mathsf{LetDepth}(t)\bigr)
  \in \mathbb{N}^3
\]
ordered lexicographically.
\end{definition}

\begin{lemma}[Administrative preservation for ordinary steps]
\label{lem:ordinary-internal-preservation}
Suppose $\Gamma \intjudge t:A$, and $t\to t'$ is a
context-lifted instance of any generating rule other than \textup{(G)}.
Then $\Gamma\intjudge t':A$.  The resulting derivation preserves the
first-order premises on sum summands and map targets, the
unit-modulus premises on phases, and every disjoint split of branch
contexts.
\end{lemma}

\begin{proof}
Induct on the enclosing typed context and inspect the root rule.
Cases (A) and (B) are precisely the one- and two-variable instances
of internal substitution (Lemma~\ref{lem:internal-substitution});
there is no value premise to re-establish.  The commuting conversions
reassociate the same linearly split premises.  For (H$^*$), exchange
the selected premise past the independent prefix premises and then
apply $\lmark$-introduction; its side condition is exactly what keeps
$z$ out of the floated scrutinee.  In (E), the two composites have
the displayed source and target types, the target summands remain
first-order, and $|\alpha\alpha'|=|\beta\beta'|=1$.  The remaining
phase-bearing rules merely move an existing constructor.  Rebuilding
the enclosing derivation preserves all context splits.
\end{proof}

\begin{lemma}[Rank of an ordinary step]
\label{lem:ordinary-step-rank}
Let $t\to t'$ be a context-lifted generating step other than
\textup{(G)}.
\begin{enumerate}
\item In case \textup{(E)}, $\Phi(t')<\Phi(t)$.
\item In cases \textup{(A)} and \textup{(B)}, $\Phi$ is preserved and
  $\mathsf{Bind}(t')<\mathsf{Bind}(t)$.
\item For every commuting conversion, $\Phi$ and $\mathsf{Bind}$ are
  preserved and
  $\mathsf{LetDepth}(t')<\mathsf{LetDepth}(t)$.
\end{enumerate}
\end{lemma}

\begin{proof}
We verify the actual context-lifted step; no monotonicity claim for
arbitrary unrelated pairs of terms is used.  Rule (E) removes one
map constructor.  Expanding $f\circ h$ as
$\lam{z}{f\,(h\,z)}$ introduces no map constructor, so the first
component decreases exactly by one.

For (A), linearity gives
\[
 \Phi(t[e/x])=\Phi(t)+\Phi(e),\qquad
 \mathsf{Bind}(t[e/x])=\mathsf{Bind}(t)+\mathsf{Bind}(e),
\]
and the contracted $\lambda$ contributes the missing one on the
left.  For (B), the two-variable form of the same calculation uses
Lemma~\ref{lem:internal-substitution}; the contracted tensor
constructor contributes the missing one.  Hence the second
component strictly decreases in both cases.

Each remaining rule preserves the map, $\lambda$, and tensor
constructors and moves a let past at least one counted node.  The
floated let therefore loses that node; lets in its scrutinee weakly
lose counted ancestors, and all other lets keep their counted
ancestors.  The nested-scrutinee case (D) loses the outer let's
scrutinee slot.  In the indexed case (H$^*$), the selected let
$L_j$ and every let in its scrutinee lose the counted $\lambda$
ancestor, while the crossed prefix lets occur on let-body paths,
which are not counted, so they and the lets in their scrutinees
retain their counted ancestors.  Induction on the actual
enclosing context shows that its common path contribution is equal
on the two sides.  Thus the third component strictly decreases.
\end{proof}

\begin{theorem}[Total deterministic hereditary normalization]
\label{thm:lo-normalizer-total}
For every $k\in\mathbb N$, the following assertions hold.
\begin{enumerate}
\item[\textup{(N$_k$)}]
  If $\Gamma\intjudge t:A$ and $\Phi(t)\leq k$, then
  $\mathsf{NF}_{\mathrm{LO}}(t)$ is defined.  Writing
  $n=\mathsf{NF}_{\mathrm{LO}}(t)$,
  \[
    t\to^*n,\qquad n\text{ is irreducible},\qquad
    \Phi(n)\leq\Phi(t),\qquad \Gamma\intjudge n:A.
  \]
\item[\textup{(S$_k$)}]
  Suppose the variable domains of $\Gamma_L$ and $\Gamma_R$ are
  disjoint,
  \[
    \Gamma_L\intjudge e_L:C,\qquad
    \Gamma_R\intjudge e_R:D,\qquad
    \Phi(e_L)+\Phi(e_R)\leq k,
  \]
  where $C,D$ are first-order and
  $\alpha,\beta\in\mathbb C_{\mathrm{static}}$ with
  $|\alpha|=|\beta|=1$.
  Then
  $s=\mathsf{SumNF}^{C,D}_{\alpha,\beta}(e_L,e_R)$ is defined and
  \[
    [\,\alpha\cdot e_L\mid\beta\cdot e_R\,]\to^*s,\qquad
    s\text{ is irreducible},\qquad
    \Phi(s)\leq\Phi(e_L)+\Phi(e_R),\qquad
    \Gamma_L,\Gamma_R\intjudge s:C\plus D.
  \]
\end{enumerate}
In both clauses every constructed phased sum has first-order
summands, every constructed phased map has first-order targets, all
accumulated phases have unit modulus, and the branch contexts remain
disjoint.
\end{theorem}

\begin{proof}
Use strong induction on $k$.  At stratum $k$, first prove
\textup{(N$_k$)} and then \textup{(S$_k$)}.  Terms of smaller
$\Phi$ use the outer induction hypothesis.  For terms with
$\Phi(t)=k$, define $\mathsf{NF}_{\mathrm{LO}}(t)$ by well-founded
recursion on
\[
  \bigl(\mathsf{Bind}(t),\mathsf{LetDepth}(t)\bigr)
\]
in lexicographic order.

If $t$ has no selected redex, the defining clause returns $t$.
This is the zero-step case and supplies irreducibility directly.
For a selected ordinary step $t\to_{\mathrm{LO}}t'$, administrative
preservation is Lemma~\ref{lem:ordinary-internal-preservation}.
Rule (E) enters a smaller $\Phi$-stratum.  Rules (A), (B), and all
floaters remain at the same stratum and enter a smaller inner rank
by Lemma~\ref{lem:ordinary-step-rank}.  The appropriate induction
hypothesis therefore defines the recursive call.  Prefixing its
reduction sequence by this one context-lifted generating rewrite
establishes $t\to^*\mathsf{NF}_{\mathrm{LO}}(t)$; the remaining
conjuncts are inherited from that recursive call.

It remains to justify a selected (G)-step.  Such a redex has
$\Phi\geq1$, so this case occurs only when $k\geq1$.  Write its occurrence as
$K[r]$, where
\[
 r=
 \oplusmap{\alpha}{f}{\beta}{g}
 [\,\alpha'\cdot R_1\mid\beta'\cdot R_2\,].
\]
Typing inversion gives branch applications of the first-order map
targets $C,D$ in disjoint contexts.  Since application nodes do not
contribute to $\Phi$,
\[
 \Phi(f\,R_1)+\Phi(g\,R_2)
 =\Phi(f)+\Phi(g)+\Phi(R_1)+\Phi(R_2)
 =\Phi(r)-1.                                      \tag{*}
\]
Thus the companion assertion \textup{(S$_{k-1}$)} from the outer
induction hypothesis defines
\[
 s=\mathsf{SumNF}^{C,D}_{\alpha\alpha',\beta\beta'}
      (f\,R_1,g\,R_2)
\]
and gives
\[
 \Phi(s)\leq\Phi(f\,R_1)+\Phi(g\,R_2)=\Phi(r)-1.
\]
If $\Phi(K[-])$ denotes the constructors contributed outside the
hole, then
\[
 \Phi(K[s])\leq\Phi(K[-])+\Phi(r)-1=\Phi(K[r])-1.
\]
Consequently the recursive call on $K[s]$ lies in a strictly smaller
outer stratum.  The companion assertion at stratum $k-1$ supplies the
bound before the recursive call at stratum $k$, so the hereditary
clause is well founded.  The (G)-instance itself is one generating
rewrite; context closure and the lower-stratum sequence give the
required reduction.
Typing is supplied by \textup{(S$_{k-1}$)} and rebuilding $K[-]$.
The map-target premise makes $C,D$ first-order, and products of the
input phases remain unit-modulus.

Having established \textup{(N$_k$)}, consider the inputs of
\textup{(S$_k$)}.  Apply \textup{(N$_k$)} first to $e_L$ and then to
$e_R$, obtaining
\[
 n_L=P_L[R_L],\qquad n_R=P_R[R_R],
\]
with maximal outer tensor-let prefixes and irreducible results that
are not headed by a let.  Context closure gives
\[
 [\,\alpha\cdot e_L\mid\beta\cdot e_R\,]
 \to^*[\,\alpha\cdot n_L\mid\beta\cdot n_R\,].
\]
Repeated (S$_L$)-steps, followed from inside $P_L$ by repeated
(S$_R$)-steps, give exactly the fixed output
\[
 P_L\!\left[P_R\!\left[
   [\,\alpha\cdot R_L\mid\beta\cdot R_R\,]
 \right]\right].
\]
The branch normal forms have no internal redex.  Their maximal
prefix scrutinees are irreducible, and the displayed construction
has no let in a sum branch; combining the two prefixes creates no
new root redex.  Hence the output is irreducible.  Sum formers and
prefix extraction add no map constructor, so
\[
 \Phi(s)=\Phi(n_L)+\Phi(n_R)
 \leq\Phi(e_L)+\Phi(e_R).
\]
The internal phased-sum rule types the sum because $C,D$ are
first-order; Lemma~\ref{lem:ordinary-internal-preservation} types
the subsequent prefix floaters.  It also preserves the disjoint
branch split.  The outer phases are unit-modulus, and every phase
multiplication made while obtaining $n_L,n_R$ was covered by the
normalizer induction.  This proves \textup{(S$_k$)} and completes
the strong induction.
\end{proof}

The construction of $\mathsf{NF}_{\mathrm{LO}}$ used the grammar
nowhere: the grammar theorem below is consumed only after
irreducibility has been established.

\begin{lemma}[Internal subject reduction]
\label{lem:internal-preservation}
If $\Gamma\intjudge t:A$ and $t\to t'$, then
$\Gamma\intjudge t':A$.  The first-order side conditions,
unit-modulus phases, and disjoint branch contexts are preserved.
\end{lemma}

\begin{proof}
For an ordinary step, use
Lemma~\ref{lem:ordinary-internal-preservation}.  For a (G)-step,
typing inversion supplies the two typed branch applications, the
first-order map targets, unit phases, and disjoint contexts.  Apply
the companion clause of
Theorem~\ref{thm:lo-normalizer-total} at
$k=\Phi(f\,R_1)+\Phi(g\,R_2)$, then rebuild the unchanged enclosing
typed context.
\end{proof}

\begin{lemma}[Deterministic typed LO lifting]
\label{lem:typed-lo-lifting}
Fix the redex order, branch order, context order, and fresh-name
convention of Definition~\ref{def:lo-normalizer}.  There is a
function on internal derivations
\[
 \mathsf{DerNF}_{\mathrm{LO}}
 \bigl(\mathcal D:(\Gamma\intjudge t:A)\bigr)
 =
 \mathcal D^\downarrow_{\mathrm{LO}}:
 \Gamma\intjudge\mathsf{NF}_{\mathrm{LO}}(t):A .
\]
Its subject is the fixed LO normal form, and it preserves the
first-order premises, unit-modulus phase premises, and disjoint
branch contexts recorded by $\mathcal D$.
\end{lemma}
\begin{proof}
Define one typed selected-step transformer by the same induction on
the enclosing typed context as
Lemma~\ref{lem:ordinary-internal-preservation}, choosing the
displayed reconstruction in each root-rule case.  The selected redex
and its position are unique.  Context order and fresh bound names are
fixed by Definition~\ref{def:lo-normalizer}.  For (G), choose exactly
the typed $\mathsf{SumNF}$ derivation constructed by the companion
clause \textup{(S$_k$)} of
Theorem~\ref{thm:lo-normalizer-total}; its left-before-right branch
order and prefix extraction are already fixed.  Thus the transformer
is a function, rather than an appeal to proof irrelevance.

Iterate it along the recursive calls defining
$\mathsf{NF}_{\mathrm{LO}}$.  Totality follows from
Theorem~\ref{thm:lo-normalizer-total}; preservation of the displayed
invariants is part of
Lemmas~\ref{lem:ordinary-internal-preservation}
and~\ref{lem:internal-preservation}.
\end{proof}

\begin{lemma}[Canonical retyping of normal terms]
\label{lem:canonical-normal-retyping}
Fix the context-variable order already used by the LO normalizer and
the following rule priority: a phased sum or phased map syntax is
always typed by its internal phased rule (also when both phases are
$1$); all other term constructors use their syntax-directed internal
rule; and an atom uses its declared signature entry.  There is a
function
\[
 \mathsf{CanDer}_{\Gamma,A}(n):
 \Gamma\intjudge n:A
\]
defined whenever $n$ is irreducible and that judgment is derivable.
Its premise derivations are recursively canonical.  Moreover, every
chosen proof $p:n\approx n'$ at the same judgment shape induces a
syntax-directed typed lift
\[
 \mathsf{CanDer}_{\Gamma,A}(p):
 \mathsf{CanDer}_{\Gamma,A}(n)
 \longleftrightarrow
 \mathsf{CanDer}_{\Gamma,A}(n')                                \tag{CR}
\]
built only from bound-name relabelling, exchange of the same
independent tensor-let premises, and equality of evaluated phase
scalars.
\end{lemma}
\begin{proof}
Recurse on the normal syntax, not on its overlapping $V,E,R$
classification.  Linearity makes each premise context the ordered
restriction of $\Gamma$ to that premise's free variables, so tensor,
application, sum/map, and let splits are fixed.  The result type fixes
the remaining rule indices.  Derivability supplies every side
condition: first-order sum summands and map targets, unit phases, and
disjoint contexts.  The stated priority removes the only overlap
between the ordinary and phased internal presentations.  Thus the
selected derivation is a function.

For (CR), induct on the chosen $\approx$ derivation.  Alpha renames the
selected binder.  An (X) generator reassembles the same recursively
canonical scrutinee and body premises in the opposite legal order.
A (P) generator changes only an annotation with the same evaluated
unit scalar.  Congruence and equivalence closure compose these typed
lifts.  Reflexivity returns the identical canonical derivation, so no
typing-proof irrelevance is used.
\end{proof}
\begin{lemma}[Each rewrite step decreases the measure]
\label{lem:measure-decrease}
If $\Gamma\intjudge t:A$ and $t \to t'$, then $\mu(t) > \mu(t')$.
\end{lemma}

\begin{proof}
For every rule other than (G), this is
Lemma~\ref{lem:ordinary-step-rank}.  For a (G)-redex $r$ in an
actual enclosing context $K[-]$, equation~\textup{(*)} in the proof
of Theorem~\ref{thm:lo-normalizer-total} and its companion
$\Phi$-bound give
\[
 \Phi\bigl(\mathsf{SumNF}(fR_1,gR_2)\bigr)
 \leq \Phi(fR_1)+\Phi(gR_2)=\Phi(r)-1.
\]
The context contributes the same number of map constructors on both
sides, so the first component of $\mu$ strictly decreases.
\end{proof}

\begin{theorem}[Termination of normalization]
\label{thm:termination}
There is no infinite reduction sequence
$t_0 \to t_1 \to t_2 \to \cdots$ whenever
$\Gamma\intjudge t_0:A$.
\end{theorem}

\begin{proof}
Lemma~\ref{lem:internal-preservation} keeps every intermediate term
internally well typed, so Lemma~\ref{lem:measure-decrease} applies at
each step and makes $\mu$ strictly decrease.  Lexicographic order on
$\mathbb N^3$ is well-founded.
\end{proof}

\paragraph{What this theorem does.}
Every reduction strategy terminates.  Existence of the particular
normal form returned by the deterministic strategy was established
constructively in Theorem~\ref{thm:lo-normalizer-total}; uniqueness
across strategies is proved separately in
Appendix~\ref{app:determinacy-proofs}.

\subsection{Analysis of normal forms}

Definition~\ref{def:nf-grammar} makes the irreducible shapes
explicit.  Ordinary neutral spines have variable or atom heads and
result operands.  A bare $\plus$-map has result operands as well: it
is a value when both are values and a blocked-map result otherwise.
An applied map may have such result operands but, in normal form, its
scrutinee is neutral.  These are grammar classifications, not new
term constructors.

\begin{definition}[Deep lambda-prefix condition \textup{(NF4*)}]
\label{def:nf4star}
Let $L_i[-]=\letpair{x_i}{y_i}{e_i}{[-]}$ use the prefix notation and
independence relation of Definition~\ref{def:rewrite}.  Condition
\textup{(NF4*)} says that, for every subterm whose $\lambda$-body has
outer prefix
\[
  \lam{z}{L_1\cdots L_k[R]},
\]
there is no index $j$ such that
\[
  z\notin\mathrm{fv}(e_j)
  \quad\text{and}\quad
  L_j\text{ is independent of every }L_i\text{ with }i<j.
\]
Equivalently, no (H$^*$)-redex is rooted at that $\lambda$-node.
For $j=1$ the independence requirement is vacuous.
\end{definition}

\begin{lemma}[Normal forms are exactly the $\to$-irreducible terms]
\label{lem:nf-iff-irreducible}
If $\Gamma\intjudge t:A$, then $t$ is $\to$-irreducible if and only if $t$ is
generated by the grammar of Definition~\ref{def:nf-grammar} and
satisfies \textup{(NF1)}--\textup{(NF3)} and \textup{(NF4*)}.
\end{lemma}

\begin{proof}
We prove both directions.

\smallskip
\noindent\textbf{($\Rightarrow$) Irreducible implies in the grammar.}
Because $\to$ is closed under contexts, every subterm of an
irreducible term is irreducible.  We classify the possible outer
shapes, using this observation recursively.

First consider lets.  A let headed subterm in a function or argument
position, a tensor component, a let scrutinee, a map scrutinee, a map
operand, or a sum branch would trigger respectively (C), (C$'$),
(C$_\tensor^{L/R}$), (D), (F), (F$_{L/R}$), or (S$_{L/R}$).  Hence
normal constructor operands are results $R$, not outer let-prefixes,
and every remaining let occurs in an $N$-prefix.  Under a
$\lambda z$, an eligible prefix let at any depth $j$ would trigger
(H$^*$); the absence of such a redex is exactly (NF4*).  Notice that
a $z$-independent deeper let may remain only when some crossed prefix
let is dependent on it, so this conclusion is derivation- rather
than type-shaped.

The absence of the two beta redexes gives (NF1).  The absence of a
map applied to a map application gives (NF2).  The absence of a map
applied to any phased sum former, whether a value or blocked, gives
(NF3).

It remains to check that every let-free outer shape is one of the
recently widened productions.
\begin{itemize}
\item Variables and atoms are values and neutrals.  An ordinary
  irreducible application cannot have a $\lambda$ head by (NF1).
  Typing excludes tensor and sum formers from function position, so,
  apart from the map-headed case below, its function is a
  variable- or atom-headed neutral.  Its argument is an irreducible
  result.  Thus it has the production $E\;R$ (and $x\,V$ also has the
  paper's overlapping value classification).
\item A tensor former has irreducible result components.  It is the
  value $V_1\tensor V_2$ when both are values and otherwise the
  widened result $R_1\tensor R_2$.
\item A phased sum former has result branches because an outer branch
  let would trigger (S$_{L/R}$).  It is a sum value when both branches
  are values and otherwise a blocked-sum result
  $[\,\alpha\cdot R_1\mid\beta\cdot R_2\,]$.  Its summand types are first-order
  by inversion of the internal sum rule.
\item A bare map former likewise has result operands because
  (F$_{L/R}$) excludes outer operand lets.  If both operands are
  values it has the value production
  $\oplusmap{\alpha}{V_1}{\beta}{V_2}$; otherwise it has the genuine
  blocked-map result production
  $\oplusmap{\alpha}{R_1}{\beta}{R_2}$.  Its function type and
  first-order target summands come from the $\plus$-Map rule, not
  from the sum-former rule.
\item In an applied map
  $\oplusmap{\alpha}{R_1}{\beta}{R_2}\,u$, the operands have just been
  classified as results.  An outer let-prefix in $u$ would trigger
  (F); a sum former would trigger (G); and a map application is
  forbidden by (NF2).  Every remaining normal term of the required
  sum source type is neutral.  Hence the term has the widened
  production $\oplusmap{\alpha}{R_1}{\beta}{R_2}\,E$.
\item A $\lambda$ has an $N$ body and satisfies (NF4*), as established
  above.  Finally, a tensor-let has an irreducible result scrutinee
  and an $N$ body.  A let-headed scrutinee is excluded by (D), and a
  direct tensor-former scrutinee is excluded by (NF1), exactly as
  irreducibility requires.
\end{itemize}
These cases exhaust the typed term constructors and establish the
grammar.

\smallskip
\noindent\textbf{($\Leftarrow$) In the grammar implies irreducible.}
Proceed by simultaneous induction over $V,E,R,N$.  The recursive
premises contain no redex by the induction hypotheses.  Conditions
(NF1), (NF2), and (NF3) exclude the beta, map-composition, and
hereditary-sum rules at every node.

It remains to exclude commuting conversions.  A result production
has no outer let-prefix.  Consequently $E\;R$, $R_1\tensor R_2$,
$[\,\alpha\cdot R_1\mid\beta\cdot R_2\,]$, the bare map
$\oplusmap{\alpha}{R_1}{\beta}{R_2}$, and the applied map
$\oplusmap{\alpha}{R_1}{\beta}{R_2}\,E$ contain no outer let in any
slot targeted by (C), (C$'$), (C$_\tensor^{L/R}$), (F),
(F$_{L/R}$), or (S$_{L/R}$).  In an $N$-prefix, every scrutinee is a
result rather than a further prefix, excluding (D).  A bare map has
no root rewrite: when both operands are values it is a value, and
when at least one is not it is precisely the blocked-map result.
The analogous statement holds for phased and blocked sum formers.
For every $\lambda$-body prefix, (NF4*) excludes (H$^*$) at every
eligible index, not merely at the leading let.  Thus no generating
rule applies anywhere, and the term is $\to$-irreducible.
\end{proof}

\begin{corollary}[Existence of normal forms]
\label{cor:nf-exists}
If $\Gamma\intjudge t:A$, then
$n=\mathsf{NF}_{\mathrm{LO}}(t)$ is defined,
$\Gamma\intjudge n:A$, $t\to^*n$, and $n\in N$.
\end{corollary}

This establishes Theorem~\ref{thm:normalization}.

\begin{proof}
Set $k=\Phi(t)$ in Theorem~\ref{thm:lo-normalizer-total}.
The result is reachable and irreducible.
Lemma~\ref{lem:nf-iff-irreducible} therefore places it in the
normal-form grammar.  The grammar lemma is applied at this final
stage, after the normalizer has been defined and proved total.
\end{proof}

\begin{definition}[Additive balance]
\label{def:additive-balance}
An occurrence of
$[\,\alpha\mathbin\cdot u\mid\beta\mathbin\cdot v\,]$ or
$\oplusmap{\alpha}{f}{\beta}{g}$ is \emph{balanced} when its two
components have the same typed free-variable context.  A term is
balanced, written $\mathsf{Bal}(t)$, when all its additive occurrences
are balanced.  Since the internal additive rules split their premise
contexts, every balanced additive occurrence is closed; in particular,
a balanced $\plus$-map has closed operands.
\end{definition}

\begin{lemma}[Closed coherent-sum-free tensor forms]
\label{lem:closed-sum-free-tensor}
If $\cdot\intjudge n:A\tensor B$, $n$ is irreducible, and $n$
contains no coherent-sum former, then
$n=R_1\tensor R_2$ for closed normal results $R_1,R_2$.
\end{lemma}
\begin{proof}
By Lemma~\ref{lem:nf-iff-irreducible}, $n$ belongs to the normal-form
grammar.  We prove simultaneously by strong induction on normal-term
size that (i) no closed coherent-sum-free normal term has sum type, (ii) no such
term has first-order type, and (iii) every such term of tensor type is
pair-headed.  A closed variable-headed neutral is impossible.  A
quantum head requires a smaller closed normal operand of first-order
type.  A $\plus$-coherence head requires a smaller sum-typed operand; a
distributor with sum result requires a tensor operand with a smaller
sum-typed component; and an inverse distributor with tensor result
requires a smaller sum-typed operand.  An applied $\plus$-map also
requires a smaller sum-typed scrutinee.  For (ii), base has no closed
producer; a first-order sum is excluded by (i), and a first-order tensor
pair would have smaller closed first-order components, excluded by
induction.  A closed let-prefix is impossible because its tensor
scrutinee is pair-headed by induction and hence forms a
\textup{(B)}-redex.  The remaining tensor form is
$R_1\tensor R_2$, and context splitting makes both components closed.
\end{proof}

\begin{lemma}[Source context and sum formation]
\label{lem:source-expansion-discipline}
If $\Gamma\sjudge t:A$, then $t^\circ$ contains no coherent-sum
former and $\mathsf{Bal}(t^\circ)$.  At every translated case, the
two branch maps abstract the same typed variable context.  Moreover,
$t^\circ$ contains no inverse distributor whose source or target is
non-first-order.
\end{lemma}
\begin{proof}
Induct on the Source derivation.  In the case rule, both branch
premises have the same typed variable context.  The expansion packs
that context once and abstracts it in both branch maps, so both maps
are closed.  It uses a forward distributor, a $\plus$-map, and inverse
distributivity, but no coherent-sum former.  The forward distributor
may have a higher-order context factor; the inverse distributor has
the first-order source
$(A\tensor C)\plus(B\tensor C)$ and first-order target
$(A\plus B)\tensor C$.

Every explicitly available Source structural atom has Source-typed
endpoints.  By the Source type grammar, an inverse distributor at a
non-first-order instance cannot have two such endpoints.
The other cases, including \textsc{Exp}, are immediate.
\end{proof}

\begin{lemma}[Reduction preserves Source discipline]
\label{lem:reduction-preserves-source-discipline}
Suppose $\Gamma\intjudge t:A$, the term $t$ contains no coherent-sum
former, and $t\to t'$.  Then $\Gamma\intjudge t':A$,
$\mathrm{fv}(t)=\mathrm{fv}(t')=\mathrm{dom}(\Gamma)$ with the types
recorded by $\Gamma$, and $t'$ contains no coherent-sum former.
If $t$ contains no inverse distributor whose source or target is
non-first-order, neither does $t'$.
\end{lemma}
\begin{proof}
Internal subject reduction gives $\Gamma\intjudge t':A$; linear
relevance identifies the typed free-variable context of both terms
with $\Gamma$.  Rules \textup{(S$_L$)}, \textup{(S$_R$)}, and
\textup{(G)} cannot
apply because their left-hand sides contain a coherent-sum former.
Every remaining rule introduces no coherent-sum former, by
capture-avoiding substitution or direct inspection.  The property is
preserved when the step is lifted into a context.  No reduction rule
creates a structural atom or changes the indexed types of an existing
one, so the final property is preserved as well.
\end{proof}

\begin{lemma}[Eventual restoration of additive balance]
\label{lem:lo-restores-balance}
Suppose $\Gamma\intjudge t:A$, the term $t$ contains no coherent-sum
former, and $\mathsf{Bal}(t)$.  If $t\to_{\mathrm{LO}}u$, then some
term $v$ on the ensuing deterministic path satisfies
\[
  u\to_{\mathrm{LO}}^*v
  \qquad\text{and}\qquad
  \mathsf{Bal}(v).
  \tag{Bal-restore}
\]
Consequently,
$\mathsf{Bal}\bigl(\mathsf{NF}_{\mathrm{LO}}(t)\bigr)$.
\end{lemma}
\begin{proof}
All rules except \textup{(F$_L$)} and \textup{(F$_R$)} preserve
balance directly.  A balanced additive occurrence is closed, so
substitution in \textup{(A)} or \textup{(B)} cannot enter it; additive
occurrences carried by the substituted terms retain their balance.
Rule \textup{(E)} composes closed operands, and the remaining ordinary
rules only relocate existing subterms.
Rules \textup{(S$_L$)}, \textup{(S$_R$)}, and \textup{(G)} cannot
occur, since their left-hand sides contain a coherent-sum former and
the remaining rules introduce none.  Thus, unless the selected rule is
\textup{(F$_L$)} or \textup{(F$_R$)}, take $v=u$.

Consider \textup{(F$_L$)}; the right case is symmetric:
\[
 \oplusmap{\alpha}{(\letpair{x}{y}{e}{h})}{\beta}{g}
 \longrightarrow
 \letpair{x}{y}{e}{\oplusmap{\alpha}{h}{\beta}{g}} .
\]
Balance and the disjoint premise contexts on the left imply
$\mathrm{fv}(e)=\mathrm{fv}(g)=\varnothing$ and
$\mathrm{fv}(h)\subseteq\{x,y\}$.  The reduct may therefore contain
an open map, but every variable free in that map is bound by the new
enclosing tensor-let, whose scrutinee is closed.

For this proof, a chain for a map occurrence is any, not necessarily
contiguous, subsequence of the enclosing tensor-lets on its
root-to-occurrence path, ordered from outermost to innermost; other
constructors and unchosen tensor-lets may occur between its members.
Such a chain $L_1,\allowbreak\ldots,\allowbreak L_k$, with
scrutinees $e_1,\ldots,e_k$, is \emph{closed} when
$\mathrm{fv}(e_1)=\varnothing$ and every
$\mathrm{fv}(e_i)$ is contained in the binders of
$L_1,\ldots,L_{i-1}$.  An open $\plus$-map is \emph{protected} when
all its free variables are bound by a closed enclosing chain.  The
reduct above has this property with a one-let chain.

Follow the ensuing deterministic path with the stronger invariant that
every open map is protected.  Induction on the decreasing measure
$\mu$ verifies it.  The ordinary floaters preserve the order and
dependencies of the chain.  Rule \textup{(D)} updates that order, and
the independence premise of \textup{(H$^*$)} preserves it across the
crossed binders.  Rule \textup{(E)} uses the ordered union of the
nested enclosing chains of its two map occurrences.  A
\textup{(B)}-step at a chain link substitutes pair components whose
free variables lie among the
preceding binders, so deleting that link preserves protection.  A later
\textup{(F$_L$)} or \textup{(F$_R$)} step creates or extends a closed
chain of the same form.  The remaining rules preserve protection by
capture-avoiding substitution and context closure.

Termination gives an irreducible endpoint $v$.  If $v$ contained an
open map, choose a closed chain protecting it.  The outermost let of
that chain has a closed, irreducible, tensor-typed scrutinee.  By
Lemma~\ref{lem:closed-sum-free-tensor} the scrutinee is pair-headed, so
this let is a \textup{(B)}-redex, a contradiction.  Hence $v$ is
balanced and proves \textup{(Bal-restore)}.  Iterating along the finite
LO path, with the zero-step case immediate, proves the final claim.
\end{proof}

\begin{proof}[Proof of Corollary~\ref{cor:source-normalization}]
Lemma~\ref{lem:source-internal-inclusion} types $t^\circ$ internally.
Theorem~\ref{thm:normalization} gives a defined, internally typed,
irreducible normal form reachable from $t^\circ$.  The discipline lemmas show that it contains no coherent-sum former
and no inverse distributor whose source or target is non-first-order;
Lemma~\ref{lem:lo-restores-balance} shows that it is balanced.
By Definition~\ref{def:additive-balance}, every $\plus$-map in it
therefore has closed operands.
\end{proof}

%% file: appendix-determinacy.tex
\section{Proofs of Determinacy}
\label{app:determinacy-proofs}

This appendix proves determinacy directly for the open,
normalization-internal calculus.  Quantum atoms and structural atoms are
inert constants for the rewrite relation, so the confluence argument does
not inspect or replace them.  All judgments below are internal judgments
$\Gamma \intjudge t:A$, and all references to $\to$ are to
Definition~\ref{def:rewrite}.  The proof is by well-founded induction on
the normalization measure $\mu$ of Definition~\ref{def:measure}.  At each
rank it establishes, in that order, compatibility with a small
administrative equivalence, closure of one-step peaks, and uniqueness of
irreducible reducts.  The hereditary rule (G) may invoke uniqueness only
at strictly smaller rank.

\subsection{Rig coherence is transport data}

The following result supplies the transport data used in
Appendices~\ref{app:normalization} and~\ref{app:unitarityNEW}.

\begin{theorem}[Rig coherence {\citep[Theorem 3.1]{Laplaza1972}}]
\label{thm:rig-coherence}
In a rig category (a symmetric monoidal category with a compatible
monoidal sum and distributivity), every diagram built from coherence
isomorphisms and distributivity maps between values of \emph{regular}
rig expressions (no repeated letters) commutes.
\end{theorem}

Regularity matters: after repeated letters are instantiated, distinct
canonical maps can have equal endpoints, for example
$\id_{S\plus S}$ and $\sigma^\plus_{S,S}$.

\subsection{The administrative quotient}

Static phase annotations are compared extensionally.  Write
$\operatorname{eval}(\alpha)\in U(1)$ for the complex number obtained
by evaluating the static expression~$\alpha$.

\begin{definition}[Administrative equivalence]
\label{def:administrative-equivalence}
Let $\approx_1$ be one application, in an arbitrary term context, of
one of the following generators.
\begin{enumerate}
\item[\textup{($\alpha$)}] Alpha-conversion of a bound variable.
\item[\textup{(X)}] Exchange of adjacent independent tensor lets:
\[
  \letpair{x}{y}{e_1}{\letpair{p}{q}{e_2}{t}}
  \;\approx_1\;
  \letpair{p}{q}{e_2}{\letpair{x}{y}{e_1}{t}},
\]
provided
$\{x,y\}\cap\mathrm{fv}(e_2)=\emptyset$ and
$\{p,q\}\cap\mathrm{fv}(e_1)=\emptyset$, with bound names chosen
fresh.
\item[\textup{(P)}] Extensional equality of static phase
annotations.  Thus corresponding annotations in either
\[
  \oplusmap{\alpha}{f}{\beta}{g}
  \qquad\text{or}\qquad
  [\,\alpha\cdot t\mid\beta\cdot u\,]
\]
may be replaced by $\alpha',\beta'$ when
$\operatorname{eval}(\alpha)=\operatorname{eval}(\alpha')$ and
$\operatorname{eval}(\beta)=\operatorname{eval}(\beta')$.
\end{enumerate}
Let $\approx$ be the reflexive, symmetric, transitive, and congruence
closure of $\approx_1$.
\end{definition}

No source coherence map is a generator of $\approx$.  In particular,
sum symmetry is computational content, not bureaucracy.

\begin{lemma}[Administrative invariants]
\label{lem:administrative-invariants}
If $t\approx t'$, then:
\begin{enumerate}
\item if $\Gamma\intjudge t:A$, then
  $\Gamma\intjudge t':A$;
\item $\mu(t)=\mu(t')$;
\item $t$ and $t'$ have the same multiset of free variables, and a
  variable occurs linearly in one exactly when it occurs linearly in
  the other;
\item $t$ is $\to$-irreducible if and only if $t'$ is
  $\to$-irreducible.
\end{enumerate}
\end{lemma}

\begin{proof}
It suffices to check the three generators.

Alpha-conversion preserves all four properties by the usual
capture-avoiding convention.  Phase replacement changes no term
constructor, binder, free variable, or redex pattern.  Both annotations
still evaluate to unit-modulus scalars, so the same internal typing rule
applies.  It also preserves the measure, whose components do not inspect
phase expressions, and is an instance of equality of static scalar data.

For (X), invert the two tensor-let typings.  The independence
hypotheses say precisely that neither pair of bound variables is used
in the other scrutinee.  The same four premise derivations can therefore
be reassembled in the opposite order, with the same context and result
type.  The exchange only permutes occurrences, so it preserves the
free-variable multiset and linearity.  It preserves $\Phi$ and
$\mathsf{Bind}$ by inspection.  It also preserves
$\mathsf{LetDepth}$: a let body is not a counted path position, and
the two exchanged binders and all binders below them consequently have
the same depths in either order.

For irreducibility, use
Lemma~\ref{lem:nf-iff-irreducible} and the deep-prefix condition of
Definition~\ref{def:nf4star}.  Alpha-conversion and phase
replacement preserve the normal-form grammar and all its side
conditions.  Exchange merely permutes adjacent members of a let
prefix.  It preserves (NF1)--(NF3).  For the strengthened
\textup{(NF4$^*$)} condition, consider
$\lam{z}{L_1\cdots L_k[R]}$.  Eligibility of $L_j$ for the deep
lambda-float depends on whether $z$ is absent from its scrutinee and
whether $L_j$ is independent of the set of preceding prefix lets.
Exchanging two mutually independent adjacent lets changes only their
order.  It neither changes that preceding set for any later let nor
changes eligibility of either exchanged let.  Hence the existence of
an eligible deep float is invariant, and so is irreducibility.

Congruence and equivalence closure preserve the four displayed
conclusions.
\end{proof}

\begin{lemma}[Reduction preserves linear support]
\label{lem:reduction-linear-support}
If $\Gamma\intjudge t:A$ and $t\to t'$, then $t$ and $t'$ have the
same multiset of free variables.
\end{lemma}

\begin{proof}
For the two beta rules, unrestricted internal substitution
(Lemma~\ref{lem:internal-substitution}) inserts each argument at the
unique occurrence of its bound variable.  Every floater
only reparents one tensor let.  Rule (E) places each of
$f,g,h,k$ once in the two compositions.  Rule (G) places
$f,R_1$ once in the left branch and $g,R_2$ once in the right
branch; induction over the defining computation of
$\mathsf{SumNF}$ shows that its local steps and prefix extraction
neither duplicate nor discard a subterm.  Closure under a linear term
context preserves the multiset equation.
\end{proof}

\subsection{Reusable commuting lemmas}

Write a tensor-let prefix as
$L[-]=L_1\cdots L_k[-]$, where
$L_i[-]=\letpair{x_i}{y_i}{e_i}{[-]}$.

\begin{lemma}[Linear let extrusion]
\label{lem:linear-let-extrusion}
Let $\mathcal C[-]$ be a well-typed linear one-hole term context and
let $\ell[-]=\letpair{x}{y}{e}{[-]}$.  Suppose the bound names are
fresh for $\mathcal C$; every crossed tensor let is independent of
$e$; and, for every crossed abstraction $\lambda z.-$,
$z\notin\mathrm{fv}(e)$.  Then there is a term $w$
such that
\[
  \mathcal C[\ell[u]]\to^*w
  \qquad\text{and}\qquad
  w\approx\ell[\mathcal C[u]].
\]
If no two independently produced let prefixes have to be interleaved,
the second comparison is equality.
\end{lemma}

\begin{proof}
By induction on the path from the hole to the root.  Each application, tensor-constructor, let-scrutinee, $\plus$-map
argument, $\plus$-map branch, and sum-former branch is crossed by its
corresponding floater in Definition~\ref{def:rewrite}.  At a lambda
the generalized rule (H$^*$) crosses the lambda and all independent
prefix lets above the occurrence in one step.  A path through the
body of another tensor let has no oriented floater; defer that
crossing until the end and exchange the two then-adjacent independent
lets by (X).  The same final exchanges reconcile prefixes contributed
by sibling paths when the fixed left-to-right extraction order places
them oppositely.  The side conditions are preserved along the
induction by Lemma~\ref{lem:reduction-linear-support}.
\end{proof}

\begin{lemma}[Independent-let Fubini]
\label{lem:floater-fubini}
Let two non-hereditary let-floaters be simultaneously applicable in
a well-typed term.
\begin{enumerate}
\item If the lets occur in independent sibling holes, or in nested
  holes with neither on the other's redex path, the two orders join
  modulo one or more applications of (X).
\item If one local floater lies on the other's redex path, including
  the self-overlap of (D), the two orders join exactly after
  completing the remaining local floaters on one side.
\end{enumerate}
\end{lemma}

\begin{proof}
For independent holes, write the common constructor as a two-hole
context $\mathcal C[-,-]$.  Both paths float both lets above
$\mathcal C$; the first-floated let is outermost.  Linearity makes
the two scrutinees independent, so the results differ by (X).
Nested off-path holes reduce to the same calculation after the inner
let reaches the shared constructor.

For an on-path overlap, every local floater has shape
\[
 \mathcal C_1[\letpair{x}{y}{e}{t}]
 \to \letpair{x}{y}{e}{\mathcal C_1[t]}.
\]
A step inside $e$ or $t$ is transported verbatim.  The remaining
case is
\[
 \mathcal C_1[
   \letpair{x}{y}{(\letpair{a}{b}{u}{v})}{t}].
\]
Floating the outer occurrence and then using (D) yields
\[
 \letpair{a}{b}{u}{
   \letpair{x}{y}{v}{\mathcal C_1[t]}}.
\]
Using (D) first and then floating through $\mathcal C_1$ twice
reaches the same term.  Taking
$\mathcal C_1=\letpair{p}{q}{[-]}{g}$ gives the
(D)$\times$(D) self-overlap.
\end{proof}

\begin{lemma}[Deep lambda-float peaks]
\label{lem:deep-float-peaks}
Consider the generalized rule
\[
\tag{H$^*$}
 \lam{z}{
   L_1\cdots L_{j-1}[
     \letpair{x}{y}{e}{t}]}
 \;\to\;
 \letpair{x}{y}{e}{
   \lam{z}{L_1\cdots L_{j-1}[t]}},
\]
where $z\notin\mathrm{fv}(e)$ and the moved let is independent of
each crossed prefix let.
Every one-step peak containing an (H$^*$)-step and an ordinary
non-hereditary step joins modulo $\approx$.  More precisely:
\begin{enumerate}
\item two eligible (H$^*$)-steps leave the two extracted independent
  lets in opposite orders and hence close by (X);
\item against beta at
  $(\lam{z}{L[\letpair{x}{y}{e}{t}]})\,a$, the two paths close
  exactly when $L$ is empty and modulo exchanges moving the selected
  let across $L[a/z]$ otherwise;
\item a step inside a crossed prefix, its scrutinee, or the final
  body commutes with (H$^*$);
\item an (H$^*$)-step in an operand of (E), (F$_L$), or (F$_R$)
  can be replayed after that rule; the lambda introduced by the
  $\circ$ abbreviation is crossed by the $j=1$ instance, and any
  additional independent prefix order is reconciled by (X).
\end{enumerate}
\end{lemma}

\begin{proof}
For (1), apply the two deep floats in opposite orders.  After the
first extraction the other selected let remains eligible: neither
its scrutinee nor the set of prefix dependencies changed.  Both
paths then extract the second let, with opposite outer order, and
(X) applies.

For (2), beta first gives
$L[a/z][\letpair{x}{y}{e}{t[a/z]}]$, since
$z\notin\mathrm{fv}(e)$.  The other path extracts the let, floats it
past the outer application by (C), and contracts beta, giving
$\letpair{x}{y}{e}{L[a/z][t[a/z]]}$.  The crossed lets are
independent of $e$, so repeated (X) relates the first result to the
second; for the local case there is no exchange.

For (3), Lemma~\ref{lem:reduction-linear-support} preserves the FV
and independence side conditions.  The deep rule transports the
changed subterm once, after which the same internal step is replayed.
For (4), expand $f\circ h$ as $\lam{w}{f(h\,w)}$.
Linearity places the affected operand once in this body.  Ordinary
floaters carry its selected let to the new lambda, the local instance
of (H$^*$) crosses it, and the corresponding map-branch floater
reaches the result of floating first.  This is the calculation
\[
 (\letpair{x}{y}{e}{f})\circ h
 =
 \lam{w}{(\letpair{x}{y}{e}{f})(h\,w)}
 \to^*
 \letpair{x}{y}{e}{\lam{w}{f(h\,w)}}.
\]
The right branch and an affected inner operand are symmetric.
\end{proof}

\subsection{Ranked coherence and peaks}

Put $\operatorname{rk}(t)=\mu(t)$, ordered lexicographically.  By
Lemma~\ref{lem:measure-decrease},
every reduction step, including (H$^*$), strictly lowers rank; by Lemma~\ref{lem:administrative-invariants},
$\approx$ preserves rank.

\begin{definition}[The three ranked assertions]
\label{def:ranked-determinacy}
For a rank $r$, let:
\begin{description}
\item[$\mathsf{LC}_r$] If
  $\operatorname{rk}(t)=r$, $t\approx_1t'$, and $t\to u$, then
  there are $u\to^*v$ and $t'\to^*v'$ with $v\approx v'$.
\item[$\mathsf{Peak}_r$] If
  $\operatorname{rk}(t)=r$ and $u\leftarrow t\to v$, then there
  are $u\to^*p$ and $v\to^*q$ with $p\approx q$.
\item[$\mathsf{UN}_r$] If
  $\operatorname{rk}(t)=r$, $t\to^*n$, $t\to^*n'$, and
  $n,n'$ are irreducible, then $n\approx n'$.
\end{description}
All terms in these assertions are required to be well typed in the
internal calculus.  We also use the derived assertion
$\mathsf{EqUN}_r$: if $t\approx t'$ has rank $r$ and $t,t'$ reduce
to irreducible $n,n'$, respectively, then $n\approx n'$.
\end{definition}

In the staged induction, $\mathsf{EqUN}_r$ is derived after
$\mathsf{UN}_r$.  The peak proof uses the following syntactic
stability property of hereditary normalization.

\begin{lemma}[Lower-rank stability and naturality of
$\mathsf{SumNF}$]
\label{lem:sumnf-lower-rank}
Fix $r$ and suppose
$\mathsf{LC}_s$, $\mathsf{UN}_s$, and the derived
$\mathsf{EqUN}_s$ hold for every $s<r$.
Assume first that the coherent sum passed to
$\mathsf{SumNF}$ has rank below $r$.
\begin{enumerate}
\item Replacing either branch by a reduct, or by an
  $\approx$-equivalent term, changes the result of
  $\mathsf{SumNF}$ only up to $\approx$.
\item Suppose additionally that
  \[
    \operatorname{rk}
    ([\,\alpha\gamma\cdot(f\,e_L)
       \mid\beta\delta\cdot(g\,e_R)\,])<r.
  \]
  If
\[
 \mathsf{SumNF}^{C,D}_{\gamma,\delta}(e_L,e_R)
 =
 L\bigl[[\,\gamma'\cdot R_L\mid\delta'\cdot R_R\,]\bigr],
\]
then, whenever the displayed terms are typed,
\[
\begin{split}
 &\mathsf{SumNF}^{E,F}_{\alpha\gamma,\,\beta\delta}
       (f\,e_L,\;g\,e_R)
\\[-0.2ex]
 &\hspace{2em}\approx
 L\Bigl[
   \mathsf{SumNF}^{E,F}_{\alpha\gamma',\,\beta\delta'}
       (f\,R_L,\;g\,R_R)
 \Bigr].
\end{split}
\]
The prefix $L$ is reassembled in the fixed left-to-right order;
any alternative interleaving on the two sides differs only by (X).
\item Let $\sigma$ be a capture-avoiding typed linear substitution
with pairwise disjoint supplying contexts.  This includes a
one-variable substitution $[u/x]$ and a simultaneous tensor
substitution $[u/x,v/y]$.  Suppose both the original coherent sum
and its $\sigma$-instance have rank below $r$.  Then
\[
 \mathsf{NF}_{\mathrm{LO}}\!\left(
   \mathsf{SumNF}^{C,D}_{\alpha,\beta}(e_L,e_R)\sigma
 \right)
 \;\approx\;
 \mathsf{SumNF}^{C,D}_{\alpha,\beta}
   (e_L\sigma,e_R\sigma).
 \tag{SN-sub}
\]
In particular, this covers substitution in the left branch, in the
right branch, and simultaneous linear substitution across the two
branches.
\end{enumerate}
\end{lemma}

\begin{proof}
The deterministic normalizer of
Definition~\ref{def:lo-normalizer} returns an irreducible term
reachable from its input and never increases $\Phi$
(Theorem~\ref{thm:lo-normalizer-total}).  For a branch reduction, the two displayed
$\mathsf{SumNF}$ outputs are irreducible reducts of the same
lower-rank coherent sum: one path first takes the branch step under
the sum former, while the other lets the deterministic branch
normalizer take it later.  Apply $\mathsf{UN}_s$ at that lower rank.
For an $\approx$-replacement, the lower-rank assertion
$\mathsf{EqUN}_s$ compares the two irreducible branch-normalizer
outputs.  This gives part~(1).

For part~(2), follow the defining computation producing $L$.
Normalize $e_L,e_R$, carry their accumulated phases
$\gamma',\delta'$, and float their tensor-let prefixes out in the
fixed order.  Applying $f,g$ before this computation or after its
branch results therefore gives two reduction sequences, modulo
independent prefix exchange, from
\[
 [\,\alpha\gamma\cdot(f\,e_L)
    \mid\beta\delta\cdot(g\,e_R)\,].
\]
Both displayed results are irreducible.  This coherent sum is
strictly below the enclosing (G)-redex: the latter has the additional
root $\plus$-map constructor.  Thus $\mathsf{UN}_s$ applies.
Static phase associativity is handled by generator (P).

For part~(3), use well-founded induction on the rank of the
substituted coherent sum.  Replay the defining computation of
$\mathsf{SumNF}$ under $\sigma$, choosing all binders fresh for the
supplying terms.  Ordinary reductions commute with typed linear
substitution, with side conditions preserved by
Lemma~\ref{lem:reduction-linear-support}.  If the computation
encounters a nested \textup{(G)} step, its branch sum lies at
strictly smaller $\Phi$, so the induction hypothesis compares the
substituted recursive call with the recursive call on the
substituted branches.

Consequently the left-hand side and the right-hand side of
\textup{(SN-sub)} are irreducible reducts, modulo
independent-prefix exchange, of the same substituted coherent sum.
Its rank is below $r$, so the lower-rank uniqueness assertion gives
the displayed equivalence.
\end{proof}

\begin{theorem}[Staged coherence, peak closure, and uniqueness]
\label{thm:ranked-determinacy}
For every rank $r$, the assertions
$\mathsf{LC}_r$, $\mathsf{Peak}_r$, and $\mathsf{UN}_r$ hold.
\end{theorem}

\begin{proof}
By well-founded induction on $r$.  Fix $r$ and assume all three
assertions, together with the derived $\mathsf{EqUN}_s$, at every
$s<r$.  We establish the three assertions at $r$ in the stated order
and derive $\mathsf{EqUN}_r$ only afterward.

\paragraph{Stage 1: $\mathsf{LC}_r$.}
Consider one generator $t\approx_1t'$ and a step $t\to u$.

For alpha-conversion, choose the same fresh names when replaying the
rewrite; the reducts are alpha-equivalent.  If the changed occurrence
lies inside one of the four components of a (G)-redex, the induced
change of its branch application is compared by
Lemma~\ref{lem:sumnf-lower-rank}(1), at the strict lower rank obtained
by removing the root map constructor.  For (P), every rule other than
(E) and (G) carries annotations unchanged.  Rule (E) multiplies
corresponding annotations.  Rule (G) passes their products to
$\mathsf{SumNF}$; equality in $U(1)$ is a congruence for
multiplication, and a phase replacement inside a component again
uses Lemma~\ref{lem:sumnf-lower-rank}(1).  Thus the reducts remain
related by (P).

It remains to consider (X).  If the rewrite is disjoint from the
exchange site, or lies in a subterm transported unchanged by the
exchange, the generator and the step commute.  If a beta rule or a
local floater overlaps the two-let prefix, linear let extrusion
(Lemma~\ref{lem:linear-let-extrusion}) and independent-let Fubini
(Lemma~\ref{lem:floater-fubini}) give the required comparison.
Against (H$^*$), exchanging two adjacent independent prefix lets
does not change which deep lets are eligible, as proved in
Lemma~\ref{lem:administrative-invariants}; replay the same selected
deep float.  If the selected let is one of the exchanged pair, the
two results differ by moving the other independent let once.

If the step is (G), the exchange site is either in its surrounding
context or inside one of $f,g,R_1,R_2$.  The former case commutes
by context closure.  In the latter case, (G) hands
$\approx$-related branch applications of rank below $r$ to
$\mathsf{SumNF}$; Lemma~\ref{lem:sumnf-lower-rank}(1) compares the
outputs.  This invokes only the lower-rank induction package.
Thus $\mathsf{LC}_r$ holds.  Notice that no assertion at rank
$r$ was borrowed in this stage.

\paragraph{Stage 2: $\mathsf{Peak}_r$.}
Disjoint redexes commute by context closure.  If one redex is
strictly inside a subterm that the other rule transports once,
replay it in the transported occurrence; linearity prevents
duplication.  The genuine overlaps are exhausted by
Table~\ref{tab:critical-pairs}.

\begin{table}[h]
\caption{Exhaustive critical-pair families.  Rule (F) and (G) have
incompatible scrutinee roots, while (S$_L$)/(S$_R$) do overlap (G)
when an arbitrary sum branch is let-headed; those two cases are
calculated below.  A blocked $\plus$-map introduces no new root
rewrite.}
\label{tab:critical-pairs}
\centering
\footnotesize
\renewcommand{\arraystretch}{1.18}
\begin{tabular}{@{}p{0.23\linewidth}p{0.30\linewidth}p{0.38\linewidth}@{}}
\toprule
\textbf{Families} & \textbf{Root/overlap shapes} & \textbf{Closure} \\
\midrule
$\beta\times\beta$
 & (A)$\times$(A), (A)$\times$(B), (B)$\times$(B)
 & incompatible roots; nested cases by linear substitution \\

$\beta\times$ nested (G)
 & (A) with (G) in the abstraction body;
   (B) with (G) in the let continuation
 & typed-linear substitution naturality
   \textup{(SN-sub)} \\

$\beta\times$ floater
 & (A)$\times$(C$'$), (A)$\times$(H$^*$);
   (B) against every let-floater in a constructor or scrutinee slot
 & Lemma~\ref{lem:linear-let-extrusion};
   Lemma~\ref{lem:deep-float-peaks} \\

floater$\times$floater
 & sibling slots, nested slots, (D)$\times$(D),
   and every pair containing (H$^*$)
 & Lemmas~\ref{lem:floater-fubini}
   and~\ref{lem:deep-float-peaks} \\

(E)$\times$floater
 & (E)$\times$(F), (F$_L$), (F$_R$);
   a floater inside any map operand
 & fuse--float calculation below; deep case by
   Lemma~\ref{lem:deep-float-peaks} \\

(E)$\times$(E)
 & three nested maps
 & beta-normalize both composition bracketings \\

(G)$\times$ordinary
 & (G)$\times$(E), (F$_L$), (F$_R$),
   (S$_L$), (S$_R$);
   a step inside $f,g,R_1,R_2$
 & lower-rank
   Lemma~\ref{lem:sumnf-lower-rank} \\

(G)$\times$(G)
 & identical root, or a (G)-redex in one component
 & identical root reduct; nested case by lower-rank stability \\
\bottomrule
\end{tabular}

\end{table}

Here are the nontrivial ordinary calculations.

For (A)$\times$(C$'$),
\[
 (\lam{x}{e})\,\letpair{a}{b}{u}{v}
\]
reduces by beta to
$e[\letpair{a}{b}{u}{v}/x]$, while floating first and then using
beta gives $\letpair{a}{b}{u}{e[v/x]}$.
The variable $x$ occurs once; Lemma~\ref{lem:linear-let-extrusion}
joins these terms.  The (A)$\times$(H$^*$) calculation is
Lemma~\ref{lem:deep-float-peaks}(2).

A representative (B) peak is
\[
 (\letpair{x}{y}{e\tensor t}{f})\tensor g.
\]
Beta first gives $f[e/x,t/y]\tensor g$.  Floating first and then
contracting gives $(f\tensor g)[e/x,t/y]$, literally the same term
because $x,y$ do not occur in the independent component~$g$.
If the floater is inside the beta scrutinee, beta embeds a let at
the unique substituted occurrence and
Lemma~\ref{lem:linear-let-extrusion} extracts it.  Changing the
enclosing constructor covers (C), (C$'$), (D),
(C$_\tensor^L$), (C$_\tensor^R$), (F), (F$_L$), and (F$_R$).

For (E)$\times$(F), fusing
\[
 \oplusmap{\alpha}{f}{\beta}{g}
 \bigl(\oplusmap{\alpha'}{h}{\beta'}{k}
       (\letpair{x}{y}{e}{t})\bigr)
\]
and then floating reaches
\[
 \letpair{x}{y}{e}{
   \oplusmap{\alpha\alpha'}{f\circ h}
     {\beta\beta'}{g\circ k}\,t}.
\]
Floating through the inner map and then the outer map, followed by
(E), reaches the same term.  For (E)$\times$(F$_L$), expand
\[
 (\letpair{x}{y}{e}{h})\circ h'
 =
 \lam{z}{(\letpair{x}{y}{e}{h})(h'z)}.
\]
One application floater and the local instance of (H$^*$) give
$\letpair{x}{y}{e}{(h\circ h')}$; (F$_L$) then agrees with
floating first and fusing inside the let body.  The right case is
symmetric.  Deeper prefix occurrences are covered by
Lemma~\ref{lem:deep-float-peaks}.

For (E)$\times$(E), the two paths from three nested maps have branch
functions $f\circ(h\circ m)$ and $(f\circ h)\circ m$.
Expanding $\circ$ and contracting one internal beta redex in either
term reaches $\lam{z}{f(h(mz))}$; phases agree by (P).

It remains to check the hereditary overlaps.  We do so with all
phases and without value assumptions on the map operands.  Write
$M_{\alpha,\beta}(f,g)$ for
$\oplusmap{\alpha}{f}{\beta}{g}$.  The (G)$\times$(E) source is
\[
\tag{$*$}
 M_{\alpha,\beta}(f,g)
 \Bigl(
   M_{\alpha',\beta'}(h,k)
   [\,\gamma\cdot R_1\mid\delta\cdot R_2\,]
 \Bigr),
\]
where $f,g,h,k,R_1,R_2$ are arbitrary terms satisfying the internal
typing premises.  Fusing first and then using (G) produces
\[
 \mathsf{SumNF}_{
   (\alpha\alpha')\gamma,\,
   (\beta\beta')\delta}
 \bigl((f\circ h)R_1,\;(g\circ k)R_2\bigr).
\]
Lemma~\ref{lem:sumnf-lower-rank}(1), applied to the $\beta$-reducts
of the two branch applications, exposes the common branch terms
$f(hR_1)$ and $g(kR_2)$.

For the other path, write
\[
 \mathsf{SumNF}_{\alpha'\gamma,\,\beta'\delta}
       (hR_1,kR_2)
 =
 L[[\,\lambda\cdot S_L\mid\kappa\cdot S_R\,]].
\]
After the inner (G), repeated (F) moves $L$ past the outer map, and
the outer (G) produces
\[
 L\Bigl[
   \mathsf{SumNF}_{\alpha\lambda,\,\beta\kappa}
       (fS_L,gS_R)
 \Bigr].
\]
Lemma~\ref{lem:sumnf-lower-rank}(2) joins the branch
normalizations.  Generator (P) supplies precisely
\[
  (\alpha\alpha')\gamma=\alpha(\alpha'\gamma),
  \qquad
  (\beta\beta')\delta=\beta(\beta'\delta).
\]
Every use of uniqueness here is below the rank of~($*$), because
the common branch terms omit its outer map constructor.

For (G)$\times$(F$_L$), start from
\[
 M_{\alpha,\beta}(\letpair{x}{y}{e}{f},g)
 [\,\gamma R_1\mid\delta R_2\,].
\]
The (G)-first path begins its left branch with
\[
 (\letpair{x}{y}{e}{f})R_1
 \to \letpair{x}{y}{e}{fR_1}.
\]
Left-to-right prefix extraction therefore puts this let before all
prefixes subsequently obtained from either branch.  Floating
(F$_L$) first also puts exactly this let outside, and (G) in its body
gives the same term.  Thus the left case is exact.

For (G)$\times$(F$_R$), the (G)-first path extracts all prefixes
from the left branch before the prefix generated by normalizing the
right map operand, including the displayed right-branch let.
Floating (F$_R$) first places that right prefix outside the whole map
and therefore before the left prefixes.  Every let in the right
prefix is independent of every left prefix let by the typing split,
so repeated (X) relates the two fixed orders.  This case is not
claimed exact.

The sum-branch floaters give the analogous comparison, but neither
overlap is generally literal equality.  From
\[
 M_{\alpha,\beta}(f,g)
 [\,\gamma\cdot(\letpair{x}{y}{e}{R_1})
       \mid\delta\cdot R_2\,],
\]
the (G)-first path normalizes
$f(\letpair{x}{y}{e}{R_1})$; rule (C$'$) exposes the displayed let.
Prefixes obtained while normalizing $f$ may therefore be extracted
before the displayed branch let.  The (S$_L$)-first path instead
places that let outside the whole map before $f$ is normalized.  The
typing split makes the two prefix blocks independent, so repeated
applications of (X) join the resulting fixed orders.

For (S$_R$), the (G)-first path places prefixes from the left branch
before the let exposed from the right branch, whereas the
(S$_R$)-first path places the right-branch let first.  The same
independence argument joins these orders by (X).  Thus both
(G)$\times$(S$_{L/R}$) overlaps close modulo (X), hence modulo
$\approx$; literal equality arises only when the intervening
prefixes are empty.

For a beta step crossing a nested \textup{(G)} occurrence, let
$\sigma$ be the substitution performed by the outer beta rule.
Reducing \textup{(G)} first and then beta produces a term containing
\[
 \mathsf{SumNF}(e_L,e_R)\sigma ,
\]
whereas contracting beta first and then \textup{(G)} produces
\[
 \mathsf{SumNF}(e_L\sigma,e_R\sigma).
\]
Clause~\textup{(SN-sub)} of Lemma~\ref{lem:sumnf-lower-rank} joins
these terms modulo $\approx$.  For \textup{(A)},
$\sigma=[u/x]$; for \textup{(B)}, $\sigma=[u/x,v/y]$.  Both
instances of \textup{(SN-sub)} lie strictly below the peak rank:
the \textup{(G)} step consumes the outer $\plus$-map former, so the
compared sums omit its contribution to $\Phi$.  A nested
\textup{(G)} occurrence lying outside the substituted body is
transported once and replays literally.

Finally, a step inside $f,g,R_1$, or $R_2$ changes one lower-rank
branch application.  Lemma~\ref{lem:sumnf-lower-rank}(1), not a
claim that hereditary normalization transports the step verbatim,
joins the two outputs.  This includes an (H$^*$)-step in a map
operand and a nested (G)-step.  A root (G)$\times$(G) peak has the
same uniquely specified macro reduct on both sides.  The
(G)$\times$(S$_{L/R}$) root/component overlaps are the two preceding
calculations.  Blocked maps have no root rule, so reductions in their
operands are already among the component cases.

The table and calculations close every one-step peak at rank $r$
using only the induction package below $r$ in hereditary cases.
Thus $\mathsf{Peak}_r$ holds.

\paragraph{Stage 3: $\mathsf{UN}_r$.}
Let $t\to^*n$ and $t\to^*n'$ with $n,n'$ irreducible.  If one
sequence is empty, $t$ is irreducible and the other is empty as well.
Otherwise write its first steps as
\[
  u\leftarrow t\to v.
\]
By $\mathsf{Peak}_r$, choose
$u\to^*p$, $v\to^*q$ with $p\approx q$.  Termination
(Theorem~\ref{thm:termination}) supplies irreducible continuations
$p\to^*\widehat p$ and $q\to^*\widehat q$.
All of $u,v,p,q$ have rank below $r$.  Lower-rank
$\mathsf{UN}$ gives
\[
 n\approx\widehat p,
 \qquad
 n'\approx\widehat q.
\]
The lower-rank assertion $\mathsf{EqUN}_{\operatorname{rk}(p)}$
applied to $p\approx q$ gives
$\widehat p\approx\widehat q$.  Transitivity yields
$n\approx n'$, proving $\mathsf{UN}_r$.

It remains to derive $\mathsf{EqUN}_r$ for use at higher ranks.
First suppose $t\approx_1t'$, $t\to^*n$, and $t'\to^*n'$, where
$n,n'$ are irreducible.  If the first reduction is empty, then $t$
is irreducible, so administrative irreducibility invariance makes
$t'$ irreducible as well; hence the second reduction is empty and
$n=t\approx t'=n'$.

Otherwise write $t\to u\to^*n$.  Apply $\mathsf{LC}_r$ once to
$t\approx_1t'$ and $t\to u$, obtaining
\[
 u\to^*v,
 \qquad
 t'\to^*v',
 \qquad
 v\approx v'.
\]
Termination gives irreducible continuations
$v\to^*q$ and $v'\to^*q'$.  All of $u,v,v'$ have rank below $r$: $u$ is the first
reduct of $t$, $v$ is a reduct of $u$, and
$v\approx v'$ gives $\mu(v')=\mu(v)<r$ by
Lemma~\ref{lem:administrative-invariants}.  Thus the joining
sequence from $t'$ is necessarily nonempty as well.  Lower-rank $\mathsf{UN}$ at $u$ compares
$n$ with $q$, while lower-rank $\mathsf{EqUN}$ applied to
$v\approx v'$ compares $q$ with $q'$.  Finally the just-proved
$\mathsf{UN}_r$ at $t'$ compares $q'$ with $n'$.  Thus $n\approx n'$.
Choosing irreducible reducts by termination at the intermediate
terms and iterating this one-generator argument along a finite
$\approx$ derivation proves $\mathsf{EqUN}_r$.  The rank-$r$
package is therefore complete; Stages~1 and~2 use only lower-rank
uniqueness.
\end{proof}

\subsection{Determinacy and strategy independence}

\begin{proof}[Proof of Theorem~\ref{thm:determinacy}]
Apply $\mathsf{UN}_{\mu(t)}$ from
Theorem~\ref{thm:ranked-determinacy} to the two reduction
sequences.
\end{proof}

\begin{remark}[Why the quotient is stated explicitly]
Theorem~\ref{thm:determinacy} is stated directly for the
administrative quotient $\approx$.  A conclusion merely up to the
symmetric closure of reduction would be immediate; the ranked proof
establishes the sharper quotient.
\end{remark}

\begin{remark}[Separation from semantic coherence]
No eta law, rig-coherence equation, or semantic equality is used in
the computational proof.  Rig coherence
(Theorem~\ref{thm:rig-coherence}) remains available for canonical
boundary transports.  Separately,
Lemma~\ref{lem:administrative-naturality} shows on
Source-generated normal forms and the lower-$\Phi$ branches generated
during their interpretation that the three generators of $\approx$ act
by canonical derivation-sector transports:
alpha-conversion renames wires,
independent-let exchange is Fubini for disjoint cuts, and phase
equality is scalar equality.  Corollary~\ref{cor:canonical-invariance}
specializes this transport to Source semantics.
\end{remark}

%% file: app-unitarityNEW.tex
\section{Boundary Unitarity}
\label{app:unitarityNEW}

This appendix proves unitarity for the canonical normal forms obtained
from Source programs.  A Source term is first expanded to Raw syntax
and deterministically normalized.  The internal clauses below are used
on that canonical normal derivation and on the strictly smaller branch
derivations generated by the applied-map and coherent-sharing clauses.
The signed occurrence bookkeeping distinguishes the branch-paired
boundary spaces selected by those derivations from their flat type
envelopes.

\begin{remark}[Necessity of the first-order restriction]
\label{rem:why-first-order-witness}
Without the first-order condition on coherent branch results, the term
\[
 \lam{x}{
   \letpair{x'}{f}
     {\caseof{x}{b_0}{\sigma^\plus_{\base,\base}}
                    {b_1}{\id_{\QBool}}}
     {f\,x'}}
 :\QBool\lmark\QBool
\]
would be typable.  Both branches are structural, but the case would
construct a branch-correlated function and the body would later apply
that function to its correlated tag.  On basis states both inputs are
sent to $\ket{1}$; hence
\[
  \alpha\ket{0}+\beta\ket{1}
  \longmapsto (\alpha+\beta)\ket{1}.
\]
The resulting operator has rank one and is not unitary.  The
first-order premises exclude precisely this higher-order branch
result.
\end{remark}

\subsection{Semantic Setup}
\label{subsec:sem-model}

\paragraph{Denotation.}
For a canonical normal derivation of $\Gamma\intjudge N:A$ used in the
Source-normal interpretation below,
Table~\ref{tab:sem-compositional} selects its branch-paired negative and
positive boundary spaces and a map
$\SEM{\Gamma\intjudge N:A}_{\mathrm{NF}}$ between them.  These
derivation-selected subspaces of the flat type envelope are the
derivation's \emph{sectors}; throughout the paper, ``sector'' always
means such a derivation-selected subspace of the corresponding
envelope (e.g.\ the code sectors of the compilation appendices).
The flat type envelope of
\S\ref{sec:boundary-semantics} supplies common coordinates only; it is
not substituted for these derivation-selected spaces.  All maps below
live in $\mathbf{FdHilb}$, where $\tensor$ is the monoidal tensor and
$\oplus$ is the additive biproduct.

The only port-closing identities used below are the typed snake
equations.  With
$\eta_H:\mathbb C\to H\tensor H^*$ and
$\epsilon_H:H^*\tensor H\to\mathbb C$,
\[
 (\id_H\tensor\epsilon_H)\circ
 (\eta_H\tensor\id_H)=\id_H,
 \qquad
 (\epsilon_H\tensor\id_{H^*})\circ
 (\id_{H^*}\tensor\eta_H)=\id_{H^*}.
 \tag{Y}
\]
Each use of (Y) is on the single typed port displayed by the canonical
normal judgment.

\subsection{Through-Maps, Completed Blocks, and the Classified Cut}
\label{subsec:through-completion}

This subsection supplies the referents for ``inactive-branch
completion'' and for the cut notation used by the proof below.  They
are stated once and used verbatim by the constructor cases.

The construction is simultaneous.  An outer strong induction on
$\Phi$ handles branches generated by applied maps and coherent
sharing; at fixed $\Phi$, ordinary clauses follow canonical premises
and classified cuts use the lexicographic order
$(|\mathcal D|,|\mathcal C|)$, producer first.

\paragraph{Context through-maps.}
For a context $\Gamma=x_1{:}A_1,\ldots,x_n{:}A_n$, write
\[
 \varphi_\Gamma=\mathsf{sgn}(A_1)^*\tensor\cdots\tensor\mathsf{sgn}(A_n)^*,
 \qquad
 \mathcal T_\Gamma^\pm=\sem{\partial^\pm(\varphi_\Gamma)} .
\]
After the fixed rig ordering, every monomial of $\varphi_\Gamma$ contributes
one negative and one positive coordinate; matching equal monomial tags
defines a canonical unitary
\[
 Y_\Gamma:\mathcal T_\Gamma^-\longrightarrow\mathcal T_\Gamma^+ .
 \tag{TY}
\]
In a fixed normal-form branch, $Y_\Gamma$ is typed identity wiring
on the existing coordinates of the inactive context.
Let $d_\Gamma^\pm:\mathcal T_\Gamma^\pm\xrightarrow{\;\cong\;}
\mathsf D_\Gamma$ be the two fixed monomial-coordinate readouts.
Then
\[
 d_\Gamma^+Y_\Gamma=d_\Gamma^- .
 \tag{TY-read}
\]

\begin{definition}[Inactive-branch completion]
\label{def:inactive-completion}
The $\plus$-map, applied-map, and coherent-sharing clauses are read
with \emph{completed blocks}.  For branch derivations
$\Gamma_1\intjudge R_1$ and $\Gamma_2\intjudge R_2$ with maps
$U_1,U_2$ on their selected spaces $\mathcal S_i^\pm$, the two
orthogonal completed blocks are
\[
 \widehat{\mathcal S}_{1\mid 2}^\pm
   =\mathcal S_1^\pm\tensor\mathcal T_{\Gamma_2}^\pm,
 \qquad
 \widehat{\mathcal S}_{2\mid 1}^\pm
   =\mathcal T_{\Gamma_1}^\pm\tensor\mathcal S_2^\pm,
\]
embedded in the left and right head-tag blocks of the conclusion
boundary, and the completed denotation is
\[
 \alpha\,(U_1\tensor Y_{\Gamma_2})
 \;\oplus\;
 \beta\,(Y_{\Gamma_1}\tensor U_2) .
 \tag{SU}
\]
The construction applies to $\plus$-map formers with first-order
targets and unrestricted sources, and to the completed branch families
used by applied maps and coherent sharing.
\end{definition}

These completed blocks implement the derivation-indexed completion of
\S\ref{sec:boundary-semantics}: both polarities carry the coordinates
assigned to the opposite branch context, paired by $Y$.

\paragraph{Blockwise boundaries.}
Whenever a boundary is displayed blockwise, its canonical inclusions
$\iota_i^\pm:\mathcal B_i^\pm\to\mathcal B^\pm$ exhibit the whole
derivation-indexed boundary as an orthogonal direct sum:
\[
 (\iota_i^\pm)^\dagger\iota_j^\pm=\delta_{ij}I,
 \qquad
 \sum_i\iota_i^\pm(\iota_i^\pm)^\dagger=I_{\mathcal B^\pm}.
 \tag{BW}
\]
This is immediate from Definition~\ref{def:inactive-completion}: the
completed blocks are, by construction, the orthogonal summands of
the displayed boundary.  A cut is \emph{blockwise} when its conclusion
is assembled from such a family, and \emph{non-block} otherwise.

\begin{definition}[Variable graph sector]
\label{def:variable-graph-sector}
For the fixed rig normal form
\[
 \mathsf{sgn}(T)\cong\bigoplus_{i\in I_T}M_i,
\]
write $M_i^\pm$ for its negative and positive tensor factors and put
\[
 \mathcal Y_T^-=
 \bigoplus_{i\in I_T}\sem{M_i^+\tensor M_i^-},
 \qquad
 \mathcal Y_T^+=
 \bigoplus_{i\in I_T}\sem{M_i^-\tensor M_i^+}.
 \tag{VG}
\]
The variable map
\[
 \mathsf{yank}_T:\mathcal Y_T^-\longrightarrow\mathcal Y_T^+
\]
is the direct sum, over $i\in I_T$, of the canonical tensor
symmetries exchanging the two displayed factors.
\end{definition}

Put $\mathsf D_T:=\mathcal Y_T^-$ and let
\[
 c_T^-:=\id_{\mathcal Y_T^-},
 \qquad
 c_T^+:=\mathsf{yank}_T^\dagger:
 \mathcal Y_T^+\longrightarrow\mathsf D_T .
\]
These are unitary and $c_T^+\mathsf{yank}_T=c_T^-$.  For a
first-order $P$, the fixed basis identifies $\mathsf D_P$ with
$\sem P$; write $\mathcal B_P$ for its basis indices.

For first-order $P$ and $V\in U(\sem P)$, its boundary graph is
derived from the identity graph:
\[
 \mathsf{Graph}_P(V)
 :=
 (I\tensor V)\mathsf{yank}_P:
 \mathcal Y_P^-\longrightarrow(I\tensor V)\mathcal Y_P^+.
 \tag{QGraph}
\]
It is unitary between the displayed spaces, and
$\mathsf{Graph}_P(I)=\mathsf{yank}_P$.  In the canonical pair
coordinates, put
\[
 \mathsf{src}_P\ket{x,y}_-:=\ket x,
 \qquad
 \mathsf{tgt}_P\ket{x,y}_+:=\ket y.
\]
Their restrictions to the displayed graph spaces are unitary and
\[
 \mathsf{tgt}_P\,\mathsf{Graph}_P(V)
 =
 V\,\mathsf{src}_P .
 \tag{QGraph-read}
\]

\paragraph{Structural graph sectors.}
Let $s:T\cong S$ be a primitive structural isomorphism and let
$\pi_s$ be its bijection of the fixed rig-normal-form monomial bases.
Write $\mathfrak S=(\cdot\intjudge s:T\lmark S)$ for this judgment
shape.
For a basis vector $\ket a\in\sem T$, define isometries into the two
polarities of the flat boundary of $s$ by
\[
 g_s^-\ket a=\ket{a,\pi_s(a)}_-,
 \qquad
 g_s^+\ket{\pi_s(a)}=\ket{a,\pi_s(a)}_+,
 \tag{SG}
\]
where the subscripts mean the fixed signed-boundary order.  The selected
atom sectors are $B_s^\pm:=\operatorname{im}(g_s^\pm)$, and the atom map is
\[
 \mathsf{Rig}_{\mathfrak S}(s)
 :=g_s^+\,\mathsf{Rig}(s)\,(g_s^-)^\dagger
 :B_s^-\longrightarrow B_s^+ .
 \tag{SG-map}
\]
The canonical graph coordinates
\[
 \epsilon_s^-:=(g_s^-)^\dagger,
 \qquad
 \epsilon_s^+:=(g_s^+)^\dagger
\]
are unitary onto $\sem T$ and $\sem S$, respectively, and satisfy
\[
 \epsilon_s^+\mathsf{Rig}_{\mathfrak S}(s)
 =\mathsf{Rig}(s)\epsilon_s^- .
 \tag{SG-read}
\]

\paragraph{Source-parametrized graph supports.}
For finite-dimensional spaces $X,Y$ in their fixed bases and an
isometry $F:X\to Y$, write
\[
 \Gamma(F):X\lhook\joinrel\longrightarrow X_{\mathrm{ref}}\tensor Y,
 \qquad
 \Gamma(F)\ket x:=\ket x_{\mathrm{ref}}\tensor F\ket x,
 \tag{Port-support}
\]
extended linearly.  The reference address is retained while the live output may be mixed.

\begin{definition}[Typed whole-port chart]
\label{def:whole-port-chart}
Fix a classified instance $\rho$ at port $P$, and write
$\mathfrak J_\rho$ for its conclusion judgment shape.  The applicable
clause supplies active spaces $H_\rho^\pm$, an inactive context $\Xi_\rho$,
and an active unitary
$\widetilde U_\rho:H_\rho^-\to H_\rho^+$.  Put
\[
 K_\rho^\pm:=H_\rho^\pm\tensor\mathcal T_{\Xi_\rho}^\pm,
 \qquad
 W_\rho:=\widetilde U_\rho\tensor Y_{\Xi_\rho}:
 K_\rho^-\xrightarrow{\;\cong\;}K_\rho^+ .
 \tag{PC-active}
\]

For a non-block first-order instance, restrict the two premise
supports to the displayed occurrences.  After the fixed
premise-boundary reorderings, these restrictions are placed source
graphs
\[
 \begin{aligned}
 h_{\mathcal D,\rho}^\epsilon&=\Gamma(F_{\mathcal D}^\epsilon),
 &F_{\mathcal D}^\epsilon&:X_\rho
      \longrightarrow Y_\rho^\epsilon\tensor\mathsf D_P,\\
 h_{\mathcal C,\rho}^\epsilon&=\Gamma(F_{\mathcal C}^\epsilon),
 &F_{\mathcal C}^\epsilon&:\mathsf D_P\tensor Z_\rho
      \longrightarrow R_\rho^\epsilon .
 \end{aligned}
\]
Their typed link is the placed source graph determined on basis
vectors by
\[
 \begin{aligned}
 &(h_{\mathcal C,\rho}^\epsilon\star_P
   h_{\mathcal D,\rho}^\epsilon)\ket{x,z}\\
 &\quad:=\ket{x,z}_{\mathrm{ref}}\tensor
   \sum_{y,r}\sum_{p\in\mathcal B_P}
   [F_{\mathcal C}^\epsilon]_{r,(p,z)}
   [F_{\mathcal D}^\epsilon]_{(y,p),x}\ket{y,r}.
 \end{aligned}
 \tag{Graph-link}
\]
Here $\sigma_\rho$ is the canonical reassociation that brings the two
displayed $P$ coordinates together; its type is displayed in the
\textup{(FO-whole)} clause below.  Equivalently, up to the fixed
conclusion placement,
\[
 h_{\mathcal C,\rho}^\epsilon\star_P
 h_{\mathcal D,\rho}^\epsilon
 =\Gamma\!\left(
   (I\tensor F_{\mathcal C}^\epsilon)\,
   \sigma_\rho\,
   (F_{\mathcal D}^\epsilon\tensor I)\right).
 \tag{Graph-link-functor}
\]
Tensoring this active link with the placed source graph of
$Y_{\Xi_\rho}$ gives the maps
\[
 h_\rho^\pm:K_\rho^-
 \longrightarrow \mathcal E_{\mathfrak J_\rho}^\pm .
\]
Define the negative- and positive-coordinate insertions by
\[
 j_\rho^-:=h_\rho^-:
 K_\rho^-\lhook\joinrel\longrightarrow
 \mathcal E_{\mathfrak J_\rho}^-,
 \qquad
 j_\rho^+:=h_\rho^+W_\rho^\dagger:
 K_\rho^+\lhook\joinrel\longrightarrow
 \mathcal E_{\mathfrak J_\rho}^+ .
 \tag{PC-yank}
\]
Lemma~\ref{lem:typed-graph-composition} proves these claims.  For a
higher-order instance closed by \textup{(PortCut)}, the variable-spine,
typed eta-unit, or structural clause supplies the same two insertions
by its displayed typed-snake calculation.  Tensor elimination first
applies \textup{(Ten-pack)} and then one of these clauses.

Put
\[
 B_\rho^\pm:=\operatorname{im}(j_\rho^\pm),
 \qquad
 \chi_\rho^\pm
 :=(j_\rho^\pm)^\dagger\!\upharpoonright_{B_\rho^\pm}.
 \tag{PC-def}
\]
Then
\[
 \chi_\rho^\pm:B_\rho^\pm\xrightarrow{\;\cong\;}K_\rho^\pm,
 \qquad
 j_\rho^\pm\chi_\rho^\pm=I_{B_\rho^\pm},
 \qquad
 \chi_\rho^\pm j_\rho^\pm=I .
 \tag{PC}
\]

For a block-family instance the construction is performed in each
component.  Writing
$B_{\rho,i}^\pm:=\operatorname{im}(j_{\rho,i}^\pm)$ and
$\jmath_i^\pm:B_{\rho,i}^\pm\to B_\rho^\pm$ for the resulting tagged
inclusions,
\[
 \begin{gathered}
 B_\rho^\pm=\bigoplus_i\jmath_i^\pm B_{\rho,i}^\pm,
 \qquad
 \chi_\rho^\pm=\bigoplus_i\chi_{\rho,i}^\pm,\\
 (\jmath_i^\pm)^\dagger\jmath_j^\pm=\delta_{ij}I,
 \qquad
 \sum_i\jmath_i^\pm(\jmath_i^\pm)^\dagger=I_{B_\rho^\pm}.
 \end{gathered}
 \tag{PC-block}
\]
For a syntactic completed block these are the tagged injections of
\textup{(BW)}; for a transported family they are the linked component
supports.

For an $\mathsf{AppCut}_A$ with residual result $B$, the complete
equal-address sector $\mathcal Y_B^\pm$
(Definition~\ref{def:variable-graph-sector}) belongs to the active
factor $H_\rho^\pm$; closing the displayed $A$-port retains that
sector.
\end{definition}

\begin{lemma}[Residual chart for a variable spine]
\label{lem:variable-spine-chart}
Let
\[
 T_{j-1}=A_j\lmark T_j\quad(1\leq j\leq k),
 \qquad T_k=C,
\]
and let $\mathcal N_j$ be the canonical derivation
\[
 x:T_0,\Gamma_1,\ldots,\Gamma_j
 \intjudge x\,R_1\cdots R_j:T_j ,
\]
where the displayed contexts are pairwise disjoint.  Let
$\mathcal R_i$ be the canonical derivation of
$\Gamma_i\intjudge R_i:A_i$, and set
\[
 \mathcal S_i^\pm:=B_{\mathcal R_i}^\pm,
 \qquad
 U_i:=U_{\mathcal R_i}
 =\SEM{\Gamma_i\intjudge R_i:A_i}_{\mathrm{NF}} .
\]
Let $r_j^\pm$ be the fixed ambient rig-boundary basis permutation
which places, in order, the complete boundary coordinates
contributed by the operand derivations, followed by the
equal-address residual $T_j$-coordinates.  Define
\[
 B_{\mathcal N_j}^\pm
 :=
 (r_j^\pm)^{-1}\!\left[
   \left(\bigotimes_{i=1}^{j}\mathcal S_i^\pm\right)
   \tensor\mathcal Y_{T_j}^\pm
 \right],
 \qquad
 \omega_j^\pm
 :=
 r_j^\pm\big|_{B_{\mathcal N_j}^\pm}.
 \tag{VS-chart}
\]
Then
\[
 \omega_j^\pm:B_{\mathcal N_j}^\pm
 \xrightarrow{\;\cong\;}
 \left(\bigotimes_{i=1}^{j}\mathcal S_i^\pm\right)
 \tensor\mathcal Y_{T_j}^\pm
\]
are unitary.  For $j=0$, the empty tensor is $\mathbb C$ and
$B_{\mathcal N_0}^\pm=\mathcal Y_{T_0}^\pm$.  Writing $v_T^\pm:=(c_T^\pm)^\dagger$, the corresponding
selected charts and source-parametrized graph supports are
\[
 \begin{aligned}
 \epsilon_{\mathcal N_j}^\pm
  &=\left(\left(\bigotimes_{i=1}^j
       \epsilon_{\mathcal R_i}^\pm\right)\tensor c_{T_j}^\pm\right)
       \omega_j^\pm,\\
 h_{\mathcal N_j}^\pm
  &=(r_j^\pm)^{-1}\left[
       \left(\bigotimes_{i=1}^j h_{\mathcal R_i}^\pm\right)
       \tensor v_{T_j}^\pm\right].
 \end{aligned}
 \tag{VS-sector}
\]
Inside a completed orthogonal block, the displayed factors form the
active part of \textup{(PC)}; the factor
$\mathcal T_{\Xi_\rho}^\pm$ and its through-map $Y_{\Xi_\rho}$ are
appended exactly as in Definition~\ref{def:whole-port-chart}.
This chart statement is purely combinatorial: it asserts an
equality of coordinate subspaces in the canonical spine chart, not
a tensor factorization of the flat envelope.
\end{lemma}
\begin{proof}
Induction on $j$.  For $j=0$, the displayed subspace is the
equal-address variable sector of
Definition~\ref{def:variable-graph-sector}.  For the step, write
$T_j=A_{j+1}\lmark T_{j+1}$.  Under the fixed polarity repartition,
the residual coordinates $\mathcal Y_{T_j}^\pm$ split into the
complete $A_{j+1}$-port coordinates and the equal-address
$T_{j+1}$-coordinates, while the operand derivation
$\mathcal R_{j+1}$ contributes the block
$\mathcal S_{j+1}^\pm$.  The permutation $r_{j+1}^\pm$ places these
blocks in the displayed order; since it is a fixed basis
permutation of the ambient rig boundary, its restriction
$\omega_{j+1}^\pm$ is unitary onto the displayed subspace.  In a
completed block, the same bookkeeping is tensored with
$\mathcal T_{\Xi_\rho}^\pm$.
\end{proof}

\begin{definition}[Classified eliminator instances]
\label{def:classified-eliminator}
A classified eliminator instance is an occurrence of the
derivation-directed port-cut construction of
Definition~\ref{def:classified-cut} while interpreting a Source normal
form, or one of the lower-$\Phi$ branch derivations generated by
$\mathsf{MapApply}$ or coherent sharing.  For a canonical consumer
\(\mathcal C:(\Gamma,p:P\intjudge N:Q)\) and canonical producer
\(\mathcal D:(\Delta\intjudge R:P)\), the selected \(P\)-port
determines a unique occurrence in the canonical derivation of
\(\mathcal C\).  The last rule exposing that occurrence determines
the cut clause.  When that clause delegates to a premise derivation,
classification continues on that premise.

We write \(\rho\) for the resulting cut instance and
\(\Phi(\rho)\) for the rank of its canonical conclusion derivation.
Its clause determines the spaces \(H_\rho^\pm\), \(\Xi_\rho\),
\(B_\rho^\pm\), the charts \(\chi_\rho^\pm\), and any
recursive cut instances used below.
\end{definition}

\begin{definition}[Classified derivation-indexed port cut]
\label{def:classified-cut}
Let \(\rho\) be the classified eliminator instance determined by
Definition~\ref{def:classified-eliminator}, with matched port
\(P\), canonical consumer
\(\mathcal C:(\Gamma,p:P\intjudge N:Q)\), and canonical producer
\(\mathcal D:(\Delta\intjudge R:P)\).  The construction assigns the
selected conclusion spaces and the map together:
\[
 \mathsf{PortCut}^{\rho}_{P}(\mathcal C,\mathcal D):
 B_{\rho}^{-}\longrightarrow B_{\rho}^{+}.
 \tag{Cut-ty}
\]
The superscript \(\rho\) is suppressed when the instance is clear.
An endpoint consumer uses \textup{(EndAct)}.  Otherwise, at a
first-order port, a non-block producer uses \textup{(FO-whole)} and a
completed or endpoint-transported producer family uses
\textup{(FO-family-cut)}.  The variable and endpoint equations below
record the corresponding specialized coordinates.  At a higher-order
port, if the distinguished consumer occurrence is the argument of a
variable-headed spine, \textup{(C-var-app)} takes priority; otherwise a
variable-headed producer consumed whole uses \textup{(C-var-whole)}.
A structural-headed producer uses \textup{(C-str)}.  Tensor elimination
first forms the packaged consumer \textup{(Ten-pack)} and then selects
one of these clauses.  A distributor on the result pair of an applied
map selects the coherent-sharing clause before \textup{(C-str)}.
Classification is repeated after every recursive call on the rebuilt
canonical premise.  The clauses below complete the construction.
\end{definition}

\paragraph{Consumer networks and grafts.}
Whenever the selected port of
Definition~\ref{def:classified-eliminator} is mentioned, it is a
distinguished axiom leaf.  A \(k\)-port consumer network is a finite
canonical typing-derivation tree
\[
 \mathcal C=(\Pi;a_1,\ldots,a_k):
 \Gamma,p_1{:}P_1,\ldots,p_k{:}P_k\intjudge N:Q,
\]
where the \(a_i\) are pairwise distinct \textsc{Var} leaves concluding
\(p_i{:}P_i\intjudge p_i:P_i\); \(k=1\) gives a one-port consumer.
All edge and root typings are inherited from \(\Pi\).  The formal
consumers below abbreviate such trees with their prescribed packages
as decorations.

For \(I\subseteq\{1,\ldots,k\}\), the partial graft
\[
 \mathcal C[\vec a_I:=\vec{\mathcal D}_I]
\]
replaces \(a_i\) by a canonical producer
\(\mathcal D_i:\Delta_i\intjudge R_i:P_i\), with disjoint fresh
contexts, and rebuilds the rule instances above it.  Its root typing is
\[
 \Gamma,(\vec p{:}\vec P)_{\bar I},
 \biguplus_{i\in I}\Delta_i
 \intjudge N[\vec R_I/\vec p_I]:Q.
\]
Its \emph{consumer skeleton} collapses each grafted producer package to
one leaf; \( |\mathcal C| \) is the number of rule nodes in this
skeleton.  If \(\mathcal A\) has root type \(S\) and \(\mathcal C\)
has marked leaf \(a:S\), define
\(\mathcal C\circ_S\mathcal A:=\mathcal C[a:=\mathcal A]\), retaining
the distinguished leaves of \(\mathcal A\).

Evaluation is bottom-up: an open axiom has its variable package, a
grafted leaf has \((B_{\mathcal D}^{\pm},U_{\mathcal D})\), and each
rebuilt rule applies its existing semantic clause.

For distinct leaves \(a,b\),
\[
 \mathcal C[a:=\mathcal D][b:=\mathcal E]
 =\mathcal C[b:=\mathcal E][a:=\mathcal D],\qquad
 \mathcal C\circ_S(\mathcal A[b:=\mathcal D])
 =(\mathcal C\circ_S\mathcal A)[b:=\mathcal D].       \tag{Graft}
\]
These are equal canonical trees after fixed context ordering and
least-fresh renaming.  Last-rule induction proves them: use the
induction hypothesis when both marks lie in one premise; otherwise
the grafts act in disjoint premise trees.  Bottom-up evaluation is
therefore independent of graft order.  At fixed $\Phi$, the cut
recursion is lexicographic in \((|\mathcal D|,|\mathcal C|)\).

\paragraph{Variable application.}
For a variable-headed neutral prefix
$E=x\,R_1\cdots R_j:A\lmark B$, let
$\mathcal A_E:(\Gamma_E,p:A\intjudge E\,p:B)$
be its canonical formal consumer.  Put
\[
 \mathsf{AppCut}^{\rho}_{A}(\mathcal E,\mathcal D)
 :=
 \mathsf{PortCut}^{\rho}_{A}(\mathcal A_E,\mathcal D),
\]
with active map
\[
 \widetilde U_\rho=
 \bigl(U_1\tensor\cdots\tensor U_j\tensor U_{\mathcal D}\bigr)
 \tensor\mathsf{yank}_B
 \tag{C-var-app}
\]
in the variable-spine chart.  The residual $B$-boundary, the
preceding operand factors, and all unmatched context factors are
retained.  For $j=0$ this is the typed snake equation.

\paragraph{First-order whole-port closure.}
Suppose that the matched port $P$ is first-order.  For the current
classified occurrence $\rho$, first separate the single inactive
factor $Y_{\Xi_\rho}$ identified by the classifier.  Write
$B_{\mathcal D,\rho}^{\pm,\mathrm{act}}$ and
$B_{\mathcal C,\rho}^{\pm,\mathrm{act}}$ for the remaining active
premise sectors.  Their occurrence charts expose the matched port:
\[
 \begin{aligned}
 e_{\mathcal D,\rho}^-:
 B_{\mathcal D,\rho}^{-,\mathrm{act}}
 &\xrightarrow{\;\cong\;}X_\rho,
 &e_{\mathcal D,\rho}^+:
 B_{\mathcal D,\rho}^{+,\mathrm{act}}
 &\xrightarrow{\;\cong\;}Y_\rho\tensor\mathsf D_P,\\
 e_{\mathcal C,\rho}^-:
 B_{\mathcal C,\rho}^{-,\mathrm{act}}
 &\xrightarrow{\;\cong\;}\mathsf D_P\tensor Z_\rho,
 &e_{\mathcal C,\rho}^+:
 B_{\mathcal C,\rho}^{+,\mathrm{act}}
 &\xrightarrow{\;\cong\;}W_\rho^{\mathrm{out}} .
 \end{aligned}
 \tag{FO-expose}
\]
Thus the active premise actions are the conjugates
\[
 \begin{aligned}
 V_{\mathcal D}^{\rho}
 &:=
 e_{\mathcal D,\rho}^+
 U_{\mathcal D}^{\mathrm{act}}
 (e_{\mathcal D,\rho}^-)^\dagger:
 X_\rho\xrightarrow{\;\cong\;}
 Y_\rho\tensor\mathsf D_P,\\
 V_{\mathcal C}^{\rho}
 &:=
 e_{\mathcal C,\rho}^+
 U_{\mathcal C}^{\mathrm{act}}
 (e_{\mathcal C,\rho}^-)^\dagger:
 \mathsf D_P\tensor Z_\rho
 \xrightarrow{\;\cong\;}W_\rho^{\mathrm{out}} .
 \end{aligned}
 \tag{FO-action}
\]
These charts are reconstructed locally from the canonical derivation;
they are not additional fields of the semantic package.

For a non-block producer, define
\[
 \widetilde U_\rho
 :=(I_{Y_\rho}\tensor V_{\mathcal C}^{\rho})\,
   \sigma_\rho\,
   (V_{\mathcal D}^{\rho}\tensor I_{Z_\rho})
 :
 X_\rho\tensor Z_\rho
 \xrightarrow{\;\cong\;}
 Y_\rho\tensor W_\rho^{\mathrm{out}} ,
 \tag{FO-whole}
\]
where
\[
 \sigma_\rho:
 (Y_\rho\tensor\mathsf D_P)\tensor Z_\rho
 \xrightarrow{\;\cong\;}
 Y_\rho\tensor(\mathsf D_P\tensor Z_\rho)
\]
is the fixed associator.  Thus
$H_\rho^-:=X_\rho\tensor Z_\rho$ and
$H_\rho^+:=Y_\rho\tensor W_\rho^{\mathrm{out}}$.

The conclusion-sector insertions for
\textup{(FO-whole)} are supplied by
Lemma~\ref{lem:typed-graph-composition} below.

For a completed or transported producer family, the active charts give
\[
 V_{\mathcal D,i}^{\rho}:X_{\rho,i}
   \xrightarrow{\;\cong\;}Y_{\rho,i}\tensor D_i,
 \qquad
 s_i:D_i\lhook\joinrel\longrightarrow\mathsf D_P,
 \tag{FO-family}
\]
where $V_{\mathcal D,i}^{\rho}$ is the active completed-branch action
after $\Xi_{\rho_i}$ has been separated.  It includes the branch
phase, but not $Y_{\Xi_{\rho_i}}$, which is inserted exactly once
below.  The port inclusions satisfy
\[
 s_i^\dagger s_j=\delta_{ij}I,
 \qquad
 \sum_i s_i s_i^\dagger=I_{\mathsf D_P}.
 \tag{FO-part}
\]
The $s_i$ are the completed inclusions, or their images under an
endpoint acting above the completed producer.

If a pair producer has family premises, regard a non-family premise
as a singleton family and use the product family.  After the fixed
tensor reordering,
\[
 D_{(i,j)}:=D_i\tensor D_j,\qquad
 s_{(i,j)}:=s_i\tensor s_j,\qquad
 V_{\mathcal D,(i,j)}^\rho
 :=\pi_\tensor
   (V_{\mathcal D_1,i}^\rho\tensor
    V_{\mathcal D_2,j}^\rho).
 \tag{FO-pair-family}
\]
Its inclusions are again orthogonal and jointly exhaustive, since
\[
 \sum_{i,j}(s_i\tensor s_j)(s_i\tensor s_j)^\dagger
 =
 \Bigl(\sum_i s_i s_i^\dagger\Bigr)\tensor
 \Bigl(\sum_j s_j s_j^\dagger\Bigr)=I .
\]

Put
\[
 \begin{aligned}
 k_i&:=V_{\mathcal C}^{\rho}(s_i\tensor I_{Z_\rho}),
 &W_i&:=\operatorname{im}(k_i),\\
 \widehat k_i&:D_i\tensor Z_\rho\xrightarrow{\;\cong\;}W_i,
 &\widetilde U_{\rho,i}
   &:=(I_{Y_{\rho,i}}\tensor\widehat k_i)\,
      \sigma_{\rho,i}\,
      (V_{\mathcal D,i}^{\rho}\tensor I_{Z_\rho}).
 \end{aligned}
 \tag{FO-family-cut}
\]
Set
\[
 H_{\rho_i}^-:=X_{\rho,i}\tensor Z_\rho,
 \qquad
 H_{\rho_i}^+:=Y_{\rho,i}\tensor W_i .
\]
The component instance of \textup{(PC-yank)} supplies
$j_{\rho,i}^\pm$ and
$B_{\rho,i}^\pm=\operatorname{im}(j_{\rho,i}^\pm)$.  Define
\[
 U_{\rho,i}
 :=
 j_{\rho,i}^+
 \bigl(\widetilde U_{\rho,i}\tensor
       Y_{\Xi_{\rho_i}}\bigr)
 (j_{\rho,i}^-)^\dagger
 :
 B_{\rho,i}^-\xrightarrow{\;\cong\;}B_{\rho,i}^+ .
 \tag{FO-family-component}
\]
With the tagged inclusions $\jmath_i^\pm$ of \textup{(PC-block)},
define the parent action
\[
 \mathsf{PortCut}_{P}^{\rho}(\mathcal C,\mathcal D)
 :=
 \sum_i\jmath_i^+U_{\rho,i}(\jmath_i^-)^\dagger
 :
 B_\rho^-\xrightarrow{\;\cong\;}B_\rho^+ .
 \tag{FO-family-close}
\]
No additional scalar occurs in this sum: the phase of component $i$
is already part of $V_{\mathcal D,i}^{\rho}$.  In the non-block case,
\textup{(PortCut)} inserts $Y_{\Xi_\rho}$ once.

\paragraph{Typed eta unit and variable-headed producer.}
Let the producer be $x\,R_1\cdots R_j:S$.  After separating the
inactive completion $\Xi_\rho$, write
$(C_\rho^-,C_\rho^+,U_{\mathcal C,\rho})$ for the active selected
consumer package and put
\[
 \begin{aligned}
 H_\rho^\epsilon
   &:=
   \left(\bigotimes_{i=1}^{j}B_{\mathcal R_i}^\epsilon\right)
   \tensor C_\rho^\epsilon,\\
 \widetilde U_\rho
   &:=
   \left(\bigotimes_{i=1}^{j}U_{\mathcal R_i}\right)
   \tensor U_{\mathcal C,\rho}.
 \end{aligned}
 \tag{C-var-whole}
\]

For every Source type $S$ and fresh $z:S$, put
$i_S:=\mathsf{NF}_{\mathrm{LO}}(\eta_S(z))$.  First construct its
identity chart by induction on Source types, treating every
first-order $P$ as a leaf.  At the current $\Phi$-stratum, define its
graft at the distinguished leaf of the canonical consumer separately
by lexicographic recursion on
\[
  (|\mathcal C|,\;\text{Source height of }S).
  \tag{Eta-order}
\]
The next rule on the marked path selects one of the
following defining rows, read from top to bottom:
\[
\begin{array}{c@{\quad}l}
 \text{root \textsc{Var}}
   & \textup{(VS-sector) followed by the typed yank},\\[0.3ex]
 P\ \text{(non-root)}
   & \textup{(PC-yank) for (FO-whole) with }\mathsf{Graph}_P(I),\\[0.3ex]
 S_1\tensor S_2\ \text{exposed by $\tensor$-E}
   & \text{the two subtype grafts, then \textup{(Ten-pack)}},\\[0.3ex]
 S_1\lmark S_2\ \text{used as the function head}
   & \text{the constructed }S_1\text{ identity chart and}\\[-0.1ex]
   & \text{the recursive }S_2\text{ graft through \textup{(C-var-app)}},\\[0.3ex]
 \text{ordinary non-cut }\kappa
   & \text{the fixed }\kappa\text{-insertion after its marked premise},\\[0.3ex]
 \text{cut }\tau\text{ at a distinct port}
   & \textup{(Eta-foreign)},\\[0.3ex]
 \text{shared block}
   & \text{the tagged direct sum of its component grafts}.
\end{array}
\tag{Eta-graft}
\]
The named constructor clauses include their fixed ambient insertions,
so these rows determine the graft uniquely.  For a classified cut
$\tau$ at a port occurrence distinct from the $S$-mark, write
\[
 h_\tau^\epsilon
 =
 \mathsf{Supp}_\tau^\epsilon
   (h_1^\epsilon,\ldots,h_m^\epsilon)
\]
for the source-support expression in its selected defining clause.
This abbreviates that clause's graph link, variable-spine insertion,
structural transport, tensor repartition, or tagged block sum; it does
not invoke $\mathsf{PortCut}$ again.  If the mark occurs in premise
$r$, let $W_r$ be the coordinate action of its recursively rebuilt
chart and put
\[
 \widehat h_r^-:=\mathsf{ug}_{\mathcal A_r,S}^-,
 \qquad
 \widehat h_r^+:=\mathsf{ug}_{\mathcal A_r,S}^+W_r .
\]
Keeping the other source supports and the classifier data fixed, the
same selected clause gives a coordinate action $W_{\tau,S}$ and
source supports
\[
 h_{\tau,S}^\epsilon
 :=
 \mathsf{Supp}_\tau^\epsilon
 \bigl(h_1^\epsilon,\ldots,\widehat h_r^\epsilon,
       \ldots,h_m^\epsilon\bigr).
\]
Define
\[
 \mathsf{ug}_{\tau,S}^-:=h_{\tau,S}^-,
 \qquad
 \mathsf{ug}_{\tau,S}^+
 :=h_{\tau,S}^+W_{\tau,S}^\dagger .
 \tag{Eta-foreign}
\]
The marked-premise consumer skeleton is strictly smaller.  The
selected sectors after rebuilding are the images of these insertions.
For example, when \(\tau\) is a first-order graph link and the marked
premise is the consumer, its two source supports are
\[
 \mathsf{ug}_{\mathcal C,S}^-\star_P h_{\mathcal D}^- ,
 \qquad
 (\mathsf{ug}_{\mathcal C,S}^+W_r)\star_P h_{\mathcal D}^+ .
\]
Equation~\textup{(Graph-link-functor)} and the recursively constructed
chart for that premise give \textup{(Eta-foreign)}.

The variable-spine chart supplies the operand insertions and its
residual $\mathcal Y_S$ factor.  The preceding recursion contracts
that factor with the typed eta unit, applies the fixed conclusion
repartition, and adjoins the placed graph of $Y_{\Xi_\rho}$ exactly
once.  This defines
\[
 \mathsf{ug}_{\rho,S}^\epsilon:
 K_\rho^\epsilon
 \lhook\joinrel\longrightarrow
 \mathcal E_{\mathfrak J_\rho}^\epsilon,
 \qquad
 j_\rho^\epsilon:=\mathsf{ug}_{\rho,S}^\epsilon,
 \qquad \epsilon\in\{-,+\}.
 \tag{C-var-place}
\]
Lemma~\ref{lem:semantic-eta-identity} proves that these maps are
isometric.  They are the insertions used by \textup{(PC-def)} and
\textup{(PortCut)}.

\paragraph{Structural and quantum endpoints.}
Selected sectors are general subspaces of the flat judgment envelope.

For a structural atom $s:T\cong S$ and context $\Delta$, let
\[
 \overline s_{\Delta}^{\pm}:
 \mathcal E_{\Delta\intjudge T}^{\pm}
 \xrightarrow{\;\cong\;}
 \mathcal E_{\Delta\intjudge S}^{\pm}
\]
be the polarity-sorted boundary permutations induced by
$\id_{\mathsf{sgn}(\Delta)^*}\tensor\mathsf{Rig}(s)$ and the fixed
rig ordering.

For a quantum atom
$a=\expi{\theta}{J}:P\lmark P$, where $P$ is first-order, put
\[
 V_a:=e^{i\theta J}=\cos\theta\,I+i\sin\theta\,J,
 \qquad
 \overline a_{\Delta}^{-}:=I,
 \qquad
 \overline a_{\Delta}^{+}:=I\tensor V_a.
\]
The certified-involution judgment gives $J^\dagger=J$ and $J^2=I$,
so $V_a$ is unitary.  These maps are read through the same canonical
boundary ordering.  For an atom
$h:P\lmark Q$, let $\mathcal V_h:(p:P\intjudge h\,p:Q)$ be its
formal-port consumer.

For either atom $h:P\lmark Q$ and any producer
$\mathcal D:(\Delta\intjudge R:P)$, define
\[
 B_{h\mathcal D}^{\pm}
 :=
 \overline h_{\Delta}^{\pm}(B_{\mathcal D}^{\pm}),
 \qquad
 e_{h,\mathcal D}^{\pm}
 :=
 \overline h_{\Delta}^{\pm}
 \!\upharpoonright_{B_{\mathcal D}^{\pm}},
\]
and
\[
 \mathsf{EndAct}_{h}(\mathcal D)
 :=
 e_{h,\mathcal D}^{+}\,
 U_{\mathcal D}\,
 (e_{h,\mathcal D}^{-})^{\dagger}
 :
 B_{h\mathcal D}^{-}
 \longrightarrow
 B_{h\mathcal D}^{+}.
 \tag{EndAct}
\]
This is unitary because $e_{h,\mathcal D}^{\pm}$ and
$U_{\mathcal D}$ are unitary between their displayed spaces.

For a variable producer, typed yanking gives
\[
 \mathsf{EndAct}_{s}(\mathcal D_x)
 =
 \mathsf{Rig}_{\mathfrak S}(s)
 \tag{C-str-app}
\]
and, for a variable producer at a first-order quantum endpoint,
\[
 B_{a\mathcal D_x}^-=\mathcal Y_P^-,
 \qquad
 B_{a\mathcal D_x}^+=(I\tensor V_a)\mathcal Y_P^+,
 \qquad
 \mathsf{EndAct}_{a}(\mathcal D_x)
 =\mathsf{Graph}_P(V_a).
 \tag{C-q-app}
\]
An atom-headed application $h\,R$ has
$\mathcal A_h=\mathcal V_h$, so its application cut is
$\mathsf{EndAct}_h(\mathcal R)$.

An endpoint consumer is terminal and is evaluated by
\textup{(EndAct)}.  If its source and target are first-order, write
\[
 V_h:=
 \begin{cases}
  \mathsf{Rig}(s),&h=s,\\
  e^{i\theta J},&h=\expi{\theta}{J}.
 \end{cases}
\]
For a non-block producer, the root coordinate in
\textup{(FO-expose)} is
\[
 V_{h\mathcal D}^{\rho}
 =(I_Y\tensor V_h)V_{\mathcal D}^{\rho}.
 \tag{C-q}
\]
For a producer family, retain its component maps and transport only
the port inclusions:
\[
 s_i^h:=V_hs_i .
 \tag{Family-transport}
\]
Their ranges remain orthogonal and jointly exhaustive.  The resulting
non-block or family coordinates are consumed by
\textup{(FO-whole)} or \textup{(FO-family-cut)}.

It remains to define structural transport at a higher-order port.
Let $s:T\cong S$, let
$\mathcal C:(\Gamma,p:S\intjudge N:Q)$, and let
$\mathcal C\triangleleft s$ be its formal consumer with distinguished
input $z:T$.  Write
\[
 \overline t_{s,\mathcal C}^{\epsilon}:
 \mathcal E_{\mathcal C\triangleleft s}^{\epsilon}
 \xrightarrow{\;\cong\;}
 \mathcal E_{\mathcal C}^{\epsilon}
\]
for the ambient boundary permutation induced by $s$ at that occurrence.
Define
\[
 \begin{aligned}
 B_{\mathcal C\triangleleft s}^{\epsilon}
   &:=(\overline t_{s,\mathcal C}^{\epsilon})^{-1}
      (B_{\mathcal C}^{\epsilon}),\\
 t_{s,\mathcal C}^{\epsilon}
   &:=\overline t_{s,\mathcal C}^{\epsilon}
      \!\upharpoonright_{B_{\mathcal C\triangleleft s}^{\epsilon}},\\
 U_{\mathcal C\triangleleft s}
   &:=(t_{s,\mathcal C}^{+})^{-1}
      U_{\mathcal C}t_{s,\mathcal C}^{-}.
 \end{aligned}
 \tag{Str-pre}
\]
If $g_{\mathcal C}^\epsilon$, $\epsilon_{\mathcal C}^\epsilon$, and
$V_{\mathcal C}$ are its selected chart, set
\[
 g_{\mathcal C\triangleleft s}^\epsilon
 :=(\overline t_{s,\mathcal C}^{\epsilon})^{-1}
   g_{\mathcal C}^\epsilon,
 \qquad
 \epsilon_{\mathcal C\triangleleft s}^\epsilon
 :=\epsilon_{\mathcal C}^\epsilon t_{s,\mathcal C}^\epsilon .
 \tag{Str-chart}
\]
The coordinate action is unchanged, and
$g_{\mathcal C\triangleleft s}^+
 V_{\mathcal C}
 =U_{\mathcal C\triangleleft s}
  g_{\mathcal C\triangleleft s}^-$.
The formal consumer $\mathcal C\triangleleft s$ has the same consumer
tree, distinguished path, and last rule exposing that path as
$\mathcal C$, with the distinguished boundary retyped from $S$ to $T$.
Only its boundary package is transported by \textup{(Str-pre)}.
Thus $t_{s,\mathcal C}^{\epsilon}$ is a unitary from the transported
sector onto $B_{\mathcal C}^{\epsilon}$, and
\[
 \mathsf{PortCut}_{S}^{\rho}
   (\mathcal C,\mathcal D_{s\,R})
 :=
 \mathsf{PortCut}_{T}^{\rho'}
   (\mathcal C\triangleleft s,\mathcal R).
 \tag{C-str}
\]
The recursive producer $\mathcal R$ is proper.  Equation
\textup{(C-q)} is the corresponding first-order endpoint coordinate;
its cut is terminal through first-order whole-port closure.

\paragraph{Tensor elimination.}
For $\mathcal N:(\Gamma,x:A,y:B\intjudge N:C)$, write
\[
 U_{\mathcal N}:B_{\mathcal N}^-\longrightarrow B_{\mathcal N}^+
\]
for its already constructed proper-premise map, and form the
canonical packaged consumer
\[
 \mathcal B_N:
 \Gamma,p:A\tensor B
 \intjudge
 \letpair{x}{y}{p}{N}:C.
\]
Let
\[
 \overline\vartheta_N^\epsilon:
 \mathcal E_{\Gamma,p:A\tensor B\intjudge C}^{\epsilon}
 \xrightarrow{\;\cong\;}
 \mathcal E_{\Gamma,x:A,y:B\intjudge C}^{\epsilon},
 \qquad \epsilon\in\{-,+\},
\]
be the canonical boundary repartition which replaces the complete
$p:A\tensor B$ context port by the ordered ports $x:A,y:B$, in the
fixed rig ordering.  Put
\[
\begin{aligned}
 B_{\mathcal B_N}^{\epsilon}
 &:=
 (\overline\vartheta_N^\epsilon)^{-1}
 (B_{\mathcal N}^{\epsilon}),\\
 \vartheta_N^\epsilon
 &:=
 \overline\vartheta_N^\epsilon
 \!\upharpoonright_{B_{\mathcal B_N}^{\epsilon}}
 :
 B_{\mathcal B_N}^{\epsilon}
 \xrightarrow{\;\cong\;}
 B_{\mathcal N}^{\epsilon},
\end{aligned}
\]
and define
\[
 U_{\mathcal B_N}
 :=
 (\vartheta_N^+)^{-1}
 U_{\mathcal N}\vartheta_N^-
 :
 B_{\mathcal B_N}^-
 \longrightarrow
 B_{\mathcal B_N}^+ .
 \tag{Ten-pack}
\]
Finally,
\[
 \mathsf{TenCut}_{A,B}^{\rho}(\mathcal N,\mathcal R)
 :=
 \mathsf{PortCut}_{A\tensor B}^{\rho}
   (\mathcal B_N,\mathcal R).
 \tag{TenCut}
\]
The construction follows the two bound leaves jointly.  At the
leaf consumer $N=x\tensor y$, whole-$A\tensor B$ yanking returns
$U_{\mathcal R}$ up to the canonical tensor repartition.  In a Source
normal form a visible pair at a destructuring port would be a tensor
beta-redex.  At every other occurrence its package, formed by the
$\tensor$-I clause, is consumed whole---by \textup{(EndAct)}, by a
first-order whole-port cut, or as an operand factor in a variable
clause.

\paragraph{Closing.}
For every instance whose clause supplies a single active map
$\widetilde U_\rho:H_\rho^-\to H_\rho^+$ and insertions
$j_\rho^\pm$---including a \textup{(C-var-whole)} instance whose
insertions are assembled by the shared-block row of
\textup{(Eta-graft)}---the conclusion map is
\[
 \mathsf{PortCut}^{\rho}_{P}(\mathcal C,\mathcal D)
 :=
 (\chi_\rho^+)^{-1}
 \bigl(\widetilde U_\rho\tensor Y_{\Xi_\rho}\bigr)
 \chi_\rho^- .
 \tag{PortCut}
\]
For $\Xi_\rho=\varnothing$, $Y_{\Xi_\rho}=\id_{\mathbb C}$, so
\textup{(PortCut)} is the typed snake equation \textup{(Y)}.
Endpoint-terminal instances return \textup{(EndAct)} directly,
with their constructed sectors.  A non-block first-order instance uses
\textup{(FO-whole)} followed by \textup{(PortCut)}; a producer family
uses \textup{(FO-family-cut)}--\textup{(FO-family-close)}.  Tensor and
shared-block descent for a variable-headed producer are internal rows
of the unit graft \textup{(Eta-graft)} used by
\textup{(C-var-whole)}; they do not create further classified cuts.

For a variable-headed spine, \textup{(C-var-app)} is applied from
the ultimate head outward;
Lemma~\ref{lem:variable-spine-chart} gives the accumulated operand
maps and the unmatched residual yank at every prefix, as displayed
in \textup{(VS-map)} (Corollary~\ref{cor:variable-spine-map},
proved below).  This is an equality in the canonical spine
coordinates.

The declared atom types (\S\ref{subsubsec:terms}) make the clauses
exhaustive.  A primitive structural atom returns a type with
outermost $\plus$ or $\tensor$, and a quantum atom returns a
first-order type; hence an atom-headed result neutral has the form
$h\,R$.  A nested application such as $a\,(b\,x)$ is constructed
inside-out: first form the package for $b\,x$ by \textup{(EndAct)}, rebuild
the outer application, and record its producer coordinate by
\textup{(C-q)} before applying \textup{(FO-whole)} or
\textup{(FO-family-cut)} at a surrounding first-order cut.  The
distributor-on-result-pair
occurrence of an applied map is defined separately by
Definition~\ref{def:coherent-sharing-clause}.

\begin{corollary}[Variable-spine map]
\label{cor:variable-spine-map}
In the notation of Lemma~\ref{lem:variable-spine-chart},
\[
 \omega_j^+\,
 U_{\mathcal N_j}\,
 (\omega_j^-)^{-1}
 =
 \left(\bigotimes_{i=1}^{j}U_i\right)
 \tensor\mathsf{yank}_{T_j}.
 \tag{VS-map}
\]
Inside a completed block, the right-hand side is further tensored
with $Y_{\Xi_\rho}$ as in \textup{(PortCut)}.
\end{corollary}
\begin{proof}
Induction on $j$.  The case $j=0$ is the variable axiom.
For the step, apply \textup{(C-var-app)} to the preceding prefix
and $U_{j+1}$.  The typed snake equation closes the displayed
$A_{j+1}$-port and leaves the residual
$\mathsf{yank}_{T_{j+1}}$ factor.  This gives the displayed
equation.  At each step, the $P$-identified product with the variable
graph is the typed snake \textup{(Y)}, so \textup{(PC-yank)} gives
\textup{(VS-sector)} as well as the displayed map.  When
$A_{j+1}$ is first-order, this is the instance of \textup{(FO-whole)}
in the same coordinates.
\end{proof}

\begin{definition}[Source consumer for an applied $\plus$-map]
\label{def:map-source-consumer}
Let
\[
 \mathcal F:
 \Gamma_1,\Gamma_2\intjudge
 F=\oplusmap{\alpha}{R_1}{\beta}{R_2}:
 (A_1\plus A_2)\lmark(C_1\plus C_2),
\]
with $C_1,C_2$ first-order and $A_1,A_2$ unrestricted.  In an
instance generated from a Source expansion, each $A_i$ is a Source
type; ``unrestricted'' means that it need not be first-order.
Choose the
fixed least-fresh $z_i:A_i$ and put
\[
 q_i:=R_i\,z_i,
 \qquad
 M_i:=\mathsf{NF}_{\mathrm{LO}}(q_i),
 \qquad
 \mathcal M_i:
 (\Gamma_i,z_i:A_i\intjudge M_i:C_i),
 \tag{Use-branch}
\]
and write $U_i:=U_{\mathcal M_i}$.  Complete each alternative
through the opposite context:
\[
\begin{aligned}
 \widehat B_1^\pm
   &:=B_{\mathcal M_1}^\pm
      \tensor\mathcal T_{\Gamma_2}^\pm,
&
 \widehat U_1
   &:=U_1\tensor Y_{\Gamma_2},
\\
 \widehat B_2^\pm
   &:=\mathcal T_{\Gamma_1}^\pm
      \tensor B_{\mathcal M_2}^\pm,
&
 \widehat U_2
   &:=Y_{\Gamma_1}\tensor U_2 .
\end{aligned}
\tag{Use-block}
\]
The canonical source consumer is the formal consumer
\[
 \mathcal U_F:
 \Gamma_1,\Gamma_2,p:A_1\plus A_2
 \intjudge F\,p:C_1\plus C_2 ,
\]
with selected spaces
$B_{\mathcal U_F}^\pm=\widehat B_1^\pm\oplus\widehat B_2^\pm$,
canonical branch inclusions
$j_i^\pm:\widehat B_i^\pm\to B_{\mathcal U_F}^\pm$ satisfying
\[
 (j_i^\pm)^\dagger j_k^\pm=\delta_{ik}I,
 \qquad
 \sum_i j_i^\pm(j_i^\pm)^\dagger
 =I_{B_{\mathcal U_F}^\pm},
 \tag{Use-part}
\]
The maps $j_i^\pm$ are the standard direct-sum injections of the
two displayed \textup{(Use-block)} spaces.  The prescribed active
map is
\[
 U_{\mathcal U_F}
 :=
 j_1^+\,\alpha\widehat U_1(j_1^-)^\dagger
 +
 j_2^+\,\beta\widehat U_2(j_2^-)^\dagger .
 \tag{Use-map}
\]
For $\mathcal E:(\Delta\intjudge E:A_1\plus A_2)$, define
\[
 \mathsf{MapApply}_{A_1\plus A_2}^{\rho}
   (\mathcal F,\mathcal E)
 :=
 \mathsf{PortCut}_{A_1\plus A_2}^{\rho}
   (\mathcal U_F,\mathcal E).
 \tag{MapApply}
\]
A bare map retains its standalone denotation.
\end{definition}

The generated branch derivations satisfy
\[
 \Phi(M_i)
 \leq\Phi(q_i)
 =\Phi(R_i)
 <\Phi(F\,E),
 \tag{Use-$\Phi$}
\]
by Theorem~\ref{thm:lo-normalizer-total}: application to the fresh
variable adds no $\plus$-map former, and the applied former of
$F\,E$ is consumed.  Thus every $\mathcal M_i$ lies strictly below
the applied map in $\Phi$, and its denotation is supplied by the
outer induction of Theorem~\ref{thm:nf-boundary-unitarity}, never
by the clause being defined.  For a neutral function-typed branch operand, its eta-exposed
representative (Proposition~\ref{prop:focused-representative}) may
be used in the local verification of the fresh application
$R_i\,z_i$.
Lemma~\ref{lem:classified-cut-eta} intertwines the two canonical
application denotations.  The recursive semantic object remains
$\mathcal M_i$, obtained from
$\mathsf{NF}_{\mathrm{LO}}(R_i\,z_i)$.

\begin{lemma}[Source cut reachability and decrease]
\label{lem:classifier-totality}
Every canonical eliminator used while interpreting a Source normal
form, or a lower-$\Phi$ branch generated by $\mathsf{MapApply}$ or
coherent sharing, selects a unique clause of
Definition~\ref{def:classified-cut}.  At a higher-order port, the
Source-reachable direct $\mathsf{PortCut}$ clauses are
\textup{(C-var-app)} and \textup{(C-var-whole)};
\textup{(C-str)} is retained for the Raw structural traversal
internal to coherent sharing.  Every recursive call is smaller in the order of
Lemma~\ref{lem:whole-port-collapse}.
\end{lemma}
\begin{proof}
Inspect the four call sites: application, tensor elimination,
$\mathsf{MapApply}$, and coherent sharing.  An endpoint is terminal.
An \textsc{Exp}-headed application returns a first-order result, while
a higher-order structural head is covered by the structural traversal
below.
At a first-order port, a non-block producer selects
\textup{(FO-whole)}, while a completed or endpoint-transported family
selects \textup{(FO-family-cut)}.

At a higher-order port, an occurrence in the argument position of a
variable-headed spine selects \textup{(C-var-app)}; a variable-headed
producer consumed whole selects
\textup{(C-var-whole)}.  Tensor elimination first forms
\textup{(Ten-pack)}.  A visible pair at its destructuring port would be
a tensor beta-redex, so its remaining higher-order producers are
variable-headed.  Tensor and shared-block descent inside
\textup{(C-var-whole)} are rows of \textup{(Eta-graft)}, not recursive
classified cuts.  A higher-order structural head belongs to the Raw
traversal internal to coherent sharing; its distinguished
case-expansion distributor is handled before \textup{(C-str)}.
Lower-$\Phi$ generated branches inherit these Source
restrictions.

At fixed $\Phi$, only \textup{(C-str)} recursively invokes
$\mathsf{PortCut}$, with a proper producer, and hence decreases the
first component of $(|\mathcal D|,|\mathcal C|)$.  The variable-spine
construction is finite, and the unit graft is governed by
\textup{(Eta-order)}.  Branches generated by $\mathsf{MapApply}$ or
coherent sharing have smaller $\Phi$ by \textup{(Use-$\Phi$)} and
\textup{(CS-NF)}.
\end{proof}

\begin{definition}[Selected boundary chart]
\label{def:complete-boundary-chart}
A selected package
$\mathcal A=(B_{\mathcal A}^-,B_{\mathcal A}^+,U_{\mathcal A})$
has a complete graph chart when its canonical derivation supplies
coordinate spaces, isometric insertions, and a unitary coordinate
action
\[
 \begin{gathered}
 g_{\mathcal A}^\pm:
 \mathsf D_{\mathcal A}^\pm
 \lhook\joinrel\longrightarrow
 \mathcal E_{\mathfrak J_{\mathcal A}}^\pm,
 \qquad
 B_{\mathcal A}^\pm=\operatorname{im}(g_{\mathcal A}^\pm),\\
 \epsilon_{\mathcal A}^\pm
 :=(g_{\mathcal A}^\pm)^\dagger
   \!\upharpoonright_{B_{\mathcal A}^\pm},
 \qquad
 U_{\mathcal A}g_{\mathcal A}^-
 =g_{\mathcal A}^+V_{\mathcal A}.
 \end{gathered}
 \tag{Chart}
\]
Equivalently,
$V_{\mathcal A}=\epsilon_{\mathcal A}^+U_{\mathcal A}
(\epsilon_{\mathcal A}^-)^\dagger$.
All source-coordinate domains used by $h_{\mathcal A}^\pm$, and the
factors exposed from them, retain the source-reference indices of the
canonical derivation, in fixed context order and fixed DNF order within
each type.  Every constructor and classified cut retains this order
through its fixed conclusion placement or reordering.  Thus the residual
factor left by exposing a family of leaves is the induced ordered
complement of their indices.  Branch computation is recorded in
$V_{\mathcal A}$, not in this coordinate order.

Both polarities are compared in source coordinates through
\[
 h_{\mathcal A}^-:=g_{\mathcal A}^-,
 \qquad
 h_{\mathcal A}^+:=g_{\mathcal A}^+V_{\mathcal A}:
 \mathsf D_{\mathcal A}^-
 \lhook\joinrel\longrightarrow
 \mathcal E_{\mathfrak J_{\mathcal A}}^\pm,
 \qquad
 U_{\mathcal A}h_{\mathcal A}^-=h_{\mathcal A}^+.
 \tag{Graph-support}
\]
At every displayed first-order occurrence used as an input or result
port, a complete chart requires $h_{\mathcal A}^\pm$ to have the
placed source-graph form \textup{(Port-support)} after the fixed
polarity repartition.  At a block root this requirement is
componentwise.

These are the supports generated by the canonical derivation.  At a
variable they are $(c_T^-)^\dagger,(c_T^+)^\dagger$.  At a
first-order quantum endpoint with action $V$, they are
\[
 h^-=(c_P^-)^\dagger,
 \qquad h^+=(I\tensor V)(c_P^+)^\dagger,
 \qquad
 g^+=h^+V^\dagger;
\]
in pair coordinates,
$h^+\ket p=\ket p_{\mathrm{ref}}\tensor V\ket p$.
At a structural endpoint they are $g_s^-$ and
$g_s^+\mathsf{Rig}(s)$.  Tensor nodes
tensor the supports, abstraction applies its fixed polarity
repartition, completed branch families take their tagged orthogonal sum, and a
classified first-order cut uses \textup{(Graph-link)}.  These clauses
define the supports recursively; restriction along a marked path gives
the occurrence supports used by the next cut.
\end{definition}

\begin{lemma}[Completion of boundary charts]
\label{lem:completion-chart}
Suppose each $U_i:B_i^-\to B_i^+$ has a complete graph chart
$(g_i^\pm,\epsilon_i^\pm,V_i)$, and let
\[
 U=\sum_i\iota_i^+\gamma_iU_i(\iota_i^-)^\dagger,
 \qquad
 V=\sum_i k_i^+\gamma_iV_i(k_i^-)^\dagger,
 \qquad |\gamma_i|=1,
\]
where both families of inclusions are orthogonal and jointly
exhaustive.  Writing $\bar\iota_i^\pm$ for the ambient tagged
inclusions, define
\[
 \begin{aligned}
 h^-&:=\sum_i\bar\iota_i^-h_i^-(k_i^-)^\dagger,\\
 h^+&:=\sum_i\bar\iota_i^+\gamma_i
             h_i^+(k_i^-)^\dagger .
 \end{aligned}
 \tag{Graph-support-block}
\]
Then $h^\pm$ are isometric source-parametrized supports and
$Uh^-=h^+$.  The target-parametrized insertions and readouts are
\[
 \begin{gathered}
 g^-:=h^-,
 \qquad
 g^+:=h^+V^\dagger
      =\sum_i\bar\iota_i^+g_i^+(k_i^+)^\dagger,\\
 B^\pm:=\operatorname{im}(g^\pm),
 \qquad
 \epsilon^\pm:=(g^\pm)^\dagger\!\upharpoonright_{B^\pm}.
 \end{gathered}
 \tag{Chart-block-ins}
\]
Equivalently,
\[
 \epsilon^\pm
 =\sum_i k_i^\pm\epsilon_i^\pm(\iota_i^\pm)^\dagger,
 \tag{Chart-block}
\]
and
\[
 \epsilon^+U(\epsilon^-)^\dagger
 =\sum_i k_i^+\gamma_i
   \bigl(\epsilon_i^+U_i(\epsilon_i^-)^\dagger\bigr)
   (k_i^-)^\dagger .
 \tag{Chart-sum}
\]
\end{lemma}
\begin{proof}
Orthogonality, joint exhaustion, and $|\gamma_i|=1$ give
$(h^\pm)^\dagger h^\pm=I$.  The displayed identities then follow
blockwise; in $h^+V^\dagger$ the phase occurs once in each factor and
cancels.
\end{proof}

\paragraph{Joint first-order exposure.}
Let $\vec a=(a_1,\ldots,a_m)$ be an ordered family of distinct
unmatched \textsc{Var} leaves of first-order types
$\vec P=(P_1,\ldots,P_m)$, and put
\[
 \mathsf D_{\vec P}
 :=\mathsf D_{P_1}\tensor\cdots\tensor\mathsf D_{P_m}.
\]
A selected chart \emph{jointly exposes} $\vec a$ when its negative
coordinate has the form
\[
 e_{\mathcal A,\vec a}^-:
 B_{\mathcal A}^-
 \xrightarrow{\;\cong\;}
 \mathsf D_{\vec P}\tensor X_{\mathcal A,\vec a}.
 \tag{FO-open}
\]
For a Source-normal chart, or a lower-$\Phi$ chart generated from one
by $\mathsf{MapApply}$ or coherent sharing,
$X_{\mathcal A,\vec a}$ is this induced residual based space.
The exposure is \emph{graph-compatible} when this factorization
is the restriction of the placed source-graph supports
$h_{\mathcal A}^\pm$ along the marked path of the canonical
derivation.  The positive occurrence
and first-order result charts are restricted in the same way.  Matching
two such occurrences is \textup{(Graph-link)}.
The positive chart has some coordinate space
$Y_{\mathcal A,\vec a}$, and hence determines
\[
 V_{\mathcal A}^{\vec a}
 :=e_{\mathcal A,\vec a}^+
 U_{\mathcal A}
 (e_{\mathcal A,\vec a}^-)^\dagger:
 \mathsf D_{\vec P}\tensor X_{\mathcal A,\vec a}
 \xrightarrow{\;\cong\;}Y_{\mathcal A,\vec a}.
\]
If the displayed result is a non-block first-order type $Q$, its
positive chart further exposes the result:
\[
 V_{\mathcal A}^{\vec a;Q}:
 \mathsf D_{\vec P}\tensor X_{\mathcal A,\vec a}
 \xrightarrow{\;\cong\;}
 Y_{\mathcal A,\vec a}\tensor\mathsf D_Q .
 \tag{FO-result}
\]
At a block-family root, \textup{(FO-result)} is read componentwise,
with inclusions into $\mathsf D_Q$ satisfying \textup{(FO-part)}.

\paragraph{Rank of a Source chart.}
Put
\[
 d(A):=\dim\sem A,
 \qquad d(\Gamma):=\prod_{x:A\in\Gamma}d(A),
 \qquad
 \kappa(\Gamma;A):=\sqrt{d(\Gamma)d(A)},
 \tag{Chart-rank-def}
\]
where the empty product is $1$.  For every Source normal derivation,
and every lower-$\Phi$ branch created from one by the construction
below, the simultaneous induction maintains
\[
 \dim\mathsf D_{\mathcal N}^-
 =\dim\mathsf D_{\mathcal N}^+
 =\kappa(\Gamma;A)
 \qquad
 (\mathcal N:\Gamma\intjudge N:A).
 \tag{Chart-rank}
\]
Thus the displayed square root is an integer whenever it is used.

\begin{lemma}[Scale of Source-generated sum maps]
\label{lem:source-map-scale}
Every $\plus$-map former
$(A_1\plus A_2)\lmark(C_1\plus C_2)$ occurring in a Source normal form
$\mathsf{NF}_{\mathrm{LO}}(t^\circ)$, or in a lower-$\Phi$ normal form
generated from it by $\mathsf{MapApply}$ or coherent sharing, has a
common scale:
\[
 d(C_1)d(A_2)=d(C_2)d(A_1).
 \tag{Map-scale}
\]
\end{lemma}
\begin{proof}
A Source case introduces branch maps
$G_\Gamma\tensor B_i\lmark B_i\tensor C$, for which both ratios are
$d(C)/d(\Gamma)$.  Branchwise map composition multiplies the two
common scales.  An \textsc{Exp} atom is an endomorphism and therefore
has scale $1$.  The remaining rewrites preserve the branch types or
remove the map former, and recursively generated branches inherit the
property.  No other Source construct introduces a $\plus$-map former.
\end{proof}

\begin{lemma}[Functoriality of typed graph composition]
\label{lem:typed-graph-composition}
Suppose the premise charts of a non-block first-order instance are
supplied by the simultaneous induction and have the placed source-graph
form \textup{(Port-support)} at the matched occurrence.  Then
\textup{(Graph-link)} constructs the source-parametrized supports of
$W_\rho$.  They are isometric and retain the same form and
source-index order at every unmatched first-order occurrence.  Consequently the maps
$j_\rho^\pm$ of \textup{(PC-yank)} are isometric, and
\textup{(PortCut)} is unitary with coordinate action $W_\rho$.  The
statement holds componentwise for a producer family.
\end{lemma}
\begin{proof}
Put
\[
 F_\rho^\epsilon
 :=(I\tensor F_{\mathcal C}^\epsilon)\,
   \sigma_\rho\,
   (F_{\mathcal D}^\epsilon\tensor I).
\]
Equation~\textup{(Graph-link-functor)} is ordinary matrix
multiplication at the matched port.  In particular, on the positive
live coordinates it reads
\[
 [\widetilde U_\rho]_{(y,w),(x,z)}
 =\sum_{p\in\mathcal B_P}
   [V_{\mathcal C}^{\rho}]_{w,(p,z)}
   [V_{\mathcal D}^{\rho}]_{(y,p),x},
 \tag{Graph-product}
\]
the matrix of \textup{(FO-whole)}.  Each $F_\rho^\epsilon$ is an
isometry because its two factors are, and therefore
\[
 \left\langle\Gamma(F_\rho^\epsilon)e_r,
                  \Gamma(F_\rho^\epsilon)e_s\right\rangle
 =\delta_{rs}
  \left\langle F_\rho^\epsilon e_r,
                  F_\rho^\epsilon e_s\right\rangle
 =\delta_{rs}.
 \tag{Graph-isom}
\]
Hence $h_\rho^\pm$ are isometric.  The negative link is the source
graph and the positive link is the source-parametrized positive graph;
thus $j_\rho^-=h_\rho^-$ and
$j_\rho^+=h_\rho^+W_\rho^\dagger$ are isometric, and
\[
 j_\rho^+W_\rho(j_\rho^-)^\dagger
 =h_\rho^+(h_\rho^-)^\dagger.
 \tag{Graph-functor}
\]
At a variable port this is the typed snake \textup{(Y)}, giving
\textup{(VS-sector)}.  In \textup{(Graph-link)} only the matched
$p$-index is summed; the reference tuple $\ket{x,z}_{\mathrm{ref}}$
retains every unmatched source index.  The residual $z$ may contain
an unclosed mark of any Source type; its type does not enter this
index calculation.  Tensor, abstraction, and
completion apply their fixed conclusion placements, and links at
distinct ports commute.  Hence the source-graph form and source-index
order are retained at every other displayed first-order occurrence.  This preservation is
proved simultaneously with classified-cut closure on the stated
outer-$\Phi$, inner $(|\mathcal D|,|\mathcal C|)$ order: the current
link uses only proper-premise supports or lower-$\Phi$ generated
branches.  For a producer family the calculation is componentwise,
and \textup{(Graph-support-block)} reassembles the parent support.
\end{proof}

\begin{lemma}[Common coordinate for Source blocks]
\label{lem:shared-source-chart}
Let
$\mathcal M_i:(\Theta,b_i:B_i\intjudge M_i:C_i)$, $i=1,2$, be the
component derivations of a Source-generated map-source or
coherent-sharing block, where $B_1,B_2$ are first-order and
$d(C_i)=\lambda d(B_i)$ for one $\lambda>0$.  Suppose their assigned
charts satisfy the source-coordinate convention following
\textup{(Chart)}.  That convention gives one residual based space
$Z_K$, and the charts
supplied by \textup{(FO-open)} at the formal leaves $b_i:B_i$ have
the form
\[
 e_i^-:B_{\mathcal M_i}^-
 \xrightarrow{\;\cong\;}
 \mathsf D_{B_i}\tensor Z_K
 \qquad(i=1,2).
 \tag{CS-common}
\]
Under \textup{(Chart-rank)},
\[
 \dim Z_K
 =\frac{\kappa(\Theta,b_i:B_i;C_i)}{d(B_i)}
 =\sqrt{d(\Theta)\lambda}.
 \tag{CS-rank}
\]
The completed source chart is
\[
 \bigoplus_i B_{\mathcal M_i}^-
 \xrightarrow{\;\bigoplus_i e_i^-\;}
 \bigoplus_i(\mathsf D_{B_i}\tensor Z_K)
 \xrightarrow{\;\cong\;}
 \mathsf D_{B_1\plus B_2}\tensor Z_K .
 \tag{CS-root-chart}
\]
\end{lemma}
\begin{proof}
A translated Source case packs the common branch context in the same
fixed order.  Removing the formal $b_i$ index therefore leaves the
same ordered source-reference indices in both components.  Give these
indices their fixed context and DNF addresses and carry the addresses
through the defining recursion of $\mathsf{NF}_{\mathrm{LO}}$.  A
\textup{(A)} step replaces the bound address by the producer's ordered
address word; a \textup{(B)} step replaces the tensor address by its two
ordered component words; and an \textup{(E)}-step deletes only the
matched intermediate address, branchwise.  The remaining applicable
conversions move unchanged subderivations, after which canonical
retyping restores the fixed context and DNF order.  Rule \textup{(G)}
is unavailable by
Lemma~\ref{lem:reduction-preserves-source-discipline}.  Hence
normalization preserves the ordered word of every unmatched source
index.  A coherent-sharing component grafts the same shared-context
derivation into each already aligned branch, so the preceding argument
applies to its lower-$\Phi$ normalization.
An \textsc{Exp} atom or structural atom is interpreted by one endpoint
or structural clause, respectively.

Thus the two residual coordinates are the same based space $Z_K$;
\textup{(CS-rank)} checks its dimension.  If other first-order leaves
are exposed simultaneously, their factors occupy the same fixed
positions, so the identification fixes $\mathsf D_{\vec P}$
pointwise.  The last arrow in \textup{(CS-root-chart)} is fixed DNF
distributivity.
\end{proof}

\begin{lemma}[Classified port-cut closure]
\label{lem:whole-port-collapse}
Let $\rho$ be a classified instance.  Assume that every proper
canonical premise and every lower-$\Phi$ branch generated by
$\mathsf{MapApply}$ or coherent sharing has already been assigned
its well-typed unitary package and the occurrence charts required
below, and those charts satisfy \textup{(Chart-rank)} and the
source-coordinate convention following \textup{(Chart)}.  Then
Definition~\ref{def:classified-cut} constructs a
well-typed unitary
\[
 \mathsf{PortCut}_{P}^{\rho}(\mathcal C,\mathcal D):
 B_\rho^-\xrightarrow{\;\cong\;}B_\rho^+ ,
\]
together with its complete graph chart.  If $P$ is first-order,
its coordinate action is \textup{(FO-whole)} or the componentwise
construction
\textup{(FO-family-cut)}--\textup{(FO-family-close)}, and its sector
insertion is constructed by Lemma~\ref{lem:typed-graph-composition}.
Simultaneously, closing $P$ preserves the graph-compatible joint
exposure \textup{(FO-open)} of every ordered finite family of
distinct unmatched first-order leaves, their fixed source-index order,
and the placed source-graph form of every displayed first-order result,
componentwise at a block root.

The proof uses outer strong induction on $\Phi$ and, at fixed
$\Phi$, lexicographic induction on
$(|\mathcal D|,|\mathcal C|)$, with producer size first.
The \textup{(C-var-whole)} clause uses the subsidiary order
\textup{(Eta-order)}.
\end{lemma}
\begin{proof}
Follow the selected occurrence from its leaves outward.  A
first-order consumer variable has the canonical $\mathsf D_P$
coordinate, and tensor and abstraction apply their fixed DNF
reorderings.  If a completed consumer is entered through the same
first-order context port $P$ in every component, the component
hypotheses give
\[
 e_i^-:B_{\mathcal C,i}^{-,\mathrm{act}}
 \xrightarrow{\;\cong\;}
 \mathsf D_P\tensor Z_{\rho,i}.
\]
Put $Z_\rho:=\bigoplus_i Z_{\rho,i}$.  The tagged sum of these charts
and fixed DNF distributivity give
\[
 \bigoplus_i B_{\mathcal C,i}^{-,\mathrm{act}}
 \xrightarrow{\;\bigoplus_i e_i^-\;}
 \bigoplus_i(\mathsf D_P\tensor Z_{\rho,i})
 \xrightarrow{\;\cong\;}
 \mathsf D_P\tensor Z_\rho .
 \tag{FO-block-expose}
\]
When a block's own sum input is first-order, balanced Source
normality makes the map-source operands closed.  Lemmas~\ref{lem:source-map-scale}
and~\ref{lem:shared-source-chart}, with $\Theta=\varnothing$, give
their common residual coordinate.  In coherent sharing the branch maps
are again closed.  Applying \textup{(Map-scale)} to the branch sources
$A'\tensor B_i$ or $B_i\tensor A'$ and absorbing the common factor
$d(A')$ into $\lambda$ gives the same conclusion with
$\Theta=\Delta$.  Thus \textup{(CS-root-chart)} supplies every
completed-consumer instance of \textup{(FO-expose)}.

The first-order producer grammar has four cases.  A variable-headed
spine uses Lemma~\ref{lem:variable-spine-chart}; a pair tensors its
premise coordinates, using the product family
\textup{(FO-pair-family)} when either premise is blocked; an
atom-headed application uses \textup{(EndAct)}; and an applied
$\plus$-map uses \textup{(Use-part)}, with the
distributor-on-result-pair subcase sent to the lower-$\Phi$
coherent-sharing clause.  These cases also cover
endpoint-transported families.  A lambda, bare atom, or bare
$\plus$-map has function type, and a tensor let is not a result \(R\).

The same marked-path induction supplies the actual occurrence
supports by restricting the premise graph charts.  At the current
first-order port, Lemma~\ref{lem:typed-graph-composition} links them
and constructs the conclusion chart.  For distinct open ports $P,Q$,
the two matrix-index contractions commute after the fixed rig
reordering; hence closing $P$ preserves the support exposed at $Q$.  Tensor and abstraction conjugate this
equation by their fixed boundary reorderings, while completion takes
its orthogonal direct sum by \textup{(Chart-block-ins)}.

For a non-block producer, \textup{(FO-whole)} is a composite of
unitaries, and Lemma~\ref{lem:typed-graph-composition} constructs
its sector insertions.  For a producer family,
\[
 k_i^\dagger k_j
 =((s_i^\dagger s_j)\tensor I_{Z_\rho})=\delta_{ij}I,
 \qquad
 \sum_i k_i k_i^\dagger
 =V_{\mathcal C}^{\rho}
  \bigl[(\sum_i s_i s_i^\dagger)\tensor I_{Z_\rho}\bigr]
  (V_{\mathcal C}^{\rho})^\dagger=I .
\]
Thus the $W_i$ are orthogonal and jointly exhaustive, every
$\widetilde U_{\rho,i}$ is unitary, and
\textup{(FO-family-close)} is a unitary.  The component definition
applies each $Y_{\Xi_{\rho_i}}$ exactly once.

It remains to verify preservation of other open ports.  Partition
the still-open marks between the producer and consumer as
$\vec a_{\mathcal D},\vec a_{\mathcal C}$.  In the non-block case
the simultaneous induction gives
\[
 \begin{aligned}
 V_{\mathcal D}^{\vec a_{\mathcal D};P}
 &: \mathsf D_{\vec P_{\mathcal D}}\tensor X
    \xrightarrow{\;\cong\;}Y\tensor\mathsf D_P,\\
 V_{\mathcal C}^{p,\vec a_{\mathcal C}}
 &: \mathsf D_P\tensor
    \mathsf D_{\vec P_{\mathcal C}}\tensor Z
    \xrightarrow{\;\cong\;}W .
 \end{aligned}
 \tag{FO-carry-prem}
\]
After the fixed permutation $\pi_{\rho,\vec a}$ restores context
order, the conclusion action is
\[
 V_\rho^{\vec a_{\mathcal D}\cdot\vec a_{\mathcal C}}
 =
 (I_Y\tensor V_{\mathcal C}^{p,\vec a_{\mathcal C}})
 \,\alpha_\rho\,
 \bigl(V_{\mathcal D}^{\vec a_{\mathcal D};P}
       \tensor
       I_{\mathsf D_{\vec P_{\mathcal C}}\tensor Z}\bigr)
 \,\pi_{\rho,\vec a}^{-1}.
 \tag{FO-carry}
\]
Hence contraction of $P$ leaves every other marked first-order
coordinate exposed.  For a producer family, use the same equation
in component $i$, replacing $\mathsf D_P$ by $D_i$ and using
$\widehat k_i$; the component actions reassemble by
\textup{(FO-family-close)}.

At a nested cut on a marked path, invoke this simultaneous
induction with all other open marks retained, construct that cut,
and continue outward by \textup{(FO-carry)}.  Tensor rules tensor
the jointly exposed factors.  In a completed consumer block admitted
by additive balance, each retained mark is present in every component;
\textup{(BW)} reassembles their charts.  Thus the rebuilt conclusion
retains one joint chart for every still-open marked leaf.

For a higher-order matched port, variable application uses
\textup{(C-var-app)}, a variable-headed producer consumed whole uses
the type-inductive unit graft \textup{(C-var-whole)}, and a structural
head uses \textup{(C-str)}.  In the unit graft, the tensor row of
\textup{(Eta-graft)} uses \textup{(Ten-pack)}, while its shared-block
row takes the tagged direct sum supplied by \textup{(BW)}.  Thus these
clauses preserve every unmatched first-order mark and its fixed
source-index order: \textup{(VS-chart)} handles spines,
\textup{(Str-chart)} handles structural transport, and the rows of
\textup{(Eta-graft)} handle the unit graft.

At fixed $\Phi$, only \textup{(C-str)} makes a recursive classified
call, with a proper producer.  The unit graft decreases by
\textup{(Eta-order)}, and every generated $\mathsf{MapApply}$ or
coherent-sharing branch has smaller $\Phi$ by \textup{(Use-$\Phi$)}
and \textup{(CS-NF)}.  Lemma~\ref{lem:classifier-totality} gives
exhaustiveness and uniqueness.

For every conclusion closed by \textup{(PortCut)}, put
\[
 \begin{aligned}
 g_\rho^\pm
  &:=j_\rho^\pm
     \bigl(I_{H_\rho^\pm}\tensor(d_{\Xi_\rho}^\pm)^\dagger\bigr),\\
 \epsilon_\rho^\pm
  &:=(I_{H_\rho^\pm}\tensor d_{\Xi_\rho}^\pm)\chi_\rho^\pm,
 &V_\rho&:=\widetilde U_\rho\tensor I_{\mathsf D_{\Xi_\rho}} .
 \end{aligned}
 \tag{Cut-chart}
\]
Equations~\textup{(PortCut)}, \textup{(PC-yank)}, and
\textup{(TY-read)} give
$\epsilon_\rho^+\mathsf{PortCut}_P^\rho
 =V_\rho\epsilon_\rho^-$.  For a producer-family conclusion closed by
\textup{(FO-family-close)}, Lemma~\ref{lem:completion-chart}
reassembles the component supports and charts through
\textup{(Graph-support-block)}, using the tagged inclusions of
\textup{(PC-block)}.
\end{proof}

\begin{lemma}[Naturality of the classified port cut]
\label{lem:classified-cut-naturality}
Let $\rho$ and $\rho'$ be corresponding classified instances:
instances with the same classification whose evaluated static phases
are equal.  Assume canonical transports for every proper premise and
generated branch, intertwining their denotations.  At a blockwise
conclusion, write $\kappa_i^\pm$ for its canonical tagged inclusions
from \textup{(BW)}---the $\jmath_i^\pm$ of
\textup{(FO-family-close)} for a producer family---and assume
\[
 \tau_\rho^\pm\kappa_i^\pm
 =\kappa_i'^\pm\tau_i^\pm .
 \tag{CN-block}
\]
Write $\bar\tau_{\mathcal A}^\pm$ for the ambient boundary
transport and
$\tau_{\mathcal A}^\pm:=\bar\tau_{\mathcal A}^\pm
 \!\upharpoonright_{B_{\mathcal A}^\pm}$ for its selected restriction.
For every package produced during the recursion, its source-coordinate
transport satisfies
\[
 \bar\tau_{\mathcal A}^\pm h_{\mathcal A}^\pm
 =h_{\mathcal A'}^\pm a_{\mathcal A}^- .
 \tag{Graph-support-Nat}
\]
The permitted transports preserve the placed source-graph form
\textup{(Port-support)}.
For an instance closed by \textup{(PortCut)}, the induced active
and residual transports satisfy
\[
 \tau_H^+\widetilde U_\rho
 =\widetilde U_{\rho'}\tau_H^-,
 \tag{CN-active}
\]
\[
 \xi^+Y_{\Xi_\rho}
 =Y_{\Xi_{\rho'}}\xi^-,
 \tag{CN-res}
\]
and the linked insertions are equivariant:
\[
 \tau_\rho^\pm j_\rho^\pm
 =j_{\rho'}^\pm(\tau_H^\pm\tensor\xi^\pm).
 \tag{CN-yank}
\]
Equivalently,
\[
 \chi_{\rho'}^\pm\tau_\rho^\pm
 =(\tau_H^\pm\tensor\xi^\pm)\chi_\rho^\pm .
 \tag{CN-chart}
\]
When the insertions are assembled by a tagged block, these equations
hold componentwise and \textup{(CN-block)} transports their orthogonal
assembly.  Consequently, for both \textup{(PortCut)} and
\textup{(FO-family-close)},
\[
 \tau_\rho^+\,
 \mathsf{PortCut}_{P}^{\rho}
 =
 \mathsf{PortCut}_{P}^{\rho'}\,
 \tau_\rho^- .
 \tag{Cut-Nat}
\]
\end{lemma}
\begin{proof}
Use the induction of Lemma~\ref{lem:whole-port-collapse}.
The constructor supports are equivariant, giving
\textup{(Graph-support-Nat)}.  Entrywise naturality of
\textup{(Graph-link)} gives
\[
 \bar\tau_\rho^\pm h_\rho^\pm=h_{\rho'}^\pm a_\rho^- .
 \tag{Graph-link-Nat}
\]
Since $a_\rho^+W_\rho=W_{\rho'}a_\rho^-$, this implies
\[
 \bar\tau_\rho^-j_\rho^-=j_{\rho'}^-a_\rho^-,
 \qquad
 \bar\tau_\rho^+j_\rho^+=j_{\rho'}^+a_\rho^+ .
 \tag{Match-Nat}
\]
For corresponding jointly exposed lists,
\[
 \begin{aligned}
 e_{\mathcal A',\vec a'}^-\tau_{\mathcal A}^-
   &=(a_{\vec P}\tensor a_X)e_{\mathcal A,\vec a}^-,\\
 e_{\mathcal A',\vec a'}^+\tau_{\mathcal A}^+
   &=a_Ye_{\mathcal A,\vec a}^+ .
 \end{aligned}
 \tag{FO-open-Nat}
\]
At a displayed first-order result, $a_Y$ has the corresponding
factorization $a_{Y_0}\tensor a_Q$.  Thus the exposed coordinate
actions intertwine.  In particular, if
$a_X,a_Y,a_P,a_Z,a_W$ are the transports in
\textup{(FO-expose)}, then
\[
 (a_Y\tensor a_P)V_{\mathcal D}^{\rho}
 =V_{\mathcal D}^{\rho'}a_X,
 \qquad
 a_WV_{\mathcal C}^{\rho}
 =V_{\mathcal C}^{\rho'}(a_P\tensor a_Z).
\]
Naturality of the associator gives
\[
 (a_Y\tensor a_W)\widetilde U_\rho
 =\widetilde U_{\rho'}(a_X\tensor a_Z)
\]
for \textup{(FO-whole)}, and naturality of
\textup{(FO-carry)} preserves all other displayed ports.

For a producer family, $a_Ps_i=s_i'a_i$.  Hence the $k_i$, their
ranges $W_i$, and the component maps of
\textup{(FO-family-cut)} transport componentwise.  Equation
\textup{(Graph-support-block)} gives
\textup{(Graph-support-Nat)} for the parent, and
\textup{(FO-family-close)} with \textup{(CN-block)} gives
\[
 \tau_\rho^+\!
 \sum_i\jmath_i^+U_{\rho,i}(\jmath_i^-)^\dagger
 =
 \sum_i\jmath_i'{}^+U_{\rho',i}
       (\jmath_i'{}^-)^\dagger\tau_\rho^- .
\]
No phase is introduced here: it is already contained in the
component action.

Restricting \textup{(Match-Nat)} to $B_\rho^\pm$ gives
\textup{(CN-yank)};
\textup{(TY-read)} gives \textup{(CN-res)}, and
\textup{(CN-chart)} follows from \textup{(PC-def)}.  Combining these equations with
\textup{(CN-active)} proves \textup{(Cut-Nat)}.

The remaining clauses use the same equations.  For
\textup{(C-var-whole)}, naturality follows by the lexicographic
induction \textup{(Eta-order)}: the first-order case is
\textup{(Graph-link-Nat)}, and tensor and implication use naturality
of tensor, typed yanking, and the fixed boundary repartitions.  The
shared-block row follows componentwise from direct-sum naturality and
\textup{(CN-block)}.
First-order endpoints use \textup{(C-q)} and
\textup{(Family-transport)}.  Higher-order structural transport
conjugates the package and its chart by
\textup{(Str-pre)}--\textup{(Str-chart)}, after which
\textup{(C-str)} invokes the induction on the proper producer.
Tensor packaging uses
\[
 \tau_{\mathcal N}^{\epsilon}\vartheta_N^\epsilon
 =\vartheta_{N'}^\epsilon\tau_{\mathcal B}^{\epsilon}.
 \tag{Ten-pack-Nat}
\]
Map use and coherent sharing assemble the assumed lower-$\Phi$
branch intertwinings through their orthogonal completed blocks.
Every recursive appeal decreases in the order of
Lemma~\ref{lem:whole-port-collapse}.
\end{proof}

\begin{definition}[Coherent-sharing clause]
\label{def:coherent-sharing-clause}
Process any outer structural endpoints preceding the distributor,
in their syntactic order, by \textup{(C-str)}.  Their
$\mathsf{Rig}$ maps transport the selected summand inclusions and
the corresponding phase labels.  Choose each formal port \(b_i:B_i\)
by the fixed least-fresh convention.  After this transport, the
residual shape is one of the following two cases:
\[
\begin{aligned}
 K_L&=
 \oplusmap{\alpha}{R_1}{\beta}{R_2}
   \bigl(\delta_L(S\tensor T)\bigr),&
 \delta_L&:A'\tensor(B_1\plus B_2)
 \longrightarrow (A'\tensor B_1)\plus(A'\tensor B_2),\\
 \Gamma_i&\intjudge R_i:(A'\tensor B_i)\lmark C_i,&
 q_i&:=R_i(S\tensor b_i);
 \\[0.4ex]
 K_R&=
 \oplusmap{\alpha}{R_1}{\beta}{R_2}
   \bigl(\delta_R(T\tensor S)\bigr),&
 \delta_R&:(B_1\plus B_2)\tensor A'
 \longrightarrow (B_1\tensor A')\plus(B_2\tensor A'),\\
 \Gamma_i&\intjudge R_i:(B_i\tensor A')\lmark C_i,&
 q_i&:=R_i(b_i\tensor S).
\end{aligned}
\tag{CS-shapes}
\]
Here $\Delta\intjudge S:A'$ and
$\Omega\intjudge T:B_1\plus B_2$.  Fix either row, write \(K\) and
\(\delta\) for its \(K_L,\delta_L\) or \(K_R,\delta_R\), and retain
its displayed definition of \(q_i\).  The contexts are pairwise
disjoint; $C_1,C_2$ are first-order; the sources are unrestricted;
$\alpha,\beta$ are static unit phases; and $R_i,S,T$ are arbitrary
normal results.  In both cases,
\[
 \Gamma_i,\Delta,b_i:B_i\intjudge q_i:C_i.        \tag{CS-ty}
\]
Let $\mathcal M_i$ be the canonical derivation of
\[
 M_i:=\mathsf{NF}_{\mathrm{LO}}(q_i),
 \qquad
 \Gamma_i,\Delta,b_i:B_i\intjudge M_i:C_i,
\]
and write
\[
 U_i:=
 \SEM{\Gamma_i,\Delta,b_i{:}B_i\intjudge M_i:C_i}_{\mathrm{NF}}
 :B_{\mathcal M_i}^-\longrightarrow B_{\mathcal M_i}^+ .
\]
These lower-$\Phi$ branch derivations determine the
shared-context formal consumer
\[
 \mathcal U_K:
 \Gamma_1,\Gamma_2,\Delta,p:B_1\plus B_2
 \intjudge K_p:C_1\plus C_2,
\]
where
\[
 K_p:=
 \begin{cases}
 \oplusmap{\alpha}{R_1}{\beta}{R_2}
   \bigl(\delta_L(S\tensor p)\bigr),&K=K_L,\\
 \oplusmap{\alpha}{R_1}{\beta}{R_2}
   \bigl(\delta_R(p\tensor S)\bigr),&K=K_R.
 \end{cases}
 \tag{CS-use}
\]
Put
\[
\begin{aligned}
 Z_{K,1}^{\pm}
   &:=
   B_{\mathcal M_1}^{\pm}
   \tensor\mathcal T_{\Gamma_2}^{\pm},\\
 Z_{K,2}^{\pm}
   &:=
   \mathcal T_{\Gamma_1}^{\pm}
   \tensor B_{\mathcal M_2}^{\pm},
\end{aligned}
\qquad
 B_{\mathcal U_K}^{\pm}
 :=
 Z_{K,1}^{\pm}\oplus Z_{K,2}^{\pm},
\]
and let
$j_i^\pm:Z_{K,i}^\pm\to B_{\mathcal U_K}^\pm$
be the standard direct-sum inclusions, so that
\[
 (j_i^\pm)^\dagger j_k^\pm=\delta_{ik}I,
 \qquad
 \sum_i j_i^\pm(j_i^\pm)^\dagger
 =I_{B_{\mathcal U_K}^\pm}.
 \tag{CS-part}
\]
The prescribed source-consumer map is
\[
\begin{aligned}
 U_{\mathcal U_K}
 &:=
 j_1^+\,\alpha
   (U_1\tensor Y_{\Gamma_2})(j_1^-)^\dagger\\
 &\quad+
 j_2^+\,\beta
   (Y_{\Gamma_1}\tensor U_2)(j_2^-)^\dagger .
\end{aligned}
\tag{CS-source}
\]
The shared context $\Delta$, arising from the single derivation of
$S$, remains in each alternative restriction.

Let $\mathcal T:(\Omega\intjudge T:B_1\plus B_2)$ be the canonical
argument derivation and $\rho_K$ the classified port-cut instance
determined by $\mathcal U_K$ and $\mathcal T$.  Let
$\bar d_{\delta,K}^\pm$ be the ambient boundary-ordering
permutation induced by $\delta$, set
$B_K^\pm:=(\bar d_{\delta,K}^\pm)^{-1}(B_{\rho_K}^\pm)$, and let
$d_{\delta,K}^\pm:B_K^\pm\xrightarrow{\;\cong\;}B_{\rho_K}^\pm$ be
its restriction to these spaces.  On active monomial labels it is,
respectively,
\[
 (a',i,b_i)\mapsto(i,a',b_i)\quad(\delta=\delta_L),\qquad
 (i,b_i,a')\mapsto(i,b_i,a')\quad(\delta=\delta_R),
\]
tensored with identities on all context coordinates.  Define
\[
 U_K:=
 (d_{\delta,K}^+)^{-1}\,
 \mathsf{PortCut}_{B_1\plus B_2}^{\rho_K}
   (\mathcal U_K,\mathcal T)\,
 d_{\delta,K}^- .
 \tag{CS-cut}
\]
The least-fresh convention makes $b_i$, and hence $q_i$, $M_i$,
and the clause, a function of $K$ and its canonical derivation.
\end{definition}

\paragraph{Reading and well-foundedness.}
The construction places the shared $S$ coordinates in both
alternative blocks $Z_{K,i}^\pm$ and closes the single derivation of $T$
once at the complete $B_1\plus B_2$ port.  Its recursive calls are
well founded because only $\plus$-map constructors contribute to
$\Phi$ (Definition~\ref{def:phi}):
\[
 \Phi(q_i)=\Phi(R_i)+\Phi(S)
 <1+\Phi(R_1)+\Phi(R_2)+\Phi(S)+\Phi(T)
 =\Phi(K).                                        \tag{CS-$\Phi$}
\]
Theorem~\ref{thm:lo-normalizer-total} then gives
\[
 \Phi(M_i)\leq\Phi(q_i)<\Phi(K).                  \tag{CS-NF}
\]

\begin{lemma}[Typed eta identity and unit]
\label{lem:semantic-eta-identity}
For every Source type $S$ and fresh $z:S$, let
\[
 i_S:=\mathsf{NF}_{\mathrm{LO}}(\eta_S(z)).
\]
There are canonical unitary repartitions
\[
 e_S^\pm:
 B_{i_S}^\pm\xrightarrow{\;\cong\;}\mathcal Y_S^\pm
\]
such that
\[
 e_S^+\,U_{i_S}
 =
 \mathsf{yank}_S\,e_S^- .
 \tag{Eta-id}
\]
At a fixed $\Phi$-stratum, suppose the selected packages for every
proper premise and lower-$\Phi$ generated branch encountered along a
marked consumer have already been constructed.  Then, for every
\textup{(C-var-whole)} instance at $S$, the recursively defined maps
$\mathsf{ug}_{\rho,S}^\pm$ of \textup{(C-var-place)} are isometric and
give a complete selected chart with coordinate action $W_\rho$ of
\textup{(PC-active)}.
\end{lemma}
\begin{proof}
The identity assertion is first proved by induction on
$S::=P\mid S_1\tensor S_2\mid S_1\lmark S_2$, with a subsidiary
induction on $P$ in the first-order case.  At fixed $\Phi$, prove the
unit-graft assertion separately by the lexicographic order
\textup{(Eta-order)}.

At the empty marked consumer skeleton, \textup{(VS-sector)} and the
typed snake give the canonical spine insertion and return the producer
package; the operand factors and $Y_{\Xi_\rho}$ are carried as in
\textup{(C-var-whole)}.

For first-order $P$, the subsidiary induction uses the variable yank
at $\base$, tensor interchange, and orthogonal direct sum at $\plus$;
the resulting graph is $\mathsf{Graph}_P(I)$.  The unit graft is
\textup{(FO-whole)} with identity coordinate action, and
Lemma~\ref{lem:typed-graph-composition} supplies its isometric
insertions.

For $S=S_1\tensor S_2$, the canonical eta derivation destructures and
re-pairs its input.  The fixed tensor-boundary repartition exposes the
two strict-subtype identity graphs.  The induction hypotheses, tensor
functoriality, and monoidal interchange reassemble them as
$\mathsf{yank}_{S_1\tensor S_2}$; the same component grafts followed
by \textup{(Ten-pack)} give the unit graft.

For $S=S_1\lmark S_2$, use
\[
 \eta_S(f)=
 \lambda a.\eta_{S_2}\bigl(f\,\eta_{S_1}(a)\bigr).
\]
The $S_1$ induction hypothesis supplies the operand identity,
\textup{(C-var-app)} retains the residual $S_2$ graph, and the $S_2$
induction hypothesis followed by the abstraction repartition gives
$\mathsf{yank}_{S_1\lmark S_2}$ and the unit graft.

At a nonprincipal non-cut constructor, follow the unique marked premise
and rebuild the same constructor.  This shortens the consumer skeleton,
and its fixed support insertion carries the chart equation outward.
Within the \textup{(C-var-whole)} unit graft, a completed shared
block uses the tagged direct sum of its component insertions.  At a
classified cut $\tau$ closing a distinct port,
\textup{(Eta-foreign)} grafts only the premise containing the $S$-mark
and rebuilds the selected support from that premise and the unchanged
premise packages.  If $\tau$ itself selects
\textup{(C-var-whole)} at another Source type, its local consumer is a
proper subnetwork and is therefore available by the smaller-skeleton
hypothesis.  For first-order closure this is
\textup{(Graph-functor)}; links at distinct ports commute by
Lemma~\ref{lem:typed-graph-composition}.  The remaining clauses use
tensor, direct-sum, and fixed-repartition functoriality.  Hence each
$\mathsf{ug}_{\rho,S}^\pm$ is obtained from isometries by typed graph
composition, tensor, direct sum, or fixed repartition, and
\[
 U_\rho\,\mathsf{ug}_{\rho,S}^-
 =
 \mathsf{ug}_{\rho,S}^+\,W_\rho .
\]
This is the required complete chart.
\end{proof}

\begin{lemma}[Eta compatibility at a fresh application]
\label{lem:classified-cut-eta}
Let
\[
 \Delta\intjudge r:A\lmark C
\]
be an occurrence in a Source-generated normal form or in a
lower-$\Phi$ branch generated during its interpretation, with $A$ a
Source type, $r$ an irreducible neutral, and $C$ first-order.  Let
$r^\eta$ be its eta-exposed representative and choose fresh
$z:A$.  Put
\[
 M:=\mathsf{NF}_{\mathrm{LO}}(r\,z),
 \qquad
 M^\eta:=\mathsf{NF}_{\mathrm{LO}}(r^\eta z),
\]
with canonical derivations $\mathcal M$ and
$\mathcal M^\eta$.  The type-directed eta construction determines
canonical unitary conclusion-chart transports
\[
 \epsilon_{r,z}^\pm:
 B_{\mathcal M}^\pm
 \xrightarrow{\;\cong\;}
 B_{\mathcal M^\eta}^\pm
\]
such that
\[
 \epsilon_{r,z}^+\,U_{\mathcal M}
 =
 U_{\mathcal M^\eta}\,\epsilon_{r,z}^- .
 \tag{Cut-$\eta$}
\]
\end{lemma}

\begin{proof}
Follow the fixed normalization of $r^\eta z$.  Its outer
beta-step exposes the application of $r$ to the type-directed
identity at $A$, followed by the type-directed identity at $C$.
By \textup{(Eta-id)}, the former is the variable graph
$\mathsf{yank}_A$, while the latter preserves the residual
$\mathsf{yank}_C$ factor.  The typed snake equation therefore
identifies the two active maps in a common fresh-application
chart.

Let
\[
 c^\pm:B_{\mathcal M}^\pm\longrightarrow H^\pm,
 \qquad
 c_\eta^\pm:B_{\mathcal M^\eta}^\pm\longrightarrow H^\pm
\]
be the resulting application charts.  The preceding calculation
gives
\[
 c^+U_{\mathcal M}(c^-)^{-1}
 =
 c_\eta^+U_{\mathcal M^\eta}(c_\eta^-)^{-1}.
\]
Therefore
\[
 \epsilon_{r,z}^\pm:=(c_\eta^\pm)^{-1}c^\pm
\]
has the required property.  The tensor and first-order-sum stages
of the type-directed identities construct these charts by tensor
and orthogonal direct sum.
\end{proof}

\subsection{Closure Lemmas}
\label{subsec:closure}

\begin{lemma}[Closure under boundary transport]
\label{lem:unitary-transport}
If $u:H\to K$ is unitary and
$\phi:H\to H'$, $\psi:K\to K'$ are unitary isomorphisms, then
$\psi \circ u \circ \phi^{-1}:H'\to K'$ is unitary.
\end{lemma}

\begin{lemma}[Closure under $\oplus$]
\label{lem:oplus-closure}
If $u$ and $v$ are unitary, then $u \oplus v$ is unitary.
\end{lemma}

\begin{lemma}[Closure under composition]
\label{lem:comp-closure}
If $u:H\to K$ and $v:K\to L$ are unitary, then
$v \circ u:H\to L$ is unitary.
\end{lemma}

\noindent
All three are immediate from $(AB)^\dagger = B^\dagger A^\dagger$
and $(u \oplus v)^\dagger = u^\dagger \oplus v^\dagger$.

\subsection{Inductive Proof of Unitarity}
\label{subsec:inductive}

\begin{theorem}[Source-normal-form boundary unitarity]
\label{thm:nf-boundary-unitarity}
Let $\Gamma\sjudge t:A$, put
\[
 N:=\mathsf{NF}_{\mathrm{LO}}(t^\circ),
\]
and let $\mathcal N:(\Gamma\intjudge N:A)$ be its canonical
derivation.  Then
$\SEM{\Gamma\intjudge N:A}_{\mathrm{NF}}$ is unitary between the
branch-paired boundary spaces selected by \(\mathcal N\), and its
construction supplies a complete graph chart in the sense of
Definition~\ref{def:complete-boundary-chart}, whose coordinate spaces
satisfy \textup{(Chart-rank)} and the source-coordinate convention
following \textup{(Chart)}.  Simultaneously, every ordered finite
family of distinct unmatched first-order
\textsc{Var} leaves has the graph-compatible joint exposure
\textup{(FO-open)}, and every first-order result has the
graph-compatible non-block or componentwise display
\textup{(FO-result)}.  The same assertions hold for every lower-$\Phi$
canonical branch derivation generated recursively by
$\mathsf{MapApply}$ or coherent sharing during this construction.
Whenever one of these derivations has singleton first-order shape
$\mathcal E:(x:P\intjudge M:P)$, its chart is the
$(\mathsf{src}_P,\mathsf{tgt}_P)$ root chart of
\textup{(QGraph-read)}.
\end{theorem}
\begin{proof}
Corollary~\ref{cor:source-normalization} supplies the internally typed
normal form \(N\), excludes coherent-sum formers, and shows that every
\(\plus\)-map in \(N\) has closed operands.  Consequently, a tensor
elimination cannot place its two bound variables in different
branch-map premises.

We prove the simultaneous assertion by outer strong induction on
$\Phi(\mathcal N)$.  At a fixed $\Phi$-stratum, ordinary constructor
clauses use structural induction on the canonical normal derivation.
Each classified port cut is handled by the nested lexicographic
induction
\[
  \bigl(|\mathcal D|,|\mathcal C|\bigr)
\]
of Lemma~\ref{lem:whole-port-collapse}.  Within
\textup{(C-var-whole)}, typed unit grafts use
\textup{(Eta-order)}.  Proper canonical premises are covered by the
structural induction hypothesis.  We simultaneously
maintain \textup{(Chart-rank)}.  Variables and closed endpoints have
the displayed graph rank; tensor introduction multiplies ranks, and
$\lmark$-introduction only repartitions the boundary.  A cut at $P$
contracts one complete $\mathsf D_P$ coordinate, so
\[
 \frac{\kappa(\Gamma,p:P;Q)\,\kappa(\Delta;P)}{d(P)}
 =\kappa(\Gamma,\Delta;Q).
 \tag{Cut-rank}
\]
Tensor elimination is the same calculation with $P=A\tensor B$.

For a Source-generated map
$(A_1\plus A_2)\lmark(C_1\plus C_2)$,
Lemma~\ref{lem:source-map-scale} supplies \textup{(Map-scale)}.
Writing $a_i=d(A_i)$ and $c_i=d(C_i)$, it gives
\[
 \sqrt{a_1c_1}+\sqrt{a_2c_2}
 =\sqrt{(a_1+a_2)(c_1+c_2)},
\]
which proves \textup{(Chart-rank)} for the completed map block.
For a coherent-sharing block, put $d_i=d(B_i)$.  By
\textup{(Map-scale)}, write
$d(C_i)=\lambda d(A')d_i$ in either row of \textup{(CS-shapes)}.
The lower-$\Phi$ branch hypotheses then have ranks $d_i r$, where
$r=\sqrt{d(\Delta)\lambda d(A')}$; their completed sum has rank
$(d_1+d_2)r$.  Cutting the scrutinee at $B_1\plus B_2$ gives the rank
of the Source conclusion.  These are the only non-atomic completed
blocks introduced by Source expansion; recursively generated branches
inherit the calculation.

We also simultaneously
maintain additive balance for every generated normal form.  In a
$\mathsf{MapApply}$ branch, \(R_i\) is closed because it is an operand
of a balanced map.  In a coherent-sharing branch, \(R_i\) is closed
and every map in \(S\) is balanced by the current induction hypothesis.
Thus the formal terms \(R_i\,z_i\) and the \(q_i\) of
\textup{(CS-shapes)} contain no coherent-sum former and are balanced.
Lemmas~\ref{lem:reduction-preserves-source-discipline} and
\ref{lem:lo-restores-balance} give the same properties for their
normal forms, whose canonical derivations lie at strictly smaller
\(\Phi\) by \textup{(Use-$\Phi$)} and \textup{(CS-NF)}.  Thus
Lemma~\ref{lem:whole-port-collapse} is invoked only after all of its
hypotheses have been discharged.

\emph{Selected charts and joint first-order exposure.}
The source supports are generated by \textup{(Graph-support)}:
variables and endpoints use \textup{(VG)}, \textup{(QGraph)}, and
\textup{(SG)}; tensor rules tensor them; completion uses
\textup{(Graph-\allowbreak support-\allowbreak block)}; and a cut uses
Lemma~\ref{lem:typed-graph-composition}.  These clauses maintain the
source-index order: \textup{(Graph-link)} removes only its matched
index, while tensor, abstraction, completion, and \textup{(FO-carry)}
apply their fixed conclusion reorderings.  The higher-order cut clauses
retain the same order by the final higher-order case of
Lemma~\ref{lem:whole-port-collapse}.  At each Source block,
Lemma~\ref{lem:shared-source-chart} therefore gives the common residual
coordinate; \textup{(CS-rank)} checks its dimension.  The completed
branch charts reassemble by \textup{(BW)}.

At a classified first-order cut, partition the other marked leaves
between producer and consumer.  Lemma~\ref{lem:whole-port-collapse}
closes the selected port and preserves their joint exposure by
\textup{(FO-carry)}.  If the marked path crosses a nested cut, that cut
is constructed first with the remaining marks retained, after which the
enclosing rule is rebuilt.

For the root clause, follow the actual \textsc{Var} leaf $x:P$ to the
actual conclusion.  At the leaf the claim is \textup{(QGraph-read)} with
$V=I$.  Tensor steps tensor the coordinate maps, while completed map,
applied-map, and coherent-sharing blocks take their tagged orthogonal
sums.  The equations
\textup{(BW)}, \textup{(FO-part)}, \textup{(Use-part)}, and
\textup{(CS-part)} give joint exhaustion.  Structural and quantum
endpoints use \textup{(SG-read)} and \textup{(QGraph-read)}.  At an
eliminator, first construct the inner cut, rebuild the enclosing rule, and
continue outward.  The non-block case is composition in the common
$\mathsf D_P$ coordinate by \textup{(FO-whole)}; the family case is its
orthogonal, jointly exhaustive assembly by
\textup{(FO-family-cut)}--\textup{(FO-family-close)}.
Applied-map and coherent-sharing branches are available at smaller $\Phi$
by \textup{(Use-$\Phi$)} and \textup{(CS-NF)}.

At the singleton judgment there is no unmatched external factor.  The
first-order type recursion---one coordinate at $\base$, tensor product at
$\tensor$, and orthogonal sum at $\plus$---therefore makes the terminal
coordinate space exactly $\sem P$.  The initial and terminal maps are
$\mathsf{src}_P\!\upharpoonright_{B_{\mathcal E}^-}$ and
$\mathsf{tgt}_P\!\upharpoonright_{B_{\mathcal E}^+}$.

The grammar overlaps introduce no further cases:
$R::=V\mid E$ and $N::=R$ are inclusions; all-value pairs and maps are
covered by their arbitrary-result cases; the overlapping \(x\,V\)
production is covered once as \(E\,R\); and applied $\plus$-Map is
treated separately.

\paragraph{Structural data.}
Canonical reassociations, symmetries, polarity repartitions,
classified cut charts, and primitive structural isomorphisms are unitary
basis permutations on the displayed selected spaces.  The eliminator
clauses are the derivation-indexed cuts of
Definition~\ref{def:classified-cut}; their typing and unitarity are
Lemma~\ref{lem:whole-port-collapse}.  Canonical boundary permutations reindex completed data blocks,
so they preserve the selected charts and local marked-port coordinates.

\paragraph{(\textsc{Var}) $x:A\intjudge x:A$.}
Table~\ref{tab:sem-compositional} gives
$\SEM{x:A\intjudge x:A}_{\mathrm{NF}}=\mathsf{yank}_A$.
Equation~\textup{(Y)} makes this the identity on the selected paired
$A$-port, hence it is unitary.  The maps $(c_A^-,c_A^+)$ give its
complete boundary chart with coordinate action $I$.

\paragraph{(\textsc{Atom}).}
The current atomic grammar has exactly two forms.
For a structural atom
$\cdot\intjudge s:T\lmark S$,
$\SEM{\cdot\intjudge s:T\lmark S}_{\mathrm{NF}}
=\mathsf{Rig}_{\mathfrak S}(s)$, the graph-sector map (SG-map): a
typed basis permutation between $B_s^-$ and $B_s^+$, hence unitary.
Equation~\textup{(SG-read)} supplies its complete chart with
coordinate action $\mathsf{Rig}(s)$.

For $\cdot\intjudge\expi{\theta}{J}:P\lmark P$, where $P$ is
first-order and $\theta\in\mathbb R_{\mathrm{static}}$, fix the certified
derivation $\mathcal J:(\ijudge J:P\lmark P)$ and let $X_P$ be the
canonical tagged tensor basis of $\sem P$.  Induction on $\mathcal J$
gives a signed basis involution
\[
 J\ket{x}=s_x\ket{j_{\mathcal J}x},\qquad
 j_{\mathcal J}^{\,2}=\id,\qquad
 s_x\in\{1,-1\},\qquad s_xs_{j_{\mathcal J}x}=1.
\]
The induction follows the six rules of
Table~\ref{tab:typing-rules-involutions}.
\textsc{Inv-Id} is immediate.  \textsc{Inv-$\sigma^\tensor$} and
\textsc{Inv-$\sigma^\plus$} exchange equal-type tensor factors,
respectively tagged summands, with sign $1$.  For
\textsc{Inv-Scalar}, the new orbit sign is $\alpha s_x$ with
$\alpha\in\{1,-1\}$, so
$(\alpha s_x)(\alpha s_{j_{\mathcal J}x})=\alpha^2=1$.  For
\textsc{Inv-$\tensor$}, with premise data $(j_{\mathcal J},s)$ and
$(j_{\mathcal K},s')$, the basis involution and sign are
\[
 (x,y)\longmapsto
 (j_{\mathcal J}x,\,j_{\mathcal K}y),
 \qquad
 s_{(x,y)}=s_x s'_y,
\]
and the two premise orbit equations give
$s_{(x,y)}s_{(j_{\mathcal J}x,\,j_{\mathcal K}y)}=1$.  For
\textsc{Inv-$\plus$}, the tagged basis is preserved branchwise;
the branch sign is $\alpha s_x$ or $\beta s'_x$, and
$\alpha^2=\beta^2=1$.  Hence every certificate constructor
preserves $j^2=\id$ and $s_xs_{jx}=1$.
The displayed signed basis equations show that $J^\dagger=J$ and
$J^2=I$.  Consequently
\[
 V_a:=e^{i\theta J}
     =\cos\theta\,I+i\sin\theta\,J
\]
is unitary, with $V_a^\dagger=e^{-i\theta J}$.  The atom's selected
boundary spaces and map are
\[
 B_{\mathcal A}^-:=\mathcal Y_P^-,
 \qquad
 B_{\mathcal A}^+:=(I\tensor V_a)\mathcal Y_P^+,
 \qquad
 U_{\mathcal A}:=\mathsf{Graph}_P(V_a).
 \tag{Exp-graph}
\]
Thus the quantum atomic case follows from \textup{(QGraph)}.
Equation~\textup{(QGraph-read)} supplies its complete chart with
coordinate action $V_a$.  These are all atomic cases.

\paragraph{(\textsc{$\tensor$-I})
$\Gamma_1\intjudge R_1:A$, $\Gamma_2\intjudge R_2:B$.}
\[
 \SEM{\Gamma_1,\Gamma_2\intjudge
       R_1\tensor R_2:A\tensor B}_{\mathrm{NF}}
 =
 \SEM{\Gamma_1\intjudge R_1:A}_{\mathrm{NF}}
 \tensor
 \SEM{\Gamma_2\intjudge R_2:B}_{\mathrm{NF}} .
\]
The tensor product of unitaries in $\mathbf{FdHilb}$ is unitary,
since $(U\tensor V)^\dagger(U\tensor V)
=(U^\dagger U)\tensor(V^\dagger V)$.  Canonical boundary
repartitions are covered by Lemma~\ref{lem:unitary-transport}.
The premise charts tensor, giving the complete product chart.
Here $R_1,R_2$ are arbitrary results, so this one case includes
value pairs and pairs containing neutral or blocked-map
components.

\paragraph{(\textsc{$\tensor$-E})
$\Gamma_1\intjudge R:A\tensor B$,
$\Gamma_2,x:A,y:B\intjudge N:C$.}
The semantic clause is
\[
 \SEM{\Gamma_1,\Gamma_2\intjudge
       \letpair{x}{y}{R}{N}:C}_{\mathrm{NF}}
 =
 \mathsf{TenCut}_{A,B}(\mathcal N,\mathcal R),
\]
for the displayed canonical premise derivations, per
\textup{(TenCut)} with the packaged consumer $\mathcal B_N$.
Normality excludes a constructor pair as $R$ --- that would be a
tensor beta-redex --- and typing excludes the other constructor forms at
outer tensor type.  If $A\tensor B$ is first-order, the cut is governed by
\textup{(FO-whole)} or \textup{(FO-family-cut)}.  At a higher-order
tensor port, the irreducible scrutinee has a variable or structural
head and is handled by \textup{(C-var-whole)} or \textup{(C-str)}.
The tensor and any shared-block descent are internal rows of
\textup{(Eta-graft)} in the former case.
Lemma~\ref{lem:whole-port-collapse} gives the required unitary
and, by \textup{(Cut-chart)}, the selected chart after the canonical
$A\tensor B$ repartition and typed snake.

\paragraph{(\textsc{$\lmark$-I})
$\Gamma,x:A\intjudge N:B$.}
The semantic table gives
\[
 \SEM{\Gamma\intjudge\lam{x}{N}:A\lmark B}_{\mathrm{NF}}
 =
 \SEM{\Gamma,x:A\intjudge N:B}_{\mathrm{NF}}
\]
under the canonical polarity repartition: abstraction changes only
which side names the typed $A$ port.  The induction hypothesis is
therefore the required unitary, and any coordinate repartition is
unitary by Lemma~\ref{lem:unitary-transport}.  The same repartition
transports the complete chart.

\paragraph{(\textsc{$\lmark$-E})
$\Gamma_1\intjudge E:A\lmark B$,
$\Gamma_2\intjudge R:A$.}
The normal application clause is
\[
 \SEM{\Gamma_1,\Gamma_2\intjudge E\,R:B}_{\mathrm{NF}}
 =
 \mathsf{AppCut}_{A}(\mathcal E,\mathcal R)
 =
 \mathsf{PortCut}_{A}(\mathcal A_E,\mathcal R),
\]
for the displayed canonical premise derivations.  Here $R$ is an
arbitrary normal result: a value, a neutral, a pair of results, or a blocked phased map.

Follow the neutral spine $E$ to its ultimate head; the grammar
leaves a variable, a structural atom, or a quantum atom.  (An
applied $\plus$-Map has sum type, not function type, and a bare
$\plus$-Map is the distinct case below.)  For an atom head $h$,
$\mathcal A_h=\mathcal V_h$.  For a variable head,
\textup{(C-var-app)} identifies the conclusion map in its
canonical chart as
$(U_1\tensor\cdots\tensor U_j\tensor U)\tensor\mathsf{yank}_B$,
unitary by the induction hypotheses, and explicitly retaining the
residual $B$-boundary.  For either atom head,
\textup{(EndAct)} acts on the complete operand package;
\textup{(C-str-app)} and \textup{(C-q-app)} are only its
variable-producer instances.  Nested atom applications occur inside
the proper operand derivation and are covered by the same clause.
In all cases Lemma~\ref{lem:whole-port-collapse} gives both the
well-typed unitary and its selected chart \textup{(Cut-chart)}.

\paragraph{(\textsc{$\plus$-Map})
$\Gamma_1\intjudge R_1:A\lmark C$,
$\Gamma_2\intjudge R_2:B\lmark D$.}
Here $C,D$ are first-order, the sources $A,B$ are unrestricted, and
$\alpha,\beta\in\mathbb C_{\mathrm{static}}$ have unit modulus.
The table gives
\[
 \begin{aligned}
 &\SEM{\Gamma_1,\Gamma_2\intjudge
   \oplusmap{\alpha}{R_1}{\beta}{R_2}:
   (A\plus B)\lmark(C\plus D)}_{\mathrm{NF}}
 \\
 &\qquad =
 \alpha\bigl(
   \SEM{\Gamma_1\intjudge R_1:A\lmark C}_{\mathrm{NF}}
   \tensor Y_{\Gamma_2}
 \bigr)
 \oplus
 \beta\bigl(
   Y_{\Gamma_1}\tensor
   \SEM{\Gamma_2\intjudge R_2:B\lmark D}_{\mathrm{NF}}
 \bigr) .
 \end{aligned}
\]
The direct sum is taken in the branch-paired spaces selected by this
derivation, including the table's inactive-branch completion.
Unit-modulus scaling and $\oplus$ preserve unitarity.  This covers
arbitrary result operands, including blocked maps; the all-value map
production is its value subcase.  The completed premise charts
reassemble by Lemma~\ref{lem:completion-chart} using
\textup{(BW)}.

\paragraph{(Applied \textsc{$\plus$-Map}).}
For
\[
 \Gamma_1\intjudge R_1:A\lmark C,\qquad
 \Gamma_2\intjudge R_2:B\lmark D,\qquad
 \Delta\intjudge E:A\plus B,
\]
the table applies the map through its source consumer.
Here $R_1,R_2$ are arbitrary normal results, $A,B$ are
unrestricted, $C,D$ are first-order, and
$\alpha,\beta\in\mathbb C_{\mathrm{static}}$ have unit modulus.
The distinct normal-form production has denotation
\[
 \SEM{\Gamma_1,\Gamma_2,\Delta\intjudge
   \oplusmap{\alpha}{R_1}{\beta}{R_2}\,E:C\plus D}_{\mathrm{NF}}
 =
 \mathsf{MapApply}_{A\plus B}(\mathcal F,\mathcal E).
\]
Normality excludes both a map-on-sum redex and a composed-map
redex.  Definition~\ref{def:map-source-consumer} constructs the
source action from the normalized formal branch applications; by
\textup{(Use-$\Phi$)}, their denotations are supplied by the outer
induction hypothesis, and their completed ranges are orthogonal
and jointly exhaustive by \textup{(Use-part)}.  Their complete
charts therefore reassemble by Lemma~\ref{lem:completion-chart}.
If this
ordered structural-endpoint traversal exposes a
distributor applied to a result pair, the denotation is
Definition~\ref{def:coherent-sharing-clause}, proved unitary
below.  In every other case
Lemma~\ref{lem:whole-port-collapse} gives the required unitary.

\emph{Coherent sharing.}
At the distributor-pair shape, the applied-map denotation is
Definition~\ref{def:coherent-sharing-clause}.  Recall $q_i$,
$M_i$, $U_i$, and the bounds
\textup{(CS-$\Phi$)} and \textup{(CS-NF)}.

By \textup{(CS-NF)}, each $\mathcal M_i$ is covered by the outer
induction hypothesis, so each $U_i$, and hence each completed
block map $\alpha(U_1\tensor Y_{\Gamma_2})$ and
$\beta(Y_{\Gamma_1}\tensor U_2)$, is unitary.
Equation~\textup{(CS-part)} gives orthogonal, jointly exhaustive
ranges at both polarities, so \textup{(CS-source)} is unitary by
Lemma~\ref{lem:oplus-closure}.
Lemma~\ref{lem:whole-port-collapse} then makes
$\mathsf{PortCut}_{B_1\plus B_2}^{\rho_K}
 (\mathcal U_K,\mathcal T)$
unitary; a nested coherent-sharing occurrence inside $T$ is a
proper lower-$\Phi$ occurrence covered by the outer hypothesis.
Finally, $d_{\delta,K}^\pm$ are unitary basis permutations, so
$U_K$ is unitary.  Lemma~\ref{lem:completion-chart} reassembles the
complete chart through \textup{(CS-part)}, and
$d_{\delta,K}^\pm$ reindex that chart.  When a quantum endpoint
mixes summands, its source consumer is cut as one operator; its
orthogonal branch restrictions retain the shared $S,\Delta$
coordinates.

These cases cover every production occurring in a Source normal form
and in the lower-$\Phi$ branch derivations generated during its
interpretation.  Corollary~\ref{cor:source-normalization} excludes the
coherent-sum-former production and supplies balance at the root; the
generated-branch argument above supplies both properties recursively.
Hence the cases establish the unitary, selected-chart, and
joint-exposure assertions.
\end{proof}

\begin{lemma}[Administrative transports preserve ports]
\label{lem:administrative-port-locality}
The canonical transport of an alpha-renaming, an exchange of
independent tensor lets, or an equality of evaluated phases fixes
every matched-port coordinate and every branch tag.
Consequently, for a conversion strictly inside the producer of a
classified cut, the induced transport enters \textup{(CN-active)}
acting only on the producer factors of the active map, with all
consumer factors transported by the identity; dually for a
conversion strictly inside the consumer.  For block instances, the
transport commutes with the canonical block inclusions, which is
\textup{(CN-block)}.
\end{lemma}
\begin{proof}
The marked-occurrence chart is derivation-selected but constructed
equivariantly along its marked path.  Alpha-renaming, independent-let
exchange, and equality of evaluated phases leave its displayed
$\mathsf D_P$ factor and branch tags fixed and act only on complementary
coordinates.  The transport therefore factors as the asserted map on one
premise tensored with the identity on the matched port, dually for the
other premise.
The result follows from naturality of tensor symmetry and
direct-sum inclusions in $\mathbf{FdHilb}$.
\end{proof}

\begin{lemma}[Naturality of administrative conversions]
\label{lem:administrative-naturality}
Let $N\approx N'$ be internally typed normal forms reached from the
same Source expansion, or from the same lower-$\Phi$ branch generated
during its interpretation.  The canonical boundary transports
$\tau^\pm$ satisfy
\[
 \tau^+\SEM{N}_{\mathrm{NF}}
 =
 \SEM{N'}_{\mathrm{NF}}\tau^- .
\]
Moreover, selected charts, jointly exposed marked-port coordinates,
residual through-\allowbreak maps, and completed-\allowbreak block inclusions are
constructed uniformly from the canonical derivation and the fixed rig
ordering.
Canonical transports therefore carry the data of one classified instance
to the other; in particular, \textup{(CN-active)}, \textup{(CN-res)},
\textup{(CN-chart)}, and \textup{(CN-block)} hold at every
$\approx$-related pair.
\end{lemma}
\begin{proof}
Strong induction on $\Phi$, with subsidiary structural induction on
the context containing the generating conversion; the generators
preserve $\Phi$, so the induction is well posed.
Alpha-conversion is canonical coordinate renaming; equality of
evaluated phases is scalar equality; and exchange of independent
tensor lets is the monoidal interchange law in $\mathbf{FdHilb}$.
Tensor and sum contexts follow by functoriality; abstraction
conjugates by its polarity repartition
(Lemma~\ref{lem:unitary-transport}); and cut contexts follow from
Lemma~\ref{lem:classified-cut-naturality}, applied stepwise along
the spine, whose hypotheses hold by
Lemma~\ref{lem:administrative-port-locality} together with the
uniformity statement above.

For coherent sharing, a conversion inside $R_i$ or $S$ induces
$M_i\approx M_i'$ by the derived assertion $\mathsf{EqUN}$ of
Theorem~\ref{thm:ranked-determinacy} at the lower rank
\textup{(CS-$\Phi$)}; equation~\textup{(CS-NF)} places this use at
strictly smaller $\Phi$.  The resulting branch transports assemble
by $\oplus$; a transport from the shared $S$ acts by the
\emph{same} $\tau_S$ in both blocks --- both arise from one
term-level $\approx$-derivation on its single occurrence --- and
Lemma~\ref{lem:administrative-port-locality} fixes the branch tags,
so the assembled transport is a spectator map, not a tag-dependent
block map.
Lemma~\ref{lem:classified-cut-naturality} then applies to
\textup{(CS-cut)}.  A conversion inside $T$ is handled directly by
the subsidiary structural induction: it is a proper-subderivation
transport, inducing the port transport at the $B_1\plus B_2$ port
and its dual action on $U_{\mathcal U_K}$, and the boundary-ordering maps
$d_{\delta,K}^\pm$ conjugate the conclusion transports.

The generator calculations extend to the least congruence
$\approx$: identities, inverses, composition, and compatible term
contexts preserve the displayed intertwining equation.
\end{proof}

\begin{corollary}[Canonical-form invariance]
\label{cor:canonical-invariance}
Let $\Gamma\sjudge t:A$ and
$N=\mathsf{NF}_{\mathrm{LO}}(t^\circ)$, so that
Definition~\ref{def:canonical-form-semantics} gives
\[
 \SEM{\Gamma\sjudge t:A}
 =
 \SEM{\Gamma\intjudge N:A}_{\mathrm{NF}}.
\]
Then: (i) this map is unitary; (ii) canonical normal forms of
internally typed reducts of \(t^\circ\) that are related by
\(\approx\) have denotations agreeing up to canonical unitary boundary
transport; hence (iii) reduction of the Raw expansion and choice of
reduction strategy preserve the denotation up to canonical unitary
boundary transport.
\end{corollary}
\begin{proof}
Part (i) is Theorem~\ref{thm:nf-boundary-unitarity}.  Part (ii) is
Lemma~\ref{lem:administrative-naturality}.  If
$t^\circ\to^*u$, then
\[
 \mathsf{NF}_{\mathrm{LO}}(t^\circ)
 \approx
 \mathsf{NF}_{\mathrm{LO}}(u)
\]
by Theorem~\ref{thm:determinacy}; part (ii) applies.  The same argument
applies to every irreducible co-reduct.
\end{proof}


\begin{lemma}[Register readback]
\label{lem:first-order-endpoint-readback}
For $n\geq1$, let
\[
  P:=\QBool^{\tensor n},
  \qquad
  \mathcal D:
  \bigl(\cdot\sjudge t:P\lmark P\bigr),
\]
and choose a fresh \(x:P\).  Let
\[
  \mathcal E:
  x:P\intjudge
  N:=\mathsf{NF}_{\mathrm{LO}}\bigl((t\,x)^\circ\bigr):P
\]
be the resulting canonical endpoint derivation.  Source application
gives \(x:P\sjudge t\,x:P\), and
Lemma~\ref{lem:source-internal-inclusion} gives
\(x:P\intjudge (t\,x)^\circ:P\); Theorem~\ref{thm:normalization} with
Lemma~\ref{lem:canonical-normal-retyping} supplies \(\mathcal E\).
Write
\[
  U_{\mathcal E}
  :=
  \SEM{x:P\intjudge N:P}_{\mathrm{NF}}:
  B_{\mathcal E}^{-}\longrightarrow B_{\mathcal E}^{+}.
\]

The fixed rig normal form of \(P\) is
\[
  \mathsf{sgn}(P)
  \cong
  \bigoplus_{w\in\{0,1\}^{n}}m_w,
  \qquad
  m_w:=(b^+)^{\tensor n},
  \qquad
  \sem{m_w}=\mathbb C.
  \tag{Register-rig}
\]
The canonical root display reads the $x$-coordinate at negative
polarity and the result coordinate at positive polarity.  Its
coordinate restrictions are therefore
\[
  \varepsilon_{\mathcal E}^{-}
  :=
  \mathsf{src}_P\!\upharpoonright_{B_{\mathcal E}^{-}},
  \qquad
  \varepsilon_{\mathcal E}^{+}
  :=
  \mathsf{tgt}_P\!\upharpoonright_{B_{\mathcal E}^{+}} .
  \tag{Register-restriction}
\]
These maps are unitary isomorphisms onto \(\sem P\).
Consequently there is a
unique unitary
\[
  U_t:\sem P\longrightarrow\sem P
\]
such that
\[
  \varepsilon_{\mathcal E}^{+}U_{\mathcal E}
  =
  U_t\varepsilon_{\mathcal E}^{-}.
  \tag{Register-readback}
\]
\end{lemma}

\begin{proof}
The singleton first-order root clause of
Theorem~\ref{thm:nf-boundary-unitarity} makes the two maps in
\textup{(Register-restriction)} unitary isomorphisms onto $\sem P$.
Define
\[
 U_t:=\varepsilon_{\mathcal E}^{+}U_{\mathcal E}
      (\varepsilon_{\mathcal E}^{-})^{-1}.
\]
Theorem~\ref{thm:nf-boundary-unitarity} makes $U_{\mathcal E}$, and
hence $U_t$, unitary.  The definition gives
\textup{(Register-readback)}, and invertibility of
$\varepsilon_{\mathcal E}^{-}$ gives uniqueness.
\end{proof}

\begin{proof}[Proof of Corollary~\ref{cor:register-unitarity}]
Apply Lemma~\ref{lem:first-order-endpoint-readback}.
\end{proof}

%% file: appendix-sum-encoding.tex
\section{Sum-Tag Encoding Invariants}
\label{app:sum-encoding}

This appendix states the sum-layout obligations used by the formal
compiler-correctness proof.  They concern the binary core.  Derived
$n$-ary forms elaborate to the fixed binary tree of
Appendix~\ref{app:nary-plus} (presented later in the appendix order;
only the binary layout is needed here) before the reference compiler
assigns a layout.

\paragraph{Reference binary layout.}
Put
\[
 w_0=\nwires(A),\qquad w_1=\nwires(B),\qquad
 W=\max\{w_0,w_1\}.
\]
The reference layout of $A\plus B$ is
\[
 [\,\text{root tag}\mid
    \text{payload}_0\mid\cdots\mid\text{payload}_{W-1}\,],
 \qquad
 \nwires(A\plus B)=1+W,
\]
with type-valid subspace
\begin{equation}
\begin{split}
 V_{A\plus B}=\operatorname{span}\bigl(&
 \{\ket{0}\tensor\ket\psi\tensor\ket{0^{W-w_0}}:
      \ket\psi\in V_A\}\;\cup\\[-2pt]
 &\{\ket{1}\tensor\ket\phi\tensor\ket{0^{W-w_1}}:
      \ket\phi\in V_B\}\bigr).
\end{split}
\label{eq:valid-sum}
\end{equation}
The summand payloads use their own recursively computed layouts.
Consequently nested source and target types are computed
independently; they are never silently recognized as a flat list of
leaves.  The reference proof uses typed logical frames for structural source and
target coordinates; \texttt{WirePerm} is reserved
for a literal symbolic wire-address update.  The prototype additionally
uses an $\plus$-outermost normalized frame, with the live tag outermost,
the selected summand and its tensor spectators in the shared payload,
and padding last: every summand occupies the same payload wires, and
the tag says which summand is live.  In that implementation frame, the multiplicative structural
isomorphisms, additive associativity, distributivity, and their
inverses themselves have empty gate lists.  An unequal-width
distributor records branch-dependent payload addresses in the
symbolic layout; later composition reconciles those addresses at the
splice.  Gates are emitted only when stored values
change; in particular a permutation of
live tag values ($\sigma^\plus$) is materialized in the gate list ---
one $X$ on a binary tag --- while the payload block is untouched.

\paragraph{Reference invariants.}
The binary encoding maintains:
\begin{description}
\item[INV-1 (Live root tag)]
  Each binary node has exactly the live tags $0$ and $1$; a nested
  summand retains its own recursively encoded tags inside the payload.
\item[INV-2 (Payload zeroing)]
  In root branch $i$, payload positions $w_i,\ldots,W-1$ are
  $\ket{0}$.
\item[INV-3 (Canonical block realization)]
  If each selected branch circuit maps its declared complete active
  canonical code sector onto its output sector, the binary row maps
  the corresponding branch-completed input code block onto its output
  block.  No behavior on the remainder of the register is required.
\item[INV-4 (Complement discipline)]
  Both root-tag bit values are live in the binary reference layout.
  The theorem assigns no behavior to invalid payload or padding
  patterns.  Layout-only structural rows carry the valid block and
  unused complement according to their declared sectorwise symbolic
  permutation and fixed complement bijection.  This metadata appends
  no gate and need not be pointwise identity on the complement.  An
  enclosing exponential may contribute its chosen scalar extension.
  After common padding, let $X^-$ and $X^+$ be the selected source
  and target subspaces, and write $\Pi_{X^\epsilon}$ for the
  orthogonal projector onto $X^\epsilon$.  For each
  $\epsilon\in\{-,+\}$, apply Gram--Schmidt to
  \[
    (I-\Pi_{X^\epsilon})\ket{0\cdots0},\ldots,
    (I-\Pi_{X^\epsilon})\ket{1\cdots1}
  \]
  in lexicographic order, discarding zero and dependent vectors and
  fixing each resulting vector's phase by making its first nonzero
  computational-basis coordinate positive real.  The fixed
  complement bijection sends the $k$th resulting source-complement
  vector to the $k$th target-complement vector.  For coordinate
  sectors this reduces to sending the $k$th lexicographically
  ordered unused source basis word to the $k$th unused target basis
  word; we call this the \emph{unused-basis-word clause}.
\end{description}

\paragraph{Codewords and structural label bijections.}
Valid basis labels and codewords are defined recursively:
$\mathcal L(\base)=\{*\}$ with $c_{\base}(*)=1$;
$\mathcal L(A\tensor B)=\mathcal L(A)\times\mathcal L(B)$ with
$c_{A\tensor B}(x,y)=c_A(x)\tensor c_B(y)$; and
$\mathcal L(A\plus B)=\mathcal L(A)\uplus\mathcal L(B)$ with
\[
 c_{A\plus B}(\mathrm{inl}\,x)=
 \ket0\tensor c_A(x)\tensor\ket{0^{W-w_0}},
 \qquad
 c_{A\plus B}(\mathrm{inr}\,y)=
 \ket1\tensor c_B(y)\tensor\ket{0^{W-w_1}} .
\]
The label and codeword definitions are formally made on compiler
layout expressions (the layout grammar of
Appendix~\ref{app:compilation-soundness}, presented later in the
appendix order).  In addition to the clauses above, put
\[
  \mathcal L(R^*)=\mathcal L(R),
  \qquad
  c_{R^*}(u)=c_R(u),
  \qquad
  V_{R^*}=V_R,
  \qquad
  \nwires(R^*)=\nwires(R),
\]
in agreement with the duality convention of that appendix.
For a source type $A$, we abbreviate
\[
  \mathcal L_A:=\mathcal L(\LayoutTy A),
  \qquad
  c_A:=c_{\LayoutTy A}.
\]
Since
$\LayoutTy{A\lmark B}=(\LayoutTy A)^*\tensor\LayoutTy B$,
it follows that
\[
  \mathcal L_{A\lmark B}
    =\mathcal L_A\times\mathcal L_B,
  \qquad
  c_{A\lmark B}(u,v)=c_A(u)\tensor c_B(v).
\]
Physical duality changes no codeword; polarity remains represented
by the boundary charts.  These codewords describe the flat physical
type envelope, not the derivation-selected semantic sector.

Each declared primitive structural atom instance $s:T\cong S$ acts
on labels by the evident bijection
$\rho_s:\mathcal L(\LayoutTy T)\to\mathcal L(\LayoutTy S)$:
$\alpha^\plus$ reassociates $\uplus$; $\sigma^\plus$ swaps
$\mathrm{inl}$ and $\mathrm{inr}$; $\distL$ sends
$(x,\mathrm{inl}\,y)$ to $\mathrm{inl}(x,y)$ and
$(x,\mathrm{inr}\,y)$ to $\mathrm{inr}(x,y)$, and $\distR$
symmetrically on the left factor; inverses invert.

\begin{lemma}[Structural codeword embeddings]
\label{lem:structural-layout}
For every declared primitive structural atom instance $s:T\cong S$,
put $q_s=\max\{\nwires(T),\nwires(S)\}$.  Identifying the fixed
semantic bases with the label sets $\mathcal L_T$ and $\mathcal L_S$,
linear extension of
\[
 \overline c_T\ket x
 :=c_T(x)\tensor\ket{0^{q_s-\nwires(T)}},
 \qquad
 \overline c_S\ket y
 :=c_S(y)\tensor\ket{0^{q_s-\nwires(S)}}
\]
defines isometries
\[
 \overline c_T:\sem T\lhook\joinrel\longrightarrow
   (\mathbb C^2)^{\tensor q_s},
 \qquad
 \overline c_S:\sem S\lhook\joinrel\longrightarrow
   (\mathbb C^2)^{\tensor q_s}.
\]
Moreover,
\[
 \mathsf{Rig}(s)\ket x=\ket{\rho_s x}
 \qquad(x\in\mathcal L_T).
 \tag{Structural-label}
\]
\end{lemma}
\begin{proof}
The recursive codewords are orthonormal, padding preserves inner
products, and $\rho_s$ is the displayed bijection of valid labels.
The final equality is the defining basis action of the primitive
structural isomorphism.
\end{proof}

\paragraph{Executable flat $n$-ary layout.}
The executable prototype's layout function is defined per node: at
each sum node, a flat $k=\lceil\log_2 n\rceil$-bit tag with shared
payload of width $W=\max_i\nwires(A_i)$ and padding inside that
node's shared block; at a tensor node, componentwise juxtaposition.
Flattening never crosses a tensor: for
$\mathbb Z_3\tensor\mathbb Z_5$ the layout is the juxtaposition
$(2{+}0)+(3{+}0)=5$ wires, not $\lceil\log_2 15\rceil=4$.  Exposed $n$-ary maps emit direct
exact-tag dispatch; otherwise the backend uses generic controlled
compilation in the same flat layout.  This implementation path lies
outside the recursive-binary reference layout and its correctness
theorem.  For example, the fixed left-associated four-way sum of
copies of $\base$ uses three recursive root tags in the reference
layout, whereas the flat prototype layout uses two tag bits.

\paragraph{Executable unequal-width distributivity.}
An unequal-width distributor has an empty gate list and records its
tag-dependent payload addresses in the symbolic layout.  Composition
through it is supported exactly.  At the incident endpoint, a difference
between producer and consumer frames that is a wire permutation is
absorbed symbolically; a genuine computational-basis word permutation is
emitted as an exact basis-word adapter and therefore has a real gate cost.
Optimizing these endpoint adapters is ongoing work.  The resulting action realizes the
branch-dependent source and target charts used by the reference
semantics.

\paragraph{Tag order and branch phases.}
The executable backend numbers flat tag words in big-endian order,
with the first tag qubit as the most significant bit.  For a nested
phased sum whose left and right subtrees contain $n_L$ and $n_R$
flattened leaves, respectively, the left phase $z=e^{i\theta}$ is
applied to every valid left tag word $0\leq i<n_L$, rather than
only to the all-zero word.  More generally, phased control with
branch maps $F_i$ and unit phases $z_i$ realizes, on the valid
tagged sector,
\[
   \bigoplus_{i=0}^{n-1} z_i F_i .
\]
Each branch body is compiled before its exact-tag control and
associated phase are applied.  The same tag-order convention is
used by \texttt{NPlusMap}, \texttt{PhasedPlusMap}, and
\texttt{PhasedControl}; the phase selector acts trivially on unused
tag words.

The structural rows implement their declared logical-frame bijections
on the reference structural code sectors.  In the normalized
prototype, the multiplicative structural isomorphisms, additive
associativity, distributivity, and their inverses themselves have
empty physical gate lists.  Unequal-width distributors record
branch-dependent symbolic addresses; any necessary basis alignment
is handled at a later splice.  Additive symmetry instead contributes
its tag permutation.  A structural row may recode dead coordinates,
and an
exponential whose primitive extension is the identity there may act
by the scalar $e^{i\theta}$.  Such complement choices lie outside the
language sector and are not consequences of \textup{(BC)}.

%% file: appendix-compilation-soundness-planned.tex
%

\section{Compilation Soundness}
\label{app:compilation-soundness}

For a Source judgment $\Gamma\sjudge t:A$, put
\[
 N_t=\mathsf{NF}_{\mathrm{LO}}(t^\circ),\qquad
 \mathcal N_t=\mathsf{CanDer}_{\Gamma,A}(N_t):
 \Gamma\intjudge N_t:A .
\]
Corollary~\ref{cor:source-normalization} and
Lemma~\ref{lem:canonical-normal-retyping} supply this derivation.
The compiler and semantics therefore use the same canonical normal
derivation.  The compiler first constructs its chart-directed carrier
plan and then emits the retained semantic constructor expression into
the incidences fixed by that plan.  The same recursion covers the
lower-$\Phi$ canonical calls generated by $\mathsf{MapApply}$ and
coherent sharing.  All semantic maps below live in $\mathbf{FdHilb}$
between the selected branch-paired spaces; target programs act on
their planned flat qubit registers.

\subsection{Flat Registers, Target Programs, and Artifact Layouts}
\label{app:valid-subspace}

\paragraph{Type-valid registers.}
The compiler uses
\[
 \LayoutTy{\base}=\base,\quad
 \LayoutTy{A\tensor B}=\LayoutTy A\tensor\LayoutTy B,\quad
 \LayoutTy{A\plus B}=\LayoutTy A\plus\LayoutTy B,\quad
 \LayoutTy{A\lmark B}=\LayoutTy A^*\tensor\LayoutTy B .
\]
Put
\[
 V_{\base}=\mathbb C,\qquad
 V_{A\tensor B}=V_A\tensor V_B,\qquad
 V_{T^*}=V_T .
\]
For a binary sum, $V_{A\plus B}$ is the type-valid subspace of the
recursive binary layout displayed in \eqref{eq:valid-sum}; nested sums
use this clause recursively.  Thus
$V_{A\lmark B}=V_{\LayoutTy A^*}\tensor V_{\LayoutTy B}$.
No behavior on invalid payload or padding patterns is asserted.

\begin{remark}[Duality does not change a register]
\label{rem:dual-valid}
The equality $V_{T^*}=V_T$ concerns physical coordinates only; it
does not identify semantic polarities or derivation-selected sectors.
\end{remark}

\begin{lemma}[Dimension of the type-valid register]
\label{lem:valid-dim}
For every compiler layout $T$, $\dim V_T=\size{T}$.
\end{lemma}
\begin{proof}
Induction on $T$.  Tensor dimensions multiply, while the two live-tag
subspaces in \eqref{eq:valid-sum} are orthogonal and their dimensions
add.  Duality preserves dimension, and implication is the translated
dual--tensor case, so its dimensions multiply.
\end{proof}

\paragraph{Flat and selected layouts.}
For a judgment shape
$\mathfrak J=(x_1{:}A_1,\ldots,x_m{:}A_m\intjudge A)$, put
\[
 \varphi_{\mathfrak J}=\mathsf{sgn}(A_1)^*\tensor\cdots\tensor
 \mathsf{sgn}(A_m)^*\tensor\mathsf{sgn}(A),\qquad
 \mathcal E_{\mathfrak J}^\pm
 =\sem{\partial^\pm(\varphi_{\mathfrak J})}.
\]
The recursive binary encoding gives flat physical spaces
$V_{\mathfrak J}^\pm$ and unitary coordinate maps
\begin{equation}
 \lay_{\mathfrak J}^\pm:
 V_{\mathfrak J}^\pm\xrightarrow{\;\cong\;}
 \mathcal E_{\mathfrak J}^\pm .
\label{eq:flat-layout}
\end{equation}
They use identity on $\base$, tensor products on $\tensor$, and
\begin{equation}
 \lay_{A\plus B}^\pm
 \bigl(\ket i\tensor\ket{\psi}\tensor
       \ket{0^{W-\nwires(A_i)}}\bigr)
 =
 \iota_i\bigl(\lay_{A_i}^\pm\ket{\psi}\bigr),
 \qquad A_0=A,\ A_1=B .
\label{eq:sum-layout}
\end{equation}
On implication the argument polarity is reversed:
\[
 \lay_{A\lmark B}^-=\lay_A^+\tensor\lay_B^-,
 \qquad
 \lay_{A\lmark B}^+=\lay_A^-\tensor\lay_B^+ .
\]
The existence of these unitary maps uses the fixed convention that
each DNF monomial contributes one negative and one positive coordinate;
hence both polarity spaces for $A\lmark B$ have the required dimension
$\size{A}\size{B}$.  Canonical reassociations and literal
wire-address permutations are suppressed.

For a canonical normal derivation
$\mathcal N:(\Gamma\intjudge N:A)$ of shape $\mathfrak J$, write
$B_{\mathcal N}^\pm\subseteq\mathcal E_{\mathfrak J}^\pm$ for its
selected boundary spaces and put
\begin{equation}
 \widetilde C_{\mathcal N}^\pm
 :=
 (\lay_{\mathfrak J}^\pm)^{-1}(B_{\mathcal N}^\pm)
 \subseteq V_{\mathfrak J}^\pm .
 \tag{Der-code}\label{eq:derivation-code}
\end{equation}

\paragraph{Target programs.}
For each $q\geq0$, let $\mathsf{TCirc}_q$ be the target IR of
$q$-wire programs.  It contains the empty program, the declared
backend gates, $X$, and, for $|\gamma|=1$,
\[
 \mathsf{Phase}_q(\gamma)\in\mathsf{TCirc}_q,\qquad
 \mathcal U_q(\mathsf{Phase}_q(\gamma))=\gamma I .
 \tag{Target-phase}
\]
It is closed under sequential and parallel composition, lifting along
an ordered wire injection, and exact computational-basis control.
Its interpretation satisfies
\[
 \mathcal U_q(C;D)=\mathcal U_q(D)\mathcal U_q(C)
 \tag{Target-seq}
\]
and
\[
 \mathcal U_{q+r}(C\mathbin{\|}D)
 =
 \mathcal U_q(C)\tensor\mathcal U_r(D).
 \tag{Target-par}
\]
For $q\geq1$, $d\in[q]$, and $b\in\{0,1\}$, let
\[
 \jmath_{d,b}:
 (\mathbb C^2)^{\tensor(q-1)}
 \lhook\joinrel\longrightarrow
 (\mathbb C^2)^{\tensor q}
\]
insert $b$ at coordinate $d$, placing the remaining coordinates in
the increasing order of $[q]\setminus\{d\}$.  Exact control means that,
for $H\in\mathsf{TCirc}_{q-1}$,
\[
\begin{aligned}
 \mathcal U_q\!\left(\mathsf{Ctrl}_{d=b}(H)\right)
 ={}&
 \jmath_{d,b}\mathcal U_{q-1}(H)\jmath_{d,b}^{\dagger}\\
 &+
 \jmath_{d,1-b}\jmath_{d,1-b}^{\dagger}.
\end{aligned}
\tag{Target-control}
\]
A displayed operator product is always the interpretation of the
corresponding temporal target program.

An \emph{exact basis-permutation supplier} is a family
\[
 \mathcal B_q:
 \operatorname{Sym}(\{0,1\}^q)\longrightarrow\mathsf{TCirc}_q
\]
such that
\[
 \mathcal U_q(\mathcal B_q(\pi))=P_\pi,\qquad
 P_\pi\ket x=\ket{\pi(x)} .
 \tag{BP-int}
\]
\begin{lemma}[Exact lowering of basis permutations]
\label{lem:basis-permutation-lowering}
An exact basis-permutation supplier exists for every $q$.
\end{lemma}
\begin{proof}
Every finite permutation is a product of transpositions.  A
transposition of two words is obtained from a simple Gray path between
them: exact-controlled one-bit flips along the path and back exchange
the endpoints and fix the other words.  Composing these target
programs along a transposition decomposition gives the supplier.
\end{proof}
Fix this supplier and define
\[
 \mathsf{Retarget}_{\pi}(C)
 :=
 \mathcal B_q(\pi^{-1});C;\mathcal B_q(\pi).
 \tag{Retarget}
\]
Then
\[
 \mathcal U_q(\mathsf{Retarget}_{\pi}(C))
 =
 P_\pi\mathcal U_q(C)P_\pi^\dagger .
 \tag{Retarget-int}
\]

\begin{definition}[Six-field target artifact]
\label{def:symbolic-layout-artifact}
A target artifact over $\mathcal N$ is a tuple
\[
 \mathcal A_{\mathcal N}
 =
 (q_{\mathcal N},G_{\mathcal N},
  L_{\mathcal N}^-,L_{\mathcal N}^+,
  p_{\mathcal N}^-,p_{\mathcal N}^+)
 \tag{Artifact-ty}
\]
where, for
$\mathcal H_{\mathcal N}
=(\mathbb C^2)^{\tensor q_{\mathcal N}}$,
\[
\begin{aligned}
 G_{\mathcal N}&\in\mathsf{TCirc}_{q_{\mathcal N}},\\
 L_{\mathcal N}^\epsilon&:
 \mathcal H_{\mathcal N}\xrightarrow{\cong}\mathcal H_{\mathcal N},\\
 p_{\mathcal N}^\epsilon&:
 \widetilde C_{\mathcal N}^\epsilon
 \lhook\joinrel\longrightarrow\mathcal H_{\mathcal N}
 \qquad(\epsilon\in\{-,+\}).
\end{aligned}
\]
The $L^\epsilon$ are total unitary frame extensions and the
$p^\epsilon$ are isometric placements with zero padding outside the
displayed live factors.  For any such tuple define
\[
 \Framed{\mathcal A_{\mathcal N}}
 :=
 (L_{\mathcal N}^+)^\dagger
 \mathcal U_{q_{\mathcal N}}(G_{\mathcal N})
 L_{\mathcal N}^- .
 \tag{Artifact-int}
\]
The target program and frame operators are total on the padded
register; all realization equations are restricted to the selected
placement domains.
\end{definition}

\paragraph{Padding and lifting.}
For an ordered wire injection
$\varrho:[k]\hookrightarrow[Q]$, let $P_\varrho$ be its fixed
order-preserving completion to a wire permutation and define
\[
\begin{aligned}
 \mathsf{Pad}_\varrho(\ket\psi)
 &:=
 P_\varrho(\ket\psi\tensor\ket{0^{Q-k}}),\\
 \mathsf{Lift}_\varrho(U)
 &:=
 P_\varrho(U\tensor I_{Q-k})P_\varrho^\dagger .
\end{aligned}
\tag{Pad-Lift}
\]
For $C\in\mathsf{TCirc}_k$, the target constructor
$\mathsf{TLift}_\varrho(C)\in\mathsf{TCirc}_Q$ satisfies
\[
 \mathcal U_Q(\mathsf{TLift}_\varrho(C))
 =
 \mathsf{Lift}_\varrho(\mathcal U_k(C)).
 \tag{Target-lift}
\]

\begin{definition}[Flat, sector, and common-register layouts]
\label{def:layout-iso}
Let $\mathcal A_{\mathcal N}$ be a six-field artifact over a canonical
normal derivation $\mathcal N$ of shape $\mathfrak J$.  Its
common-register code sectors and restricted layouts are
\[
 C_{\mathcal N}^\pm
 :=
 p_{\mathcal N}^\pm(\widetilde C_{\mathcal N}^\pm)
 \subseteq\mathcal H_{\mathcal N},
 \qquad
 \lambda_{\mathcal N}^\pm
 :=
 \lay_{\mathfrak J}^\pm
 (p_{\mathcal N}^\pm)^\dagger
 \!\upharpoonright_{C_{\mathcal N}^\pm}.
\]
Thus
$\lambda_{\mathcal N}^\pm:C_{\mathcal N}^\pm
 \xrightarrow{\cong}B_{\mathcal N}^\pm$.
\end{definition}

\begin{lemma}[Layout maps are unitary]
\label{lem:lambda-unitary}
The flat layouts are unitary, and for every six-field artifact the
placements are isometries and the restricted layouts are unitary
between their displayed spaces.
\end{lemma}
\begin{proof}
The flat claim is structural induction using \eqref{eq:sum-layout}.
The remaining claims follow from the artifact typing: restrict the
partial inverse of each isometric placement to its image and compose
with the flat layout.
\end{proof}

\paragraph{First-order identification.}
\label{par:first-order-identification}
For first-order $A$, the negative and positive data copies have the
same physical encoding; write
$\lay_A:V_A\xrightarrow{\cong}\sem A$.

\subsection{Plan-Neutral Coordinate Laws}

These are definitions and local commuting squares used by the active
proof.  They allocate no carrier and contain no recursive compiler
argument.

For every selected package $\mathcal X$ below, of judgment shape
$\mathfrak J_{\mathcal X}$, write
\[
 \ell_{\mathcal X}^{\epsilon}
 :=\lay_{\mathfrak J_{\mathcal X}}^{\epsilon}
   \!\upharpoonright_{\widetilde C_{\mathcal X}^{\epsilon}}
 :\widetilde C_{\mathcal X}^{\epsilon}
  \xrightarrow{\;\cong\;}B_{\mathcal X}^{\epsilon}.
 \tag{Restricted-layout}
\]

\paragraph{Completion and simple reparameterization.}
For an inactive context $\Xi$, let
\[
 q_\Xi:=\sum_{x{:}T\in\Xi}\nwires(T),
 \qquad
 \zeta_{\Xi,0}:\mathsf D_\Xi
 \lhook\joinrel\longrightarrow
 (\mathbb C^2)^{\tensor q_\Xi}
\]
be the tensor product, in context order, of the canonical type-word
incidences of its factors.  With the fixed readouts $d_\Xi^\epsilon$
of \textup{(TY-read)}, put
\[
 \zeta_\Xi^\epsilon
 :=\zeta_{\Xi,0}d_\Xi^\epsilon:
 \mathcal T_\Xi^\epsilon\lhook\joinrel\longrightarrow
 (\mathbb C^2)^{\tensor q_\Xi}.
\]
Then \textup{(TY-read)} gives the spectator square
\[
 \zeta_\Xi^+Y_\Xi
 =\zeta_{\Xi,0}d_\Xi^+Y_\Xi
 =\zeta_{\Xi,0}d_\Xi^-
 =\zeta_\Xi^- .
 \tag{Through-code}
\]
This is spectator presentation data carried by the identity program,
not a standalone six-field artifact.  For the empty context,
\[
 \mathsf D_\varnothing=\mathcal T_\varnothing^\epsilon=\mathbb C,
 \qquad
 \zeta_{\varnothing,0}=\zeta_\varnothing^\epsilon=\id_{\mathbb C}.
 \tag{Through-empty}
\]
Let $\mathcal A_\parallel$ be the planned parallel artifact for an
action together with this inactive spectator, with placements
$p_\parallel^\epsilon$, and let $\ell_\parallel^\epsilon$ be its
restricted flat layout.  Let
$\ell_{\mathcal A\mid\Xi}^\epsilon$ be the same selected boundary in
canonical active-first order and define
\[
 \widehat d_{\mathcal A\mid\Xi}^\epsilon
 :=
 (\ell_\parallel^\epsilon)^{-1}
 \ell_{\mathcal A\mid\Xi}^\epsilon,
 \qquad
 \ell_\parallel^\epsilon
 \widehat d_{\mathcal A\mid\Xi}^\epsilon
 =
 \ell_{\mathcal A\mid\Xi}^\epsilon .
 \tag{Complete-pullback}
\]
Define the completed artifact to have the same register, target
program, and frames as $\mathcal A_\parallel$, but placements
\[
 p_{\mathcal A\mid\Xi}^\epsilon
 :=
 p_\parallel^\epsilon
 \widehat d_{\mathcal A\mid\Xi}^\epsilon .
 \tag{Complete-place}
\]
For $\Xi\mid\mathcal A$ use the analogous pullback into the opposite
canonical factor order.  These are unchanged-image views:
$\operatorname{im}(p_{\mathcal A\mid\Xi}^\epsilon)
=\operatorname{im}(p_\parallel^\epsilon)$.
Let $W_{\mathcal A\mid\Xi}$ denote the restriction of the unchanged
framed program to this common physical sector, and define
$\lambda_{\mathcal A\mid\Xi}^\epsilon$ from
Definition~\ref{def:layout-iso} using
\textup{(Complete-place)}.  Tensoring the active realization square with
\textup{(Through-code)} and applying \textup{(Target-lift)} in the planned
ambient register gives the parallel realization equation.  Since
\[
 \ell_{\mathcal A\mid\Xi}^\epsilon
 (\widehat d_{\mathcal A\mid\Xi}^\epsilon)^\dagger
 =
 \ell_\parallel^\epsilon,
\]
substitution into the parallel realization equation gives
\[
\begin{aligned}
 \lambda_{\mathcal A\mid\Xi}^{+}
 W_{\mathcal A\mid\Xi}
 (\lambda_{\mathcal A\mid\Xi}^{-})^{-1}
 &=U_{\mathcal A}\tensor Y_{\Xi},\\
 \lambda_{\Xi\mid\mathcal A}^{+}
 W_{\Xi\mid\mathcal A}
 (\lambda_{\Xi\mid\mathcal A}^{-})^{-1}
 &=Y_{\Xi}\tensor U_{\mathcal A}.
\end{aligned}
\tag{Complete-selected}
\]

For
$\mathcal B:(\Gamma,x{:}A\intjudge N:B)$ and its abstraction
$\mathcal L:(\Gamma\intjudge\lam{x}{N}:A\lmark B)$, let
\[
 \overline r_A^\epsilon:
 \mathcal E_{\mathfrak J_{\mathcal L}}^\epsilon
 \xrightarrow{\;\cong\;}
 \mathcal E_{\mathfrak J_{\mathcal B}}^\epsilon
\]
be the canonical polarity repartition, and write
$r_{A,B}^\epsilon$ for its restriction
$B_{\mathcal L}^\epsilon\xrightarrow{\cong}B_{\mathcal B}^\epsilon$.
The abstraction clause gives
$r_{A,B}^+U_{\mathcal L}=U_{\mathcal B}r_{A,B}^-$.  Define its flat
pullback by
\[
 r_A^\epsilon
 :=(\ell_{\mathcal B}^\epsilon)^{-1}
    r_{A,B}^\epsilon\ell_{\mathcal L}^\epsilon:
 \widetilde C_{\mathcal L}^\epsilon
 \xrightarrow{\;\cong\;}
 \widetilde C_{\mathcal B}^\epsilon,
 \qquad
 \ell_{\mathcal B}^\epsilon r_A^\epsilon
 =r_{A,B}^\epsilon\ell_{\mathcal L}^\epsilon.
 \tag{Repart-pullback}
\]
Thus implication repartition retains $(q,G,L^-,L^+)$ and changes the
selected placement by
\[
 p^\epsilon\longmapsto p^\epsilon r_A^\epsilon .
 \tag{Repart}
\]
For tensor packaging, write $\ell_{\mathcal N}^\epsilon$ and
$\ell_{\mathcal B_N}^\epsilon$ for the two restricted flat layouts and
define
\[
 \widehat\vartheta_N^\epsilon
 :=
 (\ell_{\mathcal N}^\epsilon)^{-1}
 \vartheta_N^\epsilon\ell_{\mathcal B_N}^\epsilon,
 \qquad
 \ell_{\mathcal N}^\epsilon\widehat\vartheta_N^\epsilon
 =
 \vartheta_N^\epsilon\ell_{\mathcal B_N}^\epsilon .
 \tag{TenPack-pullback}
\]
The packaging view is
\[
 p^\epsilon\longmapsto p^\epsilon\widehat\vartheta_N^\epsilon .
 \tag{TenPack}
\]

\paragraph{Retained higher-order cut network.}
For a classified cut $\rho$ at a port containing $\lmark$, retain the
finite expression $\mathsf{Net}(\rho)$ constructed by
Definition~\ref{def:classified-cut}, before its semantic operators are
multiplied.  For any unitary selected-coordinate isomorphism
$\omega^\epsilon:D_{\rm out}^\epsilon\to D_{\rm in}^\epsilon$ in that
expression, $D_{\rm in}^\epsilon$ is the selected sector of the
artifact receiving the placement.  At a product node its restricted
layout $\ell_{\rm in}^\epsilon$ is the tensor product of the component
restricted layouts; it need not itself be a judgment-shape instance
of \textup{(Restricted-layout)}.  Define its flat pullback by
\[
 \ell_{\rm in}^\epsilon\widehat\omega^\epsilon
 =
 \omega^\epsilon\ell_{\rm out}^\epsilon .
 \tag{Net-repart-pullback}
\]
It is unitary and changes a placement only by
\[
 p^\epsilon\longmapsto p^\epsilon\widehat\omega^\epsilon,
 \qquad
 \operatorname{im}(p^\epsilon\widehat\omega^\epsilon)
 =
 \operatorname{im}(p^\epsilon).
 \tag{Net-Repart}
\]
For \textup{(C-var-whole)} the view is
\[
 \omega_{\rho,S}^\epsilon
 :=
 \chi_\rho^\epsilon
 =
 (\mathsf{ug}_{\rho,S}^\epsilon)^\dagger
 \!\upharpoonright_{B_\rho^\epsilon}:
 B_\rho^\epsilon\xrightarrow{\cong}K_\rho^\epsilon .
 \tag{Eta-repart}
\]
The retained node inventory is
\[
\begin{array}{@{}ll@{}}
\textup{semantic clause}&\textup{planned lowering node}\\
\hline
\textup{(Y)}
 &\text{the unchanged-image view }\widehat\omega\\
\textup{(C-var-app)}
 &\textup{(Plan-Par) on operands, consumer, and yank, then }\widehat\omega\\
\textup{(C-var-whole)}
 &\textup{(Plan-Par) on operands and consumer, then }\widehat\omega_{\rho,S}\\
\textup{(EndAct)},\ h=s
 &\mathsf{StrAct}_{s}\\
\textup{(C-str)}
 &\mathsf{StrPre}_{s}\text{ followed by the prior planned cut }\rho'.
\end{array}
\tag{HSplice-cases}
\]
The tensor and shared-block rows traversed by
\textup{(C-var-whole)} are already included in
$\omega_{\rho,S}$ through \textup{(Eta-graft)}.

For \textup{(EndAct)} and \textup{(Str-pre)}, respectively, define the
code pullbacks by
\[
 \ell_{s\mathcal D}^\epsilon\widehat e_{s,\mathcal D}^\epsilon
 =
 e_{s,\mathcal D}^\epsilon\ell_{\mathcal D}^\epsilon
 \tag{StrAct-pullback}
\]
and
\[
 \ell_{\mathcal C}^\epsilon\widehat t_{s,\mathcal C}^\epsilon
 =
 t_{s,\mathcal C}^\epsilon
 \ell_{\mathcal C\triangleleft s}^\epsilon .
 \tag{StrPre-pullback}
\]
Thus $\mathsf{StrAct}_{s}$ uses
$p^\epsilon(\widehat e_{s,\mathcal D}^\epsilon)^\dagger$ and
$\mathsf{StrPre}_{s}$ uses
$p^\epsilon\widehat t_{s,\mathcal C}^\epsilon$; both retain the
physical program and carrier.

\paragraph{Block coordinate data.}
For any binary block, let
\[
 e_i^\epsilon:
 \widehat B_i^\epsilon\lhook\joinrel\longrightarrow
 B_{\mathcal B}^\epsilon
 \qquad(i=1,2)
\]
be the component inclusions supplied by its semantic clause: the
standalone-map inclusions, or the $j_i^\epsilon$ of
\textup{(Use-part)} or \textup{(CS-part)}.  With completed branch and
block shapes $\widehat{\mathfrak J}_i$ and
$\mathfrak J_{\mathcal B}$, put
\[
 \widehat C_i^\epsilon
 :=
 (\lay_{\widehat{\mathfrak J}_i}^\epsilon)^{-1}
 (\widehat B_i^\epsilon),
 \qquad
 C_{\mathcal B}^{\epsilon,\mathrm{flat}}
 :=
 (\lay_{\mathfrak J_{\mathcal B}}^\epsilon)^{-1}
 (B_{\mathcal B}^\epsilon).
\]
Write $\widehat\ell_i^\epsilon$ and
$\ell_{\mathcal B}^\epsilon$ for the restricted layouts and let
$k_i^\epsilon$ be the canonical inclusions into
$\widehat C_1^\epsilon\oplus\widehat C_2^\epsilon$.  Define
\[
\begin{aligned}
 P_{\mathcal B}^\epsilon&:
 \widehat C_1^\epsilon\oplus\widehat C_2^\epsilon
 \xrightarrow{\cong}C_{\mathcal B}^{\epsilon,\mathrm{flat}},\\
 P_{\mathcal B}^\epsilon k_i^\epsilon
 &:=
 (\ell_{\mathcal B}^\epsilon)^{-1}
 e_i^\epsilon\widehat\ell_i^\epsilon .
\end{aligned}
\tag{Block-pack}
\]
Orthogonality and joint exhaustion of the semantic inclusions make
$P_{\mathcal B}^\epsilon$ unitary.  If
$\widehat U_i:\widehat B_i^-\to\widehat B_i^+$ is the completed branch
action, its code-coordinate action is
\[
 \widehat V_i
 :=
 (\widehat\ell_i^+)^{-1}
 \widehat U_i\widehat\ell_i^-:
 \widehat C_i^-\xrightarrow{\cong}\widehat C_i^+ .
 \tag{Block-code}
\]

\paragraph{Distributor frame.}
For the coherent-sharing distributor, define
\[
 \widehat d_{\delta,K}^\epsilon
 :=
 (\ell_{\rho_K}^\epsilon)^{-1}
 d_{\delta,K}^\epsilon\ell_K^\epsilon:
 \widetilde C_K^\epsilon
 \xrightarrow{\cong}\widetilde C_{\rho_K}^\epsilon .
 \tag{Delta-code}
\]
Choose the unchanged-image presentation
$p_K^\epsilon=p_{\rho_K}^\epsilon
 \widehat d_{\delta,K}^\epsilon$.  Its physical image
$C_K^\epsilon=p_K^\epsilon(\widetilde C_K^\epsilon)$ is the same as
that of $p_{\rho_K}^\epsilon$.  Define the unchanged artifact
\[
 \mathsf{Frame}_{\delta,K}(\mathcal A_{\rho_K})
 :=
 (q_{\rho_K},G_{\rho_K},
  L_{\rho_K}^-,L_{\rho_K}^+,p_K^-,p_K^+).
 \tag{Delta-frame}
\]
Moreover,
\[
 \lambda_K^\epsilon
 =\ell_K^\epsilon(\widehat d_{\delta,K}^\epsilon)^\dagger
   (p_{\rho_K}^\epsilon)^\dagger
 =(d_{\delta,K}^\epsilon)^{-1}\lambda_{\rho_K}^\epsilon .
 \tag{Delta-layout}
\]
Consequently the unchanged artifact gives
\[
 \lambda_K^+
 \Framed{\mathsf{Frame}_{\delta,K}(\mathcal A_{\rho_K})}
 (\lambda_K^-)^{-1}
 =
 (d_{\delta,K}^+)^{-1}
 \lambda_{\rho_K}^+
 \Framed{\mathcal A_{\rho_K}}
 (\lambda_{\rho_K}^-)^{-1}
 d_{\delta,K}^- .
 \tag{Delta-realize}
\]

\paragraph{Native artifacts.}
For the variable derivation
$\mathcal X_A:(x{:}A\intjudge x:A)$, put
\[
 \widehat g_{\mathcal X_A}^\epsilon
 :=(\ell_{\mathcal X_A}^\epsilon)^{-1}g_{\mathcal X_A}^\epsilon:
 \mathsf D_A\xrightarrow{\;\cong\;}
 \widetilde C_{\mathcal X_A}^\epsilon .
\]
Let
$\zeta_A:\mathsf D_A\lhook\joinrel\longrightarrow
(\mathbb C^2)^{\tensor\nwires(A)}$
be the canonical equal-address type-codeword incidence in the fixed
basis, including the prescribed padding.  The fixed basis of
$\mathsf D_A$ is indexed by the type labels $\mathcal L(A)$, and
Lemma~\ref{lem:valid-dim} supplies the required equality of
dimensions.  Define
\[
 p_A^\epsilon:=\zeta_A
  (\widehat g_{\mathcal X_A}^\epsilon)^\dagger,
 \qquad
 \mathsf{Wire}_A
 :=(\nwires(A),I,I,I,p_A^-,p_A^+).
 \tag{Wire}
\]
Then
\[
 p_A^-\widehat g_{\mathcal X_A}^-
 =\zeta_A
 =p_A^+\widehat g_{\mathcal X_A}^+.
 \tag{Wire-square}
\]
The variable chart has coordinate action $I$, so
\textup{(Wire-square)} realizes $\mathsf{yank}_A$.
For a primitive structural atom $s:T\cong S$, write
$\mathfrak J_s=(\cdot\intjudge s:T\lmark S)$ and put
\[
 \widehat g_s^\epsilon
 :=
 (\lay_{\mathfrak J_s}^{\epsilon})^{-1}g_s^\epsilon,
 \qquad
 \widetilde C_s^\epsilon=\operatorname{im}(\widehat g_s^\epsilon),
 \qquad
 q_s=\max\{\nwires(T),\nwires(S)\}.
\]
Let $\overline c_T,\overline c_S$ be the padded codeword embeddings of
Lemma~\ref{lem:structural-layout} and set
\[
 p_s^-=\overline c_T(\widehat g_s^-)^\dagger,\qquad
 p_s^+=\overline c_S(\widehat g_s^+)^\dagger
 \tag{Struct-place}
\]
on the selected graph codes.  Let $G_s=I$ for the associativity and
distributivity atoms and their inverses, and let $G_s$ be $X$ on the
root additive tag for additive symmetry and its inverse:
\[
 G_s=
 \begin{cases}
 I,&s\in\{\alpha^\plus,\distL,\distR
          \text{ and their inverses}\},\\
 \mathsf X_{\operatorname{tag}(s)},
   &s=\sigma^\plus\text{ or its inverse}.
 \end{cases}
 \tag{Struct-G}
\]
Take $L_s^-=I$ and define $L_s^+$ on the selected range by
\[
 G_sL_s^-p_s^-\widehat g_s^-
 =
 L_s^+p_s^+\widehat g_s^+\mathsf{Rig}(s),
 \tag{Struct-frame}
\]
then extend by INV-4.  This gives the six-field structural artifact
$\mathsf{Struct}_s$.  Thus $G_s=I$ denotes the natural framed core for
associativity and distributivity; it does not claim that their
recursive-binary codeword maps are wire-address permutations.  If an
adjacent planned incidence has a different computational-word image,
Lemma~\ref{lem:planned-endpoint-transport} emits the required exact
basis-word adapter.  In particular, additive associativity and an
unequal-width right distributor pay for their re-encoding there, not
again in \textup{(Struct-G)}.
The quantum base artifact $\mathsf{Prim}_a$ is supplied by
\textup{(BC)} below.

\subsection{Backend Correctness}
\label{app:soundness-assumptions}

The proof uses the backend-correctness hypothesis \textup{(BC)} of
\S\ref{subsubsec:unitary-primitives}.  It covers the declared quantum
artifacts and all target constructors used above.  For
$a=\expi{\theta}{J}:P\lmark P$, with $P$ first-order, let
$\mathcal A_a:(\cdot\intjudge a:P\lmark P)$ be its canonical
derivation and put
\[
 V_a=e^{i\theta J},\qquad n_P=\nwires(P),\qquad
 \mathcal H_P=(\mathbb C^2)^{\tensor n_P}.
\]
The backend supplies the six-field artifact over $\mathcal A_a$
\[
 \mathsf{Prim}_a
 =
 (n_P,G_a,L_a^-,L_a^+,p_a^-,p_a^+),
 \qquad
 G_a\in\mathsf{TCirc}_{n_P}.
 \tag{BC-carrier}
\]
Put
\[
 \varepsilon_a^-=
 \mathsf{src}_P\!\upharpoonright_{B_{\mathcal A_a}^-},
 \qquad
 \varepsilon_a^+=
 \mathsf{tgt}_P\!\upharpoonright_{B_{\mathcal A_a}^+}.
\]
Writing
$\ell_a^\epsilon:=\ell_{\mathcal A_a}^\epsilon$,
the type-canonical placement clause is the explicit equation
\[
 p_a^\epsilon
 =
 c_P\varepsilon_a^\epsilon\ell_a^\epsilon .
 \tag{BC-code}
\]
Thus both placements have the declared valid $P$-codeword image,
although their selected-coordinate parameterizations may differ.
Tag and zero-\allowbreak padding coordinates prescribed by the type
already belong to this carrier.  The backend obligation is
\[
 \varepsilon_a^+\lambda_{\mathcal A_a}^+
 \Framed{\mathsf{Prim}_a}\!\upharpoonright_{C_{\mathcal A_a}^-}
 (\lambda_{\mathcal A_a}^-)^{-1}
 (\varepsilon_a^-)^{-1}
 =
 V_a .
 \tag{BC-Exp}
\]
Equivalently, \textup{(BC-code)} reduces this obligation to the
directly checkable physical-code equation
\[
 \Framed{\mathsf{Prim}_a}\,c_P=c_PV_a,
 \qquad
 c_P^\dagger\Framed{\mathsf{Prim}_a}\,c_P=V_a .
 \tag{BC-circuit}
\]
The first equality includes preservation of the valid code sector;
conversely the second implies it because the displayed compression is
unitary.  No action on the complement of that sector is assumed;
composite widths are supplied by CarrierPlan.

\begin{lemma}[Closed exponential realization]
\label{lem:exp-correctness}
Assume \textup{(BC)}.  For the canonical derivation $\mathcal A_a$
above,
\begin{equation}
 \Framed{\mathsf{Prim}_a}\!\upharpoonright_{C_{\mathcal A_a}^-}
 =
 (\lambda_{\mathcal A_a}^+)^{-1}
 \SEM{\cdot\intjudge a:P\lmark P}_{\mathrm{NF}}
 \lambda_{\mathcal A_a}^- .
\label{eq:exp-realization}
\end{equation}
\end{lemma}
\begin{proof}
Equation~\textup{(Exp-graph)} identifies the selected normal-form
action with $\mathsf{Graph}_P(V_a)$, while
\textup{(QGraph-read)} gives
$(\varepsilon_a^+)^{-1}V_a\varepsilon_a^-
=\mathsf{Graph}_P(V_a)$.  Solving \textup{(BC-Exp)} for the framed
map gives the claim.
\end{proof}

For later use, if a unitary selected-coordinate pullback satisfies
$\ell_{\rm old}^\epsilon\widehat\omega^\epsilon
=\omega^\epsilon\ell_{\rm new}^\epsilon$ and
$p_{\rm new}^\epsilon=p_{\rm old}^\epsilon\widehat\omega^\epsilon$,
then
\[
 \lambda_{\rm new}^\epsilon
 =
 \ell_{\rm new}^\epsilon
 (p_{\rm new}^\epsilon)^\dagger
 =
 (\omega^\epsilon)^{-1}\lambda_{\rm old}^\epsilon .
 \tag{View-layout}
\]
Thus substitution transports a selected realization equation without
changing its physical image or target program.

\begin{lemma}[Plan-neutral local realization laws]
\label{lem:plan-local-laws}
Assume \textup{(BC)}, and suppose a premise artifact satisfies its
selected realization equation.  The following native or
unchanged-image rows realize the corresponding semantic clauses
without allocating any additional, emitter-owned wire:
\[
\begin{array}{@{}lll@{}}
\textbf{row}&\textbf{semantic action}&\textbf{reason}\\
\hline
\mathsf{Wire}&\mathsf{yank}&\textup{(Wire-square)}\\
\mathsf{Struct}&\mathsf{Rig}_{\mathfrak S}(s)&\textup{(Struct-frame)}\\
\mathsf{Prim}&\mathsf{Graph}(V_a)&\textup{(BC-Exp)}\\
\mathsf{Repart}&\text{implication repartition}&\textup{(Repart)}\\
\mathsf{StrAct}&\mathsf{EndAct}_s&\textup{(StrAct-pullback)}\\
\mathsf{StrPre}&(t_{s,\mathcal C}^+)^\dagger U
 t_{s,\mathcal C}^-&\textup{(StrPre-pullback)}\\
\mathsf{TenPack}&U_{\mathcal B_N}
 &\textup{(TenPack-pullback)},\ \textup{(TenPack)}\\
\mathsf{Frame}_{\delta,K}
 &(d_{\delta,K}^+)^{-1}U d_{\delta,K}^-
 &\textup{(Delta-realize)}.
\end{array}
\tag{Plan-local-laws}
\]
\end{lemma}
\begin{proof}
The Wire, Struct, and Prim rows are their displayed defining
commuting squares, using Lemma~\ref{lem:exp-correctness} for Prim.
Each remaining row follows by the substitution
\textup{(View-layout)} using the displayed pullback for that row
(and \textup{(Ten-pack)} for tensor packaging).  It changes only the
selected-coordinate presentation and therefore retains both the
physical image and the target program.
\end{proof}

\input{carrier-planned-emission}

\begin{definition}[Active reference emitter]
\label{def:planned-reference-emitter}
For every canonical call in the admitted Source recursion, define
\[
 \mathsf{Emit}(\mathcal N):=\mathsf{PlanEmit}(\mathcal N)
 =
 (Q_{\mathcal N},G_{\mathcal N},
  L_{\mathcal N}^-,L_{\mathcal N}^+,
  p_{\mathcal N}^-,p_{\mathcal N}^+).
 \tag{Emit=PlanEmit}
\]
Thus $\Framed{\mathcal N}$ abbreviates
$\Framed{\mathsf{PlanEmit}(\mathcal N)}$, and
Definition~\ref{def:layout-iso} supplies
$C_{\mathcal N}^\pm$ and $\lambda_{\mathcal N}^\pm$ from these fields.
\end{definition}

\subsection{The Normal-Form Case}
\label{app:compilation-thm}

\begin{lemma}[Circuit realization for normal forms]
\label{lem:compilation-nf}
Assume \textup{(BC)}.  Let
$\mathcal N:(\Gamma\intjudge N:A)$ be a canonical Source normal-form
call or a recursively generated lower-$\Phi$ call.  Then
$\Framed{\mathcal N}$ is unitary on its padded register, maps
$C_{\mathcal N}^-$ unitarily onto $C_{\mathcal N}^+$, and
\begin{equation}
 \Framed{\mathcal N}\!\upharpoonright_{C_{\mathcal N}^-}
 =
 (\lambda_{\mathcal N}^+)^{-1}
 \SEM{\Gamma\intjudge N:A}_{\mathrm{NF}}
 \lambda_{\mathcal N}^- .
\label{eq:nf-realization}
\end{equation}
\end{lemma}
\begin{proof}
Let
\[
 \widehat g_{\mathcal N}^\epsilon
 =
 \left(
 \lay_{\mathfrak J}^\epsilon
 \!\upharpoonright_{\widetilde C_{\mathcal N}^\epsilon}
 \right)^{-1}
 g_{\mathcal N}^\epsilon .
\]
Theorem~\ref{thm:chart-directed-plan-emit} gives
\[
 \mathcal U(G_{\mathcal N})L_{\mathcal N}^-p_{\mathcal N}^-
 \widehat g_{\mathcal N}^-
 =
 L_{\mathcal N}^+p_{\mathcal N}^+
 \widehat g_{\mathcal N}^+V_{\mathcal N}.
\]
Premultiplying by $(L_{\mathcal N}^+)^\dagger$ yields
\[
 \Framed{\mathcal N}\,
 p_{\mathcal N}^-\widehat g_{\mathcal N}^-
 =
 p_{\mathcal N}^+\widehat g_{\mathcal N}^+V_{\mathcal N}.
 \tag{NF-from-plan}
\]
For every chart coordinate $d$,
\[
 \lambda_{\mathcal N}^-
 p_{\mathcal N}^-\widehat g_{\mathcal N}^-d
 =g_{\mathcal N}^-d,
 \qquad
 \lambda_{\mathcal N}^+
 p_{\mathcal N}^+\widehat g_{\mathcal N}^+d
 =g_{\mathcal N}^+d.
\]
The semantic chart square is
\[
 \SEM{\Gamma\intjudge N:A}_{\mathrm{NF}}g_{\mathcal N}^-
 =
 g_{\mathcal N}^+V_{\mathcal N}.
\]
Hence the two sides of \eqref{eq:nf-realization} agree after
precomposition with
$p_{\mathcal N}^-\widehat g_{\mathcal N}^-$.  This map is onto
$C_{\mathcal N}^-$ because the complete chart is onto its selected
boundary, so \eqref{eq:nf-realization} follows.  The semantic action
is unitary between the selected boundaries, while $G_{\mathcal N}$
and the two frames are total unitaries.  The same equation therefore
gives both remaining assertions.
\end{proof}

\subsection{The Compilation Theorem}
\label{app:all-terms}

\begin{definition}[Canonical compilation]
\label{def:canonical-compilation}
For a Source judgment $\Gamma\sjudge t:A$, define
\[
 \mathsf{Compile}(t):=\mathsf{PlanEmit}(\mathcal N_t)
\]
and
\begin{equation}
 \Framed{t}:=\Framed{\mathsf{Compile}(t)}
            =\Framed{\mathcal N_t}.
\label{eq:normalize-emit}
\end{equation}
\end{definition}

\begin{theorem}[Circuit Realization]
\label{thm:compilation-soundness}
Assume \textup{(BC)}.  For every Source judgment
$\Gamma\sjudge t:A$, the circuit $\Framed{t}$ is unitary on its
padded register and maps
$C_{\mathcal N_t}^-$ unitarily onto $C_{\mathcal N_t}^+$.  Moreover,
\begin{equation}
\begin{aligned}
 \Framed{t}\!\upharpoonright_{C_{\mathcal N_t}^-}
 &=
 (\lambda_{\mathcal N_t}^+)^{-1}
 \SEM{\Gamma\intjudge N_t:A}_{\mathrm{NF}}
 \lambda_{\mathcal N_t}^-\\
 &=
 (\lambda_{\mathcal N_t}^+)^{-1}
 \SEM{\Gamma\sjudge t:A}
 \lambda_{\mathcal N_t}^- .
\end{aligned}
\label{eq:all-realization}
\end{equation}
\end{theorem}
\begin{proof}
The first equality is Lemma~\ref{lem:compilation-nf}; the second is
Definition~\ref{def:canonical-form-semantics}.
\end{proof}

This establishes Theorem~\ref{thm:compilation-soundness-main}.

\subsection{Qubit-Register Execution}
\label{app:qubit-execution}

For $n\geq1$, put
\[
 P:=\mathsf{QReg}_n:=\QBool^{\tensor n}.
\]
Its fixed first-order layout is
\[
 \lay_P:V_P\xrightarrow{\cong}\sem P
 \cong(\mathbb C^2)^{\tensor n}.
\]
Let $\mathcal D:(\cdot\sjudge t:P\lmark P)$, choose a fresh $x:P$,
and let
\[
 N=\mathsf{NF}_{\mathrm{LO}}((t\,x)^\circ),\qquad
 \mathcal E:(x:P\intjudge N:P)
\]
be the canonical internal derivation.  By
Lemma~\ref{lem:first-order-endpoint-readback}, there are unitaries
\[
 \varepsilon_{\mathcal E}^{\pm}:
 B_{\mathcal E}^{\pm}\xrightarrow{\cong}\sem P
\]
such that
\[
 U_t
 =
 \varepsilon_{\mathcal E}^{+}
 \SEM{x:P\intjudge N:P}_{\mathrm{NF}}
 (\varepsilon_{\mathcal E}^{-})^{-1}.
 \tag{Register-unitary}
\]
Write
\[
 \mathsf{PlanEmit}(\mathcal E)
 =
 (Q_{\mathcal E},G_{\mathcal E},
  L_{\mathcal E}^-,L_{\mathcal E}^+,
  p_{\mathcal E}^-,p_{\mathcal E}^+)
\]
and define
\[
 \rho_{\mathcal E}^{\pm}
 :=
 L_{\mathcal E}^{\pm}
 (\lambda_{\mathcal E}^{\pm})^{-1}
 (\varepsilon_{\mathcal E}^{\pm})^{-1}\lay_P .
 \tag{Register-placement}
\]

\begin{corollary}[Compiled register action]
\label{cor:compiled-register-action}
Assume \textup{(BC)}.  Then
\[
 (\rho_{\mathcal E}^{+})^{\dagger}
 \mathcal U_{Q_{\mathcal E}}(G_{\mathcal E})
 \rho_{\mathcal E}^{-}
 =
 \lay_P^{-1}U_t\lay_P .
\]
Moreover, for $\epsilon\in\{-,+\}$,
\[
 \rho_{\mathcal E}^{\epsilon}
 =
 L_{\mathcal E}^{\epsilon}p_{\mathcal E}^{\epsilon}
 \widehat g_{\mathcal E}^{\epsilon}\lay_P
 =
 \eta_{\mathcal E}^{\epsilon}\lay_P .
 \tag{Register-plan-placement}
\]
Thus both endpoint maps are zero-padded computational-basis
placements of the canonical $P$-words in the planned register: the
compiled program accepts the ordinary basis encoding and is read in
the ordinary basis encoding, with no arbitrary pre- or post-processing
unitary.
\end{corollary}
\begin{proof}
Lemma~\ref{lem:compilation-nf} gives
\[
 \lambda_{\mathcal E}^{+}
 \bigl(
  \Framed{\mathcal E}\!\upharpoonright_{C_{\mathcal E}^{-}}
 \bigr)
 (\lambda_{\mathcal E}^{-})^{-1}
 =
 \SEM{x:P\intjudge N:P}_{\mathrm{NF}} .
\]
Substitute \textup{(Register-placement)}; its right factor lands in
$C_{\mathcal E}^-$.  Equation~\textup{(Register-unitary)} then gives
the displayed equality.

For the final assertion, put
$\ell_{\mathcal E}^\epsilon
:=\lay_{\mathfrak J_{\mathcal E}}^\epsilon
\!\upharpoonright_{\widetilde C_{\mathcal E}^\epsilon}$.
Definition~\ref{def:layout-iso} gives
$(\lambda_{\mathcal E}^\epsilon)^{-1}
=p_{\mathcal E}^\epsilon(\ell_{\mathcal E}^\epsilon)^{-1}$.
The singleton-root clause of
Theorem~\ref{thm:nf-boundary-unitarity} and
\textup{(Register-restriction)} give
$(\varepsilon_{\mathcal E}^\epsilon)^{-1}
=g_{\mathcal E}^\epsilon$, while \textup{(Plan-chart)} gives
$\widehat g_{\mathcal E}^\epsilon
=(\ell_{\mathcal E}^\epsilon)^{-1}g_{\mathcal E}^\epsilon$.
Hence \textup{(Register-placement)} yields the first equality in
\textup{(Register-plan-placement)}.  Finally,
\textup{(Plan-logical-incidence)} gives
$L_{\mathcal E}^\epsilon p_{\mathcal E}^\epsilon
=\eta_{\mathcal E}^\epsilon
(\widehat g_{\mathcal E}^\epsilon)^\dagger$; composing with
$\widehat g_{\mathcal E}^\epsilon$ proves the second equality.
Now \textup{(Plan-root-word)} gives
\[
 \rho_{\mathcal E}^{\epsilon}
 =
 (\mathsf{Pad}_{\kappa_P}\zeta_P)\lay_P .
\]
By the definition of $\zeta_P$, this is precisely the zero-padded
computational-basis placement of the canonical $P$-codewords.  Thus
$(\rho_{\mathcal E}^+)^\dagger$ is inverse readback on that placement
image, not an additional total physical gate.
\end{proof}

%% file: carrier-planned-emission.tex

\subsection{Chart-Directed Carrier Planning and Emission}
\label{app:carrier-planned-emission}

The carrier plan is not extra semantic data and is not an input to the
compiler.  Retain the intermediate charts in the already fixed
semantic construction and recurse down that finite constructor
expression.  Planning and lowering are two passes.  During planning,
a parent fixes the computational-word images, logical incidences, and
inherited through coordinates of its premise faces from the retained
semantic charts, then recurses on those premises.  During lowering,
the complete plan is already available, so each premise is emitted
into its supplied incidences and the target constructor of the same
semantic clause is then applied.

\begin{definition}[Chart-directed carrier plan]
\label{def:chart-directed-carrier-plan}
Let $\mathcal R$ be a canonical call in the domain of the reference
construction, including a lower-$\Phi$ call generated by
$\mathsf{MapApply}$ or coherent sharing.  Evaluate the construction of
the complete graph chart of
Theorem~\ref{thm:nf-boundary-unitarity}, but retain its finite
constructor expression before multiplying its semantic maps.  Thus
the retained expression contains the premise and conclusion faces of
tensoring, completion, block formation, repartition, and structural
framing; the input, linked, and output faces of each classified cut;
and, for a higher-order classified cut, the vertices of the already
finite expression $\mathsf{Net}(\rho)$.  The construction of
$\mathsf{Net}(\rho)$, including every \textup{(C-var-whole)} unit
graft, has already terminated under \textup{(Eta-order)}.  Carrier
planning does not rerun that subsidiary recursion; it traverses the
resulting finite retained expression structurally.
Close this retained expression under every generated lower-$\Phi$ call
of $\mathsf{MapApply}$ and coherent sharing and every recursively prior
classified call of \textup{(C-str)}.  Their internal faces are
therefore faces of the root planning problem.

The recursion uses the order
\[
 \bigl(\Phi(N),\,|\mathcal D|,\,|\mathcal C|\bigr)
 \tag{Plan-order}
\]
lexicographically at a classified cut, with ordinary premise calls
ordered by proper-subderivation at fixed $\Phi$.  Generated
$\mathsf{MapApply}$ and coherent-sharing calls strictly decrease
$\Phi$ by \textup{(Use-$\Phi$)} and \textup{(CS-NF)}.  This is only
the well-founded traversal order of the retained semantic expression;
it carries no register-allocation data.

A retained face $s$ is read in its inherited context.  Its
\emph{contextual carrier} $\overline K_s^{\mathcal R}$ is the local
flat carrier of the semantic clause together with every factor carried
unchanged by the enclosing clauses: the other live premise of a
parallel node, inactive-completion, residual, and through factors of a
cut, and each enclosing semantic sum tag.  The local flat carrier is
the formal ordered tensor of compiler-layout factors displayed at
that retained constructor face, including its prescribed tag and
payload-padding coordinates.  It is determined before any injection
into $\mathcal H_{\mathcal R}$ is chosen and is not shorthand for the
whole judgment layout $V_{\mathfrak J}^\epsilon$ at every internal
action.  In the primitive row its local action factor is exactly the
$n_P$-wire carrier of \textup{(BC-carrier)}.  For a block, the selected
alternative occupies the common payload under that tag; the other
alternative is not a tensor factor.  Put
\[
 Q_{\mathcal R}:=\max_s\overline w_s^{\mathcal R},
 \qquad
 \mathcal H_{\mathcal R}
 :=(\mathbb C^2)^{\tensor Q_{\mathcal R}}.
 \tag{Plan-Q}
\]
Here $\overline w_s^{\mathcal R}$ denotes the number of qubit
coordinates in the displayed flat carrier
$\overline K_s^{\mathcal R}$.
This is a flat-carrier width, not
$\lceil\log_2\dim D_s\rceil$.  In particular, an internal face of one
parallel premise is counted together with the other premise's live
carrier, and a primitive $a:P\lmark P$ retains all
$n_P=\nwires(P)$ native wires of \textup{(BC-carrier)}.

The plan preserves the fixed word order within each flat factor and
assigns the factor an ordered injection into
$\mathcal H_{\mathcal R}$; unused coordinates at that face are zero
padding.  If a constructor reorders flat factors, the plan retains
that ordering as an address presentation.  If a selected-code
isomorphism
$\widehat\theta:C'\xrightarrow{\cong}C$ is only a chart
reparameterization, a placement
$p:C\lhook\joinrel\longrightarrow
(\mathbb C^2)^{\tensor q}$ in the current ambient carrier is replaced
by
\[
 p':=p\widehat\theta .
 \tag{Plan-view}
\]
Then
$\operatorname{im}(p')=\operatorname{im}(p)$, so this view allocates
no carrier and emits no target program.  No claim is made that
$\widehat\theta$ is a computational-basis permutation.

For a selected code $C_s\subseteq V_s$ visited in a
$q_s$-wire ambient carrier, write
\[
 J_{s,C}:=p_s:C_s\lhook\joinrel\longrightarrow
 (\mathbb C^2)^{\tensor q_s},
 \qquad
 J_{s,L}:C_s\lhook\joinrel\longrightarrow
 (\mathbb C^2)^{\tensor q_s}.
 \tag{Plan-stage}
\]
Write
$\widehat g_s:\mathsf D_s\xrightarrow{\cong}C_s$ for the restricted
flat pullback of the semantic chart at this face.  For a premise or
conclusion face this is the corresponding intermediate chart retained
from the recursive construction of
Theorem~\ref{thm:nf-boundary-unitarity}; a linked face or a vertex of
$\mathsf{Net}(\rho)$ inherits the chart at that row, with its displayed
tensor, direct-sum, or \textup{(Net-repart-pullback)} pullback.  Thus
no chart is postulated beyond that already finite construction.
The direct recursion chooses an ordered physical incidence
$\eta_s:\mathsf D_s\lhook\joinrel\longrightarrow
(\mathbb C^2)^{\tensor q_s}$ and defines
\[
 J_{s,L}:=\eta_s\widehat g_s^\dagger .
 \tag{Plan-logical-incidence}
\]
Concretely, $\eta_s$ sends each vector of the fixed basis of
$\mathsf D_s$ to its planned computational-basis register word and
puts $0$ on every unused coordinate.
Thus the semantic chart supplies the selected coordinate isometry and
its fixed order, while the plan supplies its physical address
presentation.  The partial isometry
$J_{s,L}J_{s,C}^\dagger$ on $\operatorname{im}(J_{s,C})$ has the fixed
INV-4 extension $L_s$, so $L_sp_s=J_{s,L}$.
Only changes between computational-basis word presentations are
candidates for the exact basis-permutation supplier.  Literal
wire-address permutations are the special case that can be absorbed
symbolically.  General selected-chart views are handled by
\textup{(Plan-view)}.

More precisely, a recursive planning call carries
$(\sigma,\iota,q)$, where $\sigma$ is the word already fixed on the
enclosing semantic tag coordinates and
$\iota:[q]\hookrightarrow[Q_{\mathcal R}]$ is the ordered injection of
the remaining active carrier.  For every descendant face $s$ it
maintains
\[
 \overline w_s^{\mathcal R}-|\sigma|\leq q .
 \tag{Plan-context-bound}
\]
At the root this triple is
$(\varnothing,\id,Q_{\mathcal R})$.  For a singleton first-order root
$\mathcal E:(x{:}P\intjudge N:P)$, the two root chart domains are the
canonical $\mathsf D_P$ domains of \textup{(QGraph-read)}.  Let
$\kappa_P:[\nwires(P)]\hookrightarrow[Q_{\mathcal E}]$ be the initial
ordered injection.  The root rule is applied before any descendant
rule and fixes both endpoint word incidences by
\[
 \eta_{\mathcal E}^-=\eta_{\mathcal E}^+
 :=
 \mathsf{Pad}_{\kappa_P}\zeta_P .
 \tag{Plan-root-word}
\]
Later first-fit assignments do not alter these root presentations.
If a block tag has local coordinate $d$, let
$c_d:[q-1]\hookrightarrow[q]$ enumerate its ordered complement;
branch $i$ is planned with
$(\sigma\cup\{\iota(d)\mapsto\tau_i\},\iota c_d,q-1)$.  All other clauses retain $q$ and
transport $\iota$ along their displayed factor incidences.  The global
presentation of a local face is recovered by inserting the fixed word
$\sigma$ outside $\iota$.

For a completed binary block $\mathcal B$ visited in a
$q_{\mathcal B}$-wire ambient carrier, retain
$P_{\mathcal B}^{\epsilon}$ and $k_i^\epsilon$ from
\textup{(Block-pack)}, together with its canonical unviewed local
incidences $F_{\mathcal B,\bullet}^{\epsilon}$,
$\bullet\in\{C,L\}$.  These are the incidences before any enclosing
selected-coordinate view is composed with the block.  Let
$\Pi_{\mathcal B}$ be the ordered family, in fixed
context-and-result order, of maximal complete external boundary
occurrences $\pi:T$ of the parent block judgment.  Completion followed
by the semantic component inclusion induces
\[
 \beta_i^\epsilon:
 \{\text{completed branch-$i$ occurrences}\}
 \rightharpoonup \Pi_{\mathcal B}.
 \tag{Plan-Block-occ}
\]
A surviving context occurrence maps to the same parent variable
occurrence, the component root maps to the parent root through the
fixed type-layout component inclusion, and a branch-local occurrence
is outside the domain.  Thus two branch restrictions of one shared
occurrence have one common image in $\Pi_{\mathcal B}$.
Let
$d_{\mathcal B}$ be the physical
coordinate occupied by the block's already present semantic
direct-sum tag under \eqref{eq:sum-layout}, and put
$\tau_1=0,\tau_2=1$.  The plan assigns the source and target semantic
tags to this same physical coordinate.  It is one tag wire carried
through the block, not two simultaneously live coordinates; exterior
endpoint presentations may still use the separate adapters $\pi_-$
and $\pi_+$.  Let $T_{\mathcal B}$ be the ordered
plan-coordinate space carried unchanged through this occurrence and
define
\[
\begin{aligned}
 \overline C_i^\epsilon
   &:=\widehat C_i^\epsilon\tensor T_{\mathcal B},\\
 \delta_{\mathcal B}^\epsilon
   &: \bigoplus_i\overline C_i^\epsilon
      \xrightarrow{\;\cong\;}
      \left(\bigoplus_i\widehat C_i^\epsilon\right)
      \tensor T_{\mathcal B},\\
 \overline P_{\mathcal B}^\epsilon
   &:=(P_{\mathcal B}^\epsilon\tensor I_{T_{\mathcal B}})
      \delta_{\mathcal B}^\epsilon .
\end{aligned}
\tag{Plan-Block-context}
\]
Let $\overline k_i^\epsilon$ be the canonical inclusion of
$\overline C_i^\epsilon$ into the displayed direct sum.  Thus all
local block data are tensored with the identity on the inherited
through coordinates, up to the fixed distributor
$\delta_{\mathcal B}^\epsilon$.  Using the insertion $\jmath_{d,b}$ of
\textup{(Target-control)}, the component incidences are required to
have the form
\[
 F_{\mathcal B,\bullet}^{\epsilon}
 \overline P_{\mathcal B}^{\epsilon}\overline k_i^\epsilon
 =
 \jmath_{d_{\mathcal B},\tau_i}
 Z_{\mathcal B,i,\bullet}^{\epsilon},
 \qquad
 \bullet\in\{C,L\}.
 \tag{Plan-Block-face}
\]
Here $Z_{\mathcal B,i,\bullet}^{\epsilon}$ is exactly the payload
presentation supplied to the recursively emitted completed branch,
padded inside the common non-tag carrier.  It is not a separately
sealed branch presentation.  For $\bullet=C$,
Equation~\eqref{eq:sum-layout} and \textup{(Block-pack)} supply the
factorization: the standalone-map inclusions and the maps
$j_i^\epsilon$ of \textup{(Use-part)} and \textup{(CS-part)} are the
canonical inclusions of the same two semantic summands.  For $\bullet=L$, the parent first assigns the common non-tag
addresses.  Using the already constructed completed-branch charts and
\textup{(Plan-logical-incidence)}, it fixes
$Z_{\mathcal B,i,L}^{\epsilon}$ in those addresses, with every
$\beta_i^\epsilon$-shared occurrence at its prescribed parent address
and with the inherited through presentation identical in both
components.  Since the
$\overline P_{\mathcal B}^{\epsilon}\overline k_i^\epsilon$
are orthogonal and jointly exhaustive, the two equations
\textup{(Plan-Block-face)} uniquely define
$F_{\mathcal B,L}^{\epsilon}$.  Their register images are orthogonal
because they have distinct values at $d_{\mathcal B}$, so the selected
incidence is isometric and INV-4 extends its frame map.  The two
$Z_{\mathcal B,i,L}^{\epsilon}$ are then supplied to the recursive
branch calls.  Thus INV-4 extends an already fixed parent incidence;
it does not determine branch endpoints bottom-up.  The occurrence family
$\Pi_{\mathcal B}$ and the maps $\beta_i^\epsilon$ specify which
component occurrences are restrictions of the same parent
occurrence.  Those restrictions use the parent flat addresses already
present in $Z_{\mathcal B,i,\bullet}^{\epsilon}$; neither
$\Pi_{\mathcal B}$ nor $\beta_i^\epsilon$ allocates a wire.
Equation~\textup{(Plan-Block-face)} is imposed on the canonical
unviewed incidence only.  An incidence requested by an enclosing
constructor has the plan-compatible form
\[
 E_{\mathcal B}^\epsilon
 =P_{\pi_\epsilon}F_{\mathcal B,L}^\epsilon\theta_\epsilon .
 \tag{Plan-Block-parent}
\]
The incidences
$F_{\mathcal B,L}^\epsilon$ and
$E_{\mathcal B}^\epsilon\theta_\epsilon^\dagger$ are
computational-word incidences indexed by the same ordered block basis.
They agree on the complete through word in $T_{\mathcal B}$, while
their branch-dependent payload correspondence is part of the partial
word bijection.  The fibrewise extension in
Lemma~\ref{lem:planned-endpoint-transport} therefore proves that the
permutation $\pi_\epsilon$ in \textup{(Plan-Block-parent)} exists.
The block is first emitted at $F_{\mathcal B,L}^\epsilon$ and the
whole resulting equation is then transported to
$E_{\mathcal B}^\epsilon$.  Thus no hypothesis that the possibly
dense view $\theta_\epsilon$ preserve individual tag sectors is
made.

For a classified first-order cut $\rho$ visited in a
$q_\rho$-wire ambient carrier, let $\mathsf M_\rho$ be the ordered
basis-label set of
\[
 Y_\rho\tensor\mathsf D_P\tensor Z_\rho
 \tensor\mathsf D_{\Xi_\rho}.
 \tag{Plan-midpoint-domain}
\]
Here $\mathsf D_{\Xi_\rho}$ is the fixed source-coordinate domain of
the inactive through-factor.  The factors and their ordered bases are
those of the Source chart, in source-reference and fixed-DNF order
following \textup{(Chart)}.  A producer family uses the tagged
disjoint union of its component label sets.  At a Source block's own
sum port, $Z_\rho$ is the common based space $Z_K$ of
Lemma~\ref{lem:shared-source-chart}, used in every branch.  This is the
single intermediate coordinate family contracted by
\textup{(Graph-link)}.  Choose its single plan presentation
\[
\mu_\rho:
\mathbb C[\mathsf M_\rho]
 \lhook\joinrel\longrightarrow
 (\mathbb C^2)^{\tensor q_\rho}.
 \tag{Plan-midpoint}
\]
The producer-positive and consumer-negative incidences are this same
map.  In particular, the canonical $\mathsf D_P$ coordinate occurs
once as a factor of the abstract domain
$\mathbb C[\mathsf M_\rho]$.  The physical incidence $\mu_\rho$ may
contain a selected-coordinate view and therefore need not expose that
factor as a computational-basis subregister.  There is no second midpoint and no
splice-level alignment presentation.
For a producer family, $\mathsf M_\rho$ is the orthogonal tagged union
of the component midpoint families and the already present semantic
family tag occurs once in $\mu_\rho$.
At a cut the direct recursion chooses one ordered incidence
$\eta_\rho$ for the canonical $\mathsf M_\rho$ domain and supplies it
to both children.  Thus, if $F_{\rm pr}^+$ and $F_{\rm co}^-$ are
their natural unviewed framed endpoint incidences, the shared
$\mathsf M_\rho$ basis indexes both endpoints and $\eta_\rho$, and
their inherited through subwords agree.  The fibrewise
through-preserving basis-word adaptation of
Lemma~\ref{lem:planned-endpoint-transport} therefore gives the common
unviewed incidence
$\mu_\rho^0:=\eta_\rho$, and any selected-coordinate view is then
shared by both sides:
\[
 \mu_\rho^0
 =P_{\pi_{\rm pr}}F_{\rm pr}^+
 =P_{\pi_{\rm co}}F_{\rm co}^- ,
 \qquad
 \mu_\rho=\mu_\rho^0\theta_\rho
 :
 \mathbb C[\mathsf M_\rho]
 \lhook\joinrel\longrightarrow
 (\mathbb C^2)^{\tensor q_\rho}.
 \tag{Plan-midpoint-compat}
\]
Here
$\theta_\rho:\mathbb C[\mathsf M_\rho]
\xrightarrow{\cong}\mathbb C[\mathsf M_\rho]$.
Thus the view at the join is one coordinate change on the same
canonical $\mathsf M_\rho$-domain, not two independently chosen
parameterizations.

These data are assigned by the following direct recursion.  Begin in
the $Q_{\mathcal R}$-wire root carrier and assign the root face before
its descendants, as in \textup{(Plan-root-word)} for a singleton
first-order root.  At a premise retain the addresses of all through
factors and place each newly live flat factor, in the clause's fixed
order, in the first unused ordered coordinates.  Once those ranges
are assigned, $\eta_s$ is not an additional choice: it writes the
fixed ordered factor labels in those ranges, the inherited word
$\sigma$ on frozen tag coordinates, and $0$ elsewhere.

Visit multi-premise nodes in displayed order: the first then second
$\mathsf{Par}$ premise, block branches $1$ then $2$, and at a cut the
producer then the consumer.  For $\mathsf{Par}$ carry the other
premise as a spectator; for $\mathsf{Complete}$ carry the inactive
factor.  For a block recurse on each tagged restriction in the ordered
$(q_{\mathcal B}-1)$-wire complement of its one through tag.  At a
cut, assign $\eta_\rho$ directly by the first-fit rule in the fixed
factor order of $\mathsf M_\rho$, before visiting either child, and
supply it to both; it is not copied from either child's natural
presentation.

Repartition and frame clauses retain the physical image and change
only its selected-coordinate view.  A coordinate may be reused only
after an outgoing materialized face at which it belongs to no live
flat factor.  At that face the planned placement puts the coordinate
in $\ket0$ padding; this is the general padding clause of
\textup{(Plan-stage)} and \textup{(Pad-Lift)}, with INV-2 for a sum
payload and \textup{(BC-carrier)} for a primitive's native carrier.
In particular, a branch-local temporary may occupy currently inactive
payload padding of its enclosing sum, but its selected realization
equation restores that padding to $\ket0$ before the branch face is
returned.  This is a property of the displayed output layout and is
therefore available during plan construction; no execution claim is
used.  No coordinate carrying a live factor is reused.
Bound~\textup{(Plan-context-bound)} says that at every call the
inherited and newly live active coordinates fit in the current
$q$-wire carrier, so the indicated extension exists.  The fixed
source-reference, DNF, constructor, left-to-right premise, branch,
basis-word, and complement orders make the recursion unique.

The collection of the stage presentations, primitive ranges,
block-tag incidences, first-order midpoints, and the unchanged-image
views is
\[
 \mathfrak P_{\mathcal R}:=
 \mathsf{CarrierPlan}(\mathcal R).
 \tag{CarrierPlan}
\]
It is the presentation component returned by this recursion before
target lowering begins.
\end{definition}

\begin{lemma}[Two-sided basis-word transport to a planned incidence]
\label{lem:planned-endpoint-transport}
Suppose
\[
 F^\epsilon:C^\epsilon
 \lhook\joinrel\longrightarrow
 (\mathbb C^2)^{\tensor q},
 \qquad
 V:C^-\xrightarrow{\;\cong\;}C^+,
 \qquad
 \mathcal U_q(G)F^-=F^+V .
 \tag{Plan-natural}
\]
Let $R_{\rm thru}\subseteq[q]$ be the coordinates assigned by the parent
to contextual tensor factors designated as carried unchanged through this
occurrence, and write
$x|_{R_{\rm thru}}$ for the corresponding subword of
$x\in\{0,1\}^q$.  For each $\epsilon$, let
\[
 E^\epsilon:D^\epsilon
 \lhook\joinrel\longrightarrow(\mathbb C^2)^{\tensor q},
 \qquad
 \theta_\epsilon:D^\epsilon\xrightarrow{\cong}C^\epsilon
\]
be the parent-assigned incidence and its selected-coordinate view.
Relative to the fixed ordered basis $\mathscr B^\epsilon$ of
$C^\epsilon$, suppose the natural and assigned unviewed incidences
are computational-word incidences
\[
 F^\epsilon\ket a=\ket{f_\epsilon(a)},
 \qquad
 E^\epsilon\theta_\epsilon^\dagger\ket a
 =\ket{e_\epsilon(a)}
 \quad(a\in\mathscr B^\epsilon),
 \tag{Plan-word-data}
\]
where $f_\epsilon,e_\epsilon$ are injective and
\[
 f_\epsilon(a)|_{R_{\rm thru}}
 =e_\epsilon(a)|_{R_{\rm thru}}
 \quad(a\in\mathscr B^\epsilon).
 \tag{Plan-through}
\]
Then there is a permutation
$\pi_\epsilon\in\operatorname{Sym}(\{0,1\}^q)$ which preserves the
complete through subword and satisfies
\[
 E^\epsilon=P_{\pi_\epsilon}F^\epsilon\theta_\epsilon,
 \tag{Plan-admissible}
\]
that is,
$\pi_\epsilon(x)|_{R_{\rm thru}}=x|_{R_{\rm thru}}$ for every word
$x$.  Such an incidence pair is called \emph{plan-compatible} with
\textup{(Plan-natural)}.  Put
\[
 G^{\mathfrak P}
 :=
 \mathcal B_q(\pi_-^{-1});G;\mathcal B_q(\pi_+).
 \tag{Plan-adapt}
\]
Then, without changing $q$,
\[
 \mathcal U_q(G^{\mathfrak P})E^-
 =
 E^+\bigl(\theta_+^\dagger V\theta_-\bigr).
 \tag{Plan-adapt-law}
\]
\end{lemma}
\begin{proof}
Put
\[
 X_\epsilon=f_\epsilon(\mathscr B^\epsilon),
 \qquad
 Y_\epsilon=e_\epsilon(\mathscr B^\epsilon),
\]
and define the partial word bijection
$b_\epsilon(f_\epsilon(a)):=e_\epsilon(a)$.
It preserves the $R_{\rm thru}$-subword by
\textup{(Plan-through)}.  For each
$u\in\{0,1\}^{R_{\rm thru}}$, let
\[
 \Omega_u:=\{x\in\{0,1\}^q:x|_{R_{\rm thru}}=u\}.
\]
The restriction of $b_\epsilon$ is a bijection
$X_\epsilon\cap\Omega_u\to Y_\epsilon\cap\Omega_u$.
The two complements in $\Omega_u$ therefore have the same
cardinality; match them in lexicographic order.  Together with
$b_\epsilon$, these fibrewise complement bijections form a total
permutation $\pi_\epsilon$ which preserves every through subword.
Equation~\textup{(Plan-word-data)} gives
$E^\epsilon\theta_\epsilon^\dagger
=P_{\pi_\epsilon}F^\epsilon$, hence
\textup{(Plan-admissible)}.

Finally, by \textup{(Target-seq)} and \textup{(BP-int)},
\[
\begin{aligned}
 \mathcal U_q(G^{\mathfrak P})E^-
 &=
 P_{\pi_+}\mathcal U_q(G)P_{\pi_-}^\dagger
 P_{\pi_-}F^-\theta_-\\
 &=P_{\pi_+}F^+V\theta_-
 =E^+\theta_+^\dagger V\theta_- .
\end{aligned}
\]
For a pure view take $\pi_-=\pi_+=I$ and use $G$ itself, so no adapter
is emitted.
Every adapter preserves the through word.  If it is induced by a
literal permutation of wire addresses, a backend may absorb it into
its symbolic addresses.  A genuine word permutation---including the
re-encoding required by additive associativity and by an
unequal-width right distributor---is lowered exactly through
$\mathcal B_q$.  No dense selected-coordinate view is synthesized.
\end{proof}

\begin{lemma}[Contextual lifting]
\label{lem:planned-context-lift}
Suppose
$G_0\in\mathsf{TCirc}_{q_0}$ and
$\mathcal U_{q_0}(G_0)F_0^-=F_0^+V$.  Let
$S:Z\lhook\joinrel\longrightarrow
(\mathbb C^2)^{\tensor(q-q_0)}$ be the common presentation of all
factors carried through this action, including zero padding, and let
$\varrho:[q_0]\hookrightarrow[q]$ be the planned active-wire
injection.  With
\[
 F_{\rm ctx}^\epsilon
 :=P_\varrho(F_0^\epsilon\tensor S),
\]
\textup{(Target-lift)} gives
\[
 \mathcal U_q(\mathsf{TLift}_\varrho(G_0))F_{\rm ctx}^-
 =
 F_{\rm ctx}^+(V\tensor I_Z).
 \tag{Plan-context-lift}
\]
Consequently Lemma~\ref{lem:planned-endpoint-transport} transports the
lifted natural equation to either parent-assigned incidence without
adding a wire.
\end{lemma}
\begin{proof}
Using \textup{(Pad-Lift)} and \textup{(Target-lift)}, both sides are
\[
 P_\varrho\bigl(F_0^+V\tensor S\bigr).
\]
\end{proof}

\begin{theorem}[Direct chart-directed planning and emission]
\label{thm:chart-directed-plan-emit}
Assume \textup{(BC)}.  For every canonical call $\mathcal R$ reached
by the Source reference
construction, the recursion first constructs
$\mathfrak P_{\mathcal R}$ and then lowers the same retained
expression.  More precisely, suppose an action occurrence
$\mathcal A$ is visited in the $q_{\mathcal A}$-wire ambient carrier
assigned by its parent, with explicit plan-coordinate domains and
framed incidences
\[
 E_{\mathcal A}^\epsilon:
 C_{\mathcal A,\mathfrak P}^\epsilon
 \lhook\joinrel\longrightarrow
 (\mathbb C^2)^{\tensor q_{\mathcal A}} .
\]
The lowering recursion constructs
$G_{\mathcal A}^{\mathfrak P}\in
\mathsf{TCirc}_{q_{\mathcal A}}$ such that
\[
 \mathcal U_{q_{\mathcal A}}(G_{\mathcal A}^{\mathfrak P})
 E_{\mathcal A}^-
 =
 E_{\mathcal A}^+V_{\mathcal A}^{\mathfrak P}.
\tag{PlanEmit-IH}
\]
Here $V_{\mathcal A}^{\mathfrak P}$ is the semantic action expressed
in the plan coordinates; under \textup{(Plan-admissible)} it is
$\theta_+^\dagger V_{\mathcal A}\theta_-$.  All factors designated
through by the parent occur as identity factors in this equation.
The presentations in this statement are the restrictions of the root
plan; they need not be the presentations obtained by compiling the
same underlying derivation as a new root.  At the root the construction
returns the six-field artifact
\[
 \mathsf{PlanEmit}(\mathcal R)
 =
 (Q_{\mathcal R},G_{\mathcal R},
  L_{\mathcal R}^-,L_{\mathcal R}^+,
  p_{\mathcal R}^-,p_{\mathcal R}^+)
 \tag{PlanEmit}
\]
in the single register $\mathcal H_{\mathcal R}$.
If
\[
 \widehat g_{\mathcal R}^\epsilon
 :=
 \left(
  \lay_{\mathfrak J_{\mathcal R}}^\epsilon
  \!\upharpoonright_{\widetilde C_{\mathcal R}^\epsilon}
 \right)^{-1}
 g_{\mathcal R}^\epsilon
 \tag{Plan-chart}
\]
is the physical pullback of its complete chart, then
\[
 \mathcal U_{Q_{\mathcal R}}(G_{\mathcal R})
 L_{\mathcal R}^-p_{\mathcal R}^-
 \widehat g_{\mathcal R}^-
 =
 L_{\mathcal R}^+p_{\mathcal R}^+
 \widehat g_{\mathcal R}^+V_{\mathcal R}.
 \tag{PlanEmit-chart}
\]
Moreover:
\begin{enumerate}
\item a block designates an existing semantic sum-tag coordinate and
      adds neither a derivation tag nor a parent-carrier bank;
\item its completed branches are emitted directly into the payload
      incidences of \textup{(Plan-Block-face)};
\item a first-order splice uses the single midpoint
      \textup{(Plan-midpoint)} and introduces no independent
      splice-level $\mathsf{Align}$; any required word adapter is
      already part of an incident endpoint transport; and
\item every node of a higher-order $\mathsf{Net}(\rho)$ either
      preserves the plan image, places genuinely simultaneous factors
      on disjoint planned ranges, or invokes the same planned splice
      recursion.
\end{enumerate}
\end{theorem}

\begin{proof}
Use the well-founded recursion of the semantic construction: outer
strong induction on $\Phi$, structural recursion
on the current canonical derivation at fixed $\Phi$, and the
producer--consumer order
$(|\mathcal D|,|\mathcal C|)$ of \textup{(Plan-order)} inside a
classified cut.  Generated applied-map and coherent-sharing calls are
prior by \textup{(Use-$\Phi$)} and \textup{(CS-NF)}.  The retained
chart expression is finite by
Theorem~\ref{thm:nf-boundary-unitarity} and
Lemma~\ref{lem:classifier-totality}, so \textup{(Plan-Q)} is defined.
This first proves that the direct plan recursion is total.  At a
recursive call the enclosing clauses have already assigned the
through-factor addresses.  Bound~\textup{(Plan-context-bound)} puts
the remaining factors in the current active $q$-wire carrier, so their
fixed ordered injection into the complement exists.  At a block the
$\beta_i^\epsilon$-images retain their inherited parent addresses and
only branch-local factors use the remaining payload coordinates.
Freezing its tag changes $(\sigma,\iota,q)$ exactly as displayed above,
so the same bound types nested $(q-1)$-wire branch calls.  Keeping
through addresses and using the fixed order for the complement proves
the claim recursively.
The block and cut clauses make the two deliberate identifications
shown in \textup{(Plan-Block-face)} and \textup{(Plan-midpoint)};
the former is the tagged orthogonal union chosen for $\eta_s$, and
the latter is the single $\eta_\rho$ supplied to both children.
Consequently both displayed equalities are definitional consequences
of the direct recursion, not comparisons of independently constructed
physical frames.
At every unviewed retained face, the natural incidence and the
parent-assigned incidence are computational-word incidences indexed
by the same fixed semantic basis.  The recursion retains the complete
through word, so \textup{(Plan-word-data)} and
\textup{(Plan-through)} hold; the fibrewise extension in
Lemma~\ref{lem:planned-endpoint-transport} therefore proves
\textup{(Plan-admissible)}.  Disjoint tensor clauses are its
wire-address special case, while repartition and frame clauses add
only unchanged-image views.  Restricted flat padding is isometric, the
selected frame maps are isometric by the complete-chart construction,
and INV-4 supplies their total unitary extensions.  This completes the
planning pass without emitting a target program.
The reuse entries of the plan occur only at the zero-padded outgoing
faces specified above.  During lowering, the induction hypothesis for
the preceding action maps its selected input incidence to precisely
that outgoing incidence, so a later first-fit assignment begins with
the required zero coordinate.

Lower the retained expression in the same recursion.  The natural
realization calculations for the unchanged rows are precisely the
corresponding rows of Lemma~\ref{lem:plan-local-laws}; only their
placement in the parent context must be checked.  At
$\mathsf{Wire}$, $\mathsf{Struct}$, and $\mathsf{Prim}$, lift the
natural target program along its planned ordered wire injection and
apply Lemma~\ref{lem:planned-context-lift}, followed by
Lemma~\ref{lem:planned-endpoint-transport}.  The primitive lift
uses exactly the $n_P$ native action wires of \textup{(BC-carrier)}
and is the identity on the carried factors.  The same transport lemma
applies to the already proved equations for $\mathsf{Repart}$,
$\mathsf{StrAct}$, $\mathsf{StrPre}$, $\mathsf{TenPack}$, and
$\mathsf{Frame}$.  The direct plan keeps the address presentation
fixed across a pure view, so these rows introduce no adapter and no
wire.

For $\mathsf{Par}$, let $E_{\parallel}^-$,
$E_{\mathrm m}$, and $E_{\parallel}^+$ be its planned input,
between-premises, and output incidences in the same $q$-wire contextual
carrier.  Apply the recursive hypothesis first with the second
premise carried through and then with the first premise carried
through:
\[
\begin{aligned}
 \mathcal U_q(G_1^{\mathfrak P})E_{\parallel}^-
   &=E_{\mathrm m}(V_1^{\mathfrak P}\tensor I),\\
 \mathcal U_q(G_2^{\mathfrak P})E_{\mathrm m}
   &=E_{\parallel}^+(I\tensor V_2^{\mathfrak P}).
\end{aligned}
\tag{Plan-Par-IH}
\]
Thus \textup{(Target-seq)} and interchange give
\[
 \mathcal U_q(G_1^{\mathfrak P};G_2^{\mathfrak P})
 E_{\parallel}^-
 =
 E_{\parallel}^+
 (V_1^{\mathfrak P}\tensor V_2^{\mathfrak P}).
\tag{Plan-Par}
\]
Only the two logical factor families are simultaneous; temporary
workspace returned to the planned between-incidence may be reused by
the second call.  This is the parallel-composition semantic
calculation in a fixed contextual register.  Completion applies
Lemma~\ref{lem:planned-context-lift} to the planned zero-\allowbreak padded lift of the common spectator incidence $\zeta_{\Xi,0}$.
Equation~\textup{(Through-code)} identifies its identity target action
with the semantic through-map $Y_\Xi$; the unchanged-image pullback
\textup{(Complete-pullback)} then gives
\textup{(Complete-selected)} and adds only the inactive factor already
counted in \textup{(Plan-Q)}.

It remains to give the two calculations whose old physical proofs
reconciled independently sealed premise presentations.

\emph{Block.}
Write $q:=q_{\mathcal B}$ for the ambient width of this occurrence.
Let $\widehat V_i:\widehat C_i^-\to\widehat C_i^+$ be the completed
branch actions of \textup{(Block-code)}, and put
$(\gamma_1,\gamma_2)=(\alpha,\beta)$.  The semantic clause first
supplies \textup{(Block-pack)}.  Put
\[
 \overline V_i:=\widehat V_i\tensor I_{T_{\mathcal B}} .
 \tag{Plan-Block-through}
\]
Pulling its two component inclusions through
\eqref{eq:sum-layout} gives the $C$ instance of
\textup{(Plan-Block-face)}; the orthogonal tagged frame construction
in the plan gives its $L$ instance.  This fixes the payload endpoints
before either recursive branch call is emitted.  In each of
$\mathsf{Block}^{\rm map}$, $\mathsf{Block}^{\rm use}$, and
$\mathsf{Block}^{\rm cs}$, component $i$ at the source is mapped to
component $i$ at the target, so the same tag value $\tau_i$ occurs at
both polarities.  Additive symmetry is the separate structural row
whose target action is $X$.

Remove coordinate $d_{\mathcal B}$ and identify its ordered complement
with $[q-1]$.  The direct recursive hypotheses are then well typed:
they return
$\widehat G_i\in\mathsf{TCirc}_{q-1}$ with
\[
 \mathcal U_{q-1}(\widehat G_i)
Z_{\mathcal B,i,L}^-
=
Z_{\mathcal B,i,L}^+\overline V_i .
 \tag{Plan-branch-IH}
\]
Define the canonical controlled core
\[
 \begin{aligned}
 G_{\mathcal B}^{0}:={}&
 \mathsf{Ctrl}_{d_{\mathcal B}=0}
  \left(
   \widehat G_1;
   \mathsf{Phase}_{q-1}(\alpha)
  \right);\\
 &\mathsf{Ctrl}_{d_{\mathcal B}=1}
  \left(
   \widehat G_2;
   \mathsf{Phase}_{q-1}(\beta)
  \right).
 \end{aligned}
 \tag{Plan-Block-prog}
\]
The controls act on the common non-tag carrier and leave
$d_{\mathcal B}$ unchanged.  By \textup{(Target-control)},
\textup{(Target-seq)}, and \textup{(Plan-branch-IH)}, for $i=1,2$,
\[
 \begin{aligned}
 &\mathcal U_q(G_{\mathcal B}^{0})
F_{\mathcal B,L}^-
\overline P_{\mathcal B}^-\overline k_i^-\\
 &\qquad =
\jmath_{d_{\mathcal B},\tau_i}
Z_{\mathcal B,i,L}^+
\gamma_i\overline V_i\\
 &\qquad =
F_{\mathcal B,L}^+
\overline P_{\mathcal B}^+\overline k_i^+
\gamma_i\overline V_i .
 \end{aligned}
 \tag{Plan-Block-component}
\]
The two component families are orthogonal and jointly exhaustive by
\textup{(Block-pack)}.  Therefore
\[
\begin{aligned}
&\mathcal U_q(G_{\mathcal B}^{0})
F_{\mathcal B,L}^-\overline P_{\mathcal B}^-\\
 &\qquad =
F_{\mathcal B,L}^+\overline P_{\mathcal B}^+
\left(
 \sum_i
 \overline k_i^+\gamma_i\overline V_i
 (\overline k_i^-)^\dagger
\right).
\end{aligned}
\tag{Plan-Block-sum}
\]
Indeed, precomposing both sides with each $\overline k_i^-$ gives
\textup{(Plan-Block-component)}, and the $\overline k_i^-$ exhaust the
domain.  By \textup{(Block-pack)}, \textup{(Block-code)}, and
\textup{(Plan-Block-context)},
\[
\begin{aligned}
 &\overline P_{\mathcal B}^+
 \left(
  \sum_i\overline k_i^+\gamma_i\overline V_i
  (\overline k_i^-)^\dagger
 \right)
 (\overline P_{\mathcal B}^-)^\dagger\\
 &\quad =
 \left[
  (\ell_{\mathcal B}^+)^{-1}
  \left(\sum_i
   e_i^+\gamma_i\widehat U_i(e_i^-)^\dagger\right)
  \ell_{\mathcal B}^-
 \right]\tensor I_{T_{\mathcal B}} .
\end{aligned}
\tag{Plan-Block-semantic}
\]
The bracketed operator is exactly the semantic block clause in flat
parent coordinates; call it $V_{\mathcal B}^0$ after tensoring with
$I_{T_{\mathcal B}}$.  Right-multiplying
\textup{(Plan-Block-sum)} by
$(\overline P_{\mathcal B}^-)^\dagger$ therefore gives the canonical
equation
\[
 \mathcal U_q(G_{\mathcal B}^{0})F_{\mathcal B,L}^-
 =F_{\mathcal B,L}^+V_{\mathcal B}^0 .
 \tag{Plan-Block-natural}
\]
For the parent incidence \textup{(Plan-Block-parent)}, apply
Lemma~\ref{lem:planned-endpoint-transport}.  It emits only the exact
basis-word adapters and gives
\[
 \mathcal U_q(G_{\mathcal B}^{\mathfrak P})E_{\mathcal B}^-
 =E_{\mathcal B}^+
   (\theta_+^\dagger V_{\mathcal B}^0\theta_-).
 \tag{Plan-Block-transport}
\]
This is \textup{(PlanEmit-IH)} for all three block constructors.  In
particular, a dense inherited view is never synthesized and need not
be tag diagonal: the controlled core acts linearly on its image before
the whole equation is reparameterized.

No relocation occurs in this calculation.  The child endpoints in
\textup{(Plan-branch-IH)} are already the payload restrictions of the
parent endpoints.  Thus the block uses neither a parent-carrier bank
nor a relocation circuit.  Its control coordinate is the semantic
tag already counted by the relevant flat stage.  If a block is nested,
the outer tag remains physically present but is frozen outside the
recursive target body; the inner semantic tag is an ordinary
coordinate of that $(q-1)$-wire payload.  Thus no derivation-tag
coordinate is added outside the recursive sum layout.

\emph{First-order splice.}
For a classified first-order instance $\rho$, put $q:=q_\rho$.  The semantic clause
\textup{(Graph-link)} supplies the one ordered coordinate family
$\mathsf M_\rho$.  The plan assigns it once as $\mu_\rho$ and asks the
producer and consumer recursive calls to end and begin, respectively,
at this same presentation in the same $q$-wire carrier.  Let $D_0$
and $C_0$ denote the natural producer and consumer actions in the
canonical $\mathsf M_\rho$ factorization, including the displayed
residual and inactive identity factors and with the fixed associator
$\sigma_\rho$ absorbed into $D_0$.  Let $\theta_{\rm in}$ and
$\theta_{\rm out}$ be the exterior endpoint views and let
$\theta_\rho$ be the single midpoint view of
\textup{(Plan-midpoint-compat)}.  A selected-code view is absorbed by
\textup{(Plan-adapt-law)}, not synthesized as a target program.  The
two induction hypotheses therefore have the form
\[
\begin{aligned}
 \mathcal U_q(G_{\rm pr}^{\mathfrak P})E_{\rm in}
 &=
 \mu_\rho
 (\theta_\rho^\dagger D_0\theta_{\rm in}),\\
 \mathcal U_q(G_{\rm co}^{\mathfrak P})\mu_\rho
 &=
 E_{\rm out}
 (\theta_{\rm out}^\dagger C_0\theta_\rho).
 \end{aligned}
 \tag{Plan-Splice-midpoint}
\]
Therefore the planned splice program is simply
\[
 G_{\rho}^{\mathfrak P}
 :=
G_{\rm pr}^{\mathfrak P};
G_{\rm co}^{\mathfrak P}.
\tag{Plan-Splice-prog}
\]
Equation~\textup{(Plan-midpoint-compat)} identifies both incident
coordinate domains with the same canonical
$\mathbb C[\mathsf M_\rho]$; hence the map inserted at their join
is the identity.  Using \textup{(Target-seq)}, the composite coordinate
action is
\[
 \theta_{\rm out}^\dagger C_0\theta_\rho
 \theta_\rho^\dagger D_0\theta_{\rm in}
 =\theta_{\rm out}^\dagger C_0D_0\theta_{\rm in}.
 \tag{Plan-Splice-cancel}
\]
In the canonical abstract $\mathsf D_P$ basis the coefficient of
$C_0D_0$ on $(y,r)$ from $x,z$ is
$\sum_p C_{r,(p,z)}D_{(y,p),x}$, exactly
\textup{(Graph-link)} and \textup{(FO-whole)}.  No claim is made that
the physical map $\mu_\rho$ exposes this basis as a word subregister.
The tagged producer family is the componentwise calculation of
\textup{(FO-family-cut)}.  Its orthogonal, jointly exhaustive parent
assembly is \textup{(FO-family-close)}, using \textup{(FO-part)}.
Since both incidences are literally $\mu_\rho$, there is no independent
$\mathsf{Align}_\rho$ and no splice-owned register enlargement.
Any exact basis-word permutation between a child's natural endpoint
and $\mu_\rho$ is already the corresponding input or output adapter of
\textup{(Plan-adapt)} inside $G_{\rm pr}^{\mathfrak P}$ or
$G_{\rm co}^{\mathfrak P}$.  Only the join itself inserts no
additional $\mathsf{Align}_\rho$.

\emph{Higher-order splice.}
Induct over the retained finite $\mathsf{Net}(\rho)$ of
\textup{(HSplice-cases)}.  Its semantic equalities are unchanged.  In
the \textup{(Y)} row use the unchanged-image
$\mathsf{Repart}$.  In \textup{(C-var-app)} use
\textup{(Plan-Par)} on the disjoint live operand, consumer, and yank
factors and then its selected-code view.  The
\textup{(C-var-whole)} row uses \textup{(Plan-Par)} on the operand and
consumer factors, followed by \textup{(Eta-repart)}.  The
\textup{(EndAct)} row is the existing $\mathsf{StrAct}$ equation.  In
\textup{(C-str)}, $\mathsf{StrPre}$ preserves the image and the
recursive splice is the simultaneous induction hypothesis at the
strictly prior classified instance $\rho'$ of
\textup{(Plan-order)}.  A terminal first-order instance uses the
first-order calculation above.  These are exactly the rows of
\textup{(HSplice-cases)}.

For \textup{(C-var-whole)}, write
$u^\epsilon:=\mathsf{ug}_{\rho,S}^\epsilon$ and let
$\ell_{\rm in}^\epsilon$ and $\ell_{\rm out}^\epsilon$ be the
restricted layouts into $K_\rho^\epsilon$ and $B_\rho^\epsilon$.
The names are relative to
$u^\epsilon:K_\rho^\epsilon\to B_\rho^\epsilon$; since
$\chi_\rho^\epsilon:B_\rho^\epsilon\to K_\rho^\epsilon$, this is
exactly the direction of \textup{(Net-repart-pullback)}.  Define
\[
 \widehat u^\epsilon
 :=
 (\ell_{\rm out}^\epsilon)^{-1}
 u^\epsilon\ell_{\rm in}^\epsilon .
\]
Since
$\widehat\omega_{\rho,S}^\epsilon
 =(\ell_{\rm in}^\epsilon)^{-1}
   \chi_\rho^\epsilon\ell_{\rm out}^\epsilon$,
Equations~\textup{(PC)} and \textup{(Eta-repart)} give
\[
 \widehat\omega_{\rho,S}^\epsilon\widehat u^\epsilon=I,
 \qquad
 \widehat u^\epsilon\widehat\omega_{\rho,S}^\epsilon=I .
 \tag{Plan-Eta-view}
\]
Thus $\widehat\omega_{\rho,S}^\epsilon$ is a surjective unitary
selected-code reparameterization and
\[
 \operatorname{im}
 \bigl(p^\epsilon\widehat\omega_{\rho,S}^\epsilon\bigr)
 =\operatorname{im}(p^\epsilon).
\tag{Plan-Eta-image}
\]
Put
\[
 \widehat W_\rho
 :=
 (\ell_{\rm in}^+)^{-1}W_\rho\ell_{\rm in}^-,
 \qquad
 V_{\rm eta}
 :=
 (\widehat\omega_{\rho,S}^+)^\dagger
 \widehat W_\rho
 \widehat\omega_{\rho,S}^- .
 \tag{Plan-Eta-action}
\]
If the preceding planned parallel program has equation
$\mathcal U(G)F^-=F^+\widehat W_\rho$, then unitarity of
$\widehat\omega_{\rho,S}^+$ gives
\[
 \mathcal U(G)(F^-\widehat\omega_{\rho,S}^-)
 =
 (F^+\widehat\omega_{\rho,S}^+)V_{\rm eta}.
 \tag{Plan-Eta-program}
\]
Moreover, \textup{(Net-repart-pullback)} and the existing
\textup{(C-var-whole)} semantic equality give
\[
\begin{aligned}
 \ell_{\rm out}^+V_{\rm eta}(\ell_{\rm out}^-)^{-1}
 &=
 (\chi_\rho^+)^\dagger W_\rho\chi_\rho^-\\
 &=
 \mathsf{PortCut}_{S}^{\rho}(\mathcal C,\mathcal D).
\end{aligned}
\tag{Plan-Eta-realize}
\]
This is the precise reason that \textup{(Net-Repart)} changes neither
the planned carrier nor the target program while still realizing the
repartitioned action.  It does not assert that
$\widehat\omega_{\rho,S}^\epsilon$ is a computational-basis
permutation.  The tensor and shared-block stages of
\textup{(C-var-whole)} are already contained in
$\omega_{\rho,S}$.  Thus every higher-order row either uses
\textup{(Plan-Par)}, retains the same image, or invokes a strictly
prior planned splice; no row allocates another carrier.

Ordinary map application is the planned use-block case followed by
the planned splice case.  Coherent sharing is the planned
coherent-sharing block followed by the same planned splice and the
unchanged-image distributor frame.  Their recursive calls are prior
in the stated order.  The retained chart expression contains no
further constructors beyond \textup{(HSplice-cases)}, completing
the recursion.

\emph{Acceptance witnesses.}
For the C--C instance
``$\mathsf{case}\ e\ \mathsf{of}\ x\mapsto g\,w
\mid y\mapsto g\,w$'', the occurrence family
$\Pi_{\mathcal B}$ has one parent occurrence for $g$ and one for $w$;
the two maps $\beta_i^\epsilon$ are their branch restrictions.
Consequently the direct recursion leaves one planned range for $g$
and one for $w$ in both payload presentations
$Z_{\mathcal B,i,\bullet}^\epsilon$.  The only branch distinction is
the already present case tag.  Equation~\textup{(Plan-Block-natural)}
therefore applies with no carrier bank, and any enclosing
\textup{(C-str)} presentation is applied afterward by
\textup{(Plan-Block-transport)}.  Thus shared occurrences are neither
duplicated nor relocated.

For the T2 instance
``$\mathsf{let}\ (h,s)=f\,a\ \mathsf{in}\
\mathsf{case}\ w\ \mathsf{of}\ x\mapsto h\,s
\mid y\mapsto h\,s$'',
Lemma~\ref{lem:classifier-totality} selects the
\textup{(C-var-whole)} shared-block row.  The occurrence maps assign
the two restrictions of $h$, and likewise of $s$, their one parent
range.  The unique case block is emitted at its natural tagged
incidence by \textup{(Plan-Block-natural)}.  Its traversal inside the
higher-order cut is already a row of $\widehat\omega_{\rho,S}$ in
\textup{(Eta-graft)} and is handled by
\textup{(Plan-Eta-program)} followed by endpoint transport; it does
not emit a second physical block.  Hence this instance adds neither a
carrier nor an additional block tag.

Finally take the abstract quantum switch with $A=\QBool$.  Its two
operation arguments each have carrier width
$\nwires(A\lmark A)=2$, and its control--payload map has width
\[
 \nwires((\QBool\tensor A)\lmark(\QBool\tensor A))=4.
\]
Thus its retained root face has width $2+2+4=8$.  Inspection of the displayed Source term in
\S\ref{subsec:quantum-switch} and its PPX carrier layout in
\S\ref{sec:impl-qswitch}, together with the complete higher-order
node inventory \textup{(HSplice-cases)}, gives exactly the following
contextual-face forms: a wire/yank, tensor, or structural face whose
active and carried-through factors together occupy at most the same
eight coordinates; the case face
$1_{\rm tag}+1_{\rm payload}+6_{\rm thru}=8$; or a cut face in which
one incident producer--consumer port is replaced by its single
midpoint while every other factor is carried through.  The cut face
has the same or smaller width than its incident face.  The Source
distributor preserves width, and this derivation contains no wider
sum reassociation.  Hence every contextual face has
width at most eight, while the root supplies the reverse inequality:
$Q_{\mathsf{QSwitch}}=8$.  The two branch restrictions of each of
$f$ and $g$ have their common parent address by $\beta_i^\epsilon$,
so these eight wires include no emitter-owned workspace.  This is a
carrier calculation; it does not assert the prototype's separately
optimized gate list or pending output permutation.

At the root, put
$p_{\mathcal R}^\epsilon:=J_{\mathcal R,C}^\epsilon$ and take
$L_{\mathcal R}^\epsilon$ to be the INV-4 extension satisfying
$L_{\mathcal R}^\epsilon p_{\mathcal R}^\epsilon
=J_{\mathcal R,L}^\epsilon$.  Both placements are isometric and both
frames are total unitary.  Apply \textup{(PlanEmit-IH)} to
$E_{\mathcal R}^\epsilon
 :=J_{\mathcal R,L}^\epsilon\widehat g_{\mathcal R}^\epsilon$.
Thus the root instance of \textup{(PlanEmit-IH)} is expressed on the
chart domain $\mathsf D_{\mathcal R}$, as are the plan coordinates in
\textup{(Plan-logical-incidence)}.  This incidence already contains
the canonical root chart pullback
$\widehat g_{\mathcal R}^\epsilon$.  Because the root has no enclosing
constructor, no additional plan view is contributed at its outer
endpoint; applying the induction hypothesis in that canonical chart
domain therefore gives
$V_{\mathcal R}^{\mathfrak P}=V_{\mathcal R}$.
Indeed, \textup{(Plan-logical-incidence)} gives directly
$J_{\mathcal R,L}^\epsilon\widehat g_{\mathcal R}^\epsilon
=\eta_{\mathcal R}^\epsilon$.
This does not assert that $\widehat g_{\mathcal R}^\epsilon$ is a
computational-word map.
This is exactly \textup{(PlanEmit-chart)} and supplies the six fields
of \textup{(PlanEmit)}, completing the proof.
\end{proof}

\begin{remark}[What the plan does not assume]
\label{rem:planned-emission-no-oracle}
The semantic charts and all direct-sum maps above live in
$\mathbf{FdHilb}$.  Their placements and unchanged-image views live
in the physical register $\mathcal H_{\mathcal R}$.  Only
$G_{\mathcal R}$ lives in
$\mathsf{TCirc}_{Q_{\mathcal R}}$.  Consequently the proof neither
identifies the semantic biproduct $\oplus$ with physical tensor
juxtaposition nor assumes that an arbitrary unitary chart
reparameterization can be synthesized.  The exact permutation
supplier is used only for computational-basis word changes, and
\textup{(BC-carrier)} remains an independent lower bound on the
planned width.  In particular, register readback uses the resulting
isometric physical placements; it does not require every
selected-coordinate view to be a computational-word permutation.
Because the reference sum layout is recursively binary,
$Q_{\mathcal R}$ is the maximum width of all retained contextual
faces and need not equal the root-interface width.  Any larger width
exposed by reassociation belongs to that prescribed semantic face; it
is not emitter-owned workspace.
\end{remark}

%% file: narymonoidal.tex

\section{Derived Forms for \texorpdfstring{$n$}{n}-ary Monoidal Sums}
\label{app:nary-plus}

This appendix gives the Raw auxiliaries used by the fixed expansion of
the $n$-ary Source case in Section~\ref{sec:datatypes}.  The forms
$\mathsf{RCase}$, $\Case_n$, and $\mathsf{factor}_n$ are Raw
abbreviations.  The form $\mathsf{TagCase}_n$ is the tag-preserving
Source construct; its Raw expansion is defined below.  Throughout,
$n\geq1$.  At Source instances every generated sum has first-order
summands and every generated map has first-order targets, so the
expansion is also internally typable by the binary rules of
Appendix~\ref{app:focused-rules}.  Every generated coherent node uses
phase~$1$; normalization may subsequently accumulate those binary
phase annotations.

The formal reference compiler respects this definitional expansion:
it assigns layouts and emits controls recursively at the binary core
nodes.  The executable prototype instead uses direct flat-tag lowering
of the derived $n$-ary syntax; this path is outside the compilation
soundness theorem.  The reference proof treats structural coherence as typed logical-frame
changes.  In the normalized prototype, additive associativity and
each distributor itself are gate-free, including unequal-width
instances; subsequent composition uses the splice alignment described
in Appendix~\ref{app:sum-encoding}.  A nontrivial permutation of
summand indices is emitted as the corresponding tag-register
permutation.

\subsection{\texorpdfstring{$n$}{n}-ary Sums}

We fix the left-associated expansion:
\[
  \bigplus_{i=0}^{0} A_i := A_0,
  \qquad
  \bigplus_{i=0}^{n} A_i
  :=
  \Big(\bigplus_{i=0}^{n-1} A_i\Big) \plus A_n
  \quad (n \ge 1).
\]
For example, $\bigplus_{i=0}^{2} A_i = (A_0 \plus A_1) \plus A_2$.
Left-association is a fixed canonical choice.  Any other association
is related by the explicit structural associator $\alpha^\plus$; the
metatheoretic results transfer along that isomorphism, and the
corresponding semantic sector maps agree under its certified transport.
In particular,
\[
  \Qn{n} := \bigplus_{i=0}^{n-1} \base.
\]

\subsection{Raw \texorpdfstring{$n$}{n}-ary Case}

Let $t : \bigplus_{i=0}^{n-1} A_i$ and suppose
$\Gamma, x_i : A_i \vdash_{\mathsf r} u_i : B_i$ for each $i < n$
(the same $\Gamma$ in every branch), with all $A_i$ and $B_i$
first-order (\S\ref{par:first-order-types}).  This hypothesis makes
each generated map target internally typable when the Raw abbreviation
is used in a Source expansion.  We define the Raw $n$-ary case
producing a result of type $\bigplus_{i=0}^{n-1} B_i$ by recursion on
$n$.

\begin{lemma}[Routed binary case]
\label{lem:routed-binary-case}
Suppose
\[
 \Gamma_0\vdash_{\mathsf r}t:A\plus B,\qquad
 \Delta,x{:}A\vdash_{\mathsf r}u:C,\qquad
 \Delta,y{:}B\vdash_{\mathsf r}v:D,
\]
where $\Gamma_0$ and $\Delta$ are disjoint and $A,B,C,D$ are
first-order.  There is an admissible derived term
\[
 \Gamma_0,\Delta\vdash_{\mathsf r}
 \mathsf{RCase}(t;\,x.u;\,y.v):C\plus D
\]
which consumes the shared context $\Delta$ once.
\end{lemma}
\begin{proof}
If $\Delta=\emptyset$, take
\[
 \mathsf{RCase}(t;\,x.u;\,y.v)
 :=\bigl((\lam{x}{u})\plus(\lam{y}{v})\bigr)\,t .
\]
Both branch functions are closed.  For nonempty
$\Delta=z_1{:}T_1,\ldots,z_k{:}T_k$, let $G$ be its fixed tensor
package (with $G=T_1$ when $k=1$), let
$\Delta\vdash_{\mathsf r}\gamma:G$ package its variables once, and let
$w[z/\Delta]$ denote the corresponding sequence of tensor lets that
unpacks $z:G$ in~$w$; for $k=1$ the package is the variable itself
and unpacking is capture-avoiding renaming.  Define the closed
branch functions
\[
\begin{split}
 \widehat f&=\lam{p}{\letpair{z}{x}{p}{u[z/\Delta]}}
      :G\tensor A\lmark C,\\
 \widehat g&=\lam{q}{\letpair{z}{y}{q}{v[z/\Delta]}}
      :G\tensor B\lmark D.
\end{split}
\]
Then
\[
 \mathsf{RCase}(t;\,x.u;\,y.v)
 :=(\widehat f\plus\widehat g)
      \bigl(\distL\,(\gamma\tensor t)\bigr).
\]
The term $\gamma\tensor t$ uses $\Delta$ and $\Gamma_0$ disjointly;
$\distL$ routes the one $G$ package into the two orthogonal summands;
and the closed Raw $\plus$-Map uses coefficient~$1$ in each branch.
Thus no free resource is copied, and the displayed conclusion follows
from the Raw tensor, application, distributivity, and $\plus$-Map
rules.  Since $C,D$ are first-order, the same expansion is internally
typable.
\end{proof}

\begin{definition}[$n$-ary case]
\label{def:case-nary}
\[
\begin{array}{rcl}
\Case_1(t;\, x_0.u_0)
&:=& (\lam{x_0}{u_0})\,t
\\[1ex]
\Case_{n+1}(t;\, x_0.u_0, \ldots, x_n.u_n)
&:=&
\mathsf{RCase}\bigl(
 t;\,x.\Case_n(x;\,x_0.u_0,\ldots,x_{n-1}.u_{n-1});
 \,y.u_n[y/x_n]\bigr)
\end{array}
\]
where in the recursive clause, $t$ is viewed at type
$(\bigplus_{i=0}^{n-1} A_i) \plus A_n$, and
$\mathsf{RCase}$ is the routed shared-context form of
Lemma~\ref{lem:routed-binary-case}.
\end{definition}

\noindent
The base case $\Case_1$ is ordinary application: when there is only
one summand, no branching is needed.  Its internal beta-contractum is
$u_0[t/x_0]$, justified by unrestricted internal substitution; the
Raw abbreviation itself does not assume that substitution preserves
the formation-value restriction on Raw sum introduction.  We write
\[
  \caseofn(t;\; x_0.u_0,\; \ldots,\; x_{n-1}.u_{n-1})
\]
as notation for $\Case_n(t;\, x_0.u_0, \ldots, x_{n-1}.u_{n-1})$.

\paragraph{Admissibility.}
Let \(n\geq1\), let every \(A_i\) and \(B_i\) be first-order, and
suppose
\[
  \Gamma_1 \vdash_{\mathsf r} t : \bigplus_{i=0}^{n-1} A_i,
  \qquad
  \Gamma_2, x_i : A_i \vdash_{\mathsf r} u_i : B_i
  \quad (0\leq i<n),
\]
where \(\Gamma_1\) and \(\Gamma_2\) are disjoint; the shared context
\(\Gamma_2\) may be empty.  Then
\[
  \Gamma_1,\Gamma_2 \vdash_{\mathsf r}
  \caseofn(t;\; x_0.u_0,\; \ldots,\; x_{n-1}.u_{n-1})
  : \bigplus_{i=0}^{n-1} B_i .
\]
Induct on \(n\).  In the base case, apply introduction and
elimination for \(\lmark\); its beta-contractum is handled only in
the internal calculus.  At the \(n+1\) step, the
induction hypothesis gives
\[
 \Gamma_2,x:\bigplus_{i<n}A_i
 \vdash_{\mathsf r}
 \Case_n(x;\,x_0.u_0,\ldots,x_{n-1}.u_{n-1})
 :\bigplus_{i<n}B_i,
\]
while renaming gives
\(\Gamma_2,y:A_n\vdash_{\mathsf r}u_n[y/x_n]:B_n\).
Lemma~\ref{lem:routed-binary-case} applies to exactly these two branch
derivations and the exposed outer sum.  Its target types
\(\bigplus_{i<n}B_i\) and \(B_n\) are first-order.  The fixed
left-associated source and target types already have the required
shape, so no implicit reassociation is used; any alternate
association would require an explicit \(\alpha^\plus\) term.

\subsection{\texorpdfstring{$n$}{n}-ary Factoring}
\label{app:nary-factoring}

The binary inverse distributivity
$\distRi : (A \tensor C) \plus (B \tensor C) \lmark (A \plus B) \tensor C$
generalizes to an $n$-ary factoring isomorphism.

\begin{definition}[$n$-ary factoring]
\label{def:factor-nary}
The structural isomorphism
\[
  \mathsf{factor}_n
  : \bigplus_{i=0}^{n-1} (A_i \tensor C)
  \;\lmark\;
  \Big(\bigplus_{i=0}^{n-1} A_i\Big) \tensor C
\]
is defined by induction:
$\mathsf{factor}_1 := \mathrm{id}_{A_0 \tensor C}$
and
$\mathsf{factor}_{n+1} := \distRi \circ (\mathsf{factor}_n \plus \mathrm{id}_{A_n \tensor C})$,
for first-order $A_i$ and $C$.  The first-order hypothesis makes every
generated map target first-order, so the Raw expansion is internally
typable.
\end{definition}

\noindent
Since each $\distRi$ is structural, so is $\mathsf{factor}_n$.
The reference proof treats it as a typed logical-frame change.  In
the normalized prototype every factoring distributor itself is
gate-free, including instances with unequal-width constituents;
subsequent composition uses splice alignment
(Appendix~\ref{app:sum-encoding}).

\subsection{\texorpdfstring{$n$}{n}-ary Tagged Case}
\label{app:tagged-case}

When all branches return the same type~$C$, the general $n$-ary case can
be composed with tag-pairing and factoring to yield a tag-preserving form.

\begin{definition}[Tagged case]
\label{def:tagged-case}
Given $\Delta \sjudge t : \bigplus_{i=0}^{n-1} A_i$ and branches
$\Gamma \sjudge u_i : C$ for each $i < n$, with $\Delta$ and
$\Gamma$ disjoint (same $\Gamma$,
same~$C$; all $A_i$ and $C$ first-order; the binder $x_i$ is a
tag/payload binder, paired with the
branch result, and is \emph{not} a resource of $u_i$ ---
cf.\ the routing discussion of the coherent case,
\S\ref{cohRouting}), the derived Source judgment is
\[
 \Delta,\Gamma\sjudge
 \mathsf{TagCase}_n(t;u_0,\ldots,u_{n-1})
 :\Bigl(\bigplus_{i=0}^{n-1}A_i\Bigr)\tensor C .
\]
Its Raw expansion is
\[
\begin{aligned}
\bigl(\mathsf{TagCase}_n(t;\,u_0,\ldots,u_{n-1})\bigr)^\circ
  &:={}
  \mathsf{factor}_n\!\Bigl(\\[-.4ex]
  &\quad \Case_n(t^\circ;\,x_0.(x_0\tensor u_0^\circ),\,\ldots,\\[-.4ex]
  &\qquad x_{n-1}.(x_{n-1}\tensor u_{n-1}^\circ))\Bigr).
\end{aligned}
\]
\end{definition}

\noindent
The inner $\Case_n$ pairs each branch result with its tag variable,
producing $\bigplus_i (A_i \tensor C)$; factoring assembles the sum
into $(\bigplus_i A_i) \tensor C$.  This is the $n$-ary generalization
of the binary tag-preserving case (\S\ref{cohRouting}).  The Raw
typing result above proves the displayed Source rule admissible.  Its
expansion has closed branch maps, no coherent-sum former, and only
first-order inverse distributors, so it satisfies the Source
discipline of Lemma~\ref{lem:source-expansion-discipline}.

\subsection{General Raw Case (Split Contexts)}
\label{app:general-case}

The Source case requires shared context and a uniform result type.
The Raw calculus also admits the split-context form
\[
  \mathsf{SplitCase}(t;\,x.u;\,y.v)
  \;\;\triangleq\;\;
  \bigl((\lam{x}{u}) \plus (\lam{y}{v})\bigr)\; t
  \;:\; C \plus D.
\]
Here $\Gamma_1 \vdash_{\mathsf r} \lam{x}{u} : A \lmark C$ and
$\Gamma_2 \vdash_{\mathsf r} \lam{y}{v} : B \lmark D$
with $\Gamma_1$ and $\Gamma_2$ disjoint.  Raw typing imposes no
restriction on $A,B,C,D$; for this abbreviation to inhabit the
internal judgment, $C$ and $D$ must be first-order.

\paragraph{Example: guarded pipeline.}
Given an error type~$X$, a resource type~$R$, and a success
type~$S$, consider a sum whose left branch already carries an error
payload and whose right branch supplies a continuation together
with its argument:
\[
  \mathsf{pipelineK}
  \;:=\;
  \lam{q}{
    \mathsf{SplitCase}\bigl(q;\,e.e;\,
      w.\letpair{k}{r}{w}{k\,r}\bigr)
  }
  \;:\;
  X \plus ((R \lmark S) \tensor R) \lmark X \plus S.
\]
The continuation $k : R \lmark S$ appears only in the right branch,
where the right summand supplies both $k$ and its argument~$r$.  The
left branch returns the error payload~$e$ directly.  The result type
is the sum $X \plus S$: the surrounding $\plus$-Map supplies the
summand structure, so no term-level injections are introduced.
Taking $X$ and $S$ first-order makes this Raw term internally typable;
$R$ may be arbitrary.  It is not a Source term because its input sum
contains
$(R\lmark S)\tensor R$.

\subsection{Derivation from Core Primitives}
\label{app:case-derivation}

The Raw forms $\mathsf{RCase}$, $\Case_n$, and
$\mathsf{factor}_n$ give the fixed expansion of the Source
tag-preserving form $\mathsf{TagCase}_n$.  Their Raw typing uses the
six rules of Table~\ref{tab:typing-rules-linear}.  At the first-order
instances used by $\mathsf{TagCase}_n$, the same expansion is
internally typable: every generated map has first-order targets.
The empty-shared-context clause of
Lemma~\ref{lem:routed-binary-case} avoids assuming a tensor unit.
The $n$-ary forms add nothing further: $\Case_n$
(Definition~\ref{def:case-nary}) iterates
Lemma~\ref{lem:routed-binary-case} along the fixed left-associated
expansion, and $\mathsf{TagCase}_n$
(Definition~\ref{def:tagged-case}) composes $\Case_n$ with the
$\distRi$-composite $\mathsf{factor}_n$
(Definition~\ref{def:factor-nary}).  The admissibility statements
and proofs are those given above; no additional primitive is used.

%% file: elaboration-proofs-appendix.tex
\section{Finite-Datatype Elaboration and Its Properties}
\label{app:elaboration-proofs}

This appendix gives the surface typing rules and the canonical
elaboration of the finite control datatypes of \S\ref{sec:datatypes},
and proves that the canonical translation of
Definition~\ref{def:canonical-elaboration} is a typed function into
the Source language.

\subsection{Surface syntax and typing}
\label{app:surface-syntax-typing}

Resolved surface staging identifiers are typed from the elaboration
environment $\Sigma$:
\begin{mathpar}
\inferrule*[right=\textsc{S-Stage}]
  {f:S\in\Sigma\quad f\text{ fully resolved}}
  {\emptyset\vdash_{\mathsf{surf}}f:S}
\end{mathpar}

Surface first-order data types and Source types are generated by
\[
 P ::= \base \mid \NAME \mid P\tensor P \mid P\plus P,
 \qquad
 S ::= P \mid S\tensor S \mid S\lmark S .
\]
Thus every surface sum belongs to a first-order data type, and the
translations of the $P$-types are exactly the first-order core types
(Lemma~\ref{lem:elab-first-order}).
A datatype $\NAME$ with $n$ labels elaborates to the canonical
type $\Qn{n} = \bigplus_{i=0}^{n-1} \base$, the canonical $n$-ary sum serving
as the qdit generalization of $\QBool = \Qn{2}$.
The translation $\ElabTy{-}$ is otherwise homomorphic on core type
constructors ($\base$, $\tensor$, $\plus$, $\lmark$).

The key construct is coherent case analysis, the $n$-ary
generalization of binary tag-preserving case (\S\ref{cohRouting}).
It preserves the selected label while coherently routing a shared
linear context.  Every branch uses the same context $\Gamma_2$ and
returns the same first-order type $C$; the compiled branch bodies
therefore act on the same live interface.  The result type is
$\NAME\tensor C$.

\begin{mathpar}
\inferrule*[right=\textsc{S-Case}]
{
  \Gamma_1 \vdash_{\mathsf{surf}} e : \NAME
  \quad \quad
  \Gamma_2 \vdash_{\mathsf{surf}} u_i : C
  \;\; (i < n)
  \quad  \quad
  \Gamma = \Gamma_1 \uplus \Gamma_2
  \quad \quad
  C \text{ first-order}
}{
  \Gamma \vdash_{\mathsf{surf}}
  \mathbf{case}\ e\ \mathbf{of}\ \{l_i \mapsto u_i\}_{i<n}
  : \NAME \tensor C
}
\end{mathpar}

\noindent
Coherent case desugars into core constructs via
distributivity and factoring, generalizing the binary $\QBool$ pattern;
the $\Qn{n}$-tagged scrutinee is a qdit generalizing the qubit case,
and each branch is guarded by the corresponding tag value
(\S\ref{sec:impl-case}).

\subsection{Well-formed environments}
\label{app:wf-environments}

\begin{definition}[Well-formed \(\Sigma\)]
\label{def:wf-sigma}
An elaboration environment \(\Sigma\) is \emph{well formed} with
respect to \(\mathcal D\) when:
\begin{enumerate}
\item if \(\mathcal D(\NAME)=(l_0,\ldots,l_{n-1})\), then
  \(\Sigma(\NAME)=n\);
\item for every fully resolved surface staging binding \(f:S\),
  \[
    \Sigma(f)=t_f,\qquad
    \emptyset\sjudge t_f:\ElabTy S,
  \]
  where \(t_f\) is a closed Source term.
\end{enumerate}
\end{definition}

Quantum exponentials in an implementation \(t_f\) have the prescribed
unitaries determined by their certified involutions.  Circuit
realization additionally assumes backend correctness \textup{(BC)}
for emitted primitives.  The typing arguments below use only the
surface and core typing rules and the declared certificates.

In the PPX interface, \texttt{[@@source.datatype]} creates the datatype
entry and its type witnesses.  The generated \texttt{permute} and
\texttt{select} combinators construct sealed closed Source terms,
which supply entries of $\Sigma$.  A permutation list must be a
bijection on the declared labels.  Given closed operations
$\emptyset\sjudge f_i:A\lmark A$ for $0\leq i<n$, with $A$
first-order, the staged call constructs
\[
 \mathsf{select}_{\NAME,A}[f_0,\ldots,f_{n-1}]:
 \NAME\tensor A\lmark\NAME\tensor A .
\]
The bracketed operations are fixed staging arguments, not linear
inputs of the resulting Source program.

\subsection{Canonical elaboration}
\label{app:canonical-elaboration}

Fix an elaboration environment $\Sigma$, well formed with respect to
$\mathcal D$, that assigns each resolved surface staging identifier a
closed Source term (Definition~\ref{def:wf-sigma}).
Elaboration operates on the typechecked surface syntax, whose case
nodes record the resolved datatype and ordered shared context.
Fix the declaration order of labels, the left-associated expansions
of Appendix~\ref{app:nary-plus}, a left-to-right order and bracketing
for every nonempty linear context, and a deterministic
capture-avoiding fresh-name convention.

\begin{definition}[Canonical elaboration]
\label{def:canonical-elaboration}
The recursive function $\mathsf{Elab}_\Sigma$ is homomorphic on
inherited Source constructors and certified-involution
syntax.  Its only additional term clauses are
\[
\begin{aligned}
 \mathsf{Elab}_\Sigma(f)
   &=\Sigma(f),\\
 \mathsf{Elab}_\Sigma
 \bigl(\mathbf{case}\ e\ \mathbf{of}\ \{l_i\mapsto u_i\}_{i<n}\bigr)
   &=
 \mathsf{TagCase}_n\bigl(
   \mathsf{Elab}_\Sigma(e);
   \mathsf{Elab}_\Sigma(u_0),\ldots,
   \mathsf{Elab}_\Sigma(u_{n-1})\bigr).
\end{aligned}
\]
The first clause is direct substitution; no staging identifier survives
in the resulting Source term.  Here $\mathsf{TagCase}_n$ is the fixed
Source construct of Appendix~\ref{app:nary-plus}, where its Raw
expansion is defined.
A datatype case has exactly one clause for each declared label and no
extra clause; $\mathsf{Elab}_\Sigma$ puts them in declaration order
before applying the case equation above.  The displayed elaboration judgment is the graph
of this function:
\[
 \Sigma;\ElabTy{\Gamma}\vdash e\Rightarrow t:\ElabTy S
 \quad\Longleftrightarrow\quad
 t=\mathsf{Elab}_\Sigma(e)
\]
for a well-typed term $\Gamma\vdash_{\mathsf{surf}}e:S$ in this
elaboration domain.
\end{definition}

For printed syntax, equality is read up to the fixed
capture-avoiding names.

Coherent case elaborates to the $n$-ary tagged case form
$\mathsf{TagCase}_n$ (Appendix~\ref{app:nary-plus}), whose Raw
expansion wraps the general $n$-ary case with tag-pairing and factoring;
the typed clause is the case clause of
Definition~\ref{def:canonical-elaboration}, at result type
$\Qn n\tensor\ElabTy C$ with $\ElabTy C$ first-order.
A resolved surface staging node elaborates by direct substitution of
the closed Source term $\Sigma(f)$.

Throughout the rest of this appendix, $\Gamma\vdash_{\mathsf{surf}}e:S$
ranges over the elaboration domain fixed in
Definition~\ref{def:canonical-elaboration}; any derived notation has
already been expanded.

\subsection{Auxiliary lemmas}
\label{app:elab-aux-lemmas}

\begin{lemma}[First-order translation]
\label{lem:elab-first-order}
For every surface type \(S\), \(S\) is first-order if and only if
\(\ElabTy S\) is a first-order core type.
\end{lemma}
\begin{proof}
Structural induction on \(S\).  The cases \(\base\), \(\tensor\), and
\(\plus\) are homomorphic.  A datatype name with \(n\geq1\) translates
to the left-associated sum \(\Qn n\) of copies of \(\base\), hence is
first-order.  Conversely, \(\ElabTy{-}\) preserves every occurrence
of \(\lmark\), so a non-first-order surface type cannot translate to a
first-order core type.
\end{proof}

\begin{lemma}[Linear support and canonical annotations]
\label{lem:elab-linear-support}
In a surface typing derivation, the domain of each premise context is
exactly the linear free-variable support of its premise term.  Hence
the context split at a multiplicative rule is determined by the
subterms, up to exchange.  With the fixed context order of
Definition~\ref{def:canonical-elaboration}, the split is unique.
\end{lemma}
\begin{proof}
Induction on the surface typing derivation.  Variables contribute
their singleton support; closed atoms and resolved surface staging
nodes contribute none.  Multiplicative rules take disjoint unions.
A coherent case has the disjoint support of its scrutinee together
with the one shared branch support; linearity requires every branch
to have exactly that same support.  The remaining inherited rules
preserve support homomorphically.  Fixing context order removes the
residual exchange choice.
\end{proof}

\subsection{Typing preservation and functionality}
\label{app:elab-preservation}

\begin{theorem}[Elaboration soundness]
\label{thm:elab-sound-full}
Assume \(\Sigma\) is well formed.
\begin{enumerate}
\item If \(\Gamma\vdash_{\mathsf{surf}}e:S\), then
  \[
    \Sigma;\ElabTy\Gamma\vdash
      e\Rightarrow\mathsf{Elab}_\Sigma(e):\ElabTy S
    \qquad\text{and}\qquad
    \ElabTy\Gamma\sjudge
      \mathsf{Elab}_\Sigma(e):\ElabTy S.
  \]
\item The homomorphic translation of an inherited surface
  certified-involution derivation is a core
  certified-involution derivation at the translated type.
\end{enumerate}
\end{theorem}

\begin{proof}
We prove both assertions simultaneously by induction on the
surface typing derivation and the
certified-involution derivation.  The graph judgment in item~1 is
immediate once the recursive clause is fixed; the content is the Source
typing assertion.

\paragraph{Inherited multiplicative rules.}
Variables translate to themselves.  Tensor introduction and
elimination, implication introduction, and application translate
homomorphically.  The induction hypotheses supply the translated
premises, Lemma~\ref{lem:elab-linear-support} supplies the same
disjoint context split, and the corresponding core rule gives the
conclusion.

\paragraph{Atoms and exponentiation.}
Primitive structural and quantum atoms translate their type indices
and retain their declared certificates.
Lemma~\ref{lem:elab-first-order} preserves the required first-order
conditions and Source type indices.  The certified-involution induction
is homomorphic for identity, sign, symmetry, tensor, sum,
and the other displayed certificate constructors.  Thus the
\textsc{Exp} case receives the translated core certificate required
by its rule.

\paragraph{Binary tag-preserving case.}
Let the translated premises be
\[
 \ElabTy{\Gamma_1}\sjudge t:\ElabTy A\plus\ElabTy B,
 \qquad
 \ElabTy{\Gamma_2}\sjudge t_1:\ElabTy C,
 \qquad
 \ElabTy{\Gamma_2}\sjudge t_2:\ElabTy C.
\]
The types \(\ElabTy A,\ElabTy B,\ElabTy C\) are first-order by
Lemma~\ref{lem:elab-first-order}.  The Source case rule therefore
derives
\[
 \ElabTy{\Gamma_1},\ElabTy{\Gamma_2}\sjudge
 \caseof{t}{x}{t_1}{y}{t_2}:
 (\ElabTy A\plus\ElabTy B)\tensor\ElabTy C .
\]

\paragraph{Datatype case.}
For
\[
 \mathbf{case}\ e\ \mathbf{of}\ \{l_i\mapsto u_i\}_{i<n},
\]
the typechecked node fixes \(n\geq1\), its resolved datatype, the
declaration order \(l_0,\ldots,l_{n-1}\), and the possibly empty shared
context.  The induction hypotheses give
\[
 \ElabTy{\Gamma_1}\sjudge
   \mathsf{Elab}_\Sigma(e):\Qn n,
 \qquad
 \ElabTy{\Gamma_2}\sjudge
   \mathsf{Elab}_\Sigma(u_i):\ElabTy C.
\]
Since \(\ElabTy C\) is first-order, the Source rule of
Definition~\ref{def:tagged-case} directly gives
\[
 \ElabTy{\Gamma_1},\ElabTy{\Gamma_2}\sjudge
 \mathsf{TagCase}_n\bigl(
   \mathsf{Elab}_\Sigma(e);
   \mathsf{Elab}_\Sigma(u_0),\ldots,
   \mathsf{Elab}_\Sigma(u_{n-1})\bigr)
 :\Qn n\tensor\ElabTy C .
\]

\paragraph{Resolved surface staging node.}
If \(f:S\) is present, canonical resolution selects its unique
staging entry.  Well-formedness gives
\[
 \mathsf{Elab}_\Sigma(f)=t_f,
 \qquad
 \emptyset\sjudge t_f:\ElabTy S.
\]

The inherited primitive cases, datatype case, and surface staging node
are precisely the constructors in the elaboration domain.  Hence the
preceding cases exhaust the surface typing and certified-involution
constructions.
\end{proof}

\begin{corollary}[Functionality of elaboration]
\label{thm:elab-unique-full}
Assume \(\Sigma\) is well formed.  If
\[
 \Sigma;\ElabTy\Gamma\vdash e\Rightarrow t:\ElabTy S
 \qquad\text{and}\qquad
 \Sigma;\ElabTy\Gamma\vdash e\Rightarrow t':\ElabTy S,
\]
then \(t=t'\) under the fixed syntax convention.  If generated binder
names are not printed canonically, then \(t={}_\alpha t'\).
\end{corollary}

\begin{proof}
Definition~\ref{def:canonical-elaboration} defines the elaboration
judgment as the graph of \(\mathsf{Elab}_\Sigma\), so
\(t=\mathsf{Elab}_\Sigma(e)=t'\).  The second clause is the stated
reading of printed syntax up to the fixed capture-avoiding names.
\end{proof}

The fixed left association therefore selects a single elaboration;
other associations are related by the explicit structural atom
\(\alpha^\plus\).

\begin{corollary}[Transfer to surface programs]
\label{cor:surface-transfer}
Assume \(\Sigma\) is well formed and
\(\Gamma\vdash_{\mathsf{surf}}e:S\).
Theorem~\ref{thm:elab-sound-full} and
Lemma~\ref{lem:source-internal-inclusion} give an internal derivation
\[
 \mathcal D_e:
 \ElabTy\Gamma\intjudge
 \bigl(\mathsf{Elab}_\Sigma(e)\bigr)^\circ:\ElabTy S.
\]
Put
\[
 \mathcal N_e:=
 \mathsf{CanDer}_{\ElabTy\Gamma,\ElabTy S}
 \bigl(
  \mathsf{NF}_{\mathrm{LO}}\bigl(
    (\mathsf{Elab}_\Sigma(e))^\circ
  \bigr)
 \bigr)
\]
and define the derivation-indexed surface meaning by
\[
 \llbracket e\rrbracket_{\mathsf{surf}}
   :=\SEM{\mathcal N_e}_{\mathrm{NF}}.
\]
This definition is single-valued by functionality of elaboration,
functionality of \(\mathsf{NF}_{\mathrm{LO}}\), and canonical
retyping.  Core normalization, determinacy, and boundary unitarity
apply to \(\mathcal D_e\); compiler correctness additionally assumes
\textup{(BC)}.
\end{corollary}